\documentclass[12pt, twoside, reqno]{amsart}
\usepackage[thm_section]{macros}
\newcommand{\bGamma}{\bm{\Gamma}}

\begin{document}
\begin{titlepage}
\begin{spacing}{1}
\title{\Large{\textbf{Semiparametric Estimation of \\
\vspace{-0.3cm} Long-Term Treatment Effects}$^*$}}
\author{
\begin{tabular}[t]{c@{\extracolsep{2em}}c} 
\large{Jiafeng Chen} &  \large{David M. Ritzwoller}\\ \vspace{-0.7em}
\normalsize{\it Harvard Business School} & \normalsize{\it Stanford Graduate School of Business} \\ \vspace{-0.7em}
\normalsize{\href{mailto:jiafengchen@g.harvard.edu}{jiafengchen@g.harvard.edu}} & \normalsize{\href{mailto: ritzwoll@stanford.edu}{ritzwoll@stanford.edu}}
\end{tabular}%
\\
}
\date{%
\today\\ $^*$We thank Isaiah Andrews, Xiaohong Chen, Raj Chetty, Pascaline Dupas, Bryan
 Graham, Han Hong, Guido Imbens, Melanie Morten, Evan Munro, Hyunseung Kang, Brad Ross, Joseph Romano, Suproteem
 Sarkar, Rahul Singh, Elie Tamer, Stefan Wager, and Chris Walker for helpful comments and
 conversations. Ritzwoller gratefully acknowledges support from the National Science
 Foundation under the Graduate Research Fellowship.}
                      
\begin{abstract}
\smalltonormalsize{Long-term outcomes of experimental evaluations are necessarily
 observed after long delays. We develop semiparametric methods for combining the
 short-term outcomes of experiments with observational measurements of
 short-term and long-term outcomes, in order to estimate long-term
 treatment effects. We characterize semiparametric efficiency bounds
 for various instances of this problem. These
 calculations facilitate the construction of several estimators.
We analyze the finite-sample performance of these estimators with a simulation
 calibrated to data from an evaluation of the long-term effects of a poverty alleviation 
 program.}\\
\\
\textbf{Keywords:} Long-Term Treatment Effects, Semiparametric Efficiency, Experimentation
\\
\textbf{JEL Codes:} C01, C13, C14, O10  
\end{abstract}
\end{spacing}
\end{titlepage}
\maketitle
\thispagestyle{empty}
\setcounter{page}{1}
\begin{spacing}{1.4}
\section{Introduction}
Empirical researchers often aim to estimate the long-term effects of policies or
interventions. Randomized experimentation provides a simple and interpretable approach to
this problem \citep*{fisher1925statistical, duflo_glennerester_2007, athey_imbens_2017}.
However, long-term outcomes of experimental evaluations are necessarily observed after
long and potentially costly delays. Consequently, there is relatively limited experimental
evidence on the long-term effects of economic and social
policies.\footnote{\cite*{bouguen_etal_2019} provide a systematic review of randomized
control trials in development economics and report that only a small proportion evaluate
long-term effects. Similarly, in a review on experimental evaluations of the effects of
early childhood educational interventions, \cite* {tanner_etal_2015} identify only one
randomized evaluation that reports long-term employment and labor market effects
\citep{gertler_etal_2014}.}

In this paper, we develop methods for estimating long-term treatment effects by
combining short-term experimental and long-term observational data sets. We consider two closely related models, proposed by \cite{athey2020combining} and \cite{athey2020estimating}. In both cases, a researcher
 is interested in the average effect of a binary treatment on a scalar long-term
 outcome. The researcher observes two samples of data, an
\emph{experimental} sample and an \emph{observational} sample. The experimental sample
 measures short-term outcomes of a randomized evaluation of the treatment. %
 the treatment The observational sample measures both short-term and long-term outcomes,
 but may be subject to unmeasured confounding. The two settings that we consider are
 distinguished by whether treatment is observed in the observational sample. Similar
 identifying assumptions are required in each case. We review these assumptions in \cref
 {sec:formulation}.
 
Methods for combining experimental and observational data to estimate long-term treatment effects are widely applicable in applied microeconomics and online platform experimentation. 
The estimators proposed in \cite{athey2020estimating} have been used to estimate the long-term
effects of free tuition on college completion \citep*{dynarski2021closing}, of an agricultural monitoring technology on farmer revenue  in Paraguay \citep*{dal2021information}, and of changes to Twitter's platform on user engagement \citep{twitter_engineering_2021}. As a result, there is considerable interest in advancing a statistical and methodological foundation for this problem \citep{gupta2019top}.

To that end, this paper offers two contributions. First, we develop semiparametric theory for estimation of long-term average treatment effects in \cref{sec: efficiency}. In particular, we derive the semiparametric efficiency bound, and
the corresponding efficient influence function, for estimating long-term average
treatment effects in each of the models that we consider. We then demonstrate that the
efficient influence function is the unique influence function in each model, indicating
that all regular and asymptotically linear estimators have the same asymptotic variance
and achieve the semiparametric efficiency bound 
\citep{chen2018overidentification, newey1994asymptotic}. In both cases, we find that the
efficient influence functions possess a ``double-robust'' structure commonly found in
causal estimation problems \citep{scharfstein1999adjusting, kang2007demystifying}.

These results are novel. In particular, our calculations correct statements concerning efficient influence functions and semiparametric efficiency bounds given in a working paper draft of \cite{athey2020estimating}.\footnote{The results given in \cite{athey2020estimating} were obtained
through a non-rigorous calculation related to standard heuristics involving a
discretization of the sample space (see e.g., \cite{kennedy2022semiparametric} and 
\cite{ichimura2022influence} for discussion). Our arguments are rigorous, following the
method developed in Section 3.4 of \cite{bickel1993efficient}.} 
Analogous results for the model considered in  \cite{athey2020combining} have not appeared before in the literature.
  
Second, in \cref{sec:estimation}, we establish the consistency and asymptotic normality of a suite of estimators. These estimators differ according to whether they are based on moment conditions 
associated with the efficient influence functions derived in \cref{sec: efficiency}. 
Moment conditions defined by efficient influence functions are often referred to as Neyman orthogonal moment conditions, due to their insensitivity to local perturbations of nuisance parameters \citep[see e.g.,][]{chernozhukov2018double,foster2019orthogonal}. The estimators that we formulate can be viewed as instances of one-step estimators \citep{le1956asymptotic, bickel1982adaptive, newey1994asymptotic}. Roughly, the conditional 
expectations that compose an efficient influence function are estimated 
by augmenting standard machine learning algorithms with cross-fitting. 
Estimates of long-term average treatment effects are then formed
 by treating estimated efficient influence functions as identifying moment functions.  Estimators obtained from Neyman orthogonal moment conditions admit a very general analysis, under high-level sufficient conditions, due to \cite{chernozhukov2018double}. We adapt these arguments to our setting.

We compare these estimators to a variety of alternative estimators based on non-orthogonal moment conditions. In particular, we consider a set of estimators that are analogous to standard inverse propensity score weighting and outcome regression estimators for average treatment effects under unconfoundedness (see e.g., \citealt{imbens2004nonparametric} for a review). These estimators include, but are not limited to, many of the estimators proposed in \cite{athey2020combining} and \cite{athey2020estimating}, and applied by e.g., \cite{dynarski2021closing}. Here, to obtain theoretical guarantees, we restrict attention to estimators that plug-in nuisance parameter estimates derived from the method of sieves \citep{chen2015sieve,chen2014sieve}. Relative to those for estimators based on orthogonal moments, sufficient conditions for estimators based on non-orthogonal moments appear more stringent.

\Cref{sec:simulation} assesses the finite sample performance of these semiparametric estimators with a simulation calibrated to data from a randomized evaluation of the long-term effects of a poverty alleviation program originally analyzed in  \cite{banerjee2015multifaceted}. We find that estimators based on orthogonal moments are substantially more accurate than estimators based on non-orthogonal moments. \cref{sec: conclusion} concludes. 

Proofs for all results stated in the main text are
provided in \cref{sec: proofs}. \cref{sec: G=0,sec: additional,sec:appendix simulation} give additional results or details and will be introduced at appropriate points throughout the paper. Code implementing the estimators developed in this paper is available on GitHub at the link \href{https://github.com/DavidRitzwoller/longterm}{https://github.com/DavidRitzwoller/longterm}.

\subsection{Related literature\label{sec:review}} Settings closely related to our own include 
\cite*{rosenman2018propensity},
\cite*{rosenman2020combining}, and \cite*{kallus2020role}. In
\cite{rosenman2018propensity}, treatment assignment is unconfounded in both
samples. In \cite{rosenman2020combining}, the outcome of interest is observed in
both samples. In \cite{kallus2020role}, treatment is unconfounded in both samples considered
jointly, but not necessarily in either sample considered individually. 
\cite{yang2020targeting},  \cite{gui2020combining}, and \cite{gechter2021combining} give
complementary analyses of related models. \cite{hou2021efficient} consider a model 
in which measurements of long-term outcomes and treatments are missing completely at 
random. \cite{imbens2022longterm} propose several alternative identification strategies for
models similar to the model that we consider. In contemporaneous
work, \cite{singh2021finite} and \cite{singh2022generalized} propose estimators
and confidence intervals similar to those developed in \cref{sec:estimation} for the case in which
treatment is not observed in the observational data set. 
 
This paper contributes to the literature on missing data models
\citep*{little2019statistical, ridder2007econometrics, hotz2005predicting}, and
more specifically to the literature on semiparametric efficiency in missing data
models \citep*{chen2008semiparametric, graham2011efficiency, muris2020efficient, bia2020double}. The models that we
consider are related to the literatures on  statistical surrogacy
\citep*{prentice1989surrogate, begg2000use} and mediation analysis
\citep*{van2004estimation, imai2010general}. 

\subsection{Notation}
Let $\lambda$ be a $\sigma$-finite measure on the measurable space
$\left(\Omega,\mathcal{F}\right)$, and let $\mathcal{M}_\lambda$ be the set of
all probability measures on $\left(\Omega,\mathcal{F}\right)$ that are
absolutely continuous with respect to $\lambda$. For an arbitrary random variable $D$
defined on $\left(\Omega,\mathcal{F}, Q\right)$ with $Q\in
\mathcal{M}_\lambda$, we let $\E_{Q}[D]$ denote its expected value. 
 The quantity $p(d \mid E)$ will denote the
density of the random variable $D$ at $d$ conditional on the event $E\in\Omega$ and
$l(d \mid E)$ will denote the corresponding log-likelihood. This
notation leaves the law of the random variable $D$ implicit, but will cause
no ambiguity. We let $\| \cdot \|_{Q,q}$ denote the $L^q(Q)$ norm and $[n]$ denote the set $\{1,\ldots,n\}$.

\section{Problem Formulation\label{sec:formulation}}

Consider a researcher who conducts a randomized
experiment aimed at assessing the effects of a policy
or intervention. For each individual in the experiment, they measure a $q$-vector $X_i$
of pre-treatment covariates, a binary variable $W_i$ denoting assignment to treatment,
and a $d$-vector $S_i$ of short-term post-treatment outcomes. The researcher is
interested in the effect of the treatment on a scalar, long-term, post-treatment
outcome $Y_i$ that is not measured in their experimental data. They are able to obtain
an auxiliary, observational data set containing measurements, for a separate population
of individuals, of the long-term outcome of interest, in addition to records of the
same pre-treatment covariates and short-term outcomes that were measured in the
experimental data set. This observational data set may or may not record whether each
individual was exposed to the treatment of interest. In this paper, we develop methods
for estimating the effect of a treatment on long-term outcomes that combine
experimental and observational data sets with this structure. 

To fix ideas, consider \cite{dynarski2021closing}, who randomize grants of free college tuition to a population of high achieving, low income high school students. They estimate the average effects of these free tuition grants on college application and enrollment rates. The long-term effects of free tuition grants on college completion rates may be of more direct interest to policy-makers considering the expansion of college aid programs. However, it will take several years before college completion is observed for the cohort of students in the experimental sample. Given an observational data set that records college enrollment and completion for a comparable population of high school students, the methods developed in this paper may facilitate a more timely quantification of the effect of free tuition grants on college completion. 

We begin this section by defining the data structure and estimands that we consider. We then review, and restate in a common notation, two sets of closely related sets of identifying assumptions proposed by \cite{athey2020combining} and \cite{athey2020estimating}.
We refer to these settings as the \emph{Latent Unconfounded Treatment} and 
\emph{Statistical Surrogacy} Models, respectively.

\subsection{Data} 
Consider the collection of random variables
\[
\{{A}_i\}_{i=1}^n = \{(Y_i(0), Y_i(1), {S}_i(0), S_i(1), W_i, G_i, {X}_i)\}_{i=1}^n
\]
consisting of the potential outcomes and characteristics of a sample of individuals drawn independently and identically from a distribution $P_\star$. Here, $G_i$ is a binary indicator
denoting whether the observation was acquired in the observational sample ($G_i=1$) or the experimental sample ($G_i=0$). The variables $Y_i(\cdot)$ are long-term
potential outcomes, $S_i(\cdot)$ are short-term {potential} outcomes, $W_i$
is a binary treatment indicator, and $X_i$ are covariates. 

The data observable to the researcher are denoted by $\{{B}_i\}_{i=1}^n$ and are i.i.d.\ according to a distribution denoted by $P$. The 
\emph{observed} outcomes $S_i$ and $Y_i$ are given by
$
S_i=W_iS_i(1)+(1-W_i)S_i(0)
$
and 
$
Y_i=W_iY_i(1)+(1-W_i)Y_i(0),
$
respectively. The short-term outcomes $S_i$ are observed in both the observational and
experimental data sets. The long-term outcome $Y_i$ is observed only in the observational
data set. Treatment $W_i$ may or may not be observed in the observational sample. 
Thus, the observable data $B_i$ are given by $(G_iY_i, S_i, W_i, G_i, X_i)$ if treatment is measured in the observational data set and by $(G_iY_i, S_i, (1-G_i)W_i, G_i, X_i)$ if treatment is not
measured in the observational data set. 

\subsection{Estimands} In the main text, we consider
estimation of the long-term average treatment effect in the observational population, given by
\begin{equation}
\label{eq: ltate_1}
\tau_1 = \E_{P_\star}\left[Y_i(1) - Y_i(0) \mid G_i = 1\right]~.
\end{equation}\cite{athey2020combining} note that there is often reason to believe that features 
of the observational population ($G_i=1$) are more ``externally valid,'' in that they
of greater interest to policymakers. 

In \cref{sec: G=0}, we give results analogous to those presented in the main text for the long-term average treatment effect in the experimental population
\begin{equation}
\tau_0 = \E_{P_\star}\left[Y_i(1) - Y_i(0) \mid G_i = 0\right]~,
\end{equation}
which may be of interest in some contexts.\footnote{An alternative estimand is the unconditional long-term average treatment effect
$\tau = \E_{P_\star}\left[Y_i(1) - Y_i(0)\right]$. The practical
interpretation of this parameter is somewhat nebulous, as it is unclear why it
would be desirable to weight the two samples in the definition of the parameter
according to their sizes.} \cite{athey2020combining} and \cite{athey2020estimating}
consider estimation of $\tau_1$ and $\tau_0$, respectively. 

\subsection{Identifying Assumptions}

We consider two sets of assumptions, proposed in \cite{athey2020combining} and 
\cite{athey2020estimating}. In both models, the
long-term average treatment effect $\tau_1$ is identified. The models differ according to whether they are applicable to
contexts in which treatment is or is not measured in the observational data set. Both
models assume that treatment assignment is unconfounded in the experimental data set
and that the probabilities of being assigned treatment or of being measured in the observational data set satisfy a strict overlap condition.
\begin{assumption}[Experimental Unconfounded Treatment]
\label{as: uex}
In the experimental data set, treatment is independent of short-term and
long-term potential outcomes conditional on pre-treatment covariates, in the sense that 
\[
W_i \indep (Y_i(0), S_i(0), Y_i(1), S_i(1)) \mid {X}_i, G_i = 0~.
\]
\end{assumption}
\begin{assumption}[Strict Overlap]
\label{as: overlap}
The probability of being assigned to treatment or of being measured in the observational data set is strictly bounded away from zero and one,
in the sense that, for each $w$ and $g$ in $\{0,1\}$, the conditional probabilities 
\[
P(W = w \mid S, X , G=g) \quad\text{and}\quad P(G= g \mid S, X, W=w)
\]
are bounded between $\varepsilon$ and $1-\varepsilon$,  $\lambda$-almost surely, for some fixed constant $0<\varepsilon<1/2$.
\end{assumption}
\noindent Unconfounded treatment and strict overlap are often
satisfied in the experimental sample by design. Assessing and accounting for
violations of overlap in observational data are important empirical and
methodological issues \citep{crump2009dealing}. In particular, strong overlap conditions can 
place stringent restrictions on the data generating process when there are many covariates or short-term outcomes \citep{d2021overlap}.
We view systematic consideration of these issues in our context as an important area for further research.

In addition, as our aim is to use the experimental sample to estimate a feature of 
the observational population, we require an assumption limiting the differences
between the two populations.

\begin{assumption}[Experimental Conditional External Validity]
\label{as: ceev}
The distribution of the potential outcomes is invariant to whether the data belong
to the experimental or observational data sets, in the sense that
\[
G_i \indep \left(Y_i(1), Y_i(0), S_i(1), S_i(0)\right) \mid {X}_i~.
\]
\end{assumption}
\noindent
\cref{as: ceev} implies that adjustments to the distribution of covariates in the experimental
data set are sufficient to obtain approximations to features of the observational 
population, thereby ruling out unobserved systematic differences between
the two populations conditional on covariates.

\subsubsection{Latent Unconfounded Treatment}

If treatment is measured in the observational data set, then the key identifying assumption
is that treatment assignment is unconfounded with respect to the long-term outcome if
conditioned on the short-term potential outcomes. We term this restriction ``Latent
Unconfounded Treatment'' following  \cite{athey2020combining}. 
\begin{assumption}[Observational Latent Unconfounded Treatment]
\label{as: luot}
In the observational data set, treatment is independent of the long-term potential outcomes
conditional on the short-term potential outcomes and pre-treatment covariates, in the sense
that, for $w\in\{0,1\}$,
\[
W_i \indep Y_i(w) \mid S_i(w), {X}_i, G_i = 1~.
\]
\end{assumption}
\noindent
Informally, \cref{as: luot} states that all unobserved confounding in the observational
sample is mediated through the short-term outcomes. \cref{as: uex,as: overlap,as: ceev,as: luot} 
are sufficient for identification of the long-term treatment effect
$\tau_1$. We summarize these assumptions with the following shorthand.
\begin{defn} 
\label{def: lut}
The collection of \cref{as: uex,as: overlap,as: ceev,as: luot}, in a 
setting where treatment is measured in the observational data set,
is referred to as the Latent Unconfounded Treatment Model.
\end{defn}
\noindent
Panel A of \cref{fig:acidag} displays a causal  Directed Acyclic Graph (DAG) that is consistent with 
the restrictions that the Latent Unconfounded Treatment Model place on the data 
generating process for the observational data set.\footnote{See \cite{pearl1995causal} for further discussion of the 
applications of manipulation of graphical models to causal inference. We do not make use of 
do-calculus methodology to establish identification.} 
The following proposition is
stated as Theorem 1 in
\cite{athey2020combining}. We state and prove the result for completeness.
\begin{prop}[\cite{athey2020combining}]
\label{thm: identification lut} Under the  Latent Unconfounded 
Treatment Model, the long-term average treatment effect $\tau_1$ is point identified.
\end{prop}

\begin{figure}
\begin{centering}
\caption{Causal DAGs Consistent with Assumptions on Observational Data}
\label{fig:acidag}
\medskip{}
\begin{tabular}{cc}
\textit{Panel A: Latent Unconfounded Treatment} & \textit{Panel B: Statistical Surrogacy} 
\tabularnewline
\begin{tikzpicture}[yscale=2, xscale = 2]
\node (w) at (-1.3,0) [label=left:\large{W},point,scale = 2];
\node (s) at (0,0) [label=above:\large{S},point,scale = 2];
\node (y) at (1.3,0) [label=right:\large{Y},point,scale = 2];
\path (w) edge[line width=1mm, color = black!30] (s);
\path (s) edge[line width=1mm, color = black!30]  (y);
\path (w) edge[bend left=60, line width=1mm, color = black!30] (y);

\path[bidirected] (w) edge[bend left=-60, line width=1mm, color = black!30] (s);
\path[bidirected] (y) edge[bend left=60, color=red, line width=1mm] (s);
\path[bidirected] (y) edge[bend left=60, color=red, line width=1mm] (w);    
\node at (.60,-.40) [red] {$\text{\Large{X}}$};
\node at (0,-0.73) [red] {$\text{\Large{X}}$};
\end{tikzpicture}
&
\begin{tikzpicture}[yscale=2, xscale = 2]
\node (w) at (-1.3,0) [label=left:\large{W},point,scale = 2];
\node (s) at (0,0) [label=above:\large{S},point,scale = 2];
\node (y) at (1.3,0) [label=right:\large{Y},point,scale = 2];
\path (w) edge[line width=1mm, color = black!30] (s);
\path (s) edge[line width=1mm, color = black!30] (y);
\path (w) edge[bend left=60, color=red, line width=1mm] (y);
\path[bidirected] (w) edge[bend left=-60, line width=1mm, color = black!30] (s);
\path[bidirected] (y) edge[bend left=60, color=red, line width=1mm] (s);
\path[bidirected] (y) edge[bend left=60, color=red, line width=1mm] (w);
\node at (.60,-.40) [red] {$\text{\Large{X}}$};
\node at (0,.73) [red] {$\text{\Large{X}}$};
\node at (0,-0.73) [red] {$\text{\Large{X}}$};
\end{tikzpicture}
\tabularnewline
\end{tabular}
\par\end{centering}
\medskip{}
\justifying
{\footnotesize{}Notes: Panels A and B of Figure \ref{fig:acidag} display causal DAGs that describe the restrictions on the data generating process for the observational data set ($G=1$) implied by the Latent Unconfounded Treatment and Statistical Surrogacy Models, respectively. Light grey arrows denote the existence of an effect of the tail variable on the head variable. Dark red arrows with x's denote that an effect of the the tail variable on the head variable is ruled out. Dashed bidirectional arrows represent existence of some unobserved common causal variable $U$, where we have a fork $\leftarrow U \rightarrow$. }{\footnotesize\par}
\end{figure}

\subsubsection{Statistical Surrogacy}

If treatment is not measured in the observational data set, an alternative ``Statistical
Surrogacy'' assumption, in the spirit of \cite{prentice1989surrogate}, is required in the
place of \cref{as: luot}. Under this restriction, the short-term outcomes can be
interpreted as ``proxies'' or ``surrogates'' for the long-term outcome.

\begin{assumption}[Experimental Statistical Surrogacy]
\label{as: es} In the experimental data set, treatment is independent of the long-term 
observed outcomes conditional on the short-term observed outcomes and pre-treatment covariates, in the sense that
\[ 
W_i \indep Y_i \mid S_i, {X}_i, G_i = 0~.
\]
\end{assumption}
\noindent
Informally, \cref{as: es} additionally rules out a causal link from treatment to the
long-term outcome that is not mediated by
the short-term outcomes.

In addition to \cref{as: ceev}, a supplementary restriction is required to ensure that the experimental and observational 
data sets are suitably comparable, conditional on the observed outcomes.
\begin{assumption}[Long-Term Outcome Comparability]
\label{as: ltoc} The distribution of the long-term outcome is invariant to
whether the data belong to the experimental or observational data sets 
conditional on the short-term outcome and covariates, in the sense that
\[ G_i \indep Y_i \mid {X}_i, S_i~.\]
\end{assumption}
\noindent
\cref{as: ceev} is not strictly stronger than \cref{as: ltoc}, as belonging to the experimental or observational data sets
is not necessarily independent of treatment assignment.
Statistical Surrogacy, in addition to \cref{as: uex,as: overlap,as: ltoc,as: ceev}, is sufficient for identification of 
the long-term treatment effect $\tau_1$ in settings where treatment is not measured in the observational 
data set. Again, we summarize these conditions with the following shorthand.
\begin{defn} 
\label{def: sur}
The collection of \cref{as: uex,as: overlap,as: ltoc,as: ceev,as: es}, in a 
setting where treatment is not measured in the observational data set,
is referred to as the Statistical Surrogacy Model.
\end{defn}
\noindent
Panel B of \cref{fig:acidag} describes the restrictions that 
the Statistical Surrogacy Model place on the data generating process for the observational data set.  
The following proposition restates Theorem 1 of \cite{athey2020estimating}. Again,
we state and prove the result for completeness.\footnote{We note that \cref{as: ceev} is unnecessary for the identification of $\tau_0$ in the Statistical Surrogacy Model.}
\begin{prop}[\cite{athey2020estimating}] 
\label{thm: identification sur} Under the Statistical Surrogacy Model, $\tau_1$ is point
identified. 
\end{prop}

\section{Semiparametric Efficiency\label{sec: efficiency}}

In this section, we derive efficient influence functions and corresponding
semiparametric efficiency bounds for estimation of long-term average treatment
effects $\tau_1$  in observational populations. In \cref{sec: G=0}, we state
comparable results for long-term average treatment effects $\tau_0$ in
experimental populations. 

\subsection{Nuisance Functions} \label{sec: nuisance}The efficient influence functions that we derive are
 expressed in terms of a set of unknown, but identified, conditional expectations. We
 classify each of these objects as being either a
 ``long-term outcome mean'' or a ``propensity score.'' Each long-term outcome mean
 is an expectation of the long-term outcome conditioned on other features of
 the data. Under the Latent Unconfounded Treatment Model, where treatment is measured
 in the observational data set, the long-term outcome means that appear in the
 efficient influence function are given by
\begin{align}
\mu_{w}(s,x) &= \E_P[Y \mid W=w, S=s, X=x, G=1]\quad\text{and} \label{eq: mu_w(s,x)}\\
\bar{\mu}_w(x) & = \E_{P}[\mu_w(S,X) \mid W=w, X=x, G=0]\label{eq: mu_w(x)}~.
\end{align}
The function $\mu_w(s,x)$ is the mean of the long term outcomes 
in the observational sample, conditioned on treatment $w$, 
the short-term outcomes $s$, and the covariates $x$. The function $\bar \mu_w(x)$ 
is the projection of $\mu_w(s,x)$ onto the experimental
population conditional on $w$ and $x$, integrated over $s$. 

In turn, under the Statistical Surrogacy model, where treatment is not measured in the observational data set, the analogous nuisance functions appearing in the efficient influence function are given by 
\begin{align}
\nu(s,x) &= \E_P[Y \mid S=s, X=x, G=1]\quad\text{and} \label{eq: nu_w(s,x)}\\
\bar{\nu}_w(x) & = \E_{P}[\nu(S,X) \mid W=w, X=x, G=0]~. \label{eq: nu(x)}
\end{align}
The function $\nu(s,x)$ is similar to $\mu_w(s,x)$, but does not condition on treatment,
as treatment is not observed in the observational sample under the Statistical Surrogacy
Model. The function $\bar \nu_w(x)$ projects $\nu(s,x)$ on the observational sample,
conditional on $x$ and $w$, and integrates over $s$.

It useful to note that the proofs of \cref{thm: identification lut,thm: identification sur} operate by establishing, in their respective models, that the equalities
\[
\mathbb{E}_{P_*}\left[Y(1)\mid G=1\right] = \mathbb{E}_{P_*}\left[\bar{\mu}_1(X) \mid G=1\right]
\quad\text{and}\quad
\mathbb{E}_{P_*} \left[Y(1)\mid G=1\right] = \mathbb{E}_{P_*}\left[\bar{\nu}_1(X) \mid G=1\right]
\]
hold and that the objects \eqref{eq: mu_w(x)} and \eqref{eq: nu(x)} are identified from the observable data. \Cref{sec: non-orth} presents a set of alternative moment conditions that analogously identify $\tau_1$.

Each member of the second class of nuisance functions---propensity scores---expresses either the probability of treatment or the probability of inclusion in the observational sample conditioned on other features of the data. In particular, let
\begin{align}
\rho_w(s,x) & = P_\star(W=w \mid S(w) = s, X = x, G=1)~,  \label{eq: q_w(s,x)}\\
\varrho(s,x)      & = P(W=1 \mid S = s, X = x, G=0)~,\quad\text{and} \label{eq: r(s,x)}\\
\varrho(x)        & = P(W=1\mid X=x, G=0) \label{eq: e(x)}
\end{align}
denote the probabilities of treatment conditional on various features of the data. Similarly, let 
\begin{align}
\gamma(s,x) & = P(G=1\mid S = s, X=x)~,  
\gamma(x)  = P(G=1\mid X = x)~,~\text{and}~~
\pi      = P(G=1) \label{eq: G propensity}
\end{align}
denote the probabilities of inclusion in the observational sample conditional on various features of the data.\footnote{ \cite{athey2020estimating} term $\nu(s,x)$ the \emph{surrogate index}, $\varrho(s,x)$ the \emph{surrogate score}, and $1-\gamma(s,x)$ the \emph{sampling score}.} 
Each of these objects is bounded away from zero and one $\lambda$-almost surely by
strict overlap (\cref{as: overlap}). Although the propensity score $\rho_w(s,x)$ includes
the unobserved random variable $S(w)$ in its conditioning set, we may write $\rho_w(s,x)$ in terms of observables as
\begin{align}
\rho_w(s,x) &=  \frac{P(G=1 \mid S=s, W=w, X=x) }{P(G=0 \mid S=s, W=w, X=x)} \nonumber 
 \\ 
 &\quad \times \frac{\P( G=0 \mid W=w, X=x)}{\P(G=1 \mid W=w, X=x)}
P(W=w\mid
X=x, G=1)\label{eq: q identified}
\end{align} by successive applications of Bayes' rule, \cref{as: uex}, and \cref
 {as: ceev}.

\subsection{Influence Functions and Efficiency Bounds}\label{sec:eif} We characterize efficient influence functions for $\tau_1$ with the approach developed in Section 3.4 of
\cite{bickel1993efficient}. In particular, we characterize the tangent spaces for the classes
 of distributions restricted by the maintained assumptions. We then verify that the
 conjectured efficient influence functions are pathwise derivatives of $\tau_1$ and are
 elements of their respective tangent spaces. Recall that the
Latent Unconfounded Treatment Model and Statistical Surrogacy Model differ according to the assumptions they impose on the data generating distribution $P_\star$ in addition to whether
treatment is assumed to have been measured in the observational sample. 

\begin{theorem}
\label{thm: EIF tau_1}~
Let $b = (y,s,w,g,x)$ denote a possible value for the observed data. 
\begin{enumerate}
\item Under the Latent Unconfounded Treatment Model, given in \cref{def: lut},
the efficient influence function for the parameter $\tau_1$ is given by
\begin{align}
\psi_1(b,\tau_1,\eta) &=
\frac{g}{\pi}
\left( \frac{w(y-\mu_1(s,x))}{\rho_1(s,x)}  - \frac{(1-w)(y-\mu_0(s,x))}{\rho_0
(s,x)} + (\bar{\mu}_1(x) - \bar{\mu}_0(x))  - \tau_1 \right) \label{eq: psi}\\
&  + \frac{1-g}{\pi}
\left(\frac{\gamma (x)}{1-\gamma(x)}\left( \frac{w(\mu_1(s,x)-\bar{\mu}_1(x))}{\varrho(x)} - \frac{(1-w)(\mu_0(s,x)-\bar{\mu}_0(x))}{1-\varrho(x)}\right)\right),\nonumber
\end{align}
where the parameter $\eta $ collects the nuisance functions appearing in \eqref{eq: psi}.
\item Under the Statistical Surrogacy Model, given in \cref{def: sur},
the efficient influence function for the
 parameter $\tau_1$ is given by
\begin{align}
& \xi_1(b,\tau_1,\varphi)
= \frac{g}{\pi}
 \left( \frac{\gamma(x)}{\gamma(s,x)} \frac{1- \gamma(s,x)}{1-\gamma(x)}\frac{(\varrho(s,x)-\varrho(x))(y-\nu(s,x))}{\varrho(x)(1-\varrho(x))}
+ (\bar{\nu}_1(x) - \bar{\nu}_0(x)) - \tau_1\right) \nonumber \\
& \quad \quad  \quad  + \frac{1-g}{\pi} \left(\frac{\gamma(x)}{1-\gamma(x)} \left(\frac{w(\nu(s,x)-\bar{\nu}_1(x))}{\varrho(x)} -\frac{(1-w)(\nu(s,x)-\bar{\nu}_0(x))}{1-\varrho(x)}\right)\right)~,\label{eq: xi}
\end{align}
where the parameter $\varphi$ collects nuisance functions appearing in \eqref{eq: xi}.\footnote{We thank Rahul Singh for noting a typo in the statement of $\xi_1(\cdot)$ in a previous draft of this paper, which was missing the factor $\gamma(x)/(1-\gamma(x))$ in the second term.}
\end{enumerate}
\end{theorem}

In our discussion, we will frequently reference the nuisance parameters $\eta$ and $\varphi$ introduced in \cref{thm: EIF tau_1}. We introduce the following notation to expedite our exposition.

\begin{defn}\label{def: partition}
Partition the nuisance function $\eta$ and $\varphi$ defined in \cref{thm: EIF tau_1} into the long-term outcome means, propensity scores, and $\pi$ by
\[
\eta = (\omega_\psi, \kappa_\psi, \pi) \quad\text{and}\quad \varphi = (\omega_\xi, \kappa_\xi, \pi).
\]
In particular, the parameters $\omega_\psi$ and $\omega_\xi$ collect long-term outcome means
\[
\omega_\psi = (\bar{\mu}_1(\cdot), \bar{\mu}_0(\cdot), \mu_1(\cdot,\cdot), \mu_0(\cdot,\cdot))
\quad\text{and}\quad
\omega_\xi = (\bar{\nu}_1(\cdot), \bar{\nu}_0(\cdot), \nu(\cdot,\cdot))~,
\]
respectively, and the parameters $\kappa_\psi$ and $\kappa_\xi$ collect propensity scores
\[
\kappa_\psi = (\rho_1(\cdot,\cdot), \rho_0(\cdot,\cdot), \varrho(\cdot), \gamma(\cdot))
\quad\text{and}\quad
\kappa_\xi = (\varrho(\cdot, \cdot), \varrho(\cdot), \gamma(\cdot, \cdot), \gamma(\cdot))~,
\]
respectively. 
\end{defn}

\begin{remark}
\label{rem: aipw comp}
The efficient influence functions $\psi_1(b,\tau_1,\eta)$ and $\xi_1(b,\tau_1,\varphi)$
derived in \cref{thm: EIF tau_1} are additively separable into two terms associated with
the observational and experimental data sets, respectively. The structure of the terms
associated with the observational data set resembles the structure of the efficient
influence function for the average treatment effect under unconfoundedness 
\citep{hahn1998role}; we discuss the relationship between these objects in \cref{sec:procedure}. \hfill\qed
\end{remark}

\begin{remark}
\label{rem: double}
The efficient influence functions $\psi_1(b,\tau_1,\eta)$ and $\xi_1(b,\tau_1,\varphi)$
possess a ``double-robust'' structure that is prevalent in causal
inference and missing data problems \citep[see e.g.,][]
{kang2007demystifying, bang2005doubly}. In particular, the mean-zero property of $\psi_{1}
(b,\tau_1,\eta)$ is maintained even if some of the nuisance functions are misspecified.
Suppose that we let arbitrary measurable functions $\tilde{\omega}$ replace the long-term outcome means $\omega_\psi$, then
$
\E_P\bk{\psi_1(B, \tau_1, (\tilde{\omega},\kappa_\psi, \pi))} = 0.
$
Similarly, if the arbitrary measurable functions $\tilde{\kappa}$ replace the propensity scores $\kappa_\psi$, then we have that
$
\E_P\bk{\psi_1(B, \tau_1, (\omega_\psi,\tilde{\kappa},\pi))} = 0.
$
The efficient influence function $\xi_{1}(b,\tau_1,\varphi)$ is similarly robust to misspecification of either the long-term outcome means or the propensity scores. This result echoes an analogous double-robustness property of the efficient influence function of the average treatment effect under ignorability \citep*{scharfstein1999adjusting}, in which the efficient influence function is mean zero under misspecification of either the conditional means of the outcome variable or the propensity score. The form of the double-robustness entailed here is slightly more general, as $\omega$ and $\kappa$ each collect several nuisance function.

Double robustness, in this form, has a useful implication for estimation. We demonstrate in \cref{sec: consistency double} that, under appropriate regularity conditions, estimators based on the efficient influence functions $\psi_1(b,\tau_1,\eta)$ or $\xi_1(b,\tau_1,\varphi)$ are consistent for $\tau_1$ if \emph{either} the outcome means ($\omega_\psi$ or $\omega_\xi$) or the propensity scores ($\kappa_\psi$ or $\kappa_\xi$) are estimated consistently. We analyze estimators of this form in further detail in \cref{sec:estimation}.
\hfill\qed
\end{remark}

The population variance of the efficient influence function is the semiparametric efficiency bound. The respective bounds are presented in \cref{cor: seb}.
\begin{cor}
\label{cor: seb}
Define the conditional variances
\begin{align*}
\sigma_w^2(s,x) &= \E_{P_\star}[(Y(w) - \mu_w(S,X))^2\mid S=s,X=x]\quad\text{and}\\
\sigma^2(s,x) &= \E_{P_\star}[(Y - \nu(S,X))^2\mid S=s,X=x]
\end{align*}
as well as the expressions 
\begin{align*}
\Gamma_{w,1}(s,x) & =   \frac{\gamma (x)}{1-\gamma(x)} \frac{\left(\mu_w(s,x) - \bar{\mu}_w
(x)\right)^2}{\varrho(x)^w(1-\varrho(x))^{1-w}}~~\text{and}~~
\Lambda_{w,1}(s,x) 
=\frac{\gamma(x)}{1-\gamma(x)} \frac{\left(\nu(s,x) - \bar{\nu}_w(x)\right)^2}{\varrho(x)^w(1-\varrho(x))^{1-w}}~.
\end{align*}
\begin{enumerate}
\item Under the Latent Unconfounded Treatment Model, given in \cref{def: lut}, the semiparametric efficiency bound for $\tau_1$ is given by 
\begin{align}
V_1^{\star} & = \E_P\Bigg[\frac{\gamma(X)}{\pi^2} \Bigg(
\frac{\sigma_1^2(S,X)}{\rho_1(S,X)} +
\frac{\sigma_0^2(S,X)}{\rho_0(S,X)}  \nonumber \\
& \quad \quad\quad \quad \quad \quad
+(\bar{\mu}_1(X) - \bar{\mu}_0(X) - \tau_1)^2 + \Gamma_{0,1}(S,X) + \Gamma_{1,1}(S,X)\Bigg) \Bigg]~.
\end{align}
\item Under the Statistical Surrogacy Model, given in \cref{def: sur}, the 
semiparametric efficiency bound for $\tau_1$ is given by
\begin{align}
V_1^{\star\star}  & 
= \E_P\Bigg[\frac{\gamma(X)}{\pi^2} 
\Bigg( \left(\frac{\gamma(X)}{\gamma(S,X)} 
\frac{1-\gamma(S,X)}{1-\gamma(X)}
\frac{\varrho(S,X)-\varrho(X)}{\varrho(X)(1-\varrho(X))}
\right)^2\sigma^2(S,X) \nonumber \\
& \quad \quad\quad \quad \quad\quad
+ (\bar{\nu}_1(X) - \bar{\nu}_0(X) - \tau_1)^2 + \Lambda_{0,1}(S,X) + \Lambda_{1,1}(S,X)\Bigg)
\Bigg]~.
\end{align}
\end{enumerate}
\end{cor}

\begin{remark}\label{rem: known nuisance} In \cref{asub:known_nuisance_bound}, we analyze how the semiparametric efficiency bounds derived
 in \cref{cor: seb} change if different components of the nuisance parameters $\eta$ or
 $\varphi$ are known a priori. In both models, if the classical propensity score $\varrho
 (X)$, i.e., the probability of being assigned to treatment in the experimental sample as a
 function of covariates, is known, then the semiparametric efficiency bounds are
 unchanged.\footnote{Invariance of the semiparametric efficiency bounds to knowledge of the
 propensity score $\varrho(X)$ would no longer hold if the estimands of interest were
 average long-term effects for the treated population.} This echoes an analogous ancillarity result for
 estimation of average treatment effects under unconfoundedness given in \cite{hahn1998role}. 

By contrast, both semiparametric efficiency bounds change if the propensity score
$\gamma(x)$, i.e., the probability of being assigned to the observational sample as a
function of covariates, is known. This result is relevant for settings where the
experimental sample is known to be drawn from the same population as the observational
sample and indicates that development of estimators tailored to this setting may be
fruitful.\footnote{On the other hand, in that context, consideration of the unconditional
long-term treatment effect $\tau=\mathbb{E}[Y(1) - Y(0)]$ is tenable and natural. We
expect the efficiency bound for this functional to be invariant to knowledge of the
propensity score $\gamma(x)$.} In the Statistical Surrogacy Model, somewhat curiously, the
efficient influence function $\xi_{1}(b,\tau_1,\varphi)$ is additionally invariant to
knowledge of the distribution, in the observational sample, of the short-term outcomes
conditional on covariates, i.e. the law $S \mid X, G=1$. This invariance is a consequence 
of the choice \eqref{eq:surprising_ancillarity} in construction of the efficient influence
function $\xi_{1}(b,\tau_1,\varphi)$ in the proof of \cref{thm: EIF tau_1}.
\hfill\qed
\end{remark}

Next, we demonstrate that the efficient influence functions $\psi_1(b,\tau_1,\eta)$ and $\xi_1(b,\tau_1,\varphi)$ expressed
in \cref{thm: EIF tau_1} are, in fact, the unique influence functions for
$\tau_1$ in their respective models. We recall that an influence for $\tau_1$ is any mean-zero and square integrable function function $\tilde{\psi}(b)$ that satisfies the condition
\begin{align*}
\tau_1^\prime = \mathbb{E}_P[\tilde{\psi}(B)l^{\prime}(B)]~,
\end{align*}
where $\tau_1^\prime$ is the pathwise derivative of $\tau_1$ along an arbitrary parametric submodel evaluated at zero and $l^{\prime}(B)$ is the score function of this submodel; see e.g., Chapter 25 of \cite{van2000asymptotic} for further discussion.

\begin{theorem}~
\label{thm: locally-overidentified} There are unique influence functions in each model:
\begin{enumerate}
\item Under the Latent Unconfounded Treatment Model, given in \cref{def: lut}, $\psi_1(b,\tau_1,\eta)$ is the unique influence function for $\tau_1$. 
\item Under the Statistical Surrogacy Model, given in \cref{def: sur}, $\xi_1(b,\tau_1,\varphi)$ is the unique influence function for $\tau_1$.
\end{enumerate}
\end{theorem}

\begin{remark}
\label{rem: local identify}
Let $\mathcal{P}\subset\mathcal{M}_\lambda$ denote the set of probability distributions
that satisfy either the Latent Unconfounded Treatment Model or the Statistical
Surrogacy Model. In the terminology of
\cite{chen2018overidentification}, \cref{thm: locally-overidentified}, Part (1),
demonstrates that $P$ is locally just-identified by $\mathcal{P}$. As a result,
by Theorem 3.1 of \cite{chen2018overidentification}, all regular and
asymptotically linear (RAL) estimators of $\tau_1$ are first-order equivalent under the
maintained assumptions. In
particular, there are no RAL estimators of
$\tau_1$ that have smaller asymptotic variances than others. 

Consequently, semiparametric efficiency is equivalent to regularity 
and asymptotic linearity under the maintained assumptions.
We note that if the propensity score $\varrho(x)$ admits
known restrictions, then the resultant model would be semiparametrically over-identified. In this case, not all RAL estimators are first-order equivalent. However, since the
efficiency bound does not change, the estimators we propose in the following section remain efficient with known propensity score.\footnote{The semiparametric literature on average
treatment effect and local average treatment effect estimation \citep[e.g.][]
{hirano2003efficient,frolich2007nonparametric} discusses efficiency as well, despite
the fact that the models in question are similarly just-identified.}

Moreover, again by Theorem 3.1 of \cite{chen2018overidentification}, the model $\mathcal
{P}$ does not have any locally testable restrictions in the sense that there are no
specification tests of the maintained identifying assumptions with nontrivial local
asymptotic power. Analogous statements follow from \cref{thm: locally-overidentified},
Part (2).  In that sense, the maintained identifying assumptions are minimal.\hfill\qed
\end{remark}

\section{Estimation\label{sec:estimation}}

We now consider estimation of the long-term average treatment effect
$\tau_1$.\footnote{In \cref{sec: G=0}, we provide an analogous treatment of estimators of the long-term average treatment effect $\tau_0$ in the experimental population.} The estimators that we consider can each be viewed as semiparametric $Z$-estimators associated with an identifying moment function. That is, each estimator is premised on determining the value $\tau_1$ that solves a sample analogue of a moment condition
\[
\mathbb{E}_P\left[g(B_i, \tau_1, \zeta)\right]= 0~,
\]
for some identifying moment function $g(\cdot)$, where $\zeta$ is an unknown nuisance parameter, replaced with its estimated counterpart in practice. 

Our treatment differs by  whether the identifying moment function $g(\cdot)$ is given by the efficient influence functions $\psi_1(\cdot)$ or $\xi_1(\cdot)$, derived in \cref{thm: EIF tau_1} or given by some other moment condition. Moment conditions defined by influence functions are often referred to as Neyman orthogonal moment conditions. We adapt very general arguments from \cite{chernozhukov2018double} to establish that these estimators are consistent and asymptotically normal. 

Our consideration of estimators based on non-orthogonal moments is selective and is more specialized. To obtain theoretical guarantees, we restrict attention to estimators that plug-in nuisance parameter estimates derived from the method of sieves \citep{chen2015sieve,chen2014sieve}. Sufficient conditions for estimators with this structure are generally available, but are more delicate and difficult to verify. We state and verify these conditions for one of the estimators that we consider.\footnote{The general high-level conditions in \citet{chen2015sieve} and \citet{chen2014sieve} apply to each of estimators that we consider. Lower-level conditions, analogous to those discussed in \cref{sec: non-orth} will exist for these estimators as well. However, as the derivation of these conditions is lengthy and cumbersome, we provide an illustration of this argument for only one estimator.}

Throughout, it is useful to keep in mind that the efficient influence functions $\psi_1(\cdot)$ and $\xi_1(\cdot)$ are the only influence functions in their respective models (i.e., \cref{thm: locally-overidentified}). Consequently, all regular and asymptotically linear estimators are first-order equivalent, that is, their asymptotic variances are all equal to the semiparametric efficiency bound. Thus, the asymptotic variances of different estimators are the same, but the conditions under which they are asymptotically normal may be different.

\subsection{Orthogonal Moments\label{sec:orthogonal}}
We begin by considering estimators that are built directly on the influence functions $\psi_1(\cdot)$ or $\xi_1(\cdot)$, derived in \cref{thm: EIF tau_1} with the ``Double/Debiased Machine Learning'' (DML) construction developed in \cite{chernozhukov2018double}.

\subsubsection{Construction\label{sec:procedure}}
The DML construction proceeds in two steps. First, estimates of the nuisance functions $\eta$ or $\varphi$, defined in \cref{thm: EIF tau_1}, are computed with cross-fitting. Second, estimates of $\tau_1$ are obtained by plugging the estimated values of $\eta$ and $\varphi$ into their respective efficient influence functions and solving for the values of $\tau_1$ that equate the sample means of these estimates of the efficient influence functions with zero.

\begin{defn}[DML Estimators]
\label{def: dml}
Let $\hat{\eta}(I)$ and $\hat{\varphi}(I)$ denote generic estimates of $\eta$ and
$\varphi$ based on the data $\{B_i\}_{i\in I}$ for some subset $I\subseteq [n]$. Let $
\{I_l\}_{l=1}^k$ denote a random $k$-fold partition of $[n]$ such that the size of each fold is $m=n/k$. The estimator $\hat{\tau}_{1,\mathsf{DML}}$ is defined as the solution to
\[
\frac{1}{k}\sum_{l=1}^k \frac{1}{m} \sum_{i\in I_l} \psi_1(B_i,\hat{\tau}_{1,\mathsf{DML}},\hat{\eta}(I_l^c)) = 0
\quad\text{or}\quad
\frac{1}{k}\sum_{l=1}^k \frac{1}{m} \sum_{i\in I_l} \xi_1(B_i,\hat{\tau}_{1,\mathsf{DML}},\hat{\varphi}(I_l^c)) = 0
\]
for the Latent Unconfounded Treatment and Statistical Surrogacy Models, respectively.
\end{defn}

\begin{remark} 
\label{rmk:structure_dml}
The fundamental structures underlying standard estimators of average treatment effects under unconfoundedness can be classified as being based on either ``inverse propensity score weighting'' or ``outcome regression'' \citep{imbens2004nonparametric}; elements of each structure appear in the estimators formulated in \cref{def: dml}.\footnote{\cite{imbens2004nonparametric} also discusses estimators based on matching and Bayesian calculations. We do not develop estimators with these structures in this paper, and view their consideration as an interesting extension.} Inverse propensity score weighted (IPW) estimators, also referred to as \cite{horvitz1952generalization} estimators, are constructed by weighting the observed values of outcomes by their inverse propensity scores; see e.g., \cite{rosenbaum1983central} and \cite*{hirano2003efficient}. By contrast, outcome regression estimators are constructed by imputing unobserved potential outcomes with estimates of their expectation conditioned on covariates.

The estimators formulated in \cref{def: dml} combine IPW and outcome regression components with an error-correcting structure comparable to the augmented inverse propensity weighted (AIPW) estimator of \cite{robins1995analysis}. To illustrate, observe that $\psi_1(b,\tau_1,\eta)$ can be interpreted as first approximating $\tau_1$ with the outcome regression $\bar{\mu}_1(x) - \bar{\mu}_0(x)$ in the observational  sample. Heuristically, the biases in this approximation, e.g., induced by regularization, are then corrected by applying IPW to the residuals of the approximation of $\bar{\mu}_w(x)$ to $\mu_w(s,x)$ with
\[
 \frac{w(\mu_1(s,x)-\bar{\mu}_1(x))}{\varrho(x)} - \frac{(1-w)(\mu_0(s,x)-\bar{\mu}_0(x))}{1-\varrho(x)}
\] computed in the experimental sample and reweighed by $\gamma(x)/(1-\gamma(x))$ to represent an
 expectation over the observational sample. However, the correction above may introduce
 additional biases through the estimation of $\mu_w(s,x)$; these are in turn corrected by
 applying IPW to the residuals of the approximation of $\mu_w(s,x)$ to $Y(w)$ with
\[
\frac{w(y-\mu_1(s,x))}{\rho_1(s,x)}  - \frac{(1-w)(y-\mu_0(s,x))}{\rho_0(s,x)}
\]
computed in the experimental sample. An analogous interpretation can be formulated for the structure of the efficient influence function $\xi_1(b,\tau_1,\varphi)$.\footnote{Note that in the term corresponding to the observational sample in $\xi_1(\cdot)$, the unobserved treatment indictor $w$ is replaced by the probability of treatment conditional on short-term outcomes and covariates.} Further discussion, at varying levels of rigor, of this ``bias-correction'' interpretation of the structure of estimators based on efficient influence functions is given in Section 4 of \cite{kennedy2023semiparametric} and Chapter 7 of \cite{bickel1993efficient}.\hfill\qed
\end{remark}

\begin{remark}
At a high-level, the cross-fitting construction used in \cref{def: dml} is implemented so that the estimation errors, e.g., $\mu_w(S_i, X_i) - \hat{\mu}_w(S_i, X_i)$, and model errors, e.g., $Y(W_i) - \mu_w(S_i, X_i)$, are unrelated for a given observation. Association between these two forms of error may have particularly pernicious effects in finite-samples if estimates of nuisance functions suffer from over-fitting. More technically, cross-fitting allows us to avoid imposing Donsker-type regularity conditions in our asymptotic analysis, which would exclude estimators with non-negligible asymptotic regularization. Standard implementations of popular machine learning algorithms may feature such regularization; see \cite{chernozhukov2016locally} for detailed discussion and illustration of this point. Further discussion of cross-fitting methods in semiparametric estimation is given in \cite{klassen1987consistent} and \cite{newey2018crossfitting}.\hfill\qed
\end{remark}

\begin{remark}
We require specialized approaches to estimate the nuisance functions $\rho_w(s,x)$, $\bar{\mu}_w(x)$, and $\bar{\nu}_w(x)$. We estimate $\rho_w(s,x)$ by combining separate estimates of each of the objects displayed in \cref{eq: q identified}. We estimate $\bar{\mu}_w(x)$ and $\bar{\nu}_w(x)$ by first computing estimates of $\mu_w(s,x)$ and $\nu(s,x)$ in the observational sample, denoted by $\hat{\mu}_w(s,x)$ and $\hat{\nu}(s,x)$, and then computing estimates of $\mathbb{E}_P[\hat{\mu}_w(s,x)\vert W = w, X = x, G = 0] $ and $\mathbb{E}_P[ \hat{\nu}(s,x)\vert W = w,X = x, G = 0]$ in the experimental sample. In \cref{sec:nested}, we derive the rate of convergence for particular implementations of estimators with this structure based on linear sieves. \hfill\qed 
\end{remark}

\subsubsection{Large-Sample Theory}\label{sec: dml lst}

We now study the asymptotic behavior of the estimators formulated in \cref{def:
dml}. First, in \cref{sec: consistency double}, we demonstrate that, under weak regularity conditions and under both models, the estimator
$\hat{\tau}_{1,\mathsf{DML}}$ is consistent for $\tau_1$ if either the long-term outcome means or the propensity scores 
are estimated consistently. Second, we establish asymptotic normality by providing conditions sufficient for the application of Theorem 3.1 of \cite{chernozhukov2018double}. We impose a set of standard bounds on 
moments of the data, and a set of conditions on the uniform rates of  convergence of nuisance parameter estimators. Throughout, for a collection of scalar-valued nuisance parameters $\theta = (\theta_1, \ldots, \theta_l)$, we let $\|\theta\|_{P,q}=\max_{i\in[l]} \pr{\E_P |\theta_i(B)|^q}^{1/q}$.

\begin{assumption}[Moment Bounds]\label{as: bounds}
Let $C, c > 0$ be constants. 
Under the Latent Unconfounded Treatment Model, the moment bounds
\begin{align*}
\|Y(w)\|_{P,q}\leq C, 
& \quad \mathbb{E}_P \left[ \sigma^2_w (S,X)\mid X\right]\leq C,\\
\mathbb{E}_P \left[ (Y(w) - \mu_w(S,X))^2\right] \geq c, 
& \quad 
\text{ and } \quad
c \le \mathbb{E}_P \left[ (\mu_w(S,X) - \bar{\mu}_w(x))^2\mid X\right] \leq C
\end{align*}
hold for each $w\in\{0,1\}$ and
$\lambda$-almost every $X$. Analogous bounds hold for the Statistical Surrogacy Model,
where $\sigma(S,X)$, $\nu(S,X)$, and $\bar{\nu}_w(x)$ replace $\sigma_w(S,X)$, 
$\mu_w(S,X)$, and $\bar{\mu}_w(x)$, respectively.
\end{assumption}

\begin{assumption}[Convergence Rates]
\label{as: rates}
Let $\mathcal{P}\subset\mathcal{M}_\lambda$ be the set of all probability distributions
$P$ that satisfy the Latent Unconfounded Treatment Model stated in \cref{def: lut}.
Consider a sequence of estimators $\hat\eta_{n}(I_n^c) = (\hat \omega_{\psi,n}, \hat\kappa_{\psi,n}, \hat\pi_n)$ indexed by $n$, where $I_n \subset [n]$ is a random subset of size $m = n/k$ and $\hat{\pi}_n = \frac{1}{n-m} \sum_{i\in I_n^c} G_i$. 
For some sequences $\Delta_n \to 0$ and $\delta_n \to 0$ and constants $\varepsilon, C > 0$ and $q>2$, that do not depend
on $P$, with $P$-probability at least $1-\Delta_n$,
\begin{enumerate}
    \item[1. (Consistency in 2-norm)]  $n^{-1/2} \le \norm{\hat \eta_n - \eta}_
    {P,2} \le \delta_n$,
    \item[2.  (Boundedness in $q$-norm)] $\norm{\hat \eta_n - \eta}_
    {P,q} < C$,
    \item[3. (Non-degeneracy)] $\varepsilon \le \hat
    \kappa_n \le 1-\varepsilon$, where the inequalities apply entry-wise, and
    \item[4.  ($o(n^{-1/2})$ product rates)]  $\| \hat{\omega}_{\psi,n} - \omega_\psi \|_
    {P,2} \cdot \| \hat{\kappa}_{\psi,n} -
    \kappa_\psi \|_{P,2} \leq \delta_n n^{-1/2}$.
\end{enumerate}
\noindent
Analogous conditions hold for the Statistical Surrogacy Model, where $\varphi$ replaces $\eta$.
\end{assumption}

\begin{remark} \label{rem: rates}
\cref{as: rates} imposes the restriction that the product of the estimation errors for the long-term outcome means and propensity scores converges at the rate $o(n^{-1/2})$.\footnote{If the true propensity score $\varrho(x)$ is known, then using this information when constructing $\hat{\tau}_{1,\mathsf{DML}}$ may produce an estimator that performs well in finite-samples. However, plugging the true propensity score into an estimator based off of a non-orthogonal moment would probably be inefficient. For example, in \cite{hahn1998role}, an IPW type estimator for the average treatment effect based off of the true propensity score is shown to be inefficient.} These rates can be achieved, even if the dimensionality of the covariates or the short-term outcomes is increasing with $n$, by many standard machine learning algorithms including the Lasso and Dantzig selector
\citep{bickel2009simultaneous,belloni2014pivotal}, boosting algorithms
\citep{ye2017boosting}, regression trees and random forests \citep{wager2015adaptive},
and neural networks \citep{chen1999improved,farrell2021deep} under appropriate conditions on the
structure or sparsity of the underlying model. 

In \cref{sec:nested}, we verify that estimates of $\bar{\mu}_w(x)$ and $\bar{\nu}_w(x)$
based on linear sieves can achieve these rates under sufficiently stringent restrictions
on the smoothness of the long-term outcome means. It is reasonable to expect that
analogous results should be available for more complicated estimators, e.g., featuring
penalization or more complicated bases. Results of this form are an interesting direction
for further research. \hfill\qed
\end{remark}

\cref{thm: large-sample} establishes the asymptotic properties of the estimators
 formulated in \cref{def: dml}. For the sake of brevity, we state the result only for
 the Latent Unconfounded Treatment Model. An analogous result holds for the Statistical Surrogacy Model. We provide proofs for both results in \cref{sec: proof of large-sample}.
\begin{theorem}
\label{thm: large-sample}~
Let $\mathcal{P}\subset\mathcal{M}_\lambda$ be the set of all probability distributions
$P$ for $\{B_i\}_{i=1}^n$ that satisfy the Latent Unconfounded Treatment Model stated in \cref{def: lut} 
in addition to \cref{as: bounds}. If \cref{as:
rates} holds for $\mathcal P$, then
\begin{align}
\sqrt{n}(\hat{\tau}_{1,\mathsf{DML}} -\tau_1) \overset{d}{\to} \mathcal{N}(0,V_1^\star)
\end{align}
uniformly over $P\in\mathcal{P}$, where $\hat{\tau}_{1,\mathsf{DML}}$ is defined in \cref{def: dml}, $V_1^\star$ is defined in \cref{cor: seb}, and $\overset{d}{\to}$ denotes convergence in distribution. Moreover, we have that 
\begin{align}\label{eq: V_1 hat}
\hat{V}_1^\star = \frac{1}{k}\sum_{l=1}^k \frac{1}{m} \sum_{i\in I_l} \left(\psi_1(B_i,\hat{\tau}_{1,\mathsf{DML}},\hat{\eta}(I_l^c))\right)^2 \overset{p}{\to} V_1^\star
\end{align}
uniformly over $P\in\mathcal{P}$, where $\overset{p}{\to}$ denotes convergence in probability. As a result, we obtain the uniform asymptotic validity of the confidence intervals
\begin{align}
\label{eq: ci}
\lim_{n\to\infty} \sup_{P\in\mathcal{P}} \Big\vert P\left(\tau_1 \in \left[\hat{\tau}_{1,\mathsf{DML}} \pm z_{1-\alpha/2}\sqrt{\hat{V}_1^\star/n} \right]\right) - (1-\alpha) \Big\vert = 0~,
\end{align}
where $z_{1-\alpha/2}$ is the $1-\alpha/2$ quantile of the standard normal distribution. 
\end{theorem}

\subsection{Non-orthogonal Moments\label{sec: non-orth}}
The moment functions $\psi_1(\cdot)$ and $\xi_1(\cdot)$ considered in \cref{sec:orthogonal} are quite complicated. Constructing the estimator $\hat{\tau}_{1,\mathsf{DML}}$ requires estimating several propensity scores and long-term outcome means. It is natural to ask whether it suffices to consider simpler estimators based on moment conditions with fewer nuisance functions. 

In this section, we consider a suite of estimators that do not use the efficient influence functions $\psi_1(\cdot)$ or $\xi_1(\cdot)$ as identifying moment functions. We emphasize estimators that bear a similarity to standard IPW or outcome regression estimators for average treatment effects, several of which were initially proposed in \cite{athey2020combining} and \cite{athey2020estimating}. We verify the asymptotic normality of one of these estimators, when nuisance parameters are estimated with the method of sieves, following ideas developed in \citep{chen2015sieve,chen2014sieve}.

\subsubsection{Construction\label{sec: non-orth construct}}

Each of the estimators that we consider can be viewed as a $Z$-estimator based on a moment function $g(\cdot,\tau_1,\zeta)$, taking as an argument an unknown nuisance parameter $\zeta$.
\begin{defn}[Non-orthogonal Moment Estimators]
\label{def: non-orth}
The estimator $\hat{\tau}_{1}(g)$ is defined as the solution to the sample moment condition
\[
\frac{1}{n}\sum_{i=1}^n g(B_i,\hat{\tau}_{1}(g),\hat{\zeta}) = 0~,
\]
where $\hat{\zeta}$ is a generic estimate of $\zeta$ based on the data $\{B_i\}_{i=1}^n$.
\end{defn}

We consider two classes of moment functions, differing in whether or not the estimators that they entail resemble IPW or outcome regression estimators for average treatment effects. In the Latent Unconfounded Treatment Model, the random variable
\begin{align}
 g_{\mathsf{w}}(B_i, \tau_1, \zeta_{\mathsf{w}}) &= 
 \frac{G_i}{\pi} \left( \frac{W_i Y_i}{\rho_1(S_i,X_i)} -  \frac{(1-W_i) Y_i}{\rho_0(S_i,X_i)} \right) -  \tau_1~,  \label{eq:lut_weight}
\end{align}
where $\zeta_{\mathsf{w}}=(\pi, \rho_1,\rho_0)$, has mean zero under $P$. The estimator $\hat{\tau}_{1}(g_{\mathsf{w}})$, proposed originally by \cite{athey2020combining}, only requires estimation of the nuisance parameters in $\zeta_{\mathsf{w}}$ and can be viewed as an analogue to the IPW estimator for average treatment effects. In turn, the functions
\begin{align}
g_{\mathsf{or,1}}(B_i, \tau_1,\zeta_{\mathsf{or,1}}) &= \frac{G_i}{\pi} \left(\bar{\mu}_1(X_i) - \bar{\mu}_0(X_i)\right) \quad \text{and} \label{eq:lut_or_1}\\
g_{\mathsf{or,0}}(B_i, \tau_1,\zeta_{\mathsf{or,0}}) &= \frac{1-G_i}{\pi} \frac{\gamma(X_i)}{1-\gamma(X_i)} \left(\mu_1(S_i,X_i) - \mu_0(S_i,X_i)\right)~,\label{eq:lut_or_0}
\end{align}
where $\zeta_{\mathsf{or,1}}$ and $\zeta_{\mathsf{or,0}}$ collect nuisance parameters, yield the outcome regression type estimators $\hat{\tau}_{1}(g_{\mathsf{or,1}})$ and $\hat{\tau}_{1}(g_{\mathsf{or,0}})$. The estimator $\hat{\tau}_{1}(g_{\mathsf{or,0}})$ was originally proposed by \cite{athey2020combining}.

Analogously, in the Statistical Surrogacy Model, the moment function
\begin{align}
 h_{\mathsf{w}}(B_i, \tau_1, \varsigma_{\mathsf{w}}) &= 
\frac{G_i Y_i}{\pi}  
\frac{\gamma(X_i)}{1-\gamma(X_i)} \frac{1-\gamma(S_i, X_i)}{\gamma(S_i, X_i)}\pr{
    \frac{\varrho(S_i, X_i)}{\varrho(X_i)} - \frac{1-\varrho(S_i, X_i)}{1-\varrho(X_i)}
} - \tau_1~,  \label{eq:ss_weight}
\end{align}
where $\varsigma_\mathsf{w}$ collects nuisance parameters, yields an IPW-type estimator $\hat{\tau}_{1}(h_{\mathsf{w}})$ that is similar to an estimator proposed by  \cite{athey2020estimating}. The moment functions 
\begin{align}
h_{\mathsf{or,1}}(B_i, \tau_1,\varsigma_{\mathsf{or,1}}) &= \frac{G_i}{\pi} \left(\bar{\nu}_1(X_i) - \bar{\nu}_0(X_i)\right) \quad \text{and}\label{eq:ss_or_1}\\
h_{\mathsf{or,0}}(B_i, \tau_1,\varsigma_{\mathsf{or,0}}) &= \frac{1-G_i}{\pi} \frac{\gamma(X_i)}{1-\gamma(X_i)} \left(\frac{W_i}{\varrho(X_i)} - \frac{1-W_i}{1-\varrho(X_i)}\right)\nu(S_i,X_i)~,\label{eq:ss_or_0}
\end{align}
result in the outcome regression type estimators $\hat{\tau}_{1}(h_{\mathsf{or,1}})$ and $\hat{\tau}_{1}(h_{\mathsf{or,0}})$, respectively. The estimator $\hat{\tau}_{1}(h_{\mathsf{or,0}})$ was originally proposed by \cite{athey2020estimating}. \cite{dynarski2021closing} use an estimator closely related to $\hat{\tau}_{1}(h_{\mathsf{or,0}})$ with an estimate of $\nu(s,x)$ based on linear regression, in their analysis of the effects of college tuition grants on college complteion rates.

\subsubsection{Large-Sample Theory}

Theoretical analysis of the large-sample performance of estimators based on non-orthogonal moment conditions requires a more specialized treatment. Sufficient conditions for their asymptotic normality in the literature are often more delicate and stronger than those for the DML estimators considered in \cref{sec:orthogonal}. We provide details of this analysis for the estimator $\hat{\tau}_{1}(g_{\mathsf{w}})$ only, as stating and verifying sufficient conditions for asymptotic linearity is cumbersome. Nevertheless, the basic structure of the conditions that we pose, and the method of their verification, is applicable to each of the estimators formulated above. 

Constructing the estimator $\hat{\tau}_{1}(g_{\mathsf{w}})$ requires an estimate of the nuisance parameter $\zeta_{\mathsf{w}}$. We restrict attention to procedures that estimate the remaining components of $\zeta_{\mathsf{w}}$, i.e., $\rho_1$ and $\rho_0$, with linear sieves. We  detail this procedure in \cref{sec: proof of sieve}.

\cref{thm:sieve_estimation} establishes the asymptotic normality of the estimator $\hat{\tau}_{1}(g_{\mathsf{w}})$ when the nuisance parameter $\zeta_{\mathsf{w}}$ is estimated with the method of sieves. Stating the precise sufficient conditions for this Theorem requires additional notation and definitions, which we defer to \cref{sec: proof of sieve}.

\begin{theorem}
\label{thm:sieve_estimation}
Under the Latent Unconfounded Treatment Model and \cref{as:sieve_as_3} stated in \cref{sec: proof of sieve}.
we have that
\[
\sqrt{n} (\hat{\tau}_{1}(g_{\mathsf{w}}) - \tau_1) \overset{d}{\to} \Norm(0, V_1^\star)~,
\]
where $V_1^\star$ is the semiparametric efficiency bound for $\tau_1$ defined in \cref{cor: seb}.
\end{theorem}

We again emphasize that, by \cref{thm: locally-overidentified}, \emph{any} regular and
asymptotically linear estimator $\hat\tau_1$ for $\tau_1$ achieves the semiparametric efficiency bound $V_1^\star$, if the assumptions that define the Latent Unconfounded Treatment Model are the  \emph{only} set of restrictions that are imposed. Thus, the semiparametric efficiency of $\hat{\tau}_{1}(g_{\mathsf{w}})$, per se, is to be expected.  
The
substance of \cref{thm:sieve_estimation} is the asymptotic linearity.

\begin{remark}
\label{rmk:comparison}
Informally, Assumption A.3 requires that:
\begin{enumerate}
    \item $\hat \zeta_{\mathsf{w}}$ is $o(n^{-1/4})$-consistent;
    \item The parameter space containing $\zeta_\mathsf{w}$ is Donsker; and
    \item The sieve space chosen to approximate $\zeta_\mathsf{w}$ has limited complexity. 
\end{enumerate}
\noindent
Condition (1) is analogous to the product rate condition in \cref{as: rates}. It is weaker in the sense that no consistent estimators for the nuisance parameters in $\eta$ that are not in $\zeta_{\mathsf{w}}$ are needed. On the other hand, Condition (1) places more stringent conditions on the rate that $\hat \zeta_{\mathsf{w}}$ estimates $\zeta_{\mathsf{w}}$. In particular, \cref{as: rates} will still hold in situations where some elements of $\zeta_{\mathsf{w}}$ are estimated at a rate slower than $o(n^{-1/4})$, so long as the product of the errors in estimation of the long-term outcome means and propensity scores is smaller than $o(n^{-1/2})$.\footnote{\citet{imbens2004nonparametric} outlines similar heuristics for average treatment effect estimation under unconfoundedness.} 
Condition (2) is imposed in order to ensure stochastic equicontinuity for the moment condition $\zeta_\mathsf{w} \mapsto g_\mathsf{w}(\cdot, \tau_1, \zeta_\mathsf{w})$ treated as a process indexed by $\zeta_\mathsf{w}$. This condition ensures that estimating $\tau_1$ and $\zeta_\mathsf{w}$ using the same data does not induce errors that are excessively large. Using sample-splitting would eliminate the need for this condition. Condition (3) is specific to the sieve approach for estimating nuisance parameters. It is not directly imposed in \cref{as: rates}, but may be needed to justify rate conditions when one estimates nuisance parameters with sieves.\hfill\qed
\end{remark}

\section{Simulation\label{sec:simulation}}

We now compare the estimators formulated in \cref{sec:estimation} with a simulation calibrated to data from \cite{banerjee2015multifaceted}. We find that the DML estimators considered in \cref{sec:orthogonal} are more accurate than the estimators based on non-orthogonal moments considered in \cref{sec: non-orth}, particularly if a nonparametric approach is taken to nuisance parameter estimation. 

\subsection{Data, Calibration, and Design}
\cite{banerjee2015multifaceted}  study randomized evaluations of several similar
poverty-alleviation programs implemented by BRAC, a large non-governmental organization.
These programs allocated productive assets (typically livestock) to participating
households and measured both short-term and long-term economic outcomes.

We restrict our attention to data from the evaluation of the program implemented in
Pakistan. For each of the $854$ households in our cleaned sample, survey measurements
of the consumption levels, food security, assets, savings, and outstanding loans were
taken prior to,  as well as two and three years after, treatment. We use the
pre-treatment measurements as covariates (i.e., $X_i$), the two-year post-treatment
measurements as short-term outcomes (i.e., $S_i$), and the three-year post-treatment
measurements as long-term outcomes (i.e., $Y_i$). There are $20$ pre-treatment
covariates and $21$ short-term outcomes. In the main text, the long-term outcome of
interest is total household assets; we give analogous results for total household
consumption in \cref{subsec:additional}. \cref{sec:data} gives further
information on the construction and content of these data.

We calibrate a generative model to these data with a Generative Adversarial Network 
\citep{goodfellow2014generative}, following a method for simulation design developed
 in \cite{athey2021using}. The details of this calibration are given in \cref{sec:calibration}.\footnote{In \cref{sec:calibration}, we demonstrate that the joint
 distribution of data drawn from this model matches the leading moments of joint
 distribution of the data from \cite{banerjee2015multifaceted} remarkably closely.}
 Crucially, a sample drawn from this model consists of covariates and both treated and
 untreated short-term and long-term potential outcomes for a hypothetical household 
 (the vector $(X_i, S_i(1), S_i(0), Y_i(1), Y_i(0))$ in our notation).
 That is, we observe the true short-term and long-term treatment effects for each
 household sampled from this model, and can measure true long-term average treatment
 effects by averaging over many simulation draws.

With this calibrated model, we generate a collection of hypothetical data sets that
satisfy either the Latent Unconfounded Treatment Model or the Statistical Surrogacy
Model, as desired. The quality of various estimators is then determined by measuring
their average accuracy in recovering long-term average treatment effects in a variety
of metrics. To generate a hypothetical data set, we draw a collection of samples from
the generative mode of size $h\cdot n$, where $h$ is some multiplier that we vary and
$n$ is the sample size of the \cite{banerjee2015multifaceted} data. Each observation is
assigned to being either ``experimental'' or ``observational'' with 
probability $1/2$, and so the experimental and observational samples have identical
distribution of covariates.
Treatment for experimental samples is always assigned uniformly at random. Treatment for
observational samples is assigned with possible confounding. Specifically, we
determine treatment probabilities with an increasing function, indexed by a parameter
$\phi$, of each hypothetical household's true short-term treatment effects. Larger
values of $\phi$ indicate more confounding. Details of this simulation design and
parameterization of confounding are given in \cref{sec:design}.

\subsection{Comparison of Methods}

We begin by comparing the estimators formulated in \cref{sec:estimation} with a simple difference between the mean long-term treated and untreated outcomes in the observational sample
\begin{equation}\label{eq:dm}
\hat{\tau}_{1,DM} 
= \frac{\sum_i G_iW_i Y_i  }{\sum_i G_iW_i} -
\frac{\sum_i G_i(1-W_i) Y_i  }{\sum_i G_i(1-W_i)}~.
\end{equation}
This estimator is naive, making no adjustment for confounding in the observational sample, and is infeasible if treatment is not observed in the observational sample. 

\cref{fig:baseline assets} compares the absolute bias and root mean squared error of the
estimators formulated in \cref{sec:estimation} with the naive estimator
\eqref{eq:dm}.\footnote{Measurements of the variance of each estimator are displayed in
\cref{fig:r3 assets,fig:r3 consumption}. We note that in the Latent Unconfounded Treatment
Model, one of the outcome regression estimators has very small variance when nuisance
parameters are estimated with linear regression. However, the bias of this estimator is
very large.} We compare the use of the Lasso \citep{tibshirani1996regression}, Generalized
Random Forests \citep{athey2019generalized}, and XGBoost \citep{chen2016xgboost} for
nuisance parameter estimation. Details about the implementation of these estimators are
given in \cref{subsec:estimation}. Panels A and B display results for the Latent
Unconfounded Treatment and Statistical Surrogacy Models, respectively. Columns within each
panel vary the sample size multiplier $h$. Each column displays results for the
confounding parameter $\phi$ set to zero, indicating no confounding in the observational
sample and labeled as ``Baseline," in addition to two parameterizations indicating
non-zero confounding.

\begin{figure}
\begin{centering}
\caption{Comparison of Estimators}
\label{fig:baseline assets}
\medskip{}
\begin{tabular}{c}
\textit{Panel A: Latent Unconfounded Treatment}\tabularnewline
\includegraphics[scale=0.32]{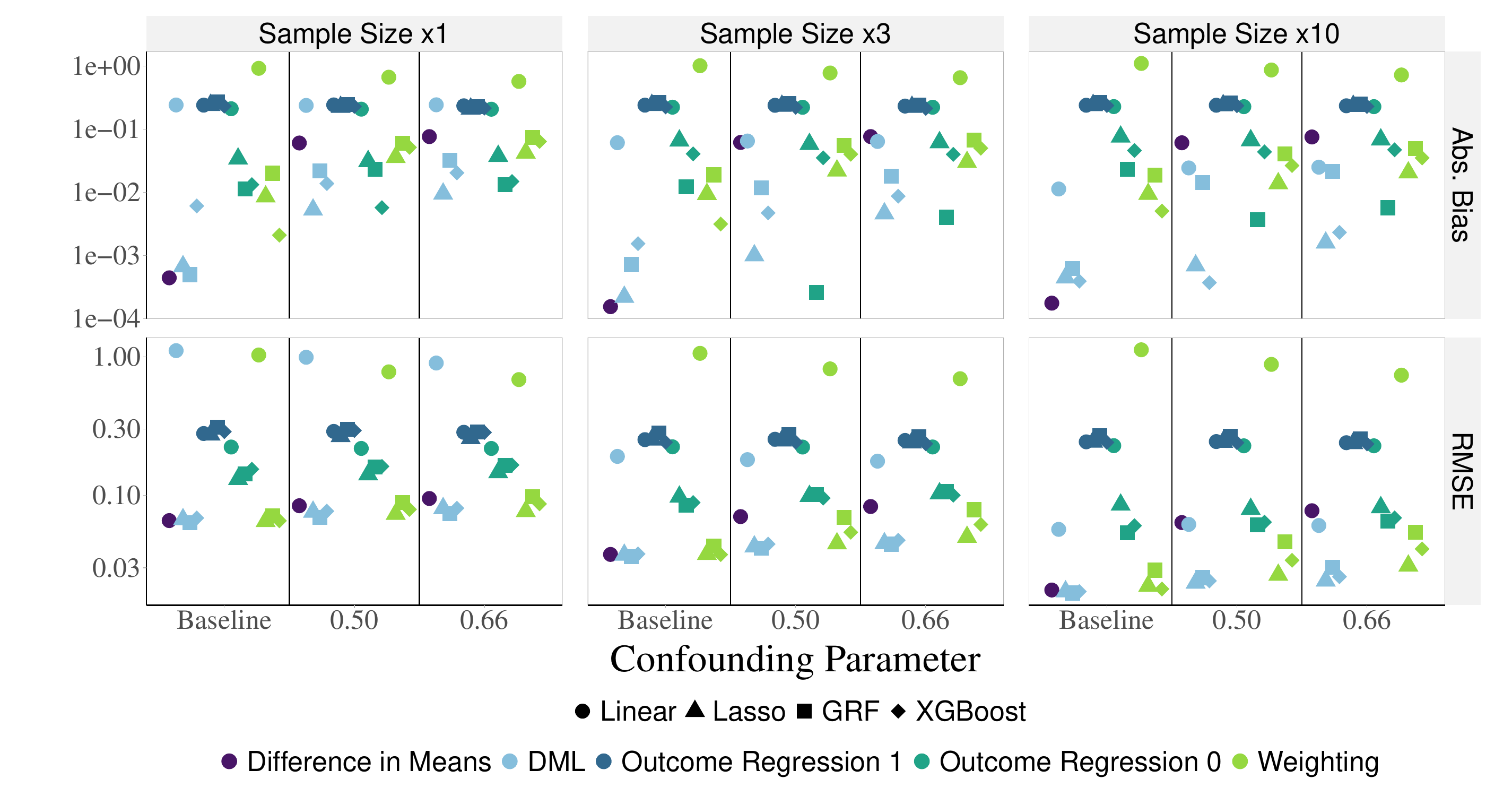}\tabularnewline
\textit{Panel B: Statistical Surrogacy}\tabularnewline
\includegraphics[scale=0.32]{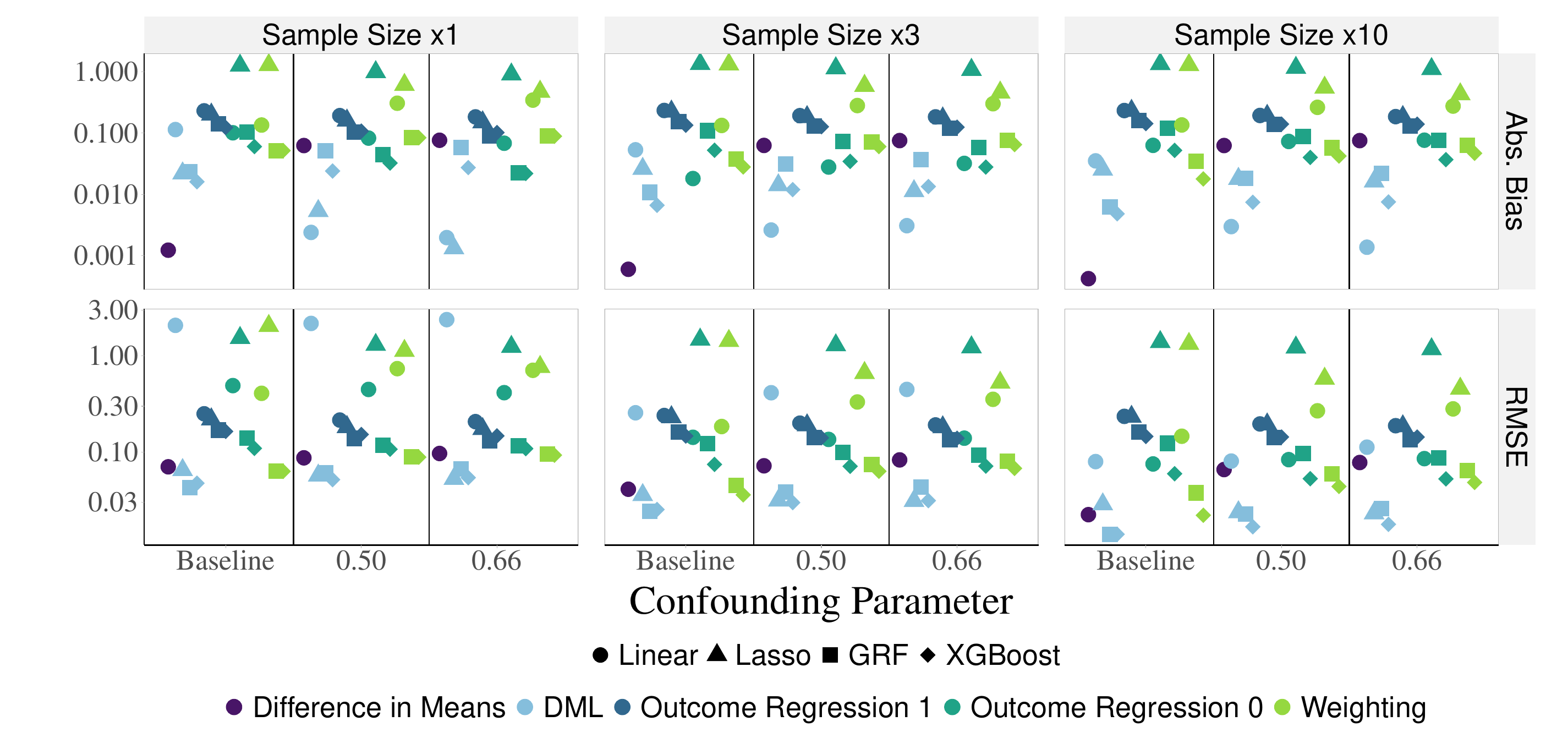}\tabularnewline
\end{tabular}
\par
\end{centering}
\medskip{}
\justifying
{\footnotesize{}Notes: \cref{fig:baseline assets} compares measurements of the absolute
bias and root mean squared error for the estimators formulated in \cref{sec:estimation}, in addition to the 
difference in means estimator defined in \eqref{eq:dm}. The y-axes are displayed in logs, base 10. The long-term outcome is total household
assets. Panels A and B display results for the Latent Unconfounded Treatment and
Statistical Surrogacy Models defined in \cref{def: lut} and \cref{def: sur},
respectively. The columns of each panel vary the sample size multiplier $h$. Each
sub-panel displays results for the baseline, unconfounded, case, as well as for the
cases that the confounding parameter $\phi$ has been set to $1/2$ and $2/3$. Results for each 
estimator are displayed with dots of different colors. Results for different nuisance parameter estimators
are displayed with dots of different shapes.}{\footnotesize\par}
\end{figure}

In both the Latent Unconfounded Treatment Model and the Statistical Surrogacy Model,
the biases of the DML estimators considered in \cref{sec:orthogonal} tend to be substantially smaller than the biases of the alternative outcome regression or weighting
estimators considered in \cref{sec: non-orth}, so long as the nuisance parameters are not estimated by linear regression. The weighting estimator has performance more comparable to the DML estimator in the Latent Unconfounded Treatment Model. The performances of the weighting estimator and the outcome regression estimators in the Statistical Surrogacy Model are more similar.

\cref{fig:results assets} displays measurements of estimator quality for just the DML estimators in a format
 analogous to \cref{fig:baseline assets}, with the addition of a third row measuring
 one minus the coverage probability of confidence intervals constructed around each
 estimator. We report these estimates of coverage probabilities in 
 Supplementary \cref{subsec:additional}.  Confidence intervals constructed with the intervals given in \eqref{eq: ci} around the
 estimators formulated in \cref{def: dml} have coverage probabilities that are reasonably close
 to the nominal level.

\begin{figure}
\begin{centering}
\caption{Finite-Sample Performance with Different Nuisance Parameter Estimators}
\label{fig:results assets}
\medskip{}
\begin{tabular}{c}
\textit{Panel A: Latent Unconfounded Treatment}\tabularnewline
\includegraphics[scale=0.32]{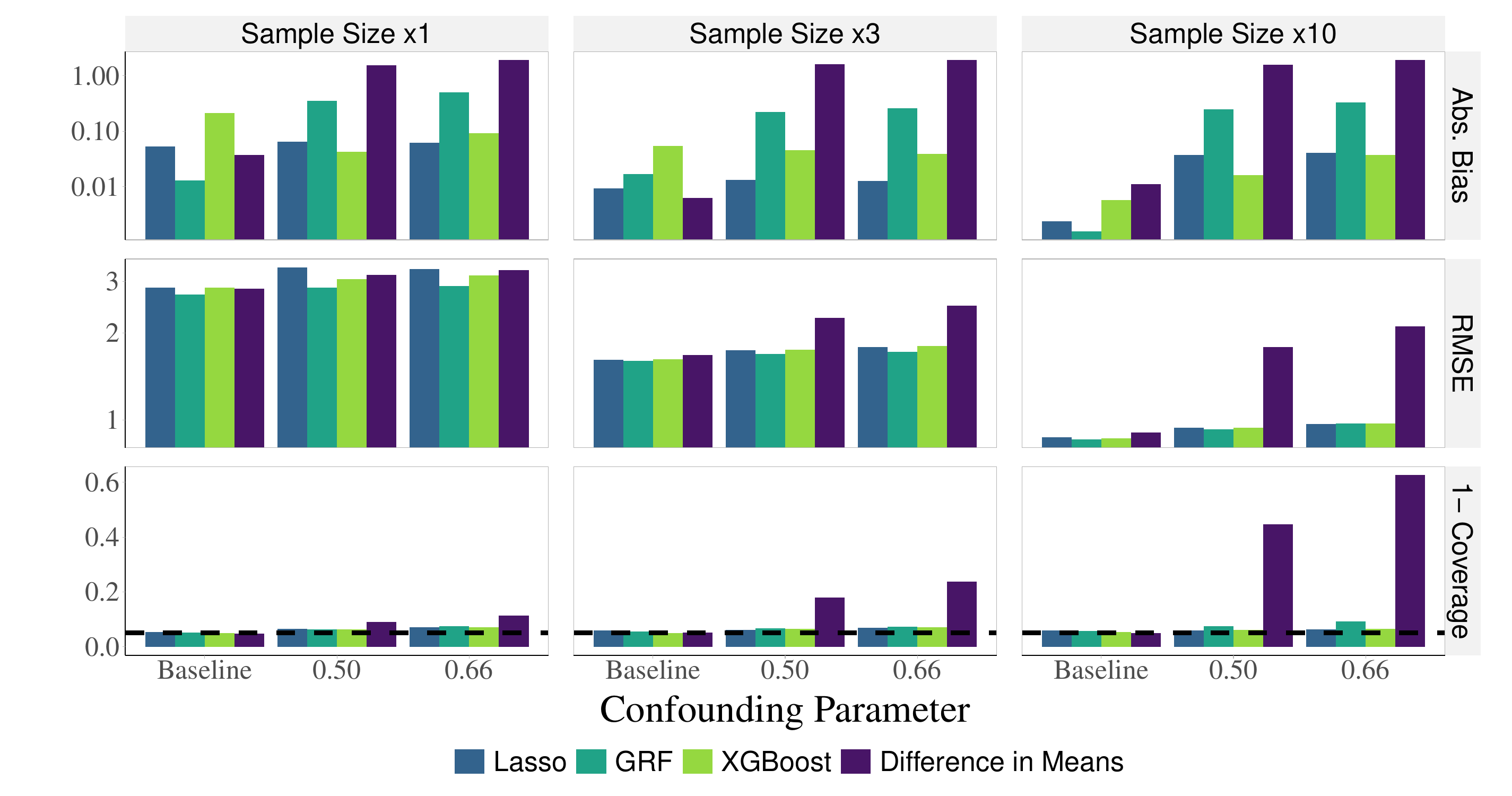}\tabularnewline
\textit{Panel B: Statistical Surrogacy}\tabularnewline
\includegraphics[scale=0.32]{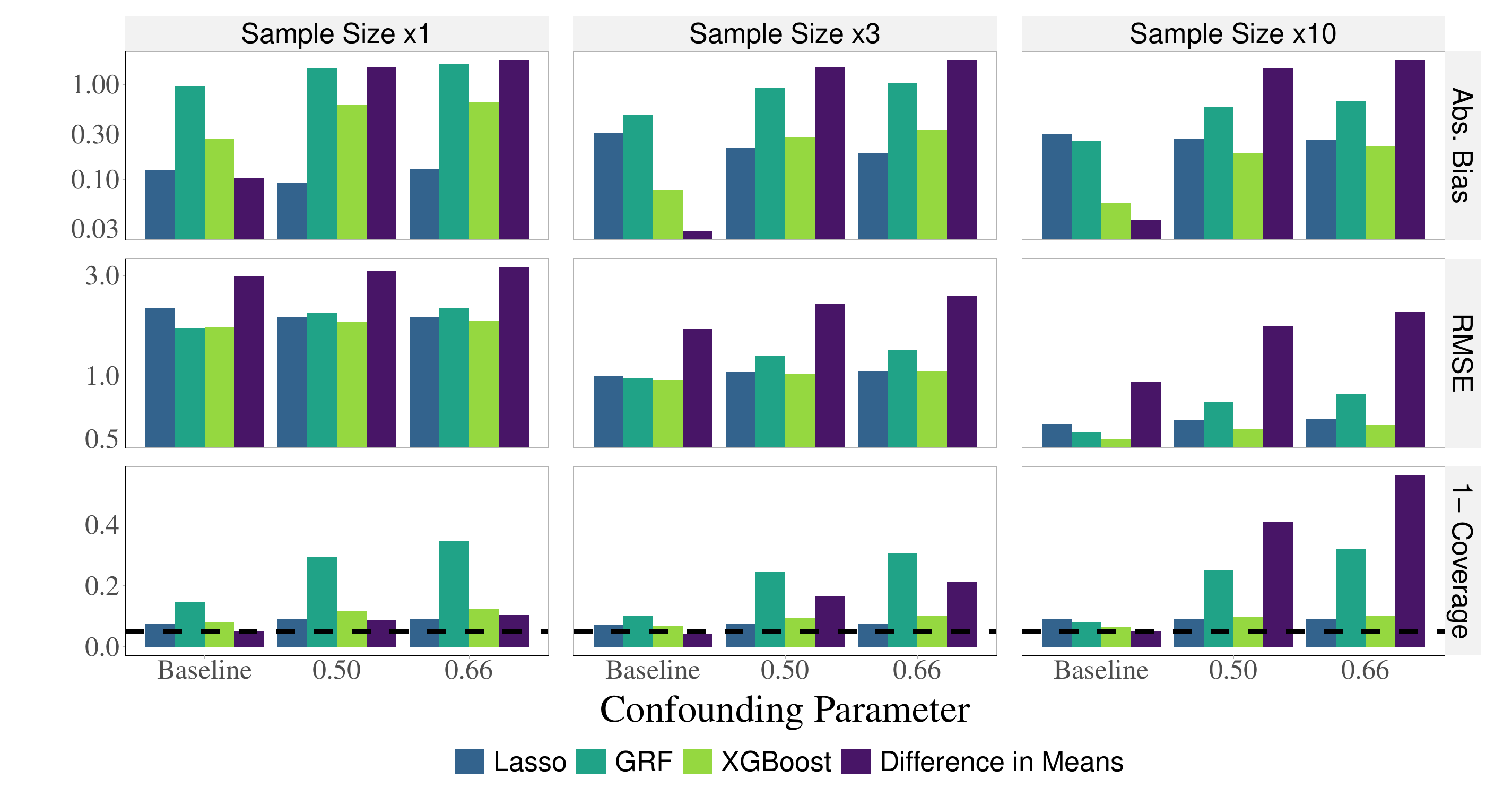}\tabularnewline
\end{tabular}
\par
\end{centering}
\medskip{}
\justifying{\footnotesize{}Notes: \cref{fig:results assets} displays measurements of the
 quality of the estimators formulated in \cref{sec:orthogonal} implemented with several
 alternative choices of nuisance parameter estimators. The long-term outcome is total
 household assets. Panels A and B display results for the estimators defined in \cref
 {def: dml} for the Latent Unconfounded Treatment and Statistical Surrogacy Models
 defined in \cref{def: lut} and \cref{def: sur}, respectively. The columns of each
 panel vary the sample size multiplier $h$. The rows of each panel display the absolute
 value of the bias, one minus the coverage probability, and the root mean squared
 error of each estimator, from top to bottom, respectively. A dotted line denoting one
 minus the nominal coverage probability, 0.05, is displayed in each sub-panel in the
third row. Each sub-panel displays a bar graph comparing measurements of the
 performance of the estimator defined in \cref{def: dml}, constructed with three types
 of nuisance parameter estimators, with the difference in means estimator \eqref
 {eq:dm}.  Each sub-panel displays results for the baseline, unconfounded, case, as
 well as for the cases that the confounding parameter $\phi$ has been set to $1/2$ and
 $2/3$.}{\footnotesize\par}
\end{figure}

\section{Conclusion\label{sec: conclusion}} We study the estimation of long-term
 treatment effects through the combination of short-term experimental and long-term
 observational data sets. We derive efficient influence functions and calculate
 corresponding semiparametric efficiency bounds for this problem. These calculations
 facilitate the development of estimators that accommodate
 the applications  of standard machine learning algorithms for estimating nuisance parameters. We demonstrate with simulation that these estimators are able to recover 
 long-term treatment effects in realistic settings.

Important unresolved practical issues remain. Methods for choosing
valid and informative short-term outcomes and assessing of the
sensitivity of estimates to violations of identifying assumptions would be valuable.  Additional
useful extensions include the incorporation of instruments and continuous treatments
and the accommodation of settings with limited covariate overlap. 

In the case of estimation of average treatment effects under unconfoundedness, recent promising work \citep[e.g., ][]{athey2018approximate, bradic2019sparsity, tan2020model} has developed estimators that are specifically optimized to handle high-dimensional covariates and are able to attain various notions of optimality under weak assumptions on, e.g., the sparsity of the outcome regression or propensity score models. It is not immediately clear how to apply these ideas to long-term average treatment effects. Further consideration of this problem would be a potentially valuable extension, as the resultant estimators may be particularly well-suited to the case where there are many short-term outcomes. Some progress on a related problem (where, effectively, there is a single short-term outcome) has been made by \cite{viviano2021dynamic}.

\end{spacing}

\newpage 
\begin{spacing}{1.172}
\bibliographystyle{apalike}
\bibliography{long-term.bib}
\end{spacing}

\newpage

\begin{appendix}
\renewcommand\thefigure{\thesection.\arabic{figure}}    

\begin{center}
\large{\it Supplemental Appendix to:}
\vskip0.2cm
\begin{spacing}{1}
\Large{\textbf{Semiparametric Estimation of \\Long-Term Treatment Effects\protect\daggerfootnote{\textit{Date}: \today}}}
\begin{tabular}[t]{c@{\extracolsep{2em}}c} 
\large{Jiafeng Chen} &  \large{David M. Ritzwoller}\\ \vspace{-0.7em}
\normalsize{\it Harvard Business School} & \normalsize{\it Stanford Graduate School of Business} \\ \vspace{-0.7em}
\normalsize{\href{mailto:jiafengchen@g.harvard.edu}{jiafengchen@g.harvard.edu}} & \normalsize{\href{mailto: ritzwoll@stanford.edu}{ritzwoll@stanford.edu}}
\end{tabular}
\end{spacing}
\end{center}
\begin{spacing}{1.13}
\DoToC
\end{spacing}
\thispagestyle{empty}
\setcounter{page}{0}
\newpage

\begin{spacing}{1.4}
\normalsize 
\section{Proofs for Results Presented in the Main Text\label{sec: proofs}} 

\subsection{Proof of Proposition \ref{thm: identification lut}} 

Unless otherwise noted, we drop the dependence on the individual $i$. It will suffice to consider the parameter 
\[
\theta_{1,1} = \mathbb{E}_{P_\star}[Y(1) \mid G = 1]~,
\]
as the point-identification of $\theta_{1,0} = \mathbb{E}[Y(0) \mid G = 1]$ will follow by an analogous argument.

We consider Latent Unconfounded Treatment Model, specified in \cref{def: lut}, and repeat the argument provided in support of Theorem 1 of \cite{athey2020combining}. Define the functions
\begin{align*}
\mu_1(s,x) & = \mathbb{E}_P[Y \mid S = s, X = x, W = 1, G = 1] \\
& = \mathbb{E}_{P_\star}[Y(1) \mid S(1)=s, X=x, W = 1, G = 1]
\end{align*}
and
\begin{align*}
\bar{\mu}_1(x) & = \mathbb{E}_P[\mu_1(S,X) \mid X = x, W = 1, G = 0]\\
& = \mathbb{E}_{P_\star}[\mu_1(S(1),X) \mid X = x, W = 1, G = 0] 
\end{align*}
and observe that, by \cref{as: overlap}, $\mu_1(s,x)$ is identified from the observational sample and $\bar{\mu}_1(x)$ is identified from the experimental sample given the identification of $\mu_1(s,x)$. Observe that
\begin{align*}
& \mathbb{E}_{P_\star}[Y(1) \mid G = 1] \\ & = \mathbb{E}_{P_\star}[ \mathbb{E}_{P_\star}[Y(1) \mid X, G = 1] \mid G = 1] \\
& =  \mathbb{E}_{P_\star}[ \mathbb{E}_{P_\star}[Y(1) \mid X, G = 0] \mid G = 1] \tag{\cref{as: ceev}}\\
& = \mathbb{E}_{P_\star}[ \mathbb{E}_{P_\star}[ \mathbb{E}_{P_\star}[ Y(1) \mid S(1), X, G = 1]\mid X, G = 0] \mid G = 1] \tag{\cref{as: ceev}}\\
& = \mathbb{E}_{P_\star}[ \mathbb{E}_{P_\star}[ \mathbb{E}_{P_\star}[ Y(1) \mid S(1), X, W=1, G = 1]\mid X, G = 0] \mid G = 1] \tag{\cref{as: luot}}\\
& = \mathbb{E}_{P_\star}[ \mathbb{E}_{P_\star}[\mu_1(S(1),X)]\mid X, G = 0] \mid G = 1]  \\
& = \mathbb{E}_{P_\star}[ \mathbb{E}_{P_\star}[ \mu_1(S(1),X)]\mid X, W=1, G = 0] \mid G = 1] \tag{\cref{as: uex}} \\
& = \mathbb{E}_{P_\star}[\bar{\mu}_1(X) \mid G = 1] ~.
\end{align*}
Hence, $\tau_1$ is identified as $\bar{\mu}_1(x)$ is identified. 

\subsection{Proof of Proposition \ref{thm: identification sur}} 

We consider the Statistical Surrogacy Model, specified in \cref{def: sur}. Define the function
\[
\nu(s, x) = \mathbb{E}_{P}[ Y \mid S=s, X=x, G = 1] 
\]
and observe that $\nu(s, x)$ is identified from the observational sample by \cref{as: overlap}. Observe that
\begin{align*}
\mathbb{E}_{P_\star}[Y(1) \mid G = 1] & = \mathbb{E}_{P_\star}[ \mathbb{E}_{P_\star}[ Y(1) \mid X, G = 1] \mid G = 1] \\
& = \mathbb{E}_{P_\star}[ \mathbb{E}_{P_\star}[ Y(1) \mid X, G = 0] \mid G = 1] \tag{\cref{as: ceev}}\\
& = \mathbb{E}_{P_\star}[ \mathbb{E}_{P_\star}[ Y(1) \mid X, W =1, G = 0] \mid G = 1] \tag{\cref{as: uex}}~,
\end{align*}
and that 
\begin{align*}
& \mathbb{E}_{P_\star}[ Y(1) \mid X, W =1, G = 0] \\ & = \mathbb{E}_{P}[ Y \mid X, W =1, G = 0] \\
& = \mathbb{E}_{P}[ \mathbb{E}_{P}[ Y \mid S, X, W =1, G = 0] \mid X, W=1, G=0 ] \\
& = \mathbb{E}_{P}[ \mathbb{E}_{P}[ Y \mid S, X, G = 0] \mid X, W=1, G=0 ] \tag{\cref{as: es}} \\
& = \mathbb{E}_{P}[ \mathbb{E}_{P}[ Y \mid S, X, G = 1] \mid X, W=1, G=0 ] \tag{\cref{as: ltoc}} \\
& = \mathbb{E}_{P}[ \nu(S, X) \mid X, W=1, G=0 ]~.
\end{align*}
It is clear that $\mathbb{E}_{P}[ \nu(S, X) \mid X, W=1, G=0 ]$ is identified from the experimental sample by \cref{as: overlap}, and that therefore $\tau_1$ is identified. \hfill\qed

\subsection{Proof of Theorem \ref{thm: EIF tau_1}. Part 1} 
\label{asub:eif_derivation} 

Throughout, for a general functional $\varphi$ on $\mathcal{M}_\lambda$ and an arbitrary random
variable $D$ with distribution $Q$ in $\mathcal{M}_\lambda$, we write
$
\varphi(Q)=\varphi(D)
$
with a minor abuse of notation.We let 
\[
\varphi^\prime(D) = \varphi^\prime(D; 0) = \frac{d}{d\varepsilon} \varphi
(Q_\varepsilon) {\small\evalbar}_{\varepsilon=0}
\]
denote the pathwise derivative of $\varphi(\cdot)$ on some regular parametric submodel $\mathcal{Q} = \{Q_\varepsilon : \varepsilon \in
[-1,1]\}$ of $\mathcal{M}_\lambda$ evaluated at $\varepsilon=0$.

It suffices to derive the efficient influence function for
 the parameter $\theta_{1,1}$, defined more generally by $\theta_{w,g}=\mathbb{E}[Y_i
 (w)\mid G_i = g]$, as an analogous argument will hold for $\theta_{0,1}$. Unless
 otherwise noted, we drop the dependence on the individual $i$. We omit the subscript
 $w$ on the nuisance functions $\mu_w(s,x)$, $\bar{\mu}_w(x)$, and $\varrho_w(s,x)$, as
 only the case $W = 1$ will be relevant. 

 The presentation of the proof will be somewhat
 constructive, as we hope to convey the rationale behind the form of the efficient
 influence function, for which we have had to take a few educated guesses in forming a
 conjecture. 
 Nonetheless, we state that the efficient influence function $\psi_{1,1}(b)$ takes the form \[
\psi_{1,1}(b) = \frac{g w (y-\mu(s,x))}{\pi \rho(s,x)} +\frac{(1-g)w}{(1-\gamma(x)) \varrho(x)} \frac{\gamma(x)}{\pi} (\mu(s,x) - \bar\mu(x)) + \frac{g}{\pi} (\bar \mu(x) - \theta_{1,1})
 \]
The reader can of
course verify \cref{thm: EIF tau_1} directly by using this stated form of the
influence function and checking a series of conditions that we lay out.

The argument proceeds as follows. First, we factorize the density of the
observed data and relate each factor to conditional densities of the complete
data distribution. Second, for an arbitrary smooth parametric submodel, we
characterize the tangent space of the distribution of the observed data as a linear space of
mean-zero square-integrable functions. Third, we conjecture a functional
form that explicitly resides in the tangent space, where we note that the efficient
influence function is the unique influence function for the target statistical functional $\theta_{1,1}$ in the tangent space of $\theta_{1,1}$ \citep[Chapter 25 of][]{van2000asymptotic}. Fourth, we state the conditions an influence function of the pathwise derivative of $\theta_{1,1}$ necessarily satisfies, and we verify that the conjectured efficient influence function
satisfies these conditions.

\noindent
\textbf{Density Factorization.}
Consider the random variable
\[
A_1 = (Y(1), S(1), W, G, X) \equiv (Y^*, S^*, W, G, X)
\]
where we observe $B_1 = (WGY^*, WS^*, W, G, X)$. Let $Y = WGY^*$ and $S = WS^*$.
Recall that the distribution of the complete data is denoted by $P_\star$, with
density $p_\star$ with respect to $\lambda$, and that the distribution of the observed
data is denoted by $P$, with density $p$ with respect to $\lambda$. Observe that the
densities under the complete data distribution and the observed data
distribution agree when data isn't missing. That is, for all
$(y,s)\in\mathbb{R}\times\mathbb{R}^d$ and events $E\in\mathcal A$, we have
that
\[
p_\star(y \mid W = 1, G=1, E)  = p(y \mid W=1, G=1, E)
\]
and 
\[
p_\star(s \mid W = 1, E) = p(s \mid W=1, E)
\]
and that for all events $E$ measurable with respect to the $\sigma$-algebra generated by $(W,G,X)$, we have that $P(E) = P_\star(E)$.

By \cref{as: uex,as: ceev,as: luot}, the density of the observed data $p$ admits the factorization
\begin{align*}
p(b_1) =\,\,  &
p(y \mid G=1, X=x, W=1, S=s)^{wg} p(s, W=1 \mid G=1, X=x)^{wg} \\&p(s \mid
G=0, X=x, W=1)^{w(1-g)}\\
& P(W=w \mid G=0, X=x)^{1-g} P(W=0 \mid G=1, X=x)^{(1-w)g}  P(G=g \mid X=x) p(x)~.
\end{align*}
Each of the terms of this expression can be rewritten in terms of the density of the complete data $p_\star$. In particular, we have that the first term can be written
\begin{align}
    p(y\mid G=1, X=x, W=1, S=s) &= p_\star(y \mid G=1, X=x, W=1, S=s) \nonumber \\ 
    &= p_\star(y \mid G=1, X=x, S=s) \tag{\cref{as: luot}} \\ 
    &= p_\star(y \mid x,s)~, \tag{\cref{as: ceev}, noting that $p_\star$ is the measure of the potential outcomes. }
\end{align}
the second and third terms can be written
\begin{align*}
p(s, W=1 \mid G=1, X=x) &= p_\star(S^*=s, W=1 \mid G=1, X=x) \\ 
&= p_\star(W=1 \mid S^* = s, G=1, X=x) p_\star(S^* = s \mid G=1, X=x)  \\
&= \rho(s, x) p_\star(s \mid x) \tag{\cref{as: ceev}}
\end{align*}
and 
\begin{align*}
p(s \mid G=0, X=x, W=1) &= p_\star(S^* = s \mid G=0, X=x) \tag{\cref{as: uex}} \\ 
&= p_\star(s \mid x)~, \tag{\cref{as: ceev}}
\end{align*}
the fourth term can be written
\begin{align*}
P(W=w \mid G=0, X=x) &= \varrho(x)^w (1-\varrho(x))^{1-w}~, \tag{\cref{as: uex}}
\end{align*}
and finally the fifth term can be written 
\begin{align*}
P(W=0 \mid G=1, X=x) &= P_\star(W=0 \mid G=1, X=x) \\
    &= \int P_\star(W=0 \mid G=1, X=x, S=s) \, dP_\star(s\mid x)\\
    &= \int (1-\rho(s,x))p_\star(s\mid x)\, ds.
\end{align*}
As a result, we obtain a factorization of the density of the observed data in terms of conditional distributions
of the complete data
\begin{align}
p(b_1) =\,\,
& p_\star(y\mid x,s)^{wg} p_\star(s\mid x)^{w} \rho(s,x)^{wg}
 \pr{\int (1-\rho(s,x))p_\star(s\mid x)\,ds}^{(1-w)g} \nonumber \\ 
& \gamma(x)^g (1-\gamma(x))^{1-g} (\varrho(x)^w(1-\varrho(x))^{1-w})^{1-g} p(x)~. \label{eq:factorization}
\end{align}

\noindent
\textbf{Characterization of the Tangent Space.} 
Let $\mathcal{P}$ be any regular parametric submodel of $\mathcal{M}_\lambda$ indexed by $\varepsilon \in\R$, with densities 
$
p_\varepsilon  := dP_\varepsilon/d\lambda
$
with $p_0=dP/d\lambda$ and such that \cref{as: uex,as: ceev,as: luot} hold for each
$P_\varepsilon  \in \mathcal{P}$. Let the subscript $_\star$ distinguish log densities of the complete data from 
log densities of the observed data. By the factorization \eqref{eq:factorization}, we obtain the score
\begin{align}
l^\prime(b_1) 
& = wg \cdot l'_*(y\mid x,s) + w \cdot l'_*(s\mid x ) 
+ l'_*(g, x ) \nonumber \\
& - g(1-w) 
\left(\frac{\E_{P_\star}[\rho(S^*,X)l'_*(S^*\mid X )  \mid G=0, W=1, X=x]}{1-\E_{P_\star}[\rho(S^*,X) \mid G=0, W=1, X=x]} \right) \nonumber \\
& + g\left(w\cdot \frac{\rho^\prime(s,x)}{\rho(s,x)} 
- (1-w)\frac{\E_{P_\star}[\rho^\prime(S^*,X) \mid G=0, W=1, X=x]}{1-\E_{P_\star}[\rho(S^*,X) \mid G=0, W=1, X=x]}\right) \nonumber \\
& + (1-g)\varrho'(x)\left(\frac{w-\varrho(x)}{\varrho(x)(1-\varrho(x))} \right)~. \label{eq: theta_1 score}
\end{align}
Thus, the tangent space $\mathcal{T}$ is given by mean-square closure of
the linear span of the functions
\begin{align*}
s(b_1) 
& = wg \cdot s_1(y\mid x,s) + w \cdot s_2(s\mid x) 
+ s_3(g, x) \\
& - g(1-w) 
\left(\frac{\E_{P_\star}[\rho(S^*,X)s_2(S^* \mid X) \mid G=0, W=1, X=x]}{1-\E_{P_\star}[\rho(S^*,X) \mid G=0, W=1, X=x]} \right)\\\
& + g\left(w\cdot \frac{s_4(s,x)}{\rho(s,x)} 
- (1-w)\frac{\E_{P_\star}[s_4(S^*,X) \mid G=0, W=1, X=x]}{1-\E_{P_\star}[\rho(S^*,X) \mid G=0, W=1, X=x]}\right) \\
& + (1-g)\left(s_5(x)\frac{w-\varrho(x)}{\varrho(x)(1-\varrho(x))} \right) ~,
\end{align*} 
where the functions $s_1$ through $s_5$ range over the space of mean-zero and square integrable functions that satisfy the restrictions
\begin{align*}
\E_P[s_1(Y \mid S, X) \mid W=1, G=1, S, X] = \E_{P_\star}[s_1(Y^* \mid S^*, X) \mid S^*=s, X=x]  &= 0 ~,\\
\E_P[s_2(S \mid X) \mid W=1, G=0, X] = \E_{P_\star}[s_2(S^* \mid X) \mid X=x] &= 0 ~\text{, and}\\
\E_P[s_3(G,X)] &= 0~.
\end{align*}

\noindent
\textbf{Pathwise Differentiability Conditions.} 
In the following calculations, it is often easier to work with the complete
data distribution. To that end, we note that for an arbitrary measurable
function $h$, we have that certain conditional means of the observed
distribution equal certain conditional means of the complete data
distribution:
\begin{align*}
\E_P[h(Y, S, X) \mid W = 1, G=1, X=x, S=s] &= \E_{P_\star}[h(Y^*,S,X) \mid G=1, X=x,
S^*=s]~,
\\
\E_P[h(S,X) \mid G=0, W=1, X=x] &= \E_{P_\star}[h(S^*,X) \mid G=0, X=x ] \\&=
\E_{P_\star}[h
(S^*,X)
\mid G=1, X=x]~, \\ 
\E_P[h(X) \mid G=1] &= \E_{P_\star}\bk{h(X) \mid G=1}~,
\end{align*}
and, as a result, for $h$,
\begin{align*}
&\E_P[\E_P[\E_P[h(Y,S,X) \mid W=1, G=1, X, S] \mid W=1, G=0,
X] \mid G=1] \\
& \quad = \E_{P_\star}\bk{h(Y^*,S^*,X) \mid G=1}  = \E_{P_\star}\bk{\frac{G}
{\pi}
h
(Y^*,S^*,X)}~.
\end{align*}
Note that $\theta_{1,1}$ is such a functional with $h(y,s,x) = y$. 

Observe that the parameter of interest can be written
\begin{align*}
\theta_{1,1}
& = \E_{P_\star}[Y^* \mid G=1] \\
&= \E_P[\E_{P}[\E_{P}[Y \mid W=1, G=1, X, S] \mid W=1, G=0, X] \mid G=1] \\
& = \iiint \frac{y}{\pi}
p_\star(y \mid s,x) p_\star(s \mid x) p(x , G = 1)\,\text{d}\lambda(y)
\text{d}\lambda(s) \text{d}\lambda(x)~.
\end{align*}
The pathwise derivative of this parameter at 0 on $\mathcal{P}$ is given by
\begin{align}
\label{eq: theta_1 prime}
\theta_{1,1}^\prime
& = \E_{P_\star}[Y^* l'(Y^* \mid S^*,X) \mid G = 1]\nonumber\\ 
&+ \E_{P_\star}[Y^* l'(S^* \mid X) \mid G = 1]
+ \E_{P_\star}[Y^* l'(X , G = 1) \mid G = 1]~.
\end{align} 
Now, observe that 
\begin{align*}
G \cdot l'_*(X, G = 1) & = G\cdot l'_*(X, G) - G\cdot l'_*(G=1) ~,
\end{align*} 
and that 
\begin{align*}
l'_*(G=1) & = \E_{P_\star}[ l'_*(X, G) \mid G = 1]~.
\end{align*} 
Thus, each of the terms in \eqref{eq: theta_1 prime} can be written in
terms of conditional scores by
\begin{align*}
\E_{P_\star}[Y^* l'(Y^* \mid S^*, X) \mid G = 1] & =  \E_{P_\star}[\pi^{-1}GY^*l'_*(Y^*\mid S^*, X)]~, \\
\E_{P_\star}[Y^* l'(S^* \mid X) \mid G = 1]  & = \E_{P_\star}[\pi^{-1}GY^* l'_*(S^* \mid X)] ~\text{, and} \\
\E_{P_\star}[Y^* l'(X, G = 1) \mid G = 1] & = \E_{P_\star}[\pi^{-1}GY^* l'_*(X, G = 1)] \\
&= \E_{P_\star}[\pi^{-1}GY^* l'_*(X, G)] - \E_{P_\star}[\pi^{-1}GY^*] l'_*(G=1)\\
&=\E_{P_\star}[\pi^{-1}GY^* l'_*(X,G)] - \theta_{1,1} \cdot \E_{P_\star}[l'_*(X,G) \mid G=1]~,
\end{align*} 
respectively. 

An influence function for $\theta_{1,1}$ is a mean-zero and square-integrable function $\tilde{\psi}_{1,1}(B_1)$ that satisfies the condition
\begin{equation}
\label{eq: pathwise deriv}
\theta_{1,1}^\prime = \mathbb{E}_P[\tilde{\psi}_{1,1}(B_1)l^{\prime}(B_1)]~. 
\end{equation}
Hence, by (\ref{eq: theta_1 score}) and (\ref{eq: theta_1 prime}), in order
to establish that a mean-zero and square-integrable function $\tilde{\psi}_{1,1}
(B_1)$ is an influence function for $\theta_{1,1}$, it suffices to show that
the terms in the right-hand side of \eqref{eq: pathwise deriv} match the
terms in the right-hand side of \eqref{eq: theta_1 prime}:
\begin{align}
    \frac{1}{\pi}\E_{P_\star}[GY^* l'_*(Y^*\mid S^*, X)]
    &= \E_{P}\bk{\tilde \psi_{1,1}(B_1) GW l'_*(Y\mid S, X)}~, 
     \label{eq:eta1} \\ 
    \frac{1}{\pi}\E_{P_\star}[GY^* l'_*(S^* \mid X)]
    &= \E_{P}\Bigg[
        \tilde \psi_{1,1}(B_1) \Bigg( W l'_*(S \mid X) \nonumber \\
    & \quad\quad\quad\quad
     - G(1-W) \frac{\int \rho(s,X) p'_*(s \mid X)\,\text{d}\lambda(s)}{1-\int \rho(s,X)p_\star(s\mid X)\,\text{d}\lambda(s)}\Bigg)\Bigg]~, 
    \label{eq:eta2}\\ 
    \E_{P}[\tilde \psi_{1,1}(B_1)  l'_*(G, X)] &= \frac{1}{\pi}\E_{P_\star}
    [GY^* l'_*(G, X) ] - \theta_{1,1}\cdot \E_{P_\star}[l'_*(G,X) \mid G=1]~, 
     \label{eq:eta3}\\ 
    0 &= \E_{P}\Bigg[\tilde \psi_{1,1}(B_1) G\Bigg(
    W \frac{\rho^\prime(S,X)}{\rho(S,X)} \nonumber \\
    & \quad\quad\quad\quad
    - (1-W) \frac{\int \rho^\prime(s,X) p_\star(s\mid X) \text{d}\lambda(s)}{1-\int \rho
    (s,X)p_\star(s\mid X)\,\text{d}\lambda(s)}\Bigg) \Bigg]~, 
    \label{eq:eta4}\\ 
    0 &= \E_{P}\bk{\tilde \psi_{1,1}(B_1) (1-G)\frac{W - \varrho(X)}{\varrho(X)(1-\varrho(X))} \varrho'(X)}~.
   \label{eq:eta5}
\end{align}

\noindent
\textbf{Conjectured Efficient Influence Function.} 
We begin by considering an element of the tangent space. For the unspecified functions $a(s,x)$ and $b(x)$, consider the choices
\begin{align}
s_1(y,s,x) &= \frac{y-\mu(s,x)}{\pi \rho(s,x)}~,  &
s_2(s,x) &= \frac{1}{\varrho(x)} \frac{a(s,x)}{(1-\gamma(x))} ~, \nonumber\\
s_3(g,x) &= g\cdot b(x), \nonumber &
s_4(s,x) &= -\rho(s,x)s_2(s,x)~\text{, and} \nonumber\\
s_5(x) &= 0~. \label{eq:aci_score_choices_eif}
\end{align}
These choices correspond to the following
conjectured efficient influence function
\begin{align}
\psi_{1,1}(b_1) &= g w \cdot s_1(y,s,x) + w \cdot s_2(s,x) - gw \cdot s_2(s,x) +
gb(x) \nonumber  \\
&= gw \cdot s_1(y,s,x) + (1-g)w \cdot s_2(s,x) + g \cdot b(x)  \nonumber \\
&= \frac{gw(y-\mu(s,x))}{\pi \rho(s,x)}  + \frac{(1-g) w}{(1-\gamma
(x)) \varrho(x)} a(s,x) + g\cdot b(x)~.\label{eq:conjecturedeif}
\end{align}
In particular, $s_4$ is chosen so that certain terms cancel to construct
the term
$(1-g)w s_2(s,x)$, which does not, prima facie, belong in the tangent
space \eqref{eq: theta_1 score}.

In order to ensure that $\psi_{1,1}(b_1)$ belongs to the tangent space, we
require \begin{align}
\E_P[a(S,X) \mid W=1, G=0, X=x] &= \E_{P_\star}[a(S^*, X) \mid X=x] = 0 ~\text{, and}
\label{eq:condphi}\\
\E_P[G\cdot b(X)]& = 0~.
\label{eq:condvarphi}
\end{align}
In the sequel, we verify that the choices 
\begin{align*}
a(s,x) &= \frac{\gamma(x)}{\pi} (\mu(s,x) - \bar{\mu}(x)) ~\text{, and} \\
b(x) &= \frac{\bar{\mu}(x) - \theta_{1,1}}{\pi}
\end{align*}
satisfy these conditions, and, moreover, give rise to the efficient influence function by satisfying 
the conditions \eqref{eq:eta1}--\eqref{eq:eta5}.

\noindent
\textbf{Verifying Pathwise Differentiability Conditions.}
We begin by stating and proving a useful lemma for handling terms of the form $GWh(Y,S,X)$. 
\begin{lemma}
\label{lemma:get_rid_of_wgy}
    Under the assumptions of \cref{thm: EIF tau_1}, if $h(y,s,x)$ is any
    measurable function, then 
    \[
    \E_P\bk{\frac{GW}{\rho(S,X)} h(Y,S,X)} = \E_{P_\star}\bk{G h(Y^*,S^*,X)} 
    = \E_{P_\star}\bk{G \E[h(Y^*,S^*,X) \mid S^*, X]}~.
    \]
\end{lemma}
\begin{proof}
    Observe that 
        \begin{align*}
         \E_P\bk{\frac{GW}{\rho(S,X)} h(Y,S,X)}
         & = \E_P\bk{
            \pi \frac{W}{\rho(S,X)} h(Y,S,X)\mid G=1
         } \\ 
         &=  \pi  \E_{P_\star}\bk{
           \frac{W}{\rho(S^*,X)} h(Y^*,S^*,X)\mid G=1
         } \\ 
         &= \pi \E_{P_\star}\bigg[
         \rho^{-1}(S^*, X) \E_{P_\star}[W \mid S^*, X, G=1] \\
         &\quad\quad\quad \E_{P_\star}[h(Y^*, S^*, X) \mid S^*,
         X,
         G=1] \mid G=1
         \bigg] \tag{\cref{as: luot}} \\ 
         &= \E_{P_\star}[\pi h(Y^*, S^*, X) \mid G=1]\\
         &=\E_{P_\star}[G h(Y^*, S^*, X)]~,
        \end{align*}
    giving the first equality. The second equality then follows from 
    \begin{align*}
         \E_{P_\star}[G h(Y^*, S^*, X)] 
         &= \pi \E_{P_\star}[\E_{P_\star}[h(Y^*, S^*, X) \mid S^*, X, G=1] \mid G=1] \\ 
         &= \E_{P_\star}[G \E_{P_\star}[h(Y^*, S^*, X) \mid S^*, X, G=1]] \\ 
         &= \E_{P_\star}[G \E_{P_\star}[h(Y^*, S^*, X) \mid S^*, X]]~,
    \end{align*}
    completing the proof.\hfill
 \end{proof}

It now suffices to verify the conditions \eqref{eq:eta1}--\eqref{eq:eta5} for
choices of $a(S,X)$ and $b(X)$ that satisfy \eqref{eq:condphi} and \eqref{eq:condvarphi}.
The remainder of the proof verifies these conditions sequentially.
 
\noindent
\textbf{Condition \eqref{eq:eta1}:}
Observe that 
\begin{align*}
\E_P\bk{
    \psi_{1,1}(B_1) GW l'_*(Y\mid S, X)
} &= \E_P\bk{
     \frac{GWY}{\pi \rho(S,X)} l'_*(Y\mid S, X)
} \\
&=  \E_{P_\star}\bk{
    \pi^{-1}GY^* l'_*(Y^*\mid S^*, X)}~, \tag{\cref{lemma:get_rid_of_wgy}} 
\end{align*}
as required. 

\noindent
\textbf{Condition \eqref{eq:eta2}:}
We derive the restrictions on $a(s,x)$ implied by \eqref{eq:eta2}. Observe that the inner product from \eqref{eq:eta2} can be written
\begin{align*}
&\E_P \bk{ \psi_{1,1}(B_1)
 \pr{W l'_*(S \mid X) 
 - G(1-W) \frac{\E_{P_\star}[\rho(S,X) l'_*(S \mid X) \mid X]}{1-\E_{P_\star}[\rho(S,X) \mid X]}}} \\
& \quad =  \E_P \bk{
\frac{GW (Y-\mu(S,X)) l'_*(S \mid X)}{\pi \rho(S,X)}} 
- \E_P \bk{
\frac{b(X) G(1-W)\E_{P_\star}[\rho(S,X) l'_*(S \mid X) \mid X]}{1-\E_{P_\star}[\rho(S,X) \mid X]}
} \\
&\quad + \E_P \bk{
\frac{(1-G)Wa(S,X)}{(1-\gamma(X)) \varrho(X)} l'_*(S \mid X)
} +  \E_P \bk{
GWb(X)l'_*(S \mid X)}~.
\end{align*}
Note that the first term can be eliminated by \cref{lemma:get_rid_of_wgy}:
\begin{align*}
\E_P\bk{\frac{GW (Y-\mu(S,X)) l'_*(S \mid X)}{\pi \rho(S,X)}} 
&= 0~.
\end{align*}

Observe that \[\E_{P_\star}[G(1-W) \mid X] = P_\star(W \mid G=1, X) \gamma(X) = 
(1-\E_{P_\star}[\rho(S,X)\mid X]) \gamma(X),\] and thus the second term is given
by
\begin{align*}
& \E_P\bk{
\frac{b(X) G(1-W)\E_{P_\star}[\rho(S^*,X) l'_*(S^* \mid X) \mid X]}{1-\E_{P_\star}[\rho(S^*,X) \mid X]}} \\
& \quad = \E_{P_\star}\bk{\gamma(X) b(X) \E_{P_\star}[\rho(S^*,X) l'_*(S^* \mid X) \mid X]} \\
& \quad = \E_{P_\star}[\gamma(X)b(X)\rho(S^*,X)l'_*(S^* \mid X)]~.
\end{align*}
Next, observe that \[\E_{P_\star}[(1-G)W \mid S^*, X] = (1-\gamma(X)) P_\star(W=1 \mid
G=0, S^*, X) = (1-\gamma(X))\varrho(X),\] and therefore the third term can be
written
\[
\E_P\bk{
\frac{(1-G)W}{(1-\gamma(X)) \varrho(x)} a(S,X)l'_*(S \mid X)} 
= \E_{P_\star}\bk{a(S^*,X) l'_*(S^* \mid X)}~.
\]
The fourth term is given by 
\[
\E_P[GWb(X) l'_*(S \mid X)] = \E_{P_\star}[\gamma(X) q(S^*, X) b(X) l'_*(S \mid X)]~,
\]
which cancels with the second term. 

Thus, the condition \eqref{eq:eta2} is equivalent to
 \begin{equation}
    \E_{P_\star}[l'_*(S^* \mid X) a(S^*, X)] = \E_{P_\star}\bk{\pi^{-1}GY^* l'_*(S^* \mid X))}.
\label{eq:implicationphi}
\end{equation}
We proceed by verifying this condition for a choice of $a(S,X)$ that satisfies \eqref{eq:condphi}, as required.
Consider the function 
\[
a(s,x) = \frac{\gamma(x)}{\pi}(\mu(s,x) - \bar{\mu}(x))~.
\] 
Then $\E_{P_\star}[a(S^*,X) \mid X] = 0$, which satisfies \eqref{eq:condphi}. Moreover,
 \begin{align*}
\E_{P_\star}[a(S^*,X)  l'_*(S^* \mid X)] 
&=\E_{P_\star}\bk{
    \frac{\gamma(X)}{\pi}
    \mu(S^*,X)
    l'_*(S^* \mid X)}\\
&=\E_{P_\star}\bk{
    \frac{\E_{P_\star}[G \mid S^*, X]}{\pi}
    \E_{P_\star}\bk{Y^* l'_*(S^* \mid X) \mid S^*, X}}\\
&=\E_{P_\star}\bk{\pi^{-1}
    GY^*l_*'(S^* \mid X)}~, 
\end{align*}
where the second to last equality follows by \cref{as: ceev}, which verifies \eqref{eq:implicationphi} and thereby \eqref{eq:eta2}.

\noindent
\textbf{Condition \eqref{eq:eta3}:}
We derive the restrictions on $b(x)$ implied by \eqref{eq:eta3}. Note
that the conjectured influence function \eqref{eq:conjecturedeif} can be written as 
\[
\psi(b_1) = gw \cdot \alpha(y,s,x) + (1-g)w\cdot\beta(s,x) + gb(x),
\]
where 
\begin{align*}
\E_P[\alpha(Y,S,X) \mid S=s, X=x, G=1, W=1] & = \E_P[\beta(S,X) \mid G=0, W=1, X=x] = 0~.
\end{align*}

We first show that the terms in \eqref{eq:eta3} involving $\alpha$ and $\beta$ are zero.
Indeed, for terms involving $\alpha$, we have that
\[
\E_P[GW\cdot\alpha(Y,S,X)l'_*(G, X)] 
= \pi \E_P[\alpha(Y,S,X)l'_*(G, X)\mid G=1, W=1] P(W=1\mid G=1)=0
\]
since $\E_P[\alpha(Y,S,X)l'_*(G, X; 0)\mid G=1, W=1] = 0$. 
Analogously, we have that
 \[
\E_P[(1-G) W \beta(S,X) l'_*(G, X; 0)] = 0.
\]
Therefore, \eqref{eq:eta3} is equivalent to 
\begin{equation}
    \frac{1}{\pi}\E_{P_\star}
    [GY^* l'_*(G, X) ] - \theta_{1,1} \E[l'_*
(G,X) \mid G=1] = \E_{P}[Gb(X) l'_*(G, X)].
    \label{eq:implicationvarphi}
\end{equation}
We proceed by verifying this condition for a choice of $b(X)$.

Consider the function
\[b(x) = \frac{1}{\pi}(\bar{\mu}(x) - \theta_{1,1})~,\]
which satisfies \eqref{eq:condvarphi}.
We can compute
\begin{align*}
\E_P[G\pi^{-1} \mu(X) l'_*(G, X)]
&= \E_{P_\star}
\bk{\frac{G}{\pi}\E_{P_\star}
[\mu(S^*,X) \mid G=0, W=1, X]
l'_*(G, X; 0)]} \\ 
&= \E_{P_\star}\bk{
    \frac{G}{\pi} l'_*(G, X; 0) \E_{P_\star}[\E_{P_\star}[Y \mid W=1, G=1, S^*, X] \mid X]
} \\
&= \E_{P_\star}\bk{
    \frac{\gamma(X)}{\pi} l'_*(G, X; 0) \E_{P_\star}\bk{
    \E_{P_\star}\bk{Y^* \mid S^*, X} \mid X
    }
} \\
&= \E_{P_\star}\bk{\frac{\gamma(X)}{\pi} l'_*(G, X; 0) Y^*} \\
&= \E_{P_\star}\bk{\frac{\E_{P_\star}[G \mid Y^*, X]}{\pi}l'_*(G, X; 0)Y^*}  \\
&= \E_{P_\star}\bk{
    \frac{G}{\pi} Y^* l'_*(G, X; 0)
},
\end{align*}
and \[
-\E_{P_\star}[G\pi^{-1}\theta_{1,1} l'_*(G,X;0)] = -\theta_{1,1} \E_{P_\star}[l'_*(1,X;0) \mid G=1]
\]
which verifies \eqref{eq:implicationvarphi} and thereby \eqref{eq:eta3}.

\noindent
\textbf{Condition \eqref{eq:eta4}:} Since the
first and second terms $\alpha(Y,S,X)$ and $\beta(S,X)$ in the conjectured efficient influence function \eqref{eq:conjecturedeif} are mean zero conditional on $S^*$ and $X$ and $X$, respectively, the right-hand side of
\eqref{eq:eta4} is equal to 
\begin{align*}
&\E_P\bk{
b(X)  
\pr {GW \frac{\rho'(S,X)}{\rho(S,X)} - G(1-W) \frac{\E_{P_\star}[\rho'(S^*,X)\mid X]}{1-\E_{P_\star}[\rho(S^*,X) \mid X]}}} \\
& = \E_{P_\star}\bk{
b(X) \gamma(X)\E_{P_\star}\bk{\rho(S^*,X)\frac{\rho'(S^*,X)}{\rho(S^*,X)} - (1-\rho(S^*,X))\frac{\E_{P_\star}
[\rho'(S^*,X)\mid X]}{1-\E_{P_\star}[\rho(S^*,X)\mid X]} \mid G=1, S^*, X}
} \\
&= \E_{P_\star}\bk{b(X)\gamma(X)
\pr{\rho'(S^*,X) - (1-\rho(S^*,X)) \frac{\E_{P_\star}[\rho'(S^*,X)\mid X]}{1-\E_{P_\star}[\rho(S^*,X) \mid X]}}} = 0
\end{align*}
by iterated expectation conditional on $X$.

\noindent
\textbf{Condition \eqref{eq:eta5}:} The only
term on the right-hand side of \eqref{eq:eta5} is 
\[
\E_P[(1-G)W\cdot\beta(S,X)\cdot (W-\rho(X)) \delta(X)]
\] for some function $\delta(\cdot)$, where $(1-G) W \cdot\beta(S,X)$
is the second term in the conjectured efficient influence function \eqref{eq:conjecturedeif}. We have:
\begin{align*}
&\E_P[(1-G)W\cdot\beta(S,X)\cdot(W-\varrho(X)) \cdot \delta(X)] \\ 
&= \E_P[\E_P[W\cdot\beta(S,X)\cdot(W-\varrho(X)) \cdot \delta(X)\mid G=0, X](1-\gamma(X))] \\
&= \E_{P_\star}[\beta(S^*,X) \varrho(X)(1-\varrho(X))\cdot  \delta(X)] =0
\end{align*}
by the condition $\E_{P_\star}[\beta(S^*,X) \mid X] = 0$, completing the proof. \hfill\qed

\subsection{Proof of \cref{thm: EIF tau_1}. Part 2}
\label{sub:surrogateproof}
We maintain the notation of the proof of \cref{thm: EIF tau_1}, Part 1.
In particular, we omit the $w$ subscript in the nuisance function $\bar{\nu}_w(x)$ and 
let 
\[
A_1 = (Y(1), S(1), W, G, X)  = (Y^*, S^*, W, G, X) 
\]
and $B_1 = (Y,S,W,G,X)$ denote the complete and observed data, respectively. The efficient influence function we construct is 
\begin{align*}
\xi_{1,1}(b) & = \frac{g}{\pi}\left( \frac{\gamma(x)}{\gamma(s,x)} \frac{1-\gamma(s,x)}{1-\gamma(x)} \frac{\varrho(s,x)(y-\nu(S,X))}{\varrho(x)} 
+ (\bar{\nu}_1(x) - \theta_{1,1})\right) \\
& + \frac{1-g}{\pi} \frac{\gamma(x)}{1-\gamma(x)}\left(\frac{w(\nu(S,X) - \bar{\nu}_1(x))}{\varrho(x)}\right)~.
\end{align*}

\noindent
\textbf{Density Factorization.} 
By \cref{as: uex,as: ltoc,as: ceev,as: es}, the density $p$ of the observed
data distribution admits the factorization
\begin{align}
\label{eq:factorization acik}
p(b_1) =\,\, & p(x) \gamma(x)^g (1-\gamma(x))^{1-g} \nonumber \\
& p(y \mid G=1, s, x)^{g} p(s \mid G=1, x)^g  \nonumber \\
& p(s \mid G=0, W=1, x)^{w(1-g)} (\varrho(x)^w(1-\varrho(x))^{1-w})^{1-g}~. 
\end{align}

\noindent
\textbf{Characterization of the Tangent Space.}
Let $\mathcal{P}$ be any regular parametric submodel of $\mathcal{M}_\lambda$ indexed by $\varepsilon \in\R$, with densities 
$
p_\varepsilon := dP_\varepsilon /d\lambda
$
with $p_0=dP/d\lambda$ and such that \cref{as: uex,as: ltoc,as: ceev,as: es} hold for each
$P_\varepsilon \in \mathcal{P}$. By the factorization \eqref{eq:factorization acik}, we obtain the score
\begin{align} 
l'(b_1) =\,\,& 
g l'(y\mid s, x, G=1) + gl'(s\mid x,  G=1) + w(1-g)l'(s \mid W=1, x, G=0) +
l'(x) \nonumber \\
+& \frac{g-\gamma(x)}{\gamma(x)(1-\gamma(x))}\gamma'(x)
+(1-g) \frac{w-\varrho(x)}{\varrho(x) (1-\varrho(x))} \varrho'(x) \label{eq:score_acik}.
\end{align}
Thus, the tangent space $\mathcal{T}$ is given by the mean-square closure of the linear
span of the functions
\begin{align} 
s(b_1) =\,\,& 
g\cdot s_1(y\mid s, x, G=1) + g\cdot s_2(s\mid x,  G=1) + w(1-g)\cdot s_3(s \mid
W=1, x, G=0) + s_4(x) \nonumber \\
+& \frac{g-\gamma(x)}{\gamma(x)(1-\gamma(x))}s_5(x)
+(1-g) \frac{w-\varrho(x)}{\varrho(x) (1-\varrho(x))} s_6(x)~,
\label{eq:tangent_space_ss}
\end{align}
where the functions $s_1$ through $s_6$ range over the space of mean-zero and square integrable 
functions that satisfy the restrictions
\begin{align}
0&= \E_P[s_1(Y\mid S, X, G=1) \mid S, X, G=1] \nonumber  \\
& = \E_P[s_1(Y\mid S, X, G=1) \mid S, X]~,  \label{eq: s_1} \\ 
0&= \E_P[s_2(S\mid X, G=1) \mid X, G=1]~,  \label{eq: s_2} \\
0&=\E_P[s_3(S\mid W=1, X, G=0) \mid W=1, X, G=0] \nonumber \\
& = \E_{P_\star}[s_3(S^*\mid W=1, X, G=0)\mid X, G=0]~, \label{eq: s_3} \\ 
0 & = \E_P[s_4(X)]~, \label{eq: s_4}
\end{align}
and $s_5,$ and $s_6$ are unconstrained, where \eqref{eq: s_1} and \eqref{eq: s_3} follow from \cref{as: ltoc} and \cref{as: uex}, respectively.

\noindent
\textbf{Pathwise Differentiability Representation.}
Recall from the proof of \cref{thm: identification lut} that, by  \cref{as: uex,as: ltoc,as: ceev,as: es}, the target parameter $\theta_{1,1}$ can be written as 
\begin{align*}
\theta_{1,1} & = \E_{P_\star}[Y^* \mid G=1] \\
& = \E_P\bk{\E_P\bk{\E_P[Y \mid S, X, G=1] \mid W=1, X, G=0} \mid G = 1} \\ 
& = \iiint y p(y \mid s, x, G=1) p(s \mid x, W=1, G=0)p(x \mid G=1) \,d\lambda(y) \, \,d\lambda(s) \, \,d\lambda(x)~. 
\end{align*}
The pathwise derivative of this parameter is given by 
\begin{align*}
\theta_{1,1}' = \E_P[\E_P[\E_P[Y U(Y,S,X) \mid S, X, G=1] \mid X, W=1, G=0] \mid G=1]
\end{align*}
with 
\begin{align*}
U(y,s,x) & =  l'(y\mid s, x, G=1) + l'(s \mid  x, W=1, G=0) + l'(x \mid  G=1) \\
& = l'(y\mid s, x, G=1) + l'(s \mid  x, W=1, G=0) \\
& + l'(x)  + \frac{\gamma'(x)}{\gamma(x)} - \frac{\E_P[\gamma'(X) + \gamma(X) l'(X)]}{\pi}~,
\end{align*}
where final three terms follow from 
\begin{align*}
p'(x \mid G=1) = l'(x\mid G=1) p(x \mid G=1)
\end{align*}
and Bayes' rule in the form of
\[
l(x \mid G=1) = \log \gamma(x) + l(x) - \log \int \gamma(x)p(x)\,d\lambda(x)~.
\]
Observe that, using the definitions of $\nu(s,x)$ and $\bar{\nu}(x)$, we may
further simplify the pathwise derivative to give
\begin{align}
\theta_{1,1}' =\,\,
& \E_P[\E_P[\E_P[Y l'(Y \mid S, X, G=1) \mid S, X, G=1] \mid
X, W=1 ,G=0
] \mid G=1] \nonumber\\ 
& + \E_P[\E_P[\nu(S,X)l'(S\mid X, G=0, W=1) \mid X, W=1, G=0] \mid G=1]
\nonumber \\
& + \pi^{-1} \E_P\bk{\bar{\nu}(X)\pr{\gamma'(X) + \gamma(X)l'(X)} } 
- \pi^{-1} \E_P[\gamma'(X) + \gamma(X) l'(X)] \theta_{1,1}.
\label{eq:pathwisederivacik}
\end{align}

An influence function for $\theta_{1,1}$ is a mean-zero and square-integrable 
function $\tilde{\xi}_{1,1}(B_1)$ that satisfies the condition
\begin{equation}
\label{acik:pathwise}
    \theta'_{1,1}  = \E[\tilde{\xi}(B_1) l'(B_1)]~.
\end{equation}
We make a conjecture for such a function and verify that it satisfies this condition and is 
an element of the tangent space $\mathcal{T}$, establishing that it is the efficient influence
function. 

\noindent
\textbf{Conjectured Efficient Influence Function.} 
We conjecture that the efficient influence function takes the form
\begin{equation}
\xi_{1,1}(b_1) = g(y-\nu(s,x)) \cdot f_1(s,x) + w(1-g)(\nu(s,x) - \bar{\nu}(x))\cdot f_2(x) + g\cdot f_3(x)~, \label{eq: conjecture acik}
\end{equation}
for functions $f_1(s,x)$, $f_2(x)$, and $f_3(x)$ to be specified, where $f_3$ will be chosen to satisfy
\begin{equation}
 \E_P[f_3(X) \mid G=1] = 0 = \E_P[f_3(X) \gamma(X)]~. \label{eq: f_3}
\end{equation}

Observe that for this choice, $\xi_{1,1}(b)$ will be an element of the
tangent space $\mathcal{T}$. This follows as \[s_1(y\mid s,x, G=1) = (y-\nu(s,x)) f_1
(s,x)\] satisfies \eqref{eq: s_1}, we make the choice \[s_2(s \mid X, G=1)=0 \numberthis
\label{eq:surprising_ancillarity}\] satisfying \eqref{eq:
s_2}, 
\[s_3(s\mid W=1, x, G=0) = (\nu(s,x) - \bar{\nu}(x))f_2(x)\] satisfies \eqref{eq:
s_3}, 
and that setting $s_4(x) = \gamma(x) f_3(x)$ and $s_5(x) = \gamma(x) (1-\gamma(x)) f_3(x)$ yields
\[
g f_3(x)= s_4(x) + \frac{g-\gamma(x)}{\gamma(x) (1-\gamma(x))} s_5(x)~,
\]
satisfying \eqref{eq: s_4}, the mean zero conditions for $s_4$ and $s_5$ by \eqref{eq: f_3}, and implicitly satisfying \eqref{eq: s_2} and the mean zero condition for $s_6$.

\noindent
\textbf{Verification of Pathwise Differentiability Representation.} 
In the sequel, we make choices of $f_1(s,x)$, $f_2(x)$, and $f_3(x)$ such that the
pathwise derivative of $\theta_{1,1}$ satisfies \eqref{acik:pathwise} and \eqref{eq: f_3}, completing the proof.

Our conjecture for $\xi_{1,1}(b_1)$ allows for some simplification of the right-hand side
of \eqref{acik:pathwise}. Note that, for any measurable functions $h(S,X)$ and $k(X)$, we have
\begin{align*}
\E_P[G(Y - \nu(S,X)) f_1(S,X) h(S,X)] = 0 = \E_P[(1-G)W(\nu(S,X) - \bar{\nu}(X)) f_2(X) k(X)]
\end{align*}
by iterated expectations. Thus, by additionally applying the mean-zero property of the scores $l'(Y \mid S, X, G=1)$, $l'(S \mid X, G=1)$, we can simplify the pathwise differentiability condition to give
\begin{align} 
\label{eq: simplified}
\E_P\bk{\xi_{1,1}(B_1) l'(B_1)}
& = \E_P\bk{G (Y - \nu(S,X))  f_1(S,X) l'(Y \mid S, X, G=1)} \nonumber \\
&+ \E_P\bk{W(1-G)(\nu(S,X) - \bar{\nu}(X))f_2(X) l'(S \mid W=1, X, G=0)} \nonumber \\
&+  \E_P\bk{\pr{Gl'(X) + \gamma'(X)}f_3(X)}~.
\end{align}
Therefore, we obtain a set of sufficient conditions for \eqref{acik:pathwise} 
by matching each term in \eqref{eq: simplified} with \eqref{eq:pathwisederivacik}, giving
\begin{align}
& \E_P\bk{
    G (Y - \nu(S,X)) f_1(S,X) l'(Y \mid S, X, G=1)
} \label{eq:acik1} \\ \nonumber &\quad\quad\quad= \E_P[\E_P[\E_P[Y l'(Y \mid S,
X, G=1) \mid S, X, G=1] \mid
X, W = 1,G=0
] \mid G=1] \\ 
& \label{eq:acik2} \E_P\bk{W(1-G)(\nu(S,X) - \bar{\nu}(X))f_2(X)l'(S \mid W=1, G=0,
X)} \\ \nonumber &\quad\quad\quad= \E_P[\E_P
[\mu
(S,X)l'(S\mid X, W=1, G=0) \mid X, W=1, G=0] \mid G=1] \\ 
& \label{eq:acik3} \E_P\bk{\pr{Gl'(X) + \gamma'(X)}f_3(X)} \\ \nonumber
&\quad\quad\quad= \E_P\bk{\pi^{-1} \bar{\nu}(X)
\pr{\gamma'
(X) + l'(X)\gamma(X)} } -
\pi^{-1} \E_P[\gamma'(X) + \gamma(X) l'(X)] \theta_{1,1}.
\end{align}
We complete the proof by considering each condition sequentially, making
appropriate choices of $f_1(s,x)$, $f_2(x)$, and $f_3(x)$, where we note 
that each function appears in only one term. 

\noindent
\textbf{Condition \eqref{eq:acik1}:} 
Evaluating the left-hand side of \eqref{eq:acik1} we find that
\begin{align*}
&\E_P\bk{
    G (Y - \nu(S,X)) f_1(S,X) l'(Y \mid S, X, G=1) 
} \\ 
&= \E_P[GYf_1(S,X)l'(Y \mid S, X, G=1)]\\ 
&= \E_P[\gamma(X) \E_P[f_1(S, X) \E_P[ Y l'(Y\mid S, X, G=1) \mid S, X, G=1] \mid X, G=1]].
\end{align*}
Thus, if we choose 
\[
f_1(s,x) = \frac{p(s \mid x, W=1, G=0)}{\gamma(x) p(s\mid x, G=1)} \frac{p
(x \mid G=1)}{p(x)}
\]
to change the measure of the two outer expectations, \eqref{eq:acik1} is satisfied. Bayes' rule calculations
show that 
\[
f_1(s,x) = \frac{1}{\pi} \frac{\varrho(s,x)}{\varrho(x)} \frac{1-\gamma(s,x)}{\gamma(s,x)}\frac{\gamma(x)}{1-\gamma(x)}~. 
\]

\noindent
\textbf{Condition \eqref{eq:acik2}:} 
Observe that, for any measurable function $h(s,x)$, we have that
\[
\E_P[W(1-G) h(S,X)] = \E_P[\varrho(X) (1-\gamma(X)) \E_P[h(S,X) \mid X, W=1, G=0]].
\]
Also, recall that since $l'(S \mid X, W=1, G=0)$  is a conditional
score, it is conditionally orthogonal to any function of $X$; that is, for any measurable function $h(x)$,
\[
\E_P[h(X) l'(S\mid X, W=1, G=0) \mid X, W=1, G=0] = 0~.
\]

With the above two observations in mind, evaluating the left-hand side of 
\eqref{eq:acik2}, we find that
\begin{align*}
&\E_P\bk{W(1-G)(\nu(S,X) - \bar{\nu}(X))f_2(X)l'(S \mid W=1, G=0,
X)} \\
&= \E_P\bk{ \varrho(x) (1-\gamma(X))\E_P[ (\nu(S,X) - \bar{\nu}(X))f_2(X)l'(S \mid W=1, G=0,
X)]} \\
&=  \E_P\big[
   \varrho(x) (1-\gamma(X)) f_2(X)\E_P[\nu(S,X) l'(S \mid X, W=1, G=0)\mid X, W=1, G=0] 
\big]~.
\end{align*}
Thus, again, if we choose
\[ 
f_2(x) = \frac{p(x\mid G=1)}{p(x) \varrho(x) (1-\gamma(x)) } 
= \frac{1}{\pi} \frac{1}{\varrho(x)} \frac{\gamma(x)}{1-\gamma(x)} ~,
\]
again to change the measure of the outer expectation, then \eqref{eq:acik2} is satisfied.  

\noindent
\textbf{Condition \eqref{eq:acik3}:} 
Lastly, consider the choice
\[f_3(x) = \frac{\bar{\nu}(X) - \theta_{1,1}}{\pi}~,\]
which satisfies \eqref{eq: f_3}. In this case, we find
that the left-hand side of \eqref{eq:acik3} is given by 
\begin{align*}
& \E\bk{\pr{Gl'(X) + \gamma'(X)}f_3(X)} \\
& = \E\bk{\frac{\gamma(X)l'(X) +
\gamma'(X)}{\pi} \bar{\nu}(X)} - \frac{1}{\pi} \E[\gamma(X)l'(X)+\gamma'(X)]
\theta_{1,1},
\end{align*}
which is exactly the right-hand side of \eqref{eq:acik3}. 

Collecting the above choices for $f_1(s,x)$, $f_2(x)$, and $f_3(x)$ gives the efficient influence function
\begin{align*}
\xi_{1,1}(b) & = \frac{g}{\pi}\left( \frac{\gamma(x)}{\gamma(s,x)} \frac{1-\gamma(s,x)}{1-\gamma(x)} \frac{\varrho(s,x)(y-\nu(S,X))}{\varrho(x)} 
+ (\mu_1(x) - \theta_{1,1})\right) \\
& + \frac{1-g}{\pi} \frac{\gamma(x)}{1-\gamma(x)}\left(\frac{w(\nu(S,X) - \mu_1(x))}{\varrho(x)}\right)~,
\end{align*}
as required. \hfill\qed

\subsection{Proof of Theorem \ref{thm: locally-overidentified}. Part 1}

We maintain the notation from the proof of \cref{thm: EIF tau_1}. Let $\tilde{\psi}_{1,1}(b_1)$ be an influence function for $\theta_{1,1}$. It is necessarily that case that
 \[
\tilde{\psi}_{1,1} (b_1) = \psi_{1,1}(b_1) + f(b_1)~,
\]
where $\psi_{1,1}(b_1)$ is the efficient influence function in derived in the proof of \cref{thm: EIF tau_1} and $f(b_1)$ is mean-zero and orthogonal to the tangent space $\mathcal{T}$, also defined in the proof of \cref{thm: EIF tau_1}. Moreover, we may
write 
\begin{align}
f(b_1) & = gw f_1(y,s,x) + (1-g)w f_2(s,x) + g(1-w) f_3(x) \nonumber \\
& + (1-g)(1-w) f_4(x) + A~, \label{eq:f tau_1}
\end{align}
where $A$ is a constant making $f$ mean-zero. 

The orthogonality between $f(b_1)$ and the tangent space $\mathcal{T}$ implies the following set of conditions that are analogous to the conditions \eqref{eq:eta1} through \eqref{eq:eta5} defined in the proof of \cref{thm: EIF tau_1}: 
\begin{align}
    0&= \E_P\bk{f(B_1) GW l'_*(Y\mid S, X)} ~,
     \label{eq:eta1_zero} \\ 
    0&= \E_P\bk{f(B_1) \pr{ W l'_*(S \mid X)
     - G(1-W) \frac{\int \rho(s,X) p'_*(s \mid X)\,d\lambda(s)}{1-\int \rho(s,X)p_\star(s\mid X)\,d\lambda(s)}}}~,
    \label{eq:eta2_zero}\\ 
    0&= \E_P[f(B_1)  l'_*(G, X)]~,
     \label{eq:eta3_zero}\\ 
    0 &= \E_P\bk{f(B_1) G\pr{
    W \frac{\rho^\prime(S,X)}{\rho(S,X)} - (1-W) \frac{\int \rho^\prime(s,X) p_\star(s\mid X) d\nu_S(s)}{1-\int \rho
    (s,X)p_\star(s\mid X)\,d\nu_S(s)}} }~,
    \label{eq:eta4_zero}\\ 
    0 &= \E_P\bk{f(B_1) (1-G)\frac{W - \varrho(X)}{\varrho(X)(1-\varrho(X))} \varrho'(X)}~.
   \label{eq:eta5_zero}
\end{align}
In the sequel, we characterize $f$ by deriving restrictions from each of these conditions for particular choices of the scores $l_*'(\cdot)$, $\varrho'(x)$, and $\rho'(s,x)$.

Combining \eqref{eq:eta1_zero} and \eqref{eq:f tau_1} yields 
\[
\E_P\bk{GW f_1(Y,S,X) l'_*(Y \mid S,X)} = 0~.
\]
As $l'_*(Y \mid S,X)$ ranges over all conditionally mean-zero, square integrable functions, we may pick 
\[
l'_*(Y \mid S,X) = f_1(Y,S,X) - \E_P[f_1(Y,S,X) \mid G=1, W=1, S, X],
\]
from which we find that
\begin{align*}
0&= \E_P\bk{GW f_1(Y,S,X) l'_*(Y \mid S,X)} \\
&=\E_P\bk{GW f_1^2(Y,S,X)  - f_1(Y,S,X) \E_P[GWf_1(Y,S,X) \mid G=1, W=1, S, X]}~,
\end{align*}
and so $f_1(Y,S,X) = f_1(S,X)$ does not depend on $Y$. 

Next, observe that plugging \eqref{eq:f tau_1} into \eqref{eq:eta4_zero}, yields 
\[
0  = \E_{P_\star}\bk{\gamma(X) \pr{ f_1(S^*, X) \rho'(S^*,X) - f_3(X) \E_{P_\star}[\rho'(S^*, X)
\mid X]}} \]
after simplification via law of iterated expectations. Choosing 
\[
\rho'(S^*, X) = \gamma(X)^{-1}\pr{f_1(S^*,X) - \E_{P_\star}[f_1(S^*,X)\mid X]}
\] 
establishes that
\[
0 = \E_{P_\star}[f_1^2(S^*, X)-f_1(S^*, X)\E[f_1(S^*,X)\mid X]]~,
\]
and so $f_1(S,X) = f_1(X)$ does not depend on $S$. This then implies that $f_1(x) = f_3(x)$ by law of iterated expectations.

Now, plugging \eqref{eq:f tau_1} into \eqref{eq:eta2_zero}
yields
\[
0 = \E_{P_\star}\bk{(1-\gamma(X)) \varrho(X) f_2(S^*,X) l'_*(S^* \mid X)}~,
\]
again after simplification via law of iterated expectations. As we may choose
\[l'_*(S^* \mid X) = f_2(S^*, X) - \E_{P_\star}[f_2(S^*, X) \mid X]~,\] we find similarly that $f_2$ is a function only of $X$. 
 
Plugging \eqref{eq:f tau_1} into \eqref{eq:eta5_zero}, gives the condition
 \[
0 = \E_P\bk{(1-\gamma(X)) (f_2(X) - f_4(X)) \varrho'(X)}
\]
where picking $\varrho'(X) = f_2(X) - f_4(X)$ implies that $f_4
(x) = f_2(x)$.

Lastly, plugging \eqref{eq:f tau_1} into \eqref{eq:eta3_zero}, we find
 \[
0 = \E_P\bk{
    l'_*(G, X) (G f_3(X) + (1-G) f_2(X))
}.
\]
We then find that that
\[
\E_P[l'_*(G, X)] = 0 = \E_P[\gamma(X) l'_*(1,X) + (1-\gamma(X)) l'_*(0,X)]~,
\]
and so $f_1 (x) = f_2(x) = f_3(x) = f_4(x) = C$ almost everywhere, for some constant $C$. Hence, $f(b) = 0$ almost everywhere
and the space of influence functions is a singleton.\hfill\qed

\subsection{Proof of \cref{thm: locally-overidentified}. Part 2} 

We demonstrate that $\xi_{1,1}$ is the unique influence function for
$\theta_{1,1}$. An analogous argument verifies the statement for $\theta_{0,1}$, and therefore for
$\tau_1$. We maintain the notation from the proof of \cref{thm: EIF tau_1}, Part 2. 

Let $\tilde{\xi}(b_1)$ be an influence function for $\theta_{1,1}$. It is necessarily the case that
\[
\tilde{\xi}_{1,1}(b_1) = \xi_{1,1}(b_1) + f(b_1)~,
\]
where $\xi_{1,1}(b_1)$ is the efficient influence function derived in \cref{thm: EIF tau_1}, Part 2, 
and $f(b_1)$ is mean-zero and orthogonal to the tangent space $\mathcal{T}$, defined in the proof of \cref{thm: EIF tau_1}, Part 2. We may without loss of generality write
\[
f(b) = g \cdot f_1(y,s,x) + (1-g) w \cdot f_2(s,x) + (1-g)(1-w) \cdot f_3(x) + A~,  
\]
where $A$ is some constant making $f$ mean-zero.

The orthogonality between  $f$ and the tangent space $\mathcal{T}$ implies that
\begin{equation}
\label{eq: orth}
\E_P[f(B_1) l'(B_1)] = 0
\end{equation}
for any $l'(\cdot)$ corresponding to a score of a parametric submodel. 

We first argue that $f_1(y,s,x) = f_1(x)$ does not depend on $y$ or $s$.
Consider all free components in the score \eqref{eq:score_acik}, and set
everything except for $l'(y\mid s, x, G=1)$ to zero. Then the orthogonality
condition \eqref{eq: orth} becomes
\[
\E_P[G f_1(Y,S,X) l'(Y \mid S, X, G=1)] = 0~.
\]
Picking $l'(Y\mid S, X, G=1) = f_1(Y,S,X) - \E_P[f_1(Y) \mid S, X, G=1]$
implies that $\E_P[G f_1(Y,S,X) l'(Y \mid S, X, G=1)] > 0$ unless $f_1(Y,S,X) =
\E[f_1(Y) \mid S, X, G=1]$ almost surely. Therefore $f_1(Y,S,X)$ does not
depend on $Y$. 
Similarly, if we set every free component in \eqref{eq:score_acik}
except for $l'(S \mid X, G=1)$ to zero, then choosing $l'(S\mid X, G=1) = f_1
(S,X) -
\E[f(S,X) \mid X, G=1]$ again implies that $f_1$ only depends on $X$.
A similar argument with the score component $l'(S \mid W=1, X, G=0)$ shows
that $f_2(s,x)=f_2(x)$ depends only on $x$. Therefore, $f(b)$ must take the form 
\[
f(b) = g\cdot f_1(x) + (1-g)w \cdot f_2(x) + (1-g)(1-w) \cdot f_3(x). 
\]

The orthogonality condition \eqref{eq: orth} then takes the form \begin{align*}
0 &= \E_P[f(B)l'(X)] + \E_P\bk{
    f(B)\frac{G-\gamma(X)}{\gamma(X) (1-\gamma(X))}\gamma'(X)
} \\&\quad + \E_P\bk{
    (1-G)(W f_2(X) + (1-W)f_3(X)) \frac{W-\varrho(x)}{\varrho(x) (1-\varrho(x))} \varrho'(x)
}
\end{align*}
for any valid choices of $l'(x), \gamma'(x), \rho'(x)$, as the other terms
are zero by the mean-zero conditions on the conditional scores.

Observe that the third term can be simplified as
\begin{align}
& \E_P\bk{
    (1-G)(W f_2(X) + (1-W)f_3(X)) \frac{W-\varrho(x)}{\varrho(x) (1-\varrho(x))} \varrho'(x)
} \nonumber \\
& = \E_P\bk{
 \frac{f_2(X) - f_3(X)}{1-\pi} \varrho'(x) \mid G=0
}.\label{eq: three terms}
\end{align}
Picking $\varrho'(x) = f_2(x) - f_3(x), l'(x) = \gamma'(x) = 0$ implies that $f_2
(x) = f_3(x)$.
Thus $f(b) = gf_1(x) + (1-g) f_2(x)$. Working with the second term in \eqref{eq: three terms} (i.e. setting $l'(x) = 0 = \rho'(x)$ and considering choices of $\gamma'(x)$) yields that \[
\E_P\bk{
    f(B)\frac{G-\gamma(X)}{\gamma(X) (1-\gamma(X))}\gamma'(X)
} = \E_P[(f_1(X) - f_2(X)) \gamma'(X)]
\]
Picking $\gamma'(x) = f_1(x) - f_2(x)$ shows that $f_1(x) = f_2(x)$. Therefore, from the first term in \eqref{eq: three terms}, we
find that for all $l'$ s.t. $\E_P[l'(X)] = 0$, 
\[
\E_P[f(B_1) l'(X)] = \E_P[f_1(X) l'(X)]= 0 \implies f_1(X) = 0.
\]
implies that  $f_1(X) = 0$ almost everywhere. Hence, $f(b_1) = 0$ almost everywhere, completing the proof. \hfill\qed

\subsection{Proof of \cref{thm: large-sample}}\label{sec: proof of large-sample}

In this section, we provide proofs for both \cref{thm: large-sample} and for the analogous result for the statistical surrogacy model, stated as follows.
\begin{theorem}
\label{thm: large-sample sur}~
Let $\mathcal{P}\subset\mathcal{M}_\lambda$ be the set of all probability distributions
$P$ for $\{B_i\}_{i=1}^n$ that satisfy the Statistical Surrogacy Model stated in \cref{def: sur} 
in addition to \cref{as: bounds}. If \cref{as:
rates} holds for $\mathcal P$, then
\begin{align}
\sqrt{n}(\hat{\tau}_{1,\mathsf{DML}} -\tau_1) \overset{d}{\to} \mathcal{N}(0,V_1^{\star\star})
\end{align}
uniformly over $P\in\mathcal{P}$, where $\hat{\tau}_{1,\mathsf{DML}}$ is defined in \cref{def: dml}, $V_1^{\star\star}$ is defined in \cref{cor: seb}, and $\overset{d}{\to}$ denotes convergence in distribution. Moreover, we have that 
\begin{align}
\hat{V}_1^{\star\star} = \frac{1}{k}\sum_{l=1}^k \frac{1}{m} \sum_{i\in I_l} \left(\xi_1(B_i,\hat{\tau}_{1,\mathsf{DML}},\hat{\varphi}(I_l^c))\right)^2 \overset{p}{\to} V_1^{\star\star}
\end{align}
uniformly over $P\in\mathcal{P}$, where $\overset{p}{\to}$ denotes convergence in probability. As a result, we obtain the uniform asymptotic validity of the confidence intervals
\begin{align}
\label{eq: ci}
\lim_{n\to\infty} \sup_{P\in\mathcal{P}} \Big\vert P\left(\tau_1 \in \left[\hat{\tau}_{1,\mathsf{DML}} \pm z_{1-\alpha/2}\sqrt{\hat{V}_1^{\star\star}/n} \right]\right) - (1-\alpha) \Big\vert = 0~,
\end{align}
where $z_{1-\alpha/2}$ is the $1-\alpha/2$ quantile of the standard normal distribution. 
\end{theorem}

To ease notation for the purposed of this proof, we change the notation introduced in 
\cref{def: partition}. In particular, partition the collection of nuisance functions $\eta$ appearing as an argument 
in $\psi_1(B,\tau_1,\eta)$ into long-term outcome means and propensity scores by $\eta=(\omega,\kappa)$, where we include the parameter $\pi$ in the propensity scores $\kappa$ to ease notation. Similarly
partition the collection of nuisance function $\varphi$ appearing as an argument 
in $\xi_1(B,\tau_1,\varphi)$ into $\varphi=(\vartheta,\zeta)$.

Observe that the efficient influence functions $\psi_1(b,\tau_1,\eta)$ and $\xi_1(b,\tau_1,\varphi)$ can be decomposed linearly into 
\[
\psi_1(b,\tau_1,\eta) = \psi_{1}^\alpha(b,\eta)\cdot\tau_1 + \psi_{1}^\beta(b,\eta)
\quad\text{and}\quad
\xi_1(b,\tau_1,\varphi = \xi_{1}^\alpha(b,\varphi)\cdot\tau_1 + \xi_{1}^\beta(b,\varphi)
\]
where
\begin{align*}
\psi_{1}^\alpha(b,\eta) &= \xi_{1}^a(b,\varphi) = -\frac{g}{\pi}~, \\
\psi_{1}^\beta(b,\eta) &=
\frac{g}{\pi}
\left( \frac{w(y-\mu_1(s,x))}{\rho_1(s,x)}  - \frac{(1-w)(y-\mu_0(s,x))}{\rho_0
(s,x)} + (\bar{\mu}_1(x) - \bar{\mu}_0(x))\right)\nonumber \\
&  + \frac{1-g}{\pi}
\left(\frac{\gamma (x)}{1-\gamma(x)}\left( \frac{w(\mu_1(s,x)-\bar{\mu}_1(x))}{\varrho(x)} - \frac{(1-w)(\mu_0(s,x)-\bar{\mu}_0(x))}{1-\varrho(x)}\right)\right),\text{ and}\\
\xi_{1}^\beta(b,\varphi) &= \frac{g}{\pi}
 \left( \frac{\gamma(x)}{\gamma(s,x)} \frac{1- \gamma(s,x)}{1-\gamma(x)}\frac{(\varrho(s,x)-\varrho(x))(y-\nu(s,x))}{\varrho(x)(1-\varrho(x)}
+ (\bar{\nu}_1(x) - \bar{\nu}_0(x))\right) \nonumber \\
& \quad \quad  + \frac{1-g}{\pi}\left(\frac{\gamma (x)}{1-\gamma(x)}\left(\frac{w(\nu(s,x)-\bar{\nu}_1(x))}{\varrho(x)} -\frac{(1-w)(\nu(s,x)-\bar{\nu}_0(x))}{1-\varrho(x)}\right)\right).
\end{align*}
Consequently, it will suffice to verify the conditions of Theorem 3.1 of \cite{chernozhukov2018double}. In particular, it will suffice to verify Assumption 3.1 and Assumption 3.2 therein. 

\vspace{1cm}
\noindent
\textbf{Verification of Assumption 3.1 of \cite{chernozhukov2018double}.}
Both $\psi_1(b,\tau_1,\eta)$ and $\xi_1(b,\tau_1,\varphi)$ are mean-zero, satisfying condition (a). By the above decomposition, both $\psi_1(b,\tau_1,\eta)$ and $\xi_1(b,\tau_1,\varphi)$ are linear in $\tau_1$ satisfying condition (b). It is clear that the mappings $\eta \mapsto \mathbb{E}_P\left[\psi_1(b,\tau_1,\eta)\right]$ and $\varphi \mapsto \mathbb{E}_P\left[\xi_1(b,\tau_1,\varphi)\right]$ are twice continuously Gateaux differentiable, as they are both smooth functions by the strict-overlap \cref{as: overlap}, satisfying condition (c). Both $\psi_1(b,\tau_1,\eta)$ and $\xi_1(b,\tau_1,\varphi)$ satisfy the ``Neyman orthogonality'' condition of \cite{chernozhukov2018double} by the orthogonality conditions (\refeq{eq: pathwise deriv}) and (\refeq{acik:pathwise}), satisfying condition (d). Finally, as 
\[
\mathbb{E}_P\left[ -\frac{g}{\pi} \right] = -1~,
\]
the singular values of $\mathbb{E}_P\left[ \psi_1(b,\tau_1,\eta) \right]$ and $\mathbb{E}_P\left[\xi_1(b,\tau_1,\varphi)) \right]$ are bounded away from zero, satisfying condition (e). 

\vspace{1cm}
\noindent
\textbf{Verification of Assumption 3.2 of \cite{chernozhukov2018double}.} 
We note that \cref{as: rates} can be restated as: For all
$P \in \mathcal P$, with probability at least $1-\Delta_n$, the nuisance
estimator $\hat \eta_n$ (resp. $\hat \varphi_n$) belongs to a 
\emph{realization set} $\mathcal{R}_{P,n}$. The realization set $
\mathcal{R}_{P,n}$ contains nuisance parameter values $\tilde \eta = 
(\tilde \omega, \tilde \kappa)$ such that \begin{enumerate}
    \item $\| \tilde{\eta} - \eta \|_{P,q} \leq C$,
    \item $\| \tilde{\eta} - \eta \|_{P,2} \leq \delta_n$, $\| \tilde{\kappa} - 1/2 \|_{P,\infty} \leq 1/2 - \epsilon$, and
    \item $\| \tilde{\omega} - \omega \|_{P,2} \cdot \| \tilde{\kappa} - \kappa \|_{P,2} \leq \delta_n / n^{1/2}$
\end{enumerate}
Define the realization set $\mathcal{R}_{P,n}$ analogously for the
Statistical Surrogacy Model, where the nuisance function $\varphi$ replaces $\eta$.

\vspace{1cm}
\noindent
\textbf{Condition (a).} The nuisance parameter estimators are elements of their respective realization sets with probability $1-\Delta_N$ by \cref{as: rates}, verifying condition (a). 

\vspace{1cm}
\noindent
\textbf{Condition (b).} To verify condition (b), in the Latent Unconfounded Treatment Model, we must demonstrate that 
\begin{flalign*}
\sup_{\eta\in\mathcal{R}_{P,n}}\left(\mathbb{E}_{P}\left[\|\psi_{1}^{\alpha}\left(B;\eta\right)\|^{q}\right]\right)^{1/q} & \leq C^{\prime},\text{ and}\\
\sup_{\eta\in\mathcal{R}_{P,n}}\left(\mathbb{E}_{P}\left[\|\psi_{1}\left(B;\tau_{1},\eta\right)\|^{q}\right]\right)^{1/q} & \leq C^{\prime}
\end{flalign*}
for some $C^{\prime}$. Analogous conditions suffice for the Statistical Surrogacy Model. Verification of the first
inequality is immediate, as
\[
\sup_{\eta\in\mathcal{R}_{P,n}}\left(\mathbb{E}_{P}\left[\|\psi_{1}^{\alpha}\left(B;\eta\right)\|^{q}\right]\right)^{1/q}=1
\quad\text{and}\quad
\sup_{\varphi\in\mathcal{R}_{P,n}}\left(\mathbb{E}_{P}\left[\|\xi_{1}^{a}\left(B;\varphi\right)\|^{q}\right]\right)^{1/q}=1~.
\]
For the Latent Unconfounded Treatment Model,
by the equalities 
\[
\mu_{w}\left(s,x\right)=\mathbb{E}_{P_{\star}}\left[Y\left(w\right)\mid S\left(w\right)=s,X=x\right]\quad\text{and}\quad\bar{\mu}_{w}\left(X\right)=\mathbb{E}_{P_{\star}}\left[Y\left(w\right)\mid X=x\right]~,
\]
we have that
\begin{alignat*}{1}
\|\mu_{w}\left(S,X\right)\|_{P,q}\vee\|\bar{\mu}_{w}\left(X\right)\|_{P,q} & \leq\|Y\left(w\right)\|_{P,q}\leq C
\end{alignat*}
by Jensen's inequality and \cref{as: bounds}, where $\vee$ denotes the
join binary operation. Thus, for any $\tilde{\eta}\in\mathcal{R}_{P,n}$, we
have that
\begin{flalign*}
\|\psi_{1}\left(B;\tau_{1},\tilde{\eta}\right)\|_{P,q} & \leq C_{1}(
\|\tilde{\mu}_{1}\left(S,X\right)\|_{P,q}+\|\tilde{\mu}_{0}\left(S,X\right)\|_{P,q}\\
&\quad+\|\tilde{\bar{\mu}}_{1}\left(X\right)\|_{P,q}+\|\tilde{\bar{\mu}}_{0}\left(X\right)\|_{P,q}
+\|Y\|_{P,q}+\vert\tau_{1}\vert)\\
& \leq C_{1}(\|\mu_{1}\left(S,X\right)-\tilde{\mu}_{1}\left(S,X\right)\|_{P,q}
+\|\tilde{\mu}_{0}\left(S,X\right)-\tilde{\mu}_{0}\left(S,X\right)\|_{P,q}\\
&\quad+\|\mu_{1}\left(S,X\right)\|_{P,q}+\|\mu_{0}\left(S,X\right)\|_{P,q} 
+ \|\bar{\mu}_{1}\left(X\right)-\tilde{\bar{\mu}}_{1}\left(X\right)\|_{P,q}\\
&\quad +\|\bar{\mu}_{1}\left(X\right)-\tilde{\bar{\mu}}_{0}\left(X\right)\|_{P,q}
+\|\tilde{\bar{\mu}}_{1}\left(X\right)\|_{P,q}+\|\tilde{\bar{\mu}}_{0}\left(X\right)\|_{P,q}\\
 &\quad +\|Y\|_{P,q}+\vert\tau_{1}\vert) \leq C_{2},
\end{flalign*}
where the constants $C_{1}$ and $C_{2}$ depend only on $C$
and $\epsilon$, by the \cref{as: overlap}, the triangle inequality, and
\cref{as: rates}.

For the Statistical Surrogacy Model,
by the equalities
\[
\nu\left(s,x\right)=\mathbb{E}_{P}\left[Y\mid S=x,X=x,G=1\right]\quad\text{and}\quad\bar{\nu}_{w}\left(X\right)=\mathbb{E}_{P_{\star}}\left[Y\left(w\right)\mid X=x,W=w,G=0\right]
\]
and the bounds
\begin{align*}
\|\mathbb{E}_{P}\left[Y\mid S=x,X=x\right]\|_{P,q}&\geq\epsilon^{1/q}\|\nu\left(s,x\right)\|_{P,q},\text{ and}\\
\|\mathbb{E}_{P_{\star}}\left[Y\left(w\right)\mid X=x\right]\|_{P,q}
\geq \epsilon^{1/q} \|\mathbb{E}_{P_{\star}}\left[Y\left(w\right)\mid X=x,W=w\right]\|_{P,q}
&\geq \epsilon^{2/q}\|\bar{\nu}_{w}\left(X\right)\|_{P,q}~,
\end{align*}
again we have that
\begin{flalign*}
\|\nu\left(S,X\right)\|_{P,q} 
& \leq \epsilon^{-1/q} \|Y\|_{P,q}\leq2\epsilon^{-1/q}\left(1-\epsilon\right)C
\quad\text{and}\quad
\|\bar{\nu}_{w}\left(X\right)\|_{P,q}\le\epsilon^{-2/q}C
\end{flalign*}
by Jensen's inequality, \cref{as: overlap}, and \cref{as: bounds}. Thus,
for any $\tilde{\varphi}\in\mathcal{R}_{P,n}$, we have by an analogous
argument that 
\begin{flalign*}
\|\xi_{1}\left(B;\tau_{1},\tilde{\varphi}\right)\|_{P,q} & \leq C_{4},
\end{flalign*}
where $C_{4}$ depends only on $C$, $\epsilon$, and $q$, by Assumption 3.1, the triangle inequality, and \cref{as: rates}.

\vspace{1cm}
\noindent
\textbf{Condition (c).} 
To verify condition (c), in the Latent Unconfounded Treatment Model, we must demonstrate that
\begin{flalign}
\sup_{\tilde{\eta}\in\mathcal{R}_{P,n}}\|E_{P}\left[\psi_{1}^{\alpha}\left(B;\tilde{\eta}\right)-\psi_{1}^{\alpha}\left(B;\eta\right)\right]\|_{P,2} & \leq\delta^\prime_{n},\label{eq: c.a}\\
\sup_{\tilde{\eta}\in\mathcal{R}_{P,n}}\left(E_{P}\left[\psi_{1}\left(B;\tau_{1},\tilde{\eta}\right)-\psi_{1}\left(B;\tau_{1},\eta\right)\right]^{2}\|^{1/2}\right) & \leq\delta^\prime_{n},\text{ and}\label{eq: c.b}\\
\sup_{h_{0}\in\left(0,1\right),\tilde{\eta}\in\mathcal{R}_{P,n}}  \norm[\bigg]{
\frac{\partial^
{2}}{\partial h^{2}}\mathbb{E}_{P}\left[\psi_{1}\left(B;\tau_{1},\eta+h\left(\tilde{\eta}-\eta\right)\right)\right]\vert_{h=h_{0}} } & \leq\delta^\prime_{n}/\sqrt{n}~,\label{eq: c.c}
\end{flalign}
for some sequence of positive constants $\{\delta^\prime_n\}_{n\geq1}$ converging to zero. Analogous conditions suffice for the Statistical Surrogacy Model. Verification of (\ref{eq: c.a})
is immediate as 
\[
\|E_{P}\left[\psi_{1}^{\alpha}\left(B;\tilde{\eta}\right)-\psi_{1}^{\alpha}\left(B;\eta_{0}\right)\right]\|_{P,2}=\|\pi^{-1}-\tilde{\pi}^{-1}\|_{P,2}\leq\epsilon^{-2}\delta_{n}
\]
and
\[
\|E_{P}\left[\psi_{1}^{a}\left(B;\tilde{\eta}\right)-\psi_{1}^{a}\left(B;\eta_{0}\right)\right]\|=\|\pi^{-1}-\tilde{\pi}^{-1}\|_{P,2}\leq\epsilon^{-2}\delta_{n}
\]
for any $\tilde{\eta}$ and $\tilde{\varphi}$ in their respective
realization sets.

For the Latent Unconfounded Treatment Model,
observe that for any $\tilde{\eta}\in\mathcal{R}_{P,n}$, the triangle
inequality gives the bound
\[
\|\psi_{1}\left(B;\tau_{1},\tilde{\eta}\right)-\psi_{1}\left(B;\tau_{1},\eta_{0}\right)\|_{P,2}\leq\mathcal{I}_{1,1}+\mathcal{I}_{1,0}+\mathcal{I}_{2}+\mathcal{I}_{3}+\mathcal{I}_{4,1}+\mathcal{I}_{4,0}
\]
where
\begin{flalign*}
\mathcal{I}_{1,1} & = \norm[\bigg]{\frac{G}{\tilde{\pi}}\left(\frac{W\left
(Y-\tilde{\mu}_{1}\left(S,X\right)\right)}{\tilde{\rho}_{1}\left(S,X\right)}\right)-\frac{G}{\pi}\left(\frac{W\left(Y-\mu_{1}\left(S,X\right)\right)}{\rho_{1}\left(S,X\right)}\right)}_{P,2},\\
\mathcal{I}_{1,0} & = \norm[\bigg]{ \frac{G}{\tilde{\pi}}\left(\frac{\left
(1-W\right)\left(Y-\tilde{\mu}_{0}\left(S,X\right)\right)}{\tilde{\rho}_
{0}\left(S,X\right)}\right)-\frac{g}{\pi}\left(\frac{\left(1-W\right)\left
(Y-\mu_{0}\left(S,X\right)\right)}{\rho_{0}\left(S,X\right)}\right) } _{P,2},\\
\mathcal{I}_{2} & = \norm[\bigg]{\frac{G}{\tilde{\pi}}\tilde{\bar{\mu}}_
{1}\left(X\right)-\frac{G}{\pi}\bar{\mu}_{1}\left(x\right)\bigg\|_{P,2}+\bigg\|
\frac{G}{\tilde{\pi}}\tilde{\bar{\mu}}_{0}\left(x\right)-\frac{G}{\pi}\bar{\mu}_{0}\left(x\right) } _{P,2},\\
\mathcal{I}_{3} & = \norm[\bigg]{\left(\pi^{-1}-\tilde{\pi}^
{-1}\right)G\tau_{1} } _{P,2},\\
\mathcal{I}_{4,1} & = \bigg\|\frac{G}{\tilde{\pi}}\left(\frac{
\tilde{\gamma}\left(X\right)}{1-\tilde{\gamma}\left(X\right)}\frac{W\left(\tilde{\mu}_{1}\left(S,X\right)-\tilde{\bar{\mu}}_{1}\left(X\right)\right)}{\tilde{\varrho}\left(X\right)}\right)\\
&-\frac{G}{\pi}\left(\frac{\gamma\left(X\right)}{1-\gamma\left(X\right)}
\frac{W\left(\mu_{1}\left(S,X\right)-\bar{\mu}_{1}\left(X\right)\right)}
{\varrho\left(X\right)}\right)\bigg \|_{P,2},\text{ and}\\
\mathcal{I}_{4,0} & = \bigg\|\frac{G}{\tilde{\pi}}\left(\frac{
\tilde{\gamma}\left(X\right)}{1-\tilde{\gamma}\left(X\right)}\frac{\left(1-W\right)\left(\tilde{\mu}_{0}\left(S,X\right)-\tilde{\bar{\mu}}_{0}\left(X\right)\right)}{\tilde{\varrho}\left(X\right)}\right)\\
& -\frac{G}{\pi}\left(\frac{\gamma\left(X\right)}{1-\gamma\left(X\right)}
\frac{\left(1-W\right)\left(\mu_{0}\left(S,X\right)-\bar{\mu}_{0}\left
(X\right)\right)}{1-\tilde{\varrho}\left(X\right)}\right) \bigg\|_{P,2}~.
\end{flalign*}
Observe that 
\begin{flalign*}
\mathcal{I}_{1,1} & \leq\epsilon^{-2} \norm[\bigg]{\pi\left(\frac{W\left(Y-
\tilde{\mu}_{1}\left(S,X\right)\right)}{\tilde{\rho}_{1}\left(S,X\right)}\right)-\tilde{\pi}\left(\frac{W\left(Y-\mu_{1}\left(S,X\right)\right)}{\rho_{1}\left(S,X\right)}\right)}_{P,2}\\
 &
 \leq\epsilon^{-2} \norm[\bigg]{\frac{W\left(Y-\tilde{\mu}_{1}\left
 (S,X\right)\right)}{\tilde{\rho}_{1}\left(S,X\right)}-\frac{W\left(Y-\mu_
 {1}\left(S,X\right)\right)}{\rho_{1}\left(S,X\right)}}_{P,2}+  \norm
 [\bigg]{\left(\pi-\tilde{\pi}\right)\left(\frac{W\left(Y-\mu_{1}\left
 (S,X\right)\right)}{\rho_{1}\left(S,X\right)}\right)}_{P,2}\\
 &
 \leq\epsilon^{-4}\left(\|\left(\rho_{1}\left(S,X\right)-\tilde{\rho}_{1}\left(S,X\right)\right)\left(Y-\mu_{1}\left(S,X\right)\right)\|_{P,2}+\|\rho_{1}\left(S,X\right)\left(\mu_{1}\left(S,X\right)-\tilde{\mu}_{1}\left(S,X\right)\right)\|_{P,2}\right)\\
 &
 +\epsilon^{-3} \norm[\bigg]{\left(\pi-\tilde{\pi}\right)\left(Y-\mu_
 {1}\left(S,X\right)\right)}_{P,2}\\
 &
 \leq\epsilon^{-4}\left(\sqrt{C}\|\rho_{1}\left(S,X\right)-\tilde{\rho}_{1}\left(S,X\right)\|_{P,2}+\|\mu_{1}\left(S,X\right)-\tilde{\mu}_{1}\left(S,X\right)\|_{P,2}\right)+\epsilon^{-3}\sqrt{C}\|\pi-\tilde{\pi}\|_{P,2}\\
 &
 \leq\epsilon^{-3}\left(\epsilon^{-1}\left(\sqrt{C}+1\right)+\sqrt{C}\right)\delta_{n},
\end{flalign*}
where the third inequality follows from $\mathbb{E}\left[\sigma^{2}\left(S,X\right)\mid X\right]\leq C$
and $\rho_{1}\left(S,X\right)\leq1$. An analogous argument verifies
that $\mathcal{I}_{1,0}\leq\epsilon^{-3}\left(\epsilon^{-1}\left(\sqrt{C}+1\right)+\sqrt{C}\right)\delta_{n}$.
Likewise, we can see that
\begin{flalign*}
\mathcal{I}_{2} & \leq\epsilon^{-2}\left(\|\tilde{\bar{\mu}}_{1}\left(X\right)-\bar{\mu}_{1}\left(x\right)\|_{P,2}+\|\tilde{\bar{\mu}}_{0}\left(X\right)-\bar{\mu}_{0}\left(X\right)\|_{P,2}\right)\leq2\epsilon^{-2}\delta_{n}
\end{flalign*}
and that 
\[
\mathcal{I}_{3}\leq\epsilon^{-2}\|\left(\tilde{\pi}-\pi\right)\tau_{1}\|_{P,2}\leq\epsilon^{-2}2C\delta_{n}.
\]
 Finally, if we define $\upsilon\left(X\right)$ by
\[
\frac{1}{\upsilon\left(X\right)}=\frac{\gamma\left(X\right)}{\pi\left(1-\gamma\left(X\right)\right)\varrho\left(X\right)},
\]
and observe that $\upsilon\left(X\right)\geq\epsilon^{3}/\left(1-\epsilon\right)=\epsilon_{\star}$,
we have that 
\begin{flalign*}
\mathcal{I}_{4,1} & \leq\epsilon_{\star}^{-2}\|\upsilon\left(X\right)\left(\tilde{\mu}_{1}\left(S,X\right)-\tilde{\bar{\mu}}_{1}\left(X\right)\right)-\tilde{\upsilon}\left(X\right)\left(\mu_{1}\left(S,X\right)-\bar{\mu}_{1}\left(X\right)\right)\|_{P,2}\\
 & \leq\epsilon_{\star}^{-2}\|\upsilon\left(X\right)\left(\tilde{\mu}_{1}\left(S,X\right)-\mu_{1}\left(S,X\right)+\mu_{1}\left(S,X\right)-\bar{\mu}_{1}\left(X\right)+\bar{\mu}_{1}\left(X\right)-\tilde{\bar{\mu}}_{1}\left(X\right)\right)\\
 & \quad\quad-\tilde{\upsilon}\left(X\right)\left(\mu_{1}\left(S,X\right)-\bar{\mu}_{1}\left(X\right)\right)\|_{P,2}\\
 & \leq\epsilon_{\star}^{-4}\Big(\|\upsilon\left(X\right)\left(\tilde{\mu}_{1}\left(S,X\right)-\mu_{1}\left(S,X\right)\right)\|_{P,2}+\|\left(\upsilon\left(X\right)-\tilde{\upsilon}\left(X\right)\right)\left(\mu_{1}\left(S,X\right)-\bar{\mu}_{1}\left(X\right)\right)\|_{P,2}\\
 & \quad\quad+\|\upsilon\left(X\right)\left(\bar{\mu}_{1}\left(X\right)-\tilde{\bar{\mu}}_{1}\left(X\right)\right)\|_{P,2}\Big)\leq C_{5}\delta_{n}~,
\end{flalign*}
where $C_{5}$ depends only on $\epsilon$ and $C$, with an analogous
argument giving $\mathcal{I}_{4,1}\leq C_{5}\delta_{N}$, verifying
(\ref{eq: c.b}) for the Latent Unconfounded Treatment Model.

In turn, for the Statistical Surrogacy Model, observe that for any $\tilde{\varphi}\in\mathcal{R}_{P,n}$,
the triangle inequality gives the bound
\[
\|\xi_{1}\left(B;\tau_{1},\tilde{\varphi}\right)-\xi_{1}\left(B;\tau_{1},\varphi_{0}\right)\|_{P,2}\leq\mathcal{J}_{1,1}+\mathcal{J}_{1,0}+\mathcal{J}_{2}+\mathcal{J}_{3}+\mathcal{J}_{4,1}+\mathcal{J}_{4,0}
\]
where
\begin{flalign*}
\mathcal{J}_{1,1} & =\bigg \|\frac{G}{\tilde{\pi}}\left(\frac{\tilde{\gamma}\left(X\right)}{\tilde{\gamma}\left(S,X\right)}\frac{1-\tilde{\gamma}\left(X\right)}{1-\tilde{\gamma}\left(S,X\right)}\frac{\tilde{\varrho}\left(S,X\right)\left(Y-\tilde{\nu}\left(S,X\right)\right)}{\tilde{\varrho}\left(X\right)}\right)\\
 & \quad-\frac{G}{\pi}\left(\frac{\gamma\left(X\right)}{\gamma\left
 (S,X\right)}\frac{1-\gamma\left(X\right)}{1-\gamma\left(S,X\right)}
 \frac{\varrho\left(S,X\right)\left(Y-\nu\left(S,X\right)\right)}
 {\varrho\left(X\right)}\right)\bigg \|_{P,2},\\
\mathcal{J}_{1,0} & =\bigg \|\frac{G}{\tilde{\pi}}\left(\frac{\tilde{\gamma}\left(X\right)}{\tilde{\gamma}\left(S,X\right)}\frac{1-\tilde{\gamma}\left(X\right)}{1-\tilde{\gamma}\left(S,X\right)}\frac{\left(1-\tilde{\varrho}\left(S,X\right)\right)\left(Y-\tilde{\nu}\left(S,X\right)\right)}{1-\tilde{\varrho}\left(X\right)}\right),\\
 & \quad-\frac{G}{\pi}\left(\frac{\gamma\left(X\right)}{\gamma\left
 (S,X\right)}\frac{1-\gamma\left(X\right)}{1-\gamma\left(S,X\right)}
 \frac{\left(1-\varrho\left(S,X\right)\right)\left(y-\nu\left(S,X\right)\right)}{1-\varrho\left(X\right)}\right)\bigg\|_{P,2}\\
\mathcal{J}_{2} & =\bigg \|\frac{G}{\tilde{\pi}}\tilde{\bar{\nu}}_{1}\left
(X\right)-\frac{G}{\pi}\bar{\nu}_{1}\left(x\right) \bigg{\|}_{P,2}+\bigg{\|}\frac{G}{
\tilde{\pi}}\tilde{\bar{\nu}}_{0}\left(x\right)-\frac{G}{\pi}\bar{\nu}_{0}\left(x\right)\bigg\|_{P,2},\\
\mathcal{J}_{3} & =\bigg \|\left(\pi^{-1}-\tilde{\pi}^{-1}\right)G\tau_
{1}\bigg\|_{P,2},\\
\mathcal{J}_{4,1} & =\bigg \|\frac{G}{\tilde{\pi}}\left(\frac{W\left(
\tilde{\nu}_{1}\left(S,X\right)-\tilde{\bar{\nu}}_{1}\left(X\right)\right)}
{\tilde{\varrho}\left(X\right)}\right)-\frac{G}{\pi}\left(\frac{W\left(\nu_{1}\left(S,X\right)-\bar{\nu}_{1}\left(X\right)\right)}{\varrho\left(X\right)}\right)\bigg\|_{P,2},\text{ and}\\
\mathcal{J}_{4,0} & =\bigg \|\frac{G}{\tilde{\pi}}\left(\frac{\left
(1-W\right)\left(\tilde{\nu}_{0}\left(S,X\right)-\tilde{\bar{\nu}}_{0}\left
(X\right)\right)}{\tilde{\varrho}\left(X\right)}\right)-\frac{G}{\pi}\left(
\frac{\left(1-W\right)\left(\nu_{0}\left(S,X\right)-\bar{\nu}_{0}\left(X\right)\right)}{1-\tilde{\varrho}\left(X\right)}\right) \bigg\|_{P,2}.
\end{flalign*}
Bounds for each term follow similar arguments to their analogues for
the case where treatment is observed in the observational sample. We omit their explicit derivation to avoid repetition, concluding the verification of (\ref{eq: c.b}).

The expressions for the first and second Gateaux derivatives
\[
\frac{\partial}{\partial h}\mathbb{E}_{p}\left[\psi_{1}\left(B;\tau_{1},\eta+h\left(\tilde{\eta}-\eta\right)\right)\right]\quad\text{and}
\quad\frac{\partial^2}{\partial h^2}\mathbb{E}_{p}\left[\psi_{1}\left(B;\tau_{1},\eta+h\left(\tilde{\eta}-\eta\right)\right)\right]
\]
as well as
\[
\frac{\partial}{\partial h}\mathbb{E}_{p}\left[\xi_{1}\left(B;\tau_{1},\varphi+h\left(\tilde{\varphi}-\varphi\right)\right)\right]
\quad\text{and}\quad
\frac{\partial^2}{\partial h^2}\mathbb{E}_{p}\left[\xi_{1}\left(B;\tau_{1},\varphi+h\left(\tilde{\varphi}-\varphi\right)\right)\right]
\]
are very tedious. In particular, the second Gateaux derivatives each
take more than a page to display in the present formatting. They are
available upon request, but are omitted for the sake of clarity. The
second Gateaux derivatives can be written as a sum of several terms,
each upper bounded, by Hölder's inequality, by the product between the $L_{2}\left(P\right)$ norms
of the differences of components of $\tilde{\omega}-\omega$ and
$\tilde{\kappa}-\kappa$, if treatment is observed in the observational
sample, and $\tilde{\vartheta}-\vartheta$ and $\tilde{\zeta}-\zeta$,
if treatment is not observed in the observational sample. Thus, by \cref{as: rates}, (\ref{eq: c.c}) is satisfied in each case, completing the verification of condition (c). We refer the reader to the final passage in the proof of Theorem 5.1 of \cite{chernozhukov2018double} for an example of an argument with an identical structure for the substantially simpler context of the average treatment effect
under ignorability. 

\vspace{1cm}
\noindent
\textbf{Condition (d).}
Finally, to verify condition (d), we must demonstrate that $V_{1}^{\star}$
and $V_{1}^{\star\star}$ are non degenerate. Observe that 
\begin{flalign*}
V_{1}^{\star} & =\mathbb{E}_{P}\left[\frac{\gamma\left(X\right)}{\pi}\left(\frac{\sigma_{1}^{2}\left(S,X\right)}{\rho_{1}\left(S,X\right)}+\frac{\sigma_{0}^{2}\left(S,X\right)}{\rho_{0}\left(S,X\right)}+\left(\bar{\mu}_{1}\left(X\right)+\bar{\mu}_{0}\left(X\right)-\tau_{1}\right)^{2}+\Gamma_{0}\left(S,X\right)+\Gamma_{1}\left(S,X\right)\right)\right]\\
 & \geq\frac{\epsilon}{\pi\left(1-\epsilon\right)}\mathbb{E}_{P}\left[\sigma_{1}^{2}\left(S,X\right)+\sigma_{0}^{2}\left(S,X\right)\right]\\
 & +\frac{\epsilon^{2}}{\pi\left(1-\epsilon\right)^{2}}\mathbb{E}_{P}\left[\left(\mu_{1}\left(S,X\right)-\bar{\mu}_{1}\left(X\right)\right)^{2}+\left(\mu_{0}\left(S,X\right)-\bar{\mu}_{0}\left(X\right)\right)^{2}\right]\\
 & \geq2\left(\frac{\epsilon}{\pi\left(1-\epsilon\right)}+\frac{\epsilon^{2}}{\pi\left(1-\epsilon\right)^{2}}\right)c,
\end{flalign*}
as required. An analogous argument holds for $V_{1}^{\star\star}$, completing the proof.\hfill\qed

\subsection{Proof of \cref{thm:sieve_estimation}}\label{sec: proof of sieve}

\newcommand{\nuisvector}{\zeta}
\newcommand{\piest}{\hat \pi_n}

We begin by giving an explicit construction for the estimator
\[
\hat{\theta}_{1,1}(g_{\mathsf{w}}) 
= \frac{1}{n} \sum_{i=1}^n  \frac{G_i}{\piest} \frac{W_i Y_i}{\hat{\rho}_{1n}(S_i,X_i)}
\]
for the parameter $\theta_{1,1} = \E[Y(1) \mid G=1]$. We then state sufficient conditions for its asymptotic linearity and Gaussianity. The construction and sufficient conditions for the corresponding estimator of the parameter $\theta_{1,0} =\E[Y(0) \mid G=1]$ are analogous. To ease notation, we drop the subscript $\mathsf{w}$ from the parameter $\zeta_{\mathsf{w}}$. 
We will use the standard empirical process notation $\Pn[n] f = \frac{1}{n} \sum_{i=1}^n f(B_i)$, $\Pn[] f = \E_P[f(B_i)]$, and $\Gn = \sqrt{n}(\Pn[n] - \Pn[])$, where expectations are taken over $B_i$ and not over $f$ if $f$ is random. Our assumptions and argument closely follow \citet{chen2015sieve}. 

Observe that
\[
\rho_1(s,x) = \frac{\zeta_{1}(s,x)}{1-\zeta_{1}(s,x)} \zeta_{2}(x) \frac{1-\zeta_{3}(x)}{\zeta_{3}(x)}~, \numberthis
\label{eq:rho_1_def ap}
\]
where 
\begin{align*}
\zeta_{1}(s,x) &= \P(G=1 \mid S=s, W=1, X=x),\\
\zeta_{2}(x) &= \P(W=1 \mid X=x, G=1),\quad\text{and}\quad
\zeta_{3}(x) = \P(G=1 \mid W=1, X=x).
\end{align*}
The parameter $\nuisvector = (\zeta_{1}, \zeta_{2}, \zeta_{3})$ can be estimated with the sieve extremum estimator
\begin{align*}
\hat\nuisvector_{n} 
&= \argmin_{h \in \mathcal H_n} \frac{1}{2n} \sum_{i=1}^n W_i(G_i-h_1(S_i,X_i))^2 +G_i(W_i - h_2(X_i))^2 + W_i(G_i - h_3(X_i))^2 \\ 
&= \argmin_{h \in \mathcal H_n} \Pn[n]\L(\cdot, h)~,
\end{align*} where $\mathcal H_n$ is a sieve approximation of the parameter space for
 $\nuisvector \in \mathcal H$.  The estimator
 $\hat \rho_{1n}$ is then derived from plugging $\hat \nuisvector_n$ into \eqref
 {eq:rho_1_def ap}. The parametric nuisance parameter $\pi$ can be estimated by its sample
 analogue $\piest = \frac{1}{n} \sum_{i=1}^n G_i.$ 

We begin by imposing a set of regularity conditions on the data generating distribution $P$, the estimator $\hat\nuisvector_{n}$, and the parameter space $\mathcal H$.
\begin{assumption}
\label{as:sieve_as_1}
Assume that 
\begin{enumerate}
    \item The parameter $\theta_{1,1} \in \mathrm{int}(\Theta) \subset \R$, where $\Theta$ is compact.
    \item The nuisance parameters $\pi$ and $\nuisvector$ are bounded between $\varepsilon$ and $1-\varepsilon$, $\lambda$-almost surely, for some fixed constant $0<\varepsilon<1/2$.
    \item The long-term outcome $Y_i$ has a finite fourth
    moment, i.e., $\E_P[|Y_i|^4] < \infty$.
    \item The estimator $\hat\nuisvector_n$ is $n^{-1/4}$-consistent in $\|\cdot\|_\infty$, where $\|\nuisvector\|_\infty = \max_j \|\zeta_j\|_\infty$, in the sense that there exists a sequence $0 \le \delta_n = o(n^{-1/2})$ such that
    \[
    \delta_n^{-1} \norm{\hat \nuisvector_n - \nuisvector}^2_\infty \overset{p}{\to} 0
    \]
    and entries of $\hat \nuisvector_{n}(\cdot)$ belong to $[\epsilon, 1-\epsilon]$, $\lambda$-almost surely. 
    \item The parameter space $\mathcal H = \mathcal H_1 \times \mathcal H_2 \times
    \mathcal H_3$ is a product of Donsker classes.
\end{enumerate}
\end{assumption}

Next, we introduce more structure on the sieve space
$\mathcal H_n$. We assume that $\mathcal H_n$ is dense in $\mathcal
H$ in $\norm{\cdot}_\infty$ as $n \to \infty$. Define 
\begin{align*}
\bGamma[v]
&= 
\diff{}{\nuisvector} \E_P\bk{
    \frac{GWY}{\pi \rho_1}
} [v] \\ 
&= \frac{d}{dt} \E_P\bk{
    \frac{GWY}{\pi \frac{\zeta_1 + tv_1}{1-\zeta_1 -tv_1} (\zeta_2 + tv_2) 
    \frac{1-\zeta_3 - tv_3}{\zeta_3 + tv_3}}
}\evalbar_{t=0}
\\ 
&= - \E_P\bk{
    \frac{GWY}{\pi \rho_1(S,X)} \br{
        \frac{1}{\zeta_1(1-\zeta_1)} v_1 + \frac{1}{\zeta_2}v_2 + \frac{1}{\zeta_3
        (1-\zeta_3)}v_3
    }
}
\numberthis \label{eq:Gamma_def}
\end{align*}
as the pathwise derivative of $\theta_{1,1}$ in $\nuisvector$ in the direction of $v = (v_1, v_2, v_3)$. 
Similarly, define the inner product \[
\ip{v,u} = \E_P[W(v_1 u_1 + v_3 u_3) + G v_2 u_2] = \diff{^2}{t_1 \partial t_2} \E_P[\L
(B, \zeta+t_1 v + t_2 u)]\evalbar_{t_1=t_2=0}
\numberthis \label{eq:inner_product_def}
\]
as the cross-derivative of the first-step criterion function $L$ and let $\norm{\cdot}$ denote the norm induced by $\ip{\cdot, \cdot}$.
Define \[\mathcal V = 
\mathrm{CL}\pr{\br{h -
\nuisvector = (h_1 - \zeta_1, h_2 - \zeta_2, h_3 - \zeta_3) : h \in \mathcal H}}\] as the
closed linear span of the re-centered parameter space under $\norm{\cdot}$.\footnote{$
\mathrm{CL}$ takes the
closed linear span under $\norm{\cdot}$.}
Note that $\bGamma : \mathcal V \to
\R$ is a bounded linear
functional on $\mathcal V$. 
Define $v^* \in \mathcal V$ as the Riesz representer of $\bGamma$ in $\ip{\cdot, \cdot}$:
That is, for all $w \in \mathcal V$, we have that
\[
\ip{v^*, w} = \Gamma[w].
\] Similarly, let $\nuisvector_n = \argmin_{h \in \mathcal H_n} \norm{h-  \nuisvector}$ be the
projection of $\nuisvector$ to the sieve space in $\norm{\cdot}$. Define \[
\mathcal V_n = \mathrm{CL}\pr{\br{
    h - \zeta_n : h \in \mathcal H_n
}}
\]
as the re-centered sieve space. Lastly, let $v^*_n$ be the Riesz representer in $\mathcal V_n$ of
$\bGamma [\cdot]$ in $\ip{\cdot,\cdot}$. We refer to $v^*$ as the Riesz representer and
$v^*_n$ as the sieve Riesz representer. Note that these objects are defined under the inner product $
\ip{\cdot, \cdot}$, which is different from how Riesz representers are defined in our
efficiency calculations. 

With this notation in place, we impose the following conditions. 
\begin{assumption}
\label{as:sieve_as_2} $\text{ }$
\begin{enumerate}
        \item The rate condition $\norm{\hat \nuisvector_n - \nuisvector} \norm{v_n^* - v^*} = o_p(n^{-1/2})$ is satisfied.
        \item Recall $\delta_n$ from \cref{as:sieve_as_1}(4). Assume that the function classes \[\mathcal F_n = \br{h - \nuisvector : h \in \mathcal H_n, 
        \norm{h-\nuisvector}_\infty < \delta_n^{1/2}}\] have finite uniform entropy
        integrals $J(1, \mathcal F_n)$. The entropy integral $J(1, \mathcal F_n)$ is defined as
        \[
J(1, \mathcal F_n) = \sup_Q \int_0^1 \sqrt{1 + \log N\pr{ \epsilon \delta_n^{1/2},
\mathcal
F_n, L_2(Q)}}\, d \epsilon < C < \infty,
        \]
        where $Q$ ranges over all discrete distributions, $N(\epsilon \delta_n^{1/2}, \mathcal F, L_2(Q))$ is
        a covering number in $L_2(Q)$-norm. We choose the envelope function $F$ for $\mathcal F_n$ as the constant
        function $F  = \delta_n^{1/2}$.
    \end{enumerate}
\end{assumption}

\cref{as:sieve_as_1,as:sieve_as_2} are sufficient for the asymptotic linearity and Gaussianity of $\hat{\theta}_{1,1}(g_{\mathsf{w}})$. The following assumption collects these conditions for reference in the main text. 

\begin{assumption}
\label{as:sieve_as_3}
\Cref{as:sieve_as_1,as:sieve_as_2}, in addition to their analogues that pertain to the corresponding weighted estimator of $\theta_{1,0}=\E[Y(0) \mid G=1]$, hold.
\end{assumption}
\noindent
Under \cref{as:sieve_as_3}, we can readily verify Theorem~\ref{thm:sieve_estimation}, stated in the
main text using Corollary \cref{cor:sieve_main} stated in \cref{sec: linearity}.
\begin{proof}[Proof of Theorem~\ref{thm:sieve_estimation}]
By an application of \cref{cor:sieve_main} (and its analogue for $\E[Y(0) \mid G=1]$), we have
that \[
\sqrt{n} (\hat{\tau}_{1}(g_{\mathsf{w}}) - \tau_1) = \frac{1}{\sqrt{n}} \sum_{i=1}^n \psi_1
(B_i, \tau_1, \eta) + o_p(1).
\]
where $\psi_1(\cdot)$ is the efficient influence function for $\tau_1$ under the Latent Unconfounded Treatment Model, derived in \cref{thm: EIF tau_1}. As we have assumed that $\E_P[Y_i^4] < \infty$ and by 
and strict overlap, i.e., \cref{as: overlap}, we can readily verify that 
\[
\E_P[|\psi_1(B_i, \tau_1, \eta)|^{2+\epsilon}] < \infty 
\]
for some $\epsilon > 0$. As a result, by the Lyapunov central limit theorem, \[
\frac{1}{\sqrt{n}} \sum_{i=1}^n \psi_1
(B_i, \tau_1, \eta) \dto \Norm(0, V_1^\star),
\]
by the definition of $V_1^\star$ given in \cref{cor: seb}. The proof is complete with an application of the continuous mapping theorem and Slutsky's theorem.\hfill
\end{proof}

\subsubsection{Asymptotic linearity\label{sec: linearity}}

Define \[
\Delta_i[v] = -\diff{\L(B_i, \nuisvector)}{\nuisvector}[v] = W_i(G_i-\zeta_1) v_1 + G_i(W_i-\zeta_2) v_2 + W_i
(G_i-\zeta_3) v_3. \numberthis
\label{eq:delta_def}
\]
as the negative pathwise derivative of the first-stage loss function with respect to $\nuisvector$ in the direction of $v$, evaluated at the true nuisance $\nuisvector$. For a test value $\tilde\nuisvector$ of the nuisance parameter $\nuisvector$,  
let an infeasible moment condition be \[
\tilde{g}_1(B_i, \theta_{1,1}, \chi) = \frac{G_iW_iY_i}{\pi \tilde\rho_1(S, X)} - \theta_{1,1}, \numberthis 
\label{eq:moment_defn_sieve}
\]
where $\tilde \rho_1$ is derived from plugging $\chi$ into \eqref{eq:rho_1_def ap}. $\tilde{g}_1$ is a moment condition that assumes knowledge of the true $\pi$. 

Let $\tilde\theta_{1,1}$ be the $Z$ estimator based on the moment condition
\begin{align}
\Pn[n] \tilde{g}_1(\cdot, \tilde\theta_{1,1}, \hat\nuisvector_n) = 0. \label{eq: moment}
\end{align}
The following theorem applies Theorem 2.1 from \citet{chen2015sieve} to the moment
condition \eqref{eq: moment}. The result is key in that it shows that the influence of the sieve
estimation of $\nuisvector$ on $\theta_{1n}$ is asymptotically linear and is given by $\Delta_i[v^*]$.

\begin{theorem}
\label{thm:sieve_main}
Under \cref{as:sieve_as_1,as:sieve_as_2}, we have that
\[
\sqrt{n} (\tilde\theta_{1,1} - \theta_{1,1}) = \frac{1}{\sqrt{n}} \sum_{i=1}^n \tilde{g}_1(B_i, \theta,
\nuisvector) + \Delta_i[v^*] + o_p(1).
\]
\end{theorem}

\begin{cor}
\label{cor:sieve_main}
Under \cref{as:sieve_as_1,as:sieve_as_2}, we have that
\begin{align}
\sqrt{n} (\hat\theta_{1,1} - \theta_{1,1}) &= \frac{1}{\sqrt{n}} \sum_{i=1}^n \bk{\tilde{g}_1(B_i, \theta,
\nuisvector) + \Delta_i[v^*] + \theta_{1,1} - \frac{G_i \theta_{1,1}  }{\pi}} + o_p(1) \label{eq: initial expansion}\\
&= \frac{1}{\sqrt{n}} \sum_{i=1}^n
\psi_{1,1}(B_i) + o_p(1)~, \label{eq: eif expansion}
\end{align}
where $\psi_{1,1}(\cdot)$ is the efficient influence function for $\theta_{1,1}$ under the Latent Unconfounded Treatment Model, derived in the the proof of \cref{thm: EIF tau_1}.
\end{cor}
\begin{proof}[Proof of \cref{cor:sieve_main}]

Note that \begin{align*}
\sqrt{n}(\hat\theta_{1,1} - \theta_{1,1} ) &= \sqrt{n} \pr{\frac{\pi}{\piest} - 1} \theta_{1,1}  + \sqrt{n} 
(\tilde\theta_{1,1} - \theta_{1,1}) + \sqrt{n} 
(\tilde\theta_{1,1} - \theta_{1,1}) \pr{\frac{\pi}{\piest} - 1}  \\ 
&= \frac{1}{\sqrt{n}} \sum_{i=1}^n \bk{\tilde{g}_1(B_i, \theta_{1,1},
\nuisvector) + \Delta_i[v^*] + \theta_{1,1} - \frac{G_i\theta_{1,1}}{\pi}} + o_p(1)~,
\end{align*}
where the second equality follows from \cref{lemma:consistency_sieve}, stated in \cref{sec: sieve aux lemmas}, and 
\cref{thm:sieve_main}. This verifies \eqref{eq: initial expansion}. The equality \eqref{eq: eif expansion} follows from 
\cref{lemma:eif_checks_out}, stated in \cref{sec: checks out}.\hfill
\end{proof}

Towards verifying \cref{thm:sieve_main}, we state the assumptions and main result of \citet
{chen2015sieve}. For that purpose, let $g(\cdot, \theta, \zeta)$ be a generic moment
condition, with true parameter values $(\theta, \zeta)$ and generic parameter values $
(\vartheta, h)$. Let $L(\cdot, \cdot)$ a generic first-stage loss function. Define the
generic counterparts of other objects analogously, relative to $g$ and to $L$. Let
$\bGamma_1(\vartheta)$ be the derivative of $\E[g(\cdot, \theta, \zeta)]$ in $\theta$,
evaluated at $\vartheta$. Let $\bGamma_2(\vartheta)[v]$ be the pathwise derivative of $\E[g
(\cdot, \theta, \zeta)]$ in $\zeta$, evaluated at $\vartheta$ and $\zeta$, in the direction
$v$. \citet{chen2015sieve} maintain the following high-level assumptions.

\begin{assumption}[Assumption A.1 in \citet{chen2015sieve}]
\label{as:a1}
Assume that $\lim_{n \to \infty} \norm{v^*_n} < \infty$. 
Additionally assume that the estimator $\hat\zeta_n$ satisfies:
\begin{enumerate}
    \item $\abs{\frac{1}{n} \sum_{i=1}^n \Delta_i[v^*_n] - \ip{v^*_n, \hat \zeta_n -
        \zeta}} = o_p(n^{-1/2})$
    \item   $\norm{\hat \zeta_n - \zeta} \cdot \norm{v^*_n - v^*} = o_p(n^{-1/2})$
\end{enumerate}
\end{assumption}

\begin{assumption}[Assumption A.2 in \citet{chen2015sieve}]
\label{as:a2}
Assume  $\theta$ is in the interior of $\Theta$ and $\tilde\theta \overset{p}{\to} \theta$. Assume
additionally that
\begin{enumerate}
    \item The derivative \[\bGamma_1(\vartheta) = \diff{\E_P[g
    (B_i,\vartheta, \zeta)]}{\vartheta}\] exists in a neighborhood of
    $\theta$ and is continuous at $\theta$. 
    \item The pathwise derivative \[
\bGamma_2(\vartheta)[v] = \diff{\E_P[g(B_i, \vartheta, \zeta)]}{\zeta}[v]
    \]
    exists in all directions $v$ and satisfies \[
\abs{\bGamma_2(\vartheta)[v] - \bGamma_2(\theta)[v]} \le
|\vartheta - \theta| \cdot o(1).
    \]
    \item We have \[
\abs{
    \Pn[n] g(\cdot, \vartheta, \hat\zeta_n) - \Pn[] g(\cdot, \vartheta, \zeta) -
    \bGamma_2(\vartheta)[\hat\zeta_n - \zeta]} = o_p(n^{-1/2}).
    \]
    for all $\vartheta = \theta + o(1)$
    \item For all sequences $0 < \kappa_n = o(1)$, \[
\sup_{\abs{\vartheta - \theta} < \kappa_n, \norm{h-\zeta}_\infty < \kappa_n} 
    \abs{\Gn[n] [g(\cdot, \vartheta, h) - g(\cdot, \theta, \zeta)]}
 = o_p(1). 
    \]
\end{enumerate}
\end{assumption}

\begin{theorem}[Theorem 2.1 in \citet{chen2015sieve}]
\label{thm:chen_liao_thm}
Let $\tilde\theta_n$ be the estimator based on $g$: That is, \[
\Pn[n] g(\cdot, \tilde\theta_n, \hat\nuisvector_n) = 0.
\]
Under \cref{as:a1,as:a2}, \[
\sqrt{n}(\tilde\theta_n - \theta) = \frac{1}{\sqrt{n}} \sum_{i=1}^n g(B_i, \theta, \zeta) +
\Delta_i[v^*].
\]
\end{theorem}

\begin{proof}[Proof intuition for \cref{thm:chen_liao_thm}]
The influence of the first-step estimation is of the form \[
\sqrt{n} \bGamma_2(\theta)[\hat\zeta_n - \zeta] = \sqrt{n} \ip{\hat\zeta_n - \zeta,
v^*_k},
\]
for which we use \cref{as:a2}. The linearity of $ \ip{\hat\zeta_n - \zeta, v^*_k}$ is
assumed in \cref{as:a1}(1). As $n\to
\infty$, $v_k^* \to v^*$, and \cref{as:a1}(2) controls the corresponding difference
between $\ip{v_n^*, \hat \zeta_n - \zeta}$ and $\ip{v^*, \hat \zeta_n - \zeta}$.\hfill
\end{proof}

\begin{proof}[Proof of \cref{thm:sieve_main}]
It suffices to show that, in our setting, with $g$ chosen as $\tilde g$,  \cref{as:a1,as:a2} follow from \cref{as:sieve_as_1,as:sieve_as_2}. \Cref{as:a1}(1) follows from \cref{lemma:sieve_main_assumption_justify}, stated in \cref{sec: sieve aux lemmas}. 
\Cref{as:a1}(2) is assumed as \cref{as:sieve_as_2}(1). Lastly, note that $\mathcal V_n$ consist of
functions that are uniformly bounded, since entries of $\nuisvector$ are functions taking values
in $[\varepsilon, 1-\varepsilon]$. Thus, $\lim_{n\to\infty} \norm{v_n^*} \le \limsup_{n\to\infty}
\norm{v_n^*}_\infty < \infty$. This verifies \cref{as:a1}.

Next, the interiority and consistency in \cref{as:a2} are assumed by \cref{as:sieve_as_1}(1) and shown in \cref{lemma:consistency_sieve}. Our moment function is 
\eqref{eq:moment_defn_sieve}, whose derivatives are particularly simple. 
The derivative $\bGamma_1(\vartheta) = -1$ is a constant, and thus \cref{as:a2}(1) is
satisfied. The derivative $\bGamma_2 (\vartheta)[v]$ in our setting does not depend on
$\vartheta$, and so \cref{as:a2}(2) is satisfied.  

Since $\bGamma_2$ does not depend on $\vartheta$ and neither does $\Pn[n] g (\cdot,
\vartheta, \hat\zeta_n) - \Pn[] g(\cdot, \vartheta, \zeta)$, we can simplify \cref{as:a2}(3)
into \[
\abs{
    \Pn[] \tilde{g}_1(\cdot, \theta_{1,1}, \hat\nuisvector_n) - \Pn[] \tilde{g}_1(\cdot, \theta_{1,1}, \nuisvector) -
    \bGamma[\hat\nuisvector_n - \nuisvector]} = o_p(n^{-1/2})
\]
Consider the functional \[
t \mapsto \Pn[]  \tilde{g}_1(\cdot, \theta, \nuisvector + t (\hat\nuisvector_n - \nuisvector)).
\]
By Taylor's theorem, there exists a $\tilde t \in (0,1)$ such that \[
\Pn[]  \tilde{g}_1(\cdot, \theta_{1,1}, \hat\nuisvector_n) - \Pn[] \tilde{g}_1(\cdot, \theta_{1,1}, \nuisvector) -  \bGamma[\hat\nuisvector_n -
\nuisvector] = \diff{^2\Pn[]  \tilde{g}_1(\cdot, \theta, \nuisvector + t (\hat\nuisvector_n - \nuisvector))}{t^2}\evalbar_
{t=\tilde t}.
\]
It is tedious, but not difficult, to see that the second derivative term is quadratic in
the deviation $\hat\nuisvector_n - \nuisvector$, and therefore satisfies \[
\abs[\bigg]{\diff{^2\Pn[]  \tilde{g}_1(\cdot, \theta_{1,1}, \nuisvector + t (\hat\nuisvector_n - \nuisvector))}{t^2}\evalbar_
{t=\tilde t}} \le C(\epsilon) \E|Y| \cdot \norm{\hat\nuisvector_n -\nuisvector}_\infty^2 = o_p(n^
{-1/2}).
\]
for some $C(\epsilon)$ that depends on $\epsilon$. Thus, it suffices to assume 
\cref{as:sieve_as_1}(2,4). 

Lastly, the stochastic equicontinuity condition \cref{as:a2}(4) is implied by the 
condition that\[
\mathcal G = \br{\tilde{g}_1(\cdot, \vartheta, h) - g(\cdot, \theta_{1,1}, \nuisvector) : \vartheta \in \Theta,
h\in \mathcal H }
\]
is Donsker. Note that by \cref{as:sieve_as_1}(2) and the fact
that $\Theta$ is compact, $g(\cdot, \vartheta, h) - g(\cdot,
\theta_{1,1}, \nuisvector)$ is a composition of Lipschitz functions (i.e. addition, multiplication,
and division) of $\vartheta, h_1, h_2 ,h_3$. Therefore, since Lipschitz functions of
Donsker classes are Donsker (Example 19.20 in \citet{van2000asymptotic}), it suffices to
assume that $\mathcal H$ is a product of Donsker classes (\cref{as:sieve_as_1}(5)). \hfill
\end{proof}

\subsubsection{The influence function for $\hat\theta_{1,1}$ is the efficient influence function\label{sec: checks out}}

In this section, we show that the influence function expression given in \eqref{eq: initial expansion} matches the efficient influence function for $\theta_{1,1}$, derived in the the proof of \cref{thm: EIF tau_1}, yielding \eqref{eq: eif expansion}.

\begin{lemma}
\label{lemma:eif_checks_out}
Under \cref{as:sieve_as_1,as:sieve_as_2}, we have that
 \[\psi_{1,1}(B_i) = \tilde{g}_1(B_i, \theta_{1,1},
\nuisvector) + \Delta_i[v^*] + \theta_{1,1} - \frac{G_i\theta_{1,1}}{\pi}.\]
\end{lemma}

\begin{proof}
By inspection of $\psi_{1,1}(B_i)$, it will suffice to show that \[ -\Delta_i[v^*] = \frac{G_i W_i \mu_1(S_i,X_i)}{\pi\rho_1
(S_i,X_i)} - (1-G_i)W_i
\frac{1}{\pi
\varrho(X_i)}
\frac{\gamma(X_i)}{1-\gamma(X_i)} (\mu_1(S_i, X_i) - \bar\mu_1(X_i)) - \frac{G_i}{\pi}\bar\mu_1(X).
\numberthis \label{eq:missing_terms}
\]

It is not difficult to verify that the choices \begin{align*}
-v_1^*(S, X) &= \frac{1}{\zeta_1(S, X) \pi} \frac{\gamma(X)}{1-\gamma(X)}\frac{1}{\varrho(S, X)} \mu_1(S,X) = \frac{\mu
(S,X)}{\pi \rho_1(S,X)} \frac{1}{1-\zeta_1(S,X)}~, \\
-v_2^*(X) &= \frac{\bar \mu_1(X)  }{\zeta_2(X) \pi}~,\quad\text{and} \\ 
-v_3^*(X) &= -\frac{1}{\varrho(X)} \frac{\gamma(X)}{1-\gamma(X)} \frac{1}{\zeta_3(X) \pi} \bar\mu_1(X)~,
\end{align*}
ensure that \eqref{eq:missing_terms} is satisfied by \eqref{eq:delta_def}. These terms can be derived by 
inspecting terms in \eqref{eq:missing_terms}
that are of the form $W f(X), W f(S, X), G f(X)$, and matching them with corresponding
expressions in \eqref{eq:delta_def}. 

The following identity
\eqref{eq:ratio} is useful for this verification. Note that, by Bayes rule, \[
\zeta_3(x) = \frac{\gamma(x) \zeta_2(x)}{\gamma(x) \zeta_2(x) + (1-\gamma(x))\varrho(x)}
\implies \frac{\zeta_3(x)}{1-\zeta_3(x)} = \frac{\gamma(x) \zeta_2(x)}{ (1-\gamma(x)) \varrho(x)} \numberthis
\label{eq:ratio}
\]
and so it is convenient to note that \[
\rho_1(s,x) \frac{1-\zeta_1(s,x)}{\zeta_1(s,x)} = \zeta_2(x) \frac{1-\zeta_3(x)}{\zeta_3(x)} = \varrho(x) 
\frac{1-\gamma(x)}{\gamma(x)}.
\]

Lastly, we need to verify that this $v^*$ is actually the Riesz representer for $\bGamma
[\cdot]$. To do so, we need to check that for any $u = (u_1(s,x), u_2(x), u_3(x))$, \begin{align*}
&\E_P[W(v_1^* u_1 + v_3^* u_3) + G v_2^* u_2] \\ &= -\E_P\bk{
    \frac{GWY}{\pi \rho_1(S, X)} \pr{
    \frac{u_1(S,X)}{\zeta_1(S,X)(1-\zeta_1(S,X))}  + \frac{u_2(X)}{\zeta_2(X)} + \frac{v_3(X)}{\zeta_3(X)(1-\zeta_3(X))} 
    }
}~.
\end{align*}

The remainder of the proof verifies that \[
\E[W v_1^* u_1] = -\E_P\bk{
    \frac{GWY}{\pi \rho_1(S, X)} \frac{u_1(S,X)}{\zeta_1(S,X)(1-\zeta_1(S,X))}.
}
\]
The other two terms can be verified analogously. 

We first analyze the left-hand side. By law of total probability, we can break the left-hand side into \begin{align*}
-\E_P\bk{
    W v_1^* u_1
} &= -\pi\E_P[W v_1^* u_1 \mid G=1] -(1-\pi) \E_P[W v_1^* u_1 \mid G=0] \\ \text{ where }
-\pi\E_P[W v_1^* u_1 \mid G=1] &= \pi\E_{P^*}[\rho_1(S(1), X) v_1^*(S(1), X) u_1(S(1), X) \mid G=1] 
\\ 
&=\E_{P^*}\bk{\frac{\mu(S(1), X) u_1(S(1), X)}{(1-\zeta_1)} \mid G=1}\end{align*}
and where
\begin{align*}
 -(1-\pi) \E_P[W v_1^* u_1 \mid G=0] &= \frac{1-\pi}{\pi} \E_{P^*}\bk{
    \varrho(X) \frac{\mu(S(1), X) u_1(S(1), X)}{(1-\zeta_1(S(1), X)) \rho_1(S(1), X)} 
 \mid G=0} \\ 
 &= \E_{P^*}\bk{
 \frac{1-\pi}{\pi} \frac{\gamma(X)}{1-\gamma(X)} \frac{1}{\zeta_1(S(1), X)} \mu_1(S(1), X)
 u_1(S(1), X) \mid G=0
 }  \\ 
 &= \E_{P^*}\bk{
 \frac{1}{\zeta_1(S(1), X)} \mu_1(S(1), X) u_1(S(1), X) \mid G=1
 } 
\end{align*}
where the last step follows from the observation that \[
p(s(1), x \mid G=1) = \frac{\gamma(x)}{1-\gamma(x)} \frac{1-\pi}{\pi} p(x \mid G=0) p(s(1) \mid
x).
\]

We can also compute from the right-hand side that \[
\E\bk{\frac{GWY}{\pi \rho_1} \frac{1}{\zeta_1(1-\zeta_1)} u_1} = \E\bk{
\frac{\mu(S(1), X) u_1(S(1), X)}{\zeta_1(1-\zeta_1)} \mid G=1}~.
\]
Therefore, the terms involving $u_1$ in $\bGamma[u]$ and $\ip{v^*, u}$ do equal. \hfill
\end{proof}

\subsubsection{Auxiliary Lemmas\label{sec: sieve aux lemmas}}

\begin{lemma}
\label{lemma:consistency_sieve}
Under \cref{as:sieve_as_1}(1--3), \[\sqrt{n}(\piest / \pi-1) = \frac{1}{\sqrt{n}}\sum_
{i=1}^n
\pr{1-\frac{G_i}{\pi}} + o_p(1) \]
and $
\hat\theta_{1,1} -\theta = o_p(1).
$
\end{lemma}

\begin{proof}
The normality of $\piest/\pi$ follows from the delta method. The consistency of 
$\hat\theta_{1n}$ follows from the sup-norm consistency of $\hat \zeta_n$. Both claims
rely on \cref{as:sieve_as_1}(2) to enforce continuity. \hfill
\end{proof}

\begin{lemma}
\label{lemma:sieve_main_assumption_justify}
Under \cref{as:sieve_as_1,as:sieve_as_2}, the condition \cref{as:a1}(1) is satisfied.
\end{lemma}

\begin{proof}
Define $\epsilon_n = \pm \delta_n$. 
For some $h - \zeta \in \mathcal F_n$, defined in
\cref{as:sieve_as_2}(2), let $\tilde h = h + \epsilon_n v_n^*$. We first consider the quantity 
\[
\sup_{h-\zeta \in \mathcal F_n} \Gn[n] \bk{
    L(\cdot, \tilde h) - L(\cdot, h) + \Delta_i[\epsilon_n v_n^*]
} \numberthis \label{eq:emp_proc_term_sieve}. 
\]
Note that we can compute \begin{align*}
&L(\cdot, \tilde h) - L(\cdot, h) + \Delta_i[\epsilon_n v_n^*] \\&= \epsilon_n \pr{W(h_1 -
\zeta_1) v_{n1}^* +
G(h_2
-
\zeta_2) v_{n2}^* + W(h_3 - \zeta_3) v_{n3}^*} + \frac{1}{2}\epsilon_n^2 ({v_{n1}^*}^2 +
{v_{n2}^*}^2
+
{v_{n3}^*}^2) \\
& \equiv \epsilon_n R_2(\cdot, h, v_n^*) + \epsilon_n^2 R_3(\cdot, v_n^*)
\end{align*}
Note that $R_3$ does not depend on $h$, and hence \[
\sup_{h - \zeta \in \mathcal F_n} \Gn[n] R_3(\cdot, v_n^*) \le \abs{\Gn[n] R_3(\cdot, v_n^*)}
= \sqrt{n}O_p\pr{\E \norm{v_n^*}_2^2} = O_p(\sqrt{n})
\]
where we note that $ \E \norm{v_n^*}_2^2 < \infty$ since \[
\infty > \norm{v_n^*}^2 = \E\bk{
    \P(W=1 \mid S, X) ({v_{n1}^*}^2 + {v_{n3}^*}^2) + \P(G=1 \mid X) {v_{n2}^*}^2
} \ge \varepsilon \E[{v_{n1}^*}^2 +
{v_{n2}^*}^2
+
{v_{n3}^*}^2 ].
\]

To bound $R_2$, observe that \begin{align*}
\norm{R_2(\cdot, h_1, v_n^*) - R_2(\cdot, h_2, v_n^*)}_{L_2(Q)}^2 
&\lesssim \norm{v_n^*}_\infty^2 \norm{h_1 - h_2}_{L_2(Q)}^2 \\
&\lesssim \norm{h_1 - h_2}_{L_2(Q)}^2
\end{align*}
where the implicit constant does not depend on $Q$.
Observe too that for $h - \zeta \in \mathcal F_n$ \[
|R_2(\cdot, h, v_n^*)| \lesssim \delta_n^{1/2}
\]
uniformly, and thus $C \delta_n^{1/2}$ serves as an envelope function for $\mathcal R_n$, defined below.
 Hence, letting \[\mathcal R_n = \br{
    R_2(\cdot, h, v_n^*) : h - \zeta \in \mathcal F_n
},\] we have that \[
\sqrt{1 + \log N\pr{ C_1 \epsilon \delta_n^{1/2},
\mathcal R_n, L_2(Q)}} \le C_2 \sqrt{1 + \log N\pr{ \epsilon \delta_n^{1/2},
\mathcal F_n, L_2(Q)}} < C < \infty.
\]
Hence, by Theorem 2.14.1 in \citet{van1996weak}, \[
\E \sup_{h-\zeta \in \mathcal F_n}  \abs{\Gn[n] R_2(\cdot, h, v_n^*)} \lesssim \delta_n^
{1/2}
\]
This implies that \[
\eqref{eq:emp_proc_term_sieve} = O_p(\delta_n^{3/2} + \delta_n^2 \sqrt{n}).
\]

Now, consider some estimator $\hat \nuisvector$ and the event that $\norm{\hat\nuisvector_n - \nuisvector}_\infty < \delta_n$. On this event, which occurs with probability tending to 1,  $\norm{\hat\zeta_n - \zeta} \le C \delta_n$. Let $\tilde \zeta_n =
\hat\zeta_n + \epsilon_n v_n^*$. Let $A_n = \one\pr{\norm{\hat\nuisvector_n - \nuisvector}_\infty < \delta_n}$.

 Since $\hat \zeta_n$
minimizes the empirical criterion
\begin{align*}
0 &\le A_n \sqrt{n} \Pn[n][\L(\cdot, \tilde\zeta_n) - \L(\cdot, \hat\zeta_n)] \\ 
&= \sqrt{n} A_n \Pn[]\bk{\L(\cdot, \tilde\zeta_n) - \L(\cdot, \hat\zeta_n)} + A_n \Gn[n] 
\bk{\L(\cdot, \tilde\zeta_n) - \L(\cdot, \hat\zeta_n)} \\ 
&= \sqrt{n} \frac{A_n}{2}\pr{\norm{\tilde\zeta_n - \zeta}^2 - \norm{\hat\zeta_n - \zeta}^2} + A_n \Gn[n]\bk{\L(\cdot, \tilde\zeta_n) - \L(\cdot,
\hat\zeta_n)} \tag{Note that $\Pn[][\L(\cdot, h) - \L(\cdot, \zeta)] = 
\frac{1}{2} \norm{h - \zeta}^2$ by definition of $\ip{\cdot,\cdot}$}\\
&\le \sqrt{n} A_n \ip{\hat\zeta_n - \zeta, \epsilon_n v_n^*} + \sqrt{n} \frac{1}{2} \epsilon_n^2 + A_n \Gn[n]\bk{\L(\cdot, \tilde\zeta_n) - \L(\cdot,
\hat\zeta_n)} \\
&\le \sqrt{n} A_n \ip{\epsilon_n v_n^*, \hat\zeta_n - \zeta} + \Gn[n] [-\Delta_i[\epsilon_n
v_n^*]]
+ O_p
(\delta_n^{3/2} + \delta_n^2 \sqrt{n}),
\end{align*}
where the last step follows from our bound on \eqref{eq:emp_proc_term_sieve}.

Finally, this implies that
\[
\Gn[n] [\Delta_i[\epsilon_n
v_n^*]] \le \sqrt{n} A_n \ip{\hat\zeta_n - \zeta, \epsilon_n v_n^*} + O_p
(\delta_n^{3/2} + \delta_n^2 \sqrt{n}) \numberthis \label{eq:pos_neg_condition}.
\]
When $\epsilon_n = \delta_n$, \eqref{eq:pos_neg_condition} implies \[
\Gn[n] \Delta_i[v_n^*] \le \sqrt{n} A_n \ip{\hat\zeta_n - \zeta, v_n^*} + O_p
(\delta_n^{1/2} + \delta_n \sqrt{n}) = \sqrt{n} A_n\ip{\hat\zeta_n - \zeta, v_n^*} + o_p(1)  
\]
When $\epsilon_n = -\delta_n$, \eqref{eq:pos_neg_condition} implies \[
\sqrt{n} A_n \ip{\hat\zeta_n - \zeta, v_n^*} \le \Gn[n] \Delta_i[v_n^*] + o_p(1).
\]
Thus, taken together \[
|\sqrt{n} A_n \ip{\hat\zeta_n - \zeta, v_n^*} - \Gn[n] \Delta_i[v_n^*]| = o_p(1)
\]
Note that (i) $\Pn[] \Delta_i[v] = 0$ by the first-order condition of the first-step problem and (ii) $\sqrt{n} A_n \ip{\hat\zeta_n - \zeta, v_n^*} = \sqrt{n} \ip{\hat\zeta_n - \zeta, v_n^*} + o_p(1)$.
As a result, we can rewrite the above display as \[
\abs{\ip{\hat\zeta_n - \zeta, v_n^*} - \Pn[n] \Delta_i[v_n^*]
} = o_p(n^{-1/2})~,
\]
completing the proof.\hfill
\end{proof}

\section{Long-Term Treatment Effect for the Experimental Population\label{sec: G=0}} 

In this section, we develop results analogous to those presented in the main text for long-term average treatment effect for the experimental population, given by
\begin{equation}
\label{eq: tau_0 app}
\tau_0 = \E_{P_\star}\left[Y_i(1) - Y_i(0) \mid G_i = 0\right]~.
\end{equation}
This estimand was considered in \cite{athey2020estimating} for the Statistical Surrogacy Model. The efficient influence functions and efficiency bounds we state in this section for that context correct those given in Theorem 1 and Theorem 3 of the February, 2020 draft of that paper.

\subsection{Identification\label{sec:G=0 identification}}

The identifying assumptions for $\tau_0$ are slightly less restrictive than for $\tau_1$. In particular, in the case that treatment is not measured in the observational data set, it is unnecessary to impose the restriction \cref{as: ceev}. Thus, the set of assumptions that compose the Statistical Surrogacy Model will not include \cref{as: ceev} in this section.

\begin{prop}~
\label{thm: identification 0}
\begin{enumerate}
\item \citep{athey2020combining} Under the Latent Unconfounded Treatment Model, $\tau_0$ is point identified. 
\item \citep{athey2020estimating} Under the Statistical Surrogacy Model, $\tau_0$ is point identified. 
\end{enumerate}
\end{prop}

\begin{proof}
The result follows almost immediately from inspection of the proof of \cref{thm: identification lut} by considering the parameter
\[
\theta_{0,1} = \mathbb{E}_{P_\star}[Y(1)\mid G=0]~.
\]
In the Latent Unconfounded Treatment Model, the only difference will occur in the final step where it is apparent that
\[
\mathbb{E}_{P_\star}[\mu_1(X) \mid G=0]
\]
is identified. In the Statistical Surrogacy Model, the only difference is that \cref{as: ceev} need not be invoked in order to condition on $G=0$, where again the expectation of $\mathbb{E}_{P}[ \mu(S, X) \mid X, W=1, G=0 ]$ conditional on $G=0$ remains identified. \hfill
\end{proof}

\subsection{Semiparametric Efficiency}

Next, we state a theorem analogous to \cref{thm: EIF tau_1} for the case where the estimand of interest is the average long-term treatment effect $\tau_0$ in the experimental population.

\begin{theorem}
\label{thm: EIF tau_0}~
\begin{enumerate}
\item Under the Latent Unconfounded Treatment Model, where treatment is observed in the observational data set, the efficient influence function for the parameter $\tau_1$ is given by
\begin{align}
\psi_0(b,\tau_0,\eta) &=
\frac{g}{1-\pi}
\left(\frac{1-\gamma(x)}{\gamma(x)}\left(
 \frac{w(y-\mu_1(s,x))}{\rho_1(s,x)}  - \frac{(1-w)(y-\mu_0(s,x))}{\rho_0
(s,x)}\right)\right)\label{psi_0}\\
&  + \frac{1-g}{1-\pi}
\left(\frac{w(\mu_1(s,x)-\bar{\mu}_1(x))}{\varrho(x)} - \frac{(1-w)(\mu_0(s,x)-\bar{\mu}_0(x))}{1-\varrho(x)} + (\bar{\mu}_1(x) - \bar{\mu}_0(x)) -\tau_0\right),\nonumber
\end{align}
where $\eta = (\omega,\kappa)$ collects nuisance functions with
\[
\omega =  \left\{\mu_w, \bar{\mu}_w\right\}_{w\in\{0,1\}} \quad\text{and}\quad \kappa = \{\{\rho_w\}_{w\in\{0,1\}}, \varrho(\cdot), \gamma(\cdot), \pi\}~. 
\]
collecting long-term outcome means and propensity scores, respectively. 
\item Under the Statistical Surrogacy Model, where treatment is not observed in the observational data set, the efficient influence function
for the parameter $\tau_1$ is given by
\begin{align}
\xi_0(b,\tau_0,\varphi) & = \frac{g}{1-\pi}
\left(\frac{1- \gamma(s,x)}{\gamma(s,x)}\frac{(\varrho(s,x)-\varrho(x))(y-\nu(s,x))}{\varrho(x)(1-\varrho(x)}\right) \\
& + \frac{1-g}{1-\pi}\left(\frac{w(\nu(s,x)-\bar{\nu}_1(x))}{\rho(x)} -\frac{(1-w)(\nu(s,x)-\bar{\nu}_0(x))}{1-\rho(x)}+ (\bar{\nu}_1(x) - \bar{\nu}_0(x)) - \tau_0\right)~,\nonumber
\end{align}
where $\varphi = (\vartheta,\zeta)$ collects nuisance functions with
\[
\vartheta =  \left\{\nu, \{\bar{\nu}_w\}_{w\in\{0,1\}}\right\} \quad\text{and}\quad \zeta = \{\varrho(\cdot,\cdot), \varrho(\cdot), \gamma(\cdot,\cdot), \gamma(\cdot), \pi\}~. 
\]
collecting long-term outcome means and propensity scores, respectively. 
\end{enumerate}
\end{theorem}

\begin{cor}
\label{cor: seb 0}
Define the functionals 
\begin{align*}
\Gamma_{w,0}(s,x) & =  \frac{\left(\mu_w(s,x) - \bar{\mu}_w(x)\right)^2}{\varrho(x)^w(1-\varrho(x))^{1-w}}~~\text{and}~~
\Lambda_{w,0}(s,x) 
= \frac{\left(\nu(s,x) - \bar{\nu}_w(x)\right)^2}{\varrho(x)^w(1-\varrho(x))^{1-w}}~.
\end{align*}
\begin{enumerate}
\item Under the Latent Unconfounded Treatment Model, the semiparametric efficiency bound for $\tau_0$ is given by 
\begin{align}
V_0^{\star} 
& = \E_P\Bigg[\frac{1-\gamma(X)}{(1-\pi)^2} \Bigg(
\frac{1-\gamma(X)}{\gamma(X)}\left(
\frac{\sigma_1^2(S,X)}{\rho_1(S,X)} +
\frac{\sigma_0^2(S,X)}{\rho_0(S,X)}\right)  \nonumber \\
& \quad \quad\quad \quad \quad \quad 
+(\bar{\mu}_1(X) - \bar{\mu}_0(X) - \tau_0)^2 + \Gamma_{0,0}(S,X) + \Gamma_{1,0}(S,X)\Bigg) \Bigg]~.
\end{align}
\item Under the Statistical Surrogacy Model, the semiparametric efficiency bound for $\tau_0$ is given by
\begin{align}
V_0^{\star\star}  & 
= \E_P\Bigg[\frac{\gamma(X)}{(1-\pi)^2} 
\Bigg( \left(\frac{1-\gamma(S,X)}{\gamma(S,X)} 
\frac{\varrho(S,X)-\varrho(X)}{\varrho(X)(1-\varrho(X))}
\right)^2\sigma^2(S,X)\Bigg) \nonumber \\
& \quad \quad
+ \frac{1-\gamma(X)}{(1-\pi)^2} 
\Bigg((\bar{\nu}_1(X) - \bar{\nu}_0(X) - \tau_0)^2 + \Lambda_{0,0}(S,X) + \Lambda_{1,0}(S,X)\Bigg)
\Bigg]~.
\end{align}
\end{enumerate}
\end{cor}

\begin{proof}
We maintain the same notation and follow the same structure as the proof of
\cref{thm: EIF tau_1}, emphasizing differences without repeating shared steps in the argument. 

We begin by proving Part 1. The analogue of \eqref{eq: theta_1 prime} for $\theta_{1,0}$ is
\begin{align}
\theta_{1,0}' & = \E_{P_\star}[Y^* l'(Y^* \mid S^*, X) \mid G=0] \\ 
& + \E_{P_\star}[Y^* l'(S^* \mid X) \mid G=0]  + \E_{P_\star}[Y^*l'(X \mid G=0) \mid G=0]~,
\end{align}
with each term now conditioned on $G=0$ and $l'(X \mid G=0)$ replacing $l'(X \mid G=1)$. Note also that
 \[
l'_*(x \mid G=0) = l'_*(x, G=0) - \E_{P_\star}[l'_*(X, G=0) \mid G=0]~.
\]
For a candidate influence function $\tilde{\psi}_{1,0}(b_1)$, the analogues of the pathwise differentiability conditions
\eqref{eq:eta1} through \eqref{eq:eta5} are 
\begin{align}
    \frac{1}{1-\pi}\E_{P_\star}[(1-G)Y^* l'_*(Y^*\mid S^*, X)]
    &= \E_P\bk{\tilde \psi_{1,0}(B_1) GW l'_*(Y\mid S, X)} 
     \label{eq:eta1_appB} \\ 
    \frac{1}{1-\pi}\E_{P_\star}[(1-G)Y^* l'_*(S^* \mid X)]
    &= \E_P\Bigg[\tilde{\psi}_{1,0}(B_1)
    \Bigg( W l'_*(S \mid X) \nonumber \\
    & \quad 
    - G(1-W) \frac{\int \rho(s,X) p'_*(s \mid X)\,d\lambda(s)}{1-\int \rho(s,X)p_\star(s \mid X)\,d\lambda(s)}\Bigg)\Bigg]
    \label{eq:eta2_appB}\\ 
    \E_P[\tilde{\psi}_{1,0}(B_1)  l'_*(G, X)] 
    &= \frac{1}{1-\pi}\E_{P_\star}[(1-G)Y^* l'_*(G=0, X) ] \nonumber\\
    & \quad - \theta_{1,0} \E_{P_\star}\left[l'_*(G=0,X) \mid G=0\right]
     \label{eq:eta3_appB}\\ 
    0 &= \E_P\Bigg[\tilde{\psi}_{1,0}(B_1) G\Bigg(W \frac{\rho^\prime(S,X)}{\rho(S,X)} \nonumber \\
    &\quad \quad 
    - (1-W) \frac{\int \rho^\prime(s,X) p_\star(s\mid X) d\lambda(s)}{1-\int \rho
    (s,X)p_\star(s\mid X)\,d\lambda(s)}\Bigg) \Bigg]
    \label{eq:eta4_appB}\\ 
    0 &= \E_P\bk{\tilde{\psi}_{1,0}(B_1)(1-G)\frac{W - \varrho(X)}{\varrho(X)(1-\varrho(X))} \varrho'(X)}
   \label{eq:eta5_appB}
\end{align}
where the only differences relative to \eqref{eq:eta1} through \eqref{eq:eta5} are the replacement of $G$ with $1-G$ and $\pi$ with
$1-\pi$ in
\eqref{eq:eta1_appB} through \eqref{eq:eta3_appB} on the sides of the
equalities that correspond to $\theta'_{1,0}$. Since the density of the data doesn't
change with the estimand, the sides of the equalities that correspond to orthogonality condition \[\E_P\left[\tilde{\psi}_{1,0}(B_1) l'(B_1)\right]\] remains
unchanged relative to \eqref{eq:eta1} through \eqref{eq:eta5}. 

We make the same choices $s_4(s,x) = -\rho(s,x) s_2(s,x)$ and $s_5(x) = 0$ as in 
\eqref{eq:aci_score_choices_eif}, resulting in an conjectured efficient influence function of the form 
\[
\psi_{1,0}(b_1) = gw \cdot s_1 (y, s,x) + (1-g)w \cdot s_2(s,x) + s_3(g,x)
\]
We make the conjecture that 
\begin{align}
s_1 (y,s,x) & = f_1(x) \cdot \frac{y-\mu(s,x)}{\rho(s,x)}~, \nonumber
& s_2 (s,x)  = \frac{\mu(s,x) - \bar{\mu}(x)}{(1-\pi) \varrho(x)}~, \quad\text{and}\quad \nonumber\\
s_3(g,x) &= \frac{1-g}{1-\pi}(\bar{\mu}(x) - \theta_{1,0})~,
\end{align}
where 
\[f_1(x) = \frac{1}{1-\pi}\frac{1-\gamma(x)}{\gamma(x)}~.\] These choices satisfy the conditional mean-zero conditions
for scores.

We verify the conditions  \eqref{eq:eta1_appB} through \eqref{eq:eta5_appB} sequentially, completing the proof. To verify condition \eqref{eq:eta1_appB}, we note that by the mean-zero property of the conditional scores, the right-hand side of \eqref{eq:eta1_appB} simplifies to 
\begin{align*}
\E_{P_\star}\bk{f_1(X) \frac{GWY^*}{\rho(s,x)} l'_*(Y^*\mid S^*,X) } 
&= \E_{P_\star}[\pi f_1(X)Y^* l'_*(Y^*\mid S^*,X) \mid G=1] \\ 
& = \E_{P_\star}\bk{\frac{1-\gamma}{\gamma} \frac{\pi}{1-\pi} Y^*l'(Y^* \mid S^*,X) \mid G=1} \\ 
& = \E\bk{Y^* l'(Y^* \mid S^*, X) \mid G=0}\\
&= \frac{1}{1-\pi}\E_{P_\star}[(1-G)Y^* l'_*(Y^*\mid S^*, X)] \\
&= \E_P\bk{\tilde \psi_{1,0}(B_1) GW l'_*(Y\mid S, X)}~.
\end{align*}
where the third equality follows from the importance sampling argument
\[
p_\star(y,s,x \mid G=1) = p_\star(y,s\mid x)p_\star(x\mid G=1) = p_\star(y,s\mid x) p_\star(x \mid G=0) \frac{1-\gamma}{\gamma}\frac{\pi}{1-\pi}~.
\]
    
To verify condition \eqref{eq:eta2_appB}, we note that by the mean-zero property of the conditional scores, the right-hand side of \eqref{eq:eta2_appB} simplifies to 
\begin{align*}
\E_{P_\star}\bk{
    \frac{1-G}{1-\pi} W \frac{\mu(S,X)}{\varrho(x)}l'(S^* \mid X)
}  & = \E_{P_\star}[\mu(S,X)
    l'(S^* \mid X) \mid
    G=0]\\ & = \E_{P_\star}[Y^* l'(S^* \mid X) \mid G=0]~, 
\end{align*}
which is equal to the left-hand side of \eqref{eq:eta2_appB}. 

To verify condition \eqref{eq:eta3_appB}, we note again that by the mean-zero properties of the conditional
scores, the left-hand side of \eqref{eq:eta3_appB} simplifies to 
\begin{align*}
\E_P\bk{\frac{1-G}{1-\pi} (\bar{\mu}(X) - \theta_{1,0}) l'(G=0, X)} 
&= \E_{P_\star}\bk{(\bar{\mu}(X) - \theta_{1,0}) l'_*(G=0, X) \mid G=0} \\ 
&= \E_{P_\star}[\bar{\mu}(X) l'_*(G=0, X) \mid G=0] \\
&- \theta_{1,0} \E_{P_\star}[l'_*(G=0, X) \mid G=0]
\end{align*}
which is equal to the right-hand side of \eqref{eq:eta3_appB}. 

Finally, condition \eqref{eq:eta4_appB} holds by mean-zero properties of $s_1(y,s,x)$ and
condition \eqref{eq:eta5_appB} holds by mean-zero properties of $s_2(y,s,x)$ and the
fact that $\E_P[W-\varrho(x) \mid X, G=0] =0 $, completing the argument.

Next, we prove Part 2. The analogue of the pathwise derivative \eqref{eq:pathwisederivacik} for $\theta_{1,0}$ is now
\begin{align}
\theta_{1,0}' =\,\,& \E_P[\E_P[\E_P[Y l'(Y \mid S, X, G=1) \mid S, X, G=1] \mid X, W=1 ,G=0] \mid G=0] \nonumber \\ 
& + \E_P[\E_P[\nu(S,X)l'(S\mid X, G=0, W=1) \mid X, W=1, G=0] \mid G=0] \nonumber \\
& + \E_P\bk{(1-\pi)^{-1} \bar{\nu}(X)\pr{-\gamma'(X) + l'(X)(1-\gamma(X))} } \nonumber \\
& + (1-\pi)^{-1} \E_P[\gamma'(X) + \gamma(X) l'(X)] \theta_{1,0}.
\end{align}
The conjectured efficient influence function is given by 
 \[
\xi_{1,0}(b_1) = g(y - \nu(s,x))\cdot  f_1(s,x) + (1-g) w(\nu(s,x) - \bar{\nu}(x))\cdot f_2(x) +(1-g)\cdot f_3(x)
\]
with $f_1(s,x)$, $f_2(x)$, and $f_3(x)$ to be specified, where $f_3(x)$ is chosen such that
\[
\E[f_3(X) \mid G=0] = 0~. 
\]
An argument very similar to that given in \cref{sub:surrogateproof} demonstrates that
$\xi_{1,0}(b_1)$ is in the tangent space.

Again, following a similar argument, we may verify that the choices 
\begin{align*}
f_1(s,x) &= \frac{\varrho(s,x)}{(1-\pi) \varrho(x)} \frac{1-\gamma(s,x)}{\gamma(s,x)}~, \\ 
f_2(x) &= \frac{1}{(1-\pi) \varrho(x)}~, \\ 
f_3(x) &= \frac{\nu(x) - \theta_{1,0}}{1-\pi}~,\text{and}
\end{align*}
satisfy the pathwise differentiability conditions. The conditions for
$f_1(s,x)$ and $f_2(x)$ are such that multiplication of the density ratio 
\[\frac{p
(x\mid G=0)}{p(x \mid G=1)} = \frac{\pi}
{1-\pi} \frac{1-\gamma(x)}{\gamma(x)}\]
with the choices for $f_1(s,x)$ and $f_2(x)$ given in \cref{sub:surrogateproof} yields the corresponding choices here. \hfill
\end{proof}

\subsection{Estimation\label{sec: dml lst G=0}}

In this section, we define a semiparametric estimator of the long-term treatment effect $\tau_0$ in the experimental population. This estimator is analogous to the estimator formulated in \cref{sec:procedure}. Again, we use the ``Double/Debiased Machine Learning'' (DML) construction developed in \cite{chernozhukov2018double}. 

\begin{defn}[DML Estimators]
\label{def: dml 0}
Let $\hat{\eta}(I)$ and $\hat{\varphi}(I)$ denote generic estimates of $\eta$ and $\varphi$ based on the data $\{B_i\}_{i\in I}$ for some subset $I\in[n]$. Let $\{I_l\}_{l=1}^k$ denote a random $k$-fold partition of $[n]$ such that the size of each fold is $m=n/k$. The estimator $\hat{\tau}_0$ is defined as the solution to
\[
\frac{1}{k}\sum_{l=1}^k \frac{1}{m} \sum_{i\in I_l} \psi_0(B_i,\hat{\tau}_0,\hat{\eta}(I_l^c)) = 0
\quad\text{or}\quad
\frac{1}{k}\sum_{l=1}^k \frac{1}{m} \sum_{i\in I_l} \xi_0(B_i,\hat{\tau}_0,\hat{\varphi}(I_l^c)) = 0
\]
for the Latent Unconfounded Treatment and Statistical Surrogacy Models, respectively.
\end{defn}

\cref{thm: large-sample 0} demonstrates that $\hat{\tau}_0$ is semiparametrically efficient for $\tau_0$ under the Latent Unconfounded Treatment Model.

\begin{theorem}[DML Estimation and Inference]
\label{thm: large-sample 0}~
Let $\mathcal{P}\subset\mathcal{M}_\lambda$ be the set of all probability distributions
$P$ for $\{B_i\}_{i=1}^n$ that satisfy the Latent Unconfounded Treatment Model stated in \cref{def: lut} 
in addition to \cref{as: bounds}. If \cref{as: rates} holds for every $\mathcal P$, then
\begin{align}
\sqrt{n}(\hat{\tau}_0 -\tau_0) \overset{d}{\to} \mathcal{N}(0,V_0^\star)
\end{align}
uniformly over $P\in\mathcal{P}$, where $\hat{\tau}_0$ is defined in \ref{def: dml 0}, $V_0^*$ is defined in \cref{cor: seb 0}, and $\overset{d}{\to}$ denotes convergence in distribution. Moreover, we have that 
\begin{align}
\hat{V}_0^* = \frac{1}{k}\sum_{l=1}^k \frac{1}{m} \sum_{i\in I_l} \left(\psi_0(B_i,\hat{\tau}_0,\hat{\eta}(I_l^c))\right)^2 \overset{p}{\to} V_0^\star
\end{align}
uniformly over $P\in\mathcal{P}$, where $\overset{p}{\to}$ denotes convergence in probability, and as a result we obtain the uniform asymptotic validity of the confidence intervals
\begin{align}
\lim_{n\to\infty} \sup_{P\in\mathcal{P}} \Big\vert P\left(\tau_0 \in \left[\hat{\tau}_1 \pm z_{1-\alpha/2}\sqrt{\hat{V}_0^\star/n} \right]\right) - (1-\alpha) \Big\vert = 0~,
\end{align}
where $z_{1-\alpha/2}$ is the $1-\alpha/2$ quantile of the standard normal distribution. 
\end{theorem}

\begin{theorem}[DML Estimation and Inference]
\label{thm: large-sample 0 sur}~
Let $\mathcal{P}\subset\mathcal{M}_\lambda$ be the set of all probability distributions
$P$ for $\{B_i\}_{i=1}^n$ that satisfy the Statistical Surrogacy Model stated in \cref{def: sur} 
in addition to \cref{as: bounds}. If \cref{as: rates} holds for every $\mathcal P$, then
\begin{align}
\sqrt{n}(\hat{\tau}_0 -\tau_0) \overset{d}{\to} \mathcal{N}(0,V_0^{\star\star})
\end{align}
uniformly over $P\in\mathcal{P}$, where $\hat{\tau}_0$ is defined in \ref{def: dml 0}, $V_0^{\star\star}$ is defined in \cref{cor: seb 0}, and $\overset{d}{\to}$ denotes convergence in distribution. Moreover, we have that 
\begin{align}
\hat{V}_0^{\star\star} = \frac{1}{k}\sum_{l=1}^k \frac{1}{m} \sum_{i\in I_l} \left(\xi_0(B_i,\hat{\tau}_0,\hat{\varphi}(I_l^c))\right)^2 \overset{p}{\to} V_0^{\star\star}
\end{align}
uniformly over $P\in\mathcal{P}$, where $\overset{p}{\to}$ denotes convergence in probability, and as a result we obtain the uniform asymptotic validity of the confidence intervals
\begin{align}
\lim_{n\to\infty} \sup_{P\in\mathcal{P}} \Big\vert P\left(\tau_0 \in \left[\hat{\tau}_1 \pm z_{1-\alpha/2}\sqrt{\hat{V}_0^{\star\star}/n} \right]\right) - (1-\alpha) \Big\vert = 0~,
\end{align}
where $z_{1-\alpha/2}$ is the $1-\alpha/2$ quantile of the standard normal distribution. 
\end{theorem}
 
\begin{proof}[Proof of \cref{thm: large-sample 0} and \cref{thm: large-sample 0 sur}]
The proof of \cref{thm: large-sample 0} is very similar to the proof of \cref{thm: large-sample}, being obtained by verifying the conditions of Theorem 3.1 of \cite{chernozhukov2018double}. The arguments in support of the verification of Assumption 3.1 are identical, with $1-g$ replacing $g$ in the statement supporting condition (e). Similarly, the arguments supporting the verification of Assumption 3.2 are very similar, where the constants premultiplying the various bounds derived in support of conditions (b), (c), and (d) will be slightly different, but will be obtained with the same arguments, which we omit to avoid repetition. The second Gateaux derivatives of the efficient influence functions have the same structure, and so the third inequality in condition (c) will again follow from Hölder's inequality and \cref{as: rates}. \hfill
\end{proof} 

\section{Additional Results\label{sec: additional}}

\subsection{Consistency Double Robustness} \label{sec: consistency double}

Recall that, in \cref{rem: double}, we divide the nuisance parameters $\eta$ or $\varphi$ into $(\omega, \kappa, \pi)$, where $\omega$ consists of a set of outcome regression nuisance parameters, and $\kappa$ consists of a set of propensity score-type nuisance parameters. We note that the mean-zero properties of the efficient influence function is preserved as long as one of $\omega, \kappa$ is set at its respective true value. In this section, we show a consistency analogue of this double robustness result.  

To economize on notation, we note that the solutions $\hat\tau_{1n}$ to both $\psi_1 = 0$ and $\xi_1 = 0$ take the form \[
\hat\tau_{1n} = \tau_{1n}(\hat \omega_n, \hat\kappa_n) = \frac{1}{n} \sum_{i=1}^n g(B_i;  \hat \omega_n(B_i), \hat \kappa_n(B_i), \hat{\pi}_n) \numberthis \label{eq:consistency_tau_def}
\]
when $\hat\pi_n = \frac{1}{n} \sum_i G_i$ is used. Here, $\hat\omega_n (B_i), \hat\kappa_n(B_i)$ evaluates the nuisance parameters at $B_i$. 
Define $\norm{\cdot}_\infty$ entrywise as $\norm{\theta}_\infty = \max_j \norm{\theta_j}_\infty$ for a vector of nuisance parameters $\theta$.

\begin{theorem}
Let $\hat\tau_{1n}$ be as in \eqref{eq:consistency_tau_def}, derived either from either $g$ equal to $\psi_1$ or $\xi_1$ under either the Latent Unconfoundedness or the Statistical Surrogacy Models. Assume that $\hat\pi_n = \frac{1}{n} \sum_i G_i$. Assume additionally that (i) entries in $\kappa$ are bounded uniformly between $[1-\varepsilon, \varepsilon]$ for some $\varepsilon > 0$, (ii) entries in $\omega$ are bounded uniformly by some finite $C > 0$, (iii) entries in $\hat\kappa_n$ are bounded uniformly between $[1-\varepsilon, \varepsilon]$ almost surely, (iv) entries in $\hat\omega_n$ are bounded uniformly by $C$ almost surely, (v) all nuisance parameters, except for $\hat\pi_n$, are estimated on a hold-out sample, and (vi) $\E_P |Y|^2 < \infty$. Then:  
    \begin{enumerate}
         \item If $\norm{\hat\omega_n - \omega}_\infty \pto 0$, then $\hat\tau_{1n} \pto \tau_1$. 
         \item If $\norm{\hat\kappa_n - \kappa}_\infty \pto 0$, then $\hat\tau_{1n} \pto \tau_1$.
     \end{enumerate} 
\end{theorem}

\begin{proof}
    Let $A_n = \one(\hat\pi_n \in [\varepsilon/2, 1-\varepsilon/2])$ be an event that occurs with probability tending to one. For either (1) or (2), let $\eta_C$ collect the consistently estimated nuisance parameters, and let $\eta_I$ collect the rest. Let $\hat\eta_{Cn}$, $\hat\eta_{In}$ be their estimated analogues. Note that $\pi$ is always in $\eta_C$.  

    Note that under either model, by Taylor's theorem, on the event $A_n$, \[
g(B_i; \hat\omega_n(B_i), \hat\kappa_n(B_i)) = g(B_i; \eta_C(B_i), \hat\eta_{In}(B_i)) + \diff{}{\eta_{Ci}} g(B_i; \eta_{Ci}, \hat\eta_{In}(B_i))\evalbar_{\eta_{Ci} = \tilde \eta_{Ci}} (\hat \eta_{Cn} (B_i) - \eta_C(B_i))
    \]
    where $\tilde \eta_{Ci}$ lies somewhere on the line segment connecting $\hat \eta_{Cn} (B_i)$ and $ \eta_C(B_i)$. Hence, by inspecting the derivative of $g$, we find that, for some constant $M$ that depends on $(C, \varepsilon)$, \begin{align*}
    &\abs[\bigg]{\hat\tau_{1n} - \frac{1}{n} \sum_{i=1}^n g(B_i; \eta_C(B_i), \hat\eta_{In}(B_i)) }\\
    &= A_n \abs[\bigg]{\hat\tau_{1n} - \frac{1}{n} \sum_{i=1}^n g(B_i; \eta_C(B_i), \hat\eta_{In}(B_i)) } + o_p(1) \\&\le \norm{\hat \eta_{Cn} - \eta_C}_\infty \frac{A_n}{n} \sum_{i=1}^n \abs[\bigg]{\diff{}{\eta_{Ci}} g(B_i; \eta_{Ci}, \hat\eta_{In}(B_i))\evalbar_{\eta_{Ci} = \tilde \eta_{Ci}}} + o_p(1)
\\
&\le M(C, \varepsilon) \norm{\hat \eta_{Cn} - \eta_C}_\infty \frac{1}{n} \sum_{i=1}^n |Y_i| + o_p(1) \tag{(i)--(iv)}\\ 
&=o_p(1) \tag{vi}. 
    \end{align*}
It suffices to then show that \[
\abs[\bigg]{\tau_1 - \frac{1}{n} \sum_{i=1}^n g(B_i; \eta_C(B_i), \hat\eta_{In}(B_i))} = o_p(1).
\]
Note that, conditional on $\hat\eta_{In}$, $\frac{1}{n} \sum_{i=1}^n g(B_i; \eta_C(B_i), \hat\eta_{In}(B_i))$ is an average of i.i.d. random variables with mean $\tau_1$, by the double robust property and (v). By the law of large numbers applied to triangular arrays, which uses (i)--(iv) and (vi), for almost every sequence $\hat\eta_{In}$, for every $\epsilon > 0$,  \[
\P\pr{\abs[\bigg]{\tau_1 - \frac{1}{n} \sum_{i=1}^n g(B_i; \eta_C(B_i), \hat\eta_{In}(B_i))} > \epsilon \mid \hat\eta_{In}}  \to 0.
\]
Thus, by the dominated convergence theorem applied to left-hand side of the previous display, treating it as a random variable indexed by $n$, \begin{align*}
&\P\pr{\abs[\bigg]{\tau_1 - \frac{1}{n} \sum_{i=1}^n g(B_i; \eta_C(B_i), \hat\eta_{In}(B_i))} > \epsilon } \\ 
&= \E\bk{
    \P\pr{\abs[\bigg]{\tau_1 - \frac{1}{n} \sum_{i=1}^n g(B_i; \eta_C(B_i), \hat\eta_{In}(B_i))} > \epsilon \mid \hat\eta_{In}}
} \to 0~,
\end{align*}
as required.\hfill
\end{proof}

\subsection{Known Nuisance Functions}
\label{asub:known_nuisance_bound}

In this section, we assess how the efficient influence functions derived in \cref{thm: EIF tau_1} change if different components of the nuisance functions $\eta$ and $\varphi$ are known. We consider only estimation of the long-term treatment effect in the observational sample $\tau_1$. The case of the treatment effect in the experimental sample $\tau_1$ is analogous. 

First, we consider the Latent Unconfounded Treatment Model. We show that the efficient influence function, and thereby, the semiparametric efficiency bound, is unchanged if the propensity score in the experimental sample $\varrho(x)$ is known. This result echoes an analogous result for estimation of the average treatment effects under ignorability given in \cite{hahn1998role}. On the other hand, if the probability of being included in observational sample $\gamma(x)$ is known, then there is a change in the efficient influence function. In particular, the term $(g/\pi)((\bar{\mu}_1(x) - \bar{\mu}_0(x))  - \tau_1)$ in $\psi_1(b,\tau_1,\eta)$ is replaced by $(\gamma(x)/\pi)((\bar{\mu}_1(x) - \bar{\mu}_0(x))  - \tau_1)$. Proofs for these results are given in \cref{sec: lutm check i,sec: lutm check ii}, respectively.

\begin{theorem} \label{thm: lutm check} Consider the Latent Unconfounded Treatment Model, given in \cref{def: lut}. 

\noindent \textbf{(i)} If the propensity score in the experimental sample $\varrho(x)$ is known, then the efficient influence function for the parameter $\tau_1$ is unchanged. 

\noindent \textbf{(ii)} If the probability of being included in observational sample $\gamma(x)$ is known, then the efficient influence function for the parameter $\tau_1$ is given by 
\begin{align}
\check{\psi}_1(b,\tau_1,\eta) &=
\frac{1}{\pi}
\Bigg(g\left( \frac{w(y-\mu_1(s,x))}{\rho_1(s,x)}  - \frac{(1-w)(y-\mu_0(s,x))}{\rho_0
(s,x)}\right) +  \gamma(x)\left( (\bar{\mu}_1(x) - \bar{\mu}_0(x))  - \tau_1\right) \nonumber\\
&  + 
\frac{(1-g)\gamma (x)}{1-\gamma(x)}\left( \frac{w(\mu_1(s,x)-\bar{\mu}_1(x))}{\varrho(x)} - \frac{(1-w)(\mu_0(s,x)-\bar{\mu}_0(x))}{1-\varrho(x)}\right)\Bigg), \label{eq: psi check}
\end{align}
where, again, the parameter $\eta $ collects the nuisance functions appearing in \eqref{eq: psi check}.
\end{theorem}

Next, we provide analogous results for the Statistical Surrogacy Model. Here, however, the efficient influence function is additionally invariant to knowledge of the distribution in the observational sample of the short-term outcomes conditional on the covariates, i.e. the law $S\mid X, G=1$. Proofs are given in \cref{sec: ss check i,sec: ss check ii,sec: ss check iii}.

\begin{theorem} \label{thm: ss check} Consider the Statistical Surrogacy Model, given in \cref{def: sur}. 

\noindent \textbf{(i)} If the propensity score in the experimental sample $\varrho(x)$ is known, then the efficient influence function for the parameter $\tau_1$ is unchanged. 

\noindent \textbf{(ii)} If the distribution of the short-term outcomes, conditional on the pretreatment covariates, is known in observational dataset, then the efficient influence function for the parameter $\tau_1$ is unchanged. 

\noindent \textbf{(iii)} If the probability of being included in observational sample $\gamma(x)$ is known, then the efficient influence function for the parameter $\tau_1$ is given by 
\begin{align}
\check{\xi}_1(b,\tau_1,\eta) &=
 \frac{g}{\pi}
 \left( \frac{\gamma(x)}{\gamma(s,x)} \frac{1- \gamma(s,x)}{1-\gamma(x)}\frac{(\varrho(s,x)-\varrho(x))(y-\nu(s,x))}{\varrho(x)(1-\varrho(x)}\right) + \frac{\gamma(x)}{\pi}\left((\bar{\nu}_1(x) - \bar{\nu}_0(x)) - \tau_1\right) \nonumber \\
& + \frac{1-g}{\pi} \left(\frac{\gamma(x)}{1-\gamma(x)} \left(\frac{w(\nu(s,x)-\bar{\nu}_1(x))}{\varrho(x)} -\frac{(1-w)(\nu(s,x)-\bar{\nu}_0(x))}{1-\varrho(x)}\right)\right)~,\label{eq: xi check}
\end{align}
where, again, the parameter $\varphi$ collects the nuisance functions appearing in \eqref{eq: xi check}.
\end{theorem}

\subsubsection{\label{sec: lutm check i}Proof of \cref{thm: lutm check}, Part (i)}

Let $\mathcal{P}$ be a regular parametric submodel of $\mathcal{M}_{\lambda}$,
indexed by $\varepsilon\in\mathbb{R}$ and such that \cref{as: uex,as: ceev,as: luot} hold for each $P_{\varepsilon}\in\mathcal{P}$. Additionally,
suppose that 
\[
p_{\varepsilon}\left(W=1\mid x,G=0\right)=\varrho\left(x\right)
\]
for all $\varepsilon\in\mathbb{R}$ and some fixed function $\varrho\left(x\right)$. 

In this case, we have the score 
\begin{flalign*}
\ell^{\prime}\left(b_{1}\right) & =wg\cdot\ell_{*}^{\prime}\left(y\mid x,s\right)+w\cdot\ell_{*}^{\prime}\left(s\mid x\right)+\ell_{*}^{\prime}\left(g,x\right)\\
 & -g\left(1-w\right)\left(\frac{\mathbb{E}_{P_{\star}}\left[\rho\left(S^{*},X\right)\ell_{*}^{\prime}\left(S^{*}\mid X\right)\mid G=0,W=1,X=x\right]}{1-\mathbb{E}_{P_{\star}}\left[\rho\left(S^{*},X\right)\mid G=0,W=1,X=x\right]}\right)\\
 & +g\left(w\cdot\frac{\rho^{\prime}\left(s,x\right)}{\rho\left(s,x\right)}-\left(1-w\right)\frac{\mathbb{E}_{P_{\star}}\left[\rho^{\prime}\left(S^{*},X\right)\mid G=0,W=1,X=x\right]}{1-\mathbb{E}_{P_{\star}}\left[\rho\left(S^{*},X\right)\mid G=0,W=1,X=x\right]}\right)
\end{flalign*}
and so the tangent space $\mathcal{T}$ is given by the mean-square
closure of the linear space of the functions 
\begin{align*}
s\left(b_{1}\right) & =wg\cdot s_{1}\left(y\mid x,s\right)+w\cdot s_{2}\left(s\mid x\right)+s_{3}\left(g,x\right)\\
 & -g\left(1-w\right)\left(\frac{\mathbb{E}_{P_{\star}}\left[\rho\left(S^{*},X\right)s_{2}\left(S^{*}\mid X\right)\mid G=0,W=1,X=x\right]}{1-\mathbb{E}_{P_{\star}}\left[\rho\left(S^{*},X\right)\mid G=0,W=1,X=x\right]}\right)\\
 & +g\left(w\cdot\frac{s_{4}\left(s,x\right)}{\rho\left(s,x\right)}-\left(1-w\right)\frac{\mathbb{E}_{P_{\star}}\left[s_{4}\left(S^{*},X\right)\mid G=0,W=1,X=x\right]}{1-\mathbb{E}_{P_{\star}}\left[\rho\left(S^{*},X\right)\mid G=0,W=1,X=x\right]}\right)
\end{align*}
where the functions $s_{1}$ through $s_{4}$ range over the space
of mean-zero and square integrable functions that additionally satisfy
the restrictions
\begin{flalign}
\mathbb{E}_{P}\left[s_{1}\left(Y\mid X,S\right)\mid W=1,G=1,S,X\right] & =0\quad\text{and}\label{eq: conditional mean 1}\\
\mathbb{E}_{P}\left[s_{2}\left(S\mid X\right)\mid W=1,G=0,X\right] & =0.\label{eq: conditional mean 2}
\end{flalign}
Recall the integral representation
\begin{equation}
\theta_{1,1}=\int\int\int\frac{y}{\pi}p_{\star}\left(y\mid s,x\right)p_{\star}\left(s\mid x\right)p\left(x,G=1\right)\text{d}\lambda\left(y\right)\text{d}\lambda\left(s\right)\text{d}\lambda\left(x\right).\label{eq: integral rep}
\end{equation}
Observe that $p\left(w=1\mid x,G=0\right)$ does not appear in (\ref{eq: integral rep}).
Thus, the pathwise derivative of $\theta_{1,1}$ at $0$ on the submodel
of $\mathcal{P}$ is again given by \eqref{eq: theta_1 prime}. Thus, it will suffice to
choose functions $s_{1}\left(\cdot\right)$ through $s_{4}\left(\cdot\right)$
satisfying the above conditions and whose resultant score function
satisfies \eqref{eq:eta1} through \eqref{eq:eta4}. 

Consider the choices 
\begin{align*}
s_{1}\left(y,s,x\right) & =\frac{y-\mu\left(s,x\right)}{\pi\rho\left(s,x\right)}, & 
s_{2}\left(y,s\right) & =\frac{1}{\pi}\frac{1}{\varrho\left(x\right)}\frac{\gamma\left(x\right)}{1-\gamma\left(x\right)}\left(\mu\left(s,x\right)-\bar{\mu}\left(x\right)\right),\\
s_{3}\left(g,x\right) & =\frac{g}{\pi}\left(\bar{\mu}\left(x\right)-\theta_{1,1}\right)\text{, and} & s_{4}\left(y,s\right) & =-\rho\left(s,x\right)s_{2}\left(s,x\right),
\end{align*}
which again yield the conjectured influence function 
\[
\psi_{1,1}\left(b_{1}\right)=\frac{gw\left(y-\bar{\mu}\left(s,x\right)\right)}{\pi\rho\left(s,x\right)}+\frac{\left(1-g\right)w}{\left(1-\gamma\left(x\right)\right)\varrho\left(x\right)}\frac{\gamma\left(x\right)}{\pi}\left(\mu\left(s,x\right)-\bar{\mu}\left(x\right)\right)+\frac{g}{\pi}\left(\bar{\mu}\left(x\right)-\theta_{1,1}\right).
\]
Each of these choices are mean-zero and square integrable, and satisfy
the conditions (\ref{eq: conditional mean 1}) and (\ref{eq: conditional mean 2}),
as before. The arguments verifying the conditions \eqref{eq:eta1} through \eqref{eq:eta4}
given in proof of \cref{thm: EIF tau_1}, Part 1, completeing the proof. \hfill\qed

\subsubsection{\label{sec: lutm check ii}Proof of \cref{thm: lutm check}, Part (ii)}

Let $\mathcal{P}$ be a regular parametric submodel of $\mathcal{M}_{\lambda}$,
indexed by $\varepsilon\in\mathbb{R}$ and such that \cref{as: uex,as: ceev,as: luot} hold for each $P_{\varepsilon}\in\mathcal{P}$. Additionally,
suppose that 
\[
p_{\varepsilon}\left(G=1\mid X=x\right)=\gamma\left(x\right)
\]
for all $\varepsilon\in\mathbb{R}$ and some fixed function $\gamma\left(x\right)$. 

In this case, we have the score 
\begin{flalign*}
\ell^{\prime}\left(b_{1}\right) & =wg\cdot\ell_{*}^{\prime}\left(y\mid x,s\right)+w\cdot\ell_{*}^{\prime}\left(s\mid x\right)+\ell_{*}^{\prime}\left(x\right)\\
 & -g\left(1-w\right)\left(\frac{\mathbb{E}_{P_{\star}}\left[\rho\left(S^{*},X\right)\ell_{*}^{\prime}\left(S^{*}\mid X\right)\mid G=0,W=1,X=x\right]}{1-\mathbb{E}_{P_{\star}}\left[\rho\left(S^{*},X\right)\mid G=0,W=1,X=x\right]}\right)\\
 & +g\left(w\cdot\frac{\rho^{\prime}\left(s,x\right)}{\rho\left(s,x\right)}-\left(1-w\right)\frac{\mathbb{E}_{P_{\star}}\left[\rho^{\prime}\left(S^{*},X\right)\mid G=0,W=1,X=x\right]}{1-\mathbb{E}_{P_{\star}}\left[\rho\left(S^{*},X\right)\mid G=0,W=1,X=x\right]}\right)\\
 & +\left(1-g\right)\varrho^{\prime}\left(x\right)\left(\frac{w-\varrho\left(x\right)}{\varrho\left(x\right)\left(1-\varrho\left(x\right)\right)}\right)
\end{flalign*}
and so the tangent space $\mathcal{T}$ is given by the mean-square
closure of the linear space of the functions 
\begin{flalign*}
s\left(b_{1}\right) & =wg\cdot s_{1}\left(y\mid x,s\right)+w\cdot s_{2}\left(s\mid x\right)+s_{3}\left(x\right)\\
 & -g\left(1-w\right)\left(\frac{\mathbb{E}_{P_{\star}}\left[\rho\left(S^{*},X\right)s_{2}\left(S^{*}\mid X\right)\mid G=0,W=1,X=x\right]}{1-\mathbb{E}_{P_{\star}}\left[\rho\left(S^{*},X\right)\mid G=0,W=1,X=x\right]}\right)\\
 & +g\left(w\cdot\frac{s_{4}\left(s,x\right)}{\rho\left(s,x\right)}-\left(1-w\right)\frac{\mathbb{E}_{P_{\star}}\left[s_{4}\left(S^{*},X\right)\mid G=0,W=1,X=x\right]}{1-\mathbb{E}_{P_{\star}}\left[\rho\left(S^{*},X\right)\mid G=0,W=1,X=x\right]}\right)\\
 & +\left(1-g\right)\left(s_{5}\left(x\right)\cdot\frac{w-\varrho\left(x\right)}{\varrho\left(x\right)\left(1-\varrho\left(x\right)\right)}\right)
\end{flalign*}
where the functions $s_{1}$ through $s_{5}$ range over the space
of mean-zero and square integrable functions that additionally satisfy
the restrictions (\ref{eq: conditional mean 1}) and (\ref{eq: conditional mean 2}).
By the integral representation (\ref{eq: integral rep}) and the fact
that
\[
p_{\varepsilon}\left(x,G=1\right)=\gamma\left(x\right)p_{\varepsilon}\left(x\right)
\]
for all $\varepsilon\in\mathbb{R}$, the pathwise derivative of $\theta_{1,1}$
at $0$ on $\mathcal{P}$ is given by 
\begin{flalign}
\theta_{1,1}^{\prime} & =\mathbb{E}_{P_{\star}}\left[Y^{*}\ell^{\prime}\left(Y^{*}\mid S^{*},X\right)\mid G=1\right]\nonumber \\
 & +\mathbb{E}_{P_{\star}}\left[Y^{*}\ell^{\prime}\left(S^{*}\mid X\right)\mid G=1\right]+\mathbb{E}_{P_{\star}}\left[Y^{*}\ell^{\prime}\left(X\right)\mid G=1\right].\label{eq: pwd, known sample score}
\end{flalign}
Each of the terms in (\ref{eq: pwd, known sample score}) can be written
in terms of conditional scores by 
\begin{flalign*}
\mathbb{E}_{P_{\star}}\left[Y^{*}\ell^{\prime}\left(Y^{*}\mid S^{*},X\right)\mid G=1\right] & =\mathbb{E}_{P_{\star}}\left[\pi^{-1}GY^{*}\ell_{*}^{\prime}\left(Y^{*}\mid S^{*},X\right)\right],\\
\mathbb{E}_{P_{\star}}\left[Y^{*}\ell^{\prime}\left(S^{*}\mid X\right)\mid G=1\right] & =\mathbb{E}_{P_{\star}}\left[\pi^{-1}GY^{*}\ell_{*}^{\prime}\left(S^{*}\mid X\right)\right],\quad\text{and}\\
\mathbb{E}_{P_{\star}}\left[Y^{*}\gamma\left(X\right)\ell^{\prime}\left(X\right)\mid G=1\right]. & =\mathbb{E}_{P_{\star}}\left[\pi^{-1}GY^{*}\ell_{*}^{\prime}\left(X\right)\right],
\end{flalign*}
respectively. Thus, in order to establish that a mean-zero and square-integrable
function $\tilde{\psi}_{1,1}\left(B_{1}\right)$ is an influence function
for $\theta_{1,1}$ it suffices to verify the conditions \eqref{eq:eta1} through \eqref{eq:eta5}, with  \eqref{eq:eta3} replaced by
\begin{equation}
\mathbb{E}_{P_{\star}}\left[\pi^{-1}\gamma^{2}\left(X\right)Y^{*}\ell_{*}^{\prime}\left(X\right)\right]=\mathbb{E}_{P}\left[\tilde{\psi}_{1,1}\left(B_{1}\right)\ell_{*}^{\prime}\left(X\right)\right].\label{eq: new A.7}
\end{equation}

Consider the choices 
\begin{align*}
s_{1}\left(y,s,x\right) & =\frac{y-\mu\left(s,x\right)}{\pi\rho\left(s,x\right)}, & s_{2}\left(y,s\right) & =\frac{1}{\pi}\frac{1}{\varrho\left(x\right)}\frac{\gamma\left(x\right)}{1-\gamma\left(x\right)}\left(\mu\left(s,x\right)-\bar{\mu}\left(x\right)\right),\\
s_{3}\left(x\right) & =\frac{\gamma\left(x\right)}{\pi}\bar{\mu}\left(x\right)-\theta_{1,1} & s_{4}\left(y,s\right) & =-\rho\left(s,x\right)s_{2}\left(s,x\right),\quad\text{and}\\
s_{5}\left(x\right) & =0,
\end{align*}
which yield the conjectured influence function
\begin{flalign*}
\check{\psi}_{1,1}\left(b_{1}\right) & =\frac{gw\left(y-\mu\left(s,x\right)\right)}{\pi\rho\left(s,x\right)}+\frac{\left(1-g\right)w}{\pi\varrho\left(x\right)}\frac{\gamma\left(x\right)}{1-\gamma\left(x\right)}\left(\mu\left(s,x\right)-\bar{\mu}\left(x\right)\right)+\frac{\gamma\left(x\right)}{\pi}\left(\bar{\mu}\left(x\right)-\theta_{1,1}\right).
\end{flalign*}
These choices are again mean-zero and square integrable, and satisfy
the conditions (\ref{eq: conditional mean 1}) and (\ref{eq: conditional mean 2}).
Thus, it suffices to check that the conditions \eqref{eq:eta1}, \eqref{eq:eta2}, \eqref{eq:eta4}, \eqref{eq:eta5}, and (\ref{eq: new A.7}) continue to hold. 

The conditions \eqref{eq:eta1}, \eqref{eq:eta2}, and \eqref{eq:eta4} continue to hold through iterated
expectation conditional on $X$ by the fact the only term changed
relative to $\psi_{1,1}\left(b_{1}\right)$ is $s_{3}\left(x\right)$.
The condition \eqref{eq:eta5} continues to hold as 
\begin{flalign*}
 & \mathbb{E}_{P}\left[\frac{1-G}{\pi}\left(\gamma\left(X\right)\bar{\mu}\left(X\right)-\theta_{1,1}\right)\frac{W-\varrho\left(X\right)}{\varrho\left(X\right)\left(1-\varrho\left(X\right)\right)}\varrho^{\prime}\left(X\right)\right]\\
 & =\mathbb{E}_{P}\left[\frac{\left(1-G\right)}{\pi}\left(\mathbb{E}_{P}\left[G\mid X\right]\gamma\left(X\right)\bar{\mu}\left(X\right)-\theta_{1,1}\right)\frac{W-\varrho\left(X\right)}{\varrho\left(X\right)\left(1-\varrho\left(X\right)\right)}\varrho^{\prime}\left(X\right)\right]\\
 & =-\mathbb{E}_{P}\left[\frac{\theta_{1,1}\left(1-G\right)}{\pi}\frac{\theta_{1,1}\left(W-\varrho\left(X\right)\right)}{\varrho\left(X\right)\left(1-\varrho\left(X\right)\right)}\varrho^{\prime}\left(X\right)\right]=0.
\end{flalign*}
Thus, it remains to verify the new condition (\ref{eq: new A.7}).
This condition reduces to 
\[
\mathbb{E}_{P_{\star}}\left[\pi^{-1}GY^{*}\ell_{*}^{\prime}\left(X\right)\right]=\mathbb{E}_{P}\left[\frac{1}{\pi}\gamma\left(X\right)\bar{\mu}\left(X\right)\ell_{*}^{\prime}\left(X\right)\right],
\]
as before. Observe that 
\begin{flalign*}
\mathbb{E}_{P}\left[\pi^{-1}\gamma\left(X\right)\bar{\mu}\left(X\right)\ell_{*}^{\prime}\left(X\right)\right] & =\mathbb{E}_{P}\left[\pi^{-1}\gamma\left(X\right)\mathbb{E}_{P}\left[\mu\left(S,X\right)\mid G=0,W=1,X\right]\ell_{*}^{\prime}\left(X\right)\right]\\
 & =\mathbb{E}_{P}\left[\pi^{-1}\gamma\left(X\right)\mathbb{E}_{P^{*}}\left[\mu\left(S^{*},X\right)\mid X\right]\ell_{*}^{\prime}\left(X\right)\right]\\
 & =\mathbb{E}_{P}\left[\pi^{-1}\gamma\left(X\right)\mathbb{E}_{P^{*}}\left[\mathbb{E}_{P^{*}}\left[Y^{*}\mid W=1,G=1,S,X\right]\mid X\right]\ell_{*}^{\prime}\left(X\right)\right]\\
 & =\mathbb{E}_{P}\left[\pi^{-1}\mathbb{E}\left[G\mid X,Y^{*}\right]Y^{*}\ell_{*}^{\prime}\left(X\right)\right]\\
 & =\mathbb{E}_{P_{\star}}\left[\pi^{-1}GY^{*}\ell_{*}^{\prime}\left(X\right)\right]
\end{flalign*}
where the second to last equality follows from Assumption 2.3, completing
the proof.  \hfill\qed

\subsubsection{\label{sec: ss check i}Proof of \cref{thm: ss check}, Part (i)}

Let $\mathcal{P}$ be a regular parametric submodel of $\mathcal{M}_{\lambda}$,
indexed by $\varepsilon\in\mathbb{R}$ and such that  \cref{as: uex,as: ltoc,as: ceev,as: es} hold for each $P_{\varepsilon}\in\mathcal{P}$.
Additionally, suppose that 
\[
p_{\varepsilon}\left(W=1\mid x,G=0\right)=\varrho\left(x\right)
\]
for all $\varepsilon\in\mathbb{R}$ and some fixed function $\varrho\left(x\right)$. 

In this case, we have the score 
\begin{flalign*}
\ell^{\prime}\left(b_{1}\right) & =g\cdot\ell^{\prime}\left(y\mid x,s,G=1\right)+g\cdot\ell^{\prime}\left(s\mid x,G=1\right)+w\left(1-g\right)\cdot\ell^{\prime}\left(s\mid W=1,x,G=0\right)\\
 & +\ell^{\prime}\left(x\right)+\frac{g-\gamma\left(x\right)}{\gamma\left(x\right)\left(1-\gamma\left(x\right)\right)}\gamma^{\prime}\left(x\right)
\end{flalign*}
and so the tangent space $\mathcal{T}$ is given by the mean-square
closure of the linear space of the functions 
\begin{flalign*}
s\left(b_{1}\right) & =g\cdot s_{1}\left(y\mid x,s,G=1\right)+g\cdot s_{2}\left(s\mid x,G=1\right)+w\left(1-g\right)s_{3}\left(s\mid W=1,x,G=0\right)\\
 & +s_{4}\left(x\right)+\frac{g-\gamma\left(x\right)}{\gamma\left(x\right)\left(1-\gamma\left(x\right)\right)}s_{5}\left(x\right)
\end{flalign*}
where the functions $s_{1}$ through $s_{5}$ range over the space
of mean-zero and square integrable functions that additionally satisfy
the restrictions 
\begin{flalign}
\mathbb{E}_{P}\left[s_{1}\left(Y\mid X,S,G=1\right)\mid S,X,G=1\right] & =0,\label{eq: cond mean ss 1}\\
\mathbb{E}_{P}\left[s_{2}\left(S\mid X,G=1\right)\mid X,G=1\right] & =0,\quad\text{and}\label{eq: cond mean ss 2}\\
\mathbb{E}_{P}\left[s_{3}\left(S\mid X,W=1G=0\right)\mid X,W=1G=0\right] & =0.\label{eq: cond mean ss 3}
\end{flalign}
Consider the choices 
\begin{flalign}
s_{1}\left(y\mid s,x,G=1\right) & =\frac{1}{\pi}\frac{\varrho\left(s,x\right)}{\varrho\left(x\right)}\frac{1-\gamma\left(s,x\right)}{\gamma\left(s,x\right)}\frac{\gamma\left(x\right)}{1-\gamma\left(x\right)}\left(y-\nu\left(s,x\right)\right),\label{eq: s_1 choice ss}\\
s_{2}\left(s\mid X,G=1\right) & =0,\label{eq: s_2 choice ss}\\
s_{3}\left(s\mid x,W=1,G=1\right) & =\frac{1}{\pi}\frac{1}{\varrho\left(x\right)}\frac{\gamma\left(x\right)}{1-\gamma\left(x\right)}\left(\nu\left(s,x\right)-\bar{\nu}\left(x\right)\right),\label{eq: s_3 choice ss}\\
s_{4}\left(x\right) & =\gamma\left(x\right)\left(\frac{\bar{\nu}\left(x\right)-\theta_{1,1}}{\pi}\right),\quad\text{and}\label{eq: s_4 choice ss}\\
s_{5}\left(x\right) & =\gamma\left(x\right)\left(1-\gamma\left(x\right)\right)\left(\frac{\bar{\nu}\left(x\right)-\theta_{1,1}}{\pi}\right).\label{eq: s_5 choice ss}
\end{flalign}
These choices are again mean-zero and square integrable, and satisfy
the conditions (\ref{eq: cond mean ss 1}), (\ref{eq: cond mean ss 2}),
and (\ref{eq: cond mean ss 3}).

Recall the integral representation 
\begin{equation}
\theta_{1,1}=\int\int\int yp\left(y\mid s,x,G=1\right)p\left(s\mid x,W=1,G=0\right)p\left(x\mid G=1\right)\text{d\ensuremath{\lambda\left(x\right)\text{d}\lambda\left(s\right)\text{d}\lambda\left(x\right).}}\label{eq: integral rep ss}
\end{equation}
Observe that $p\left(w=1\mid x,G=0\right)$ does not appear in (\ref{eq: integral rep ss}).
Thus, the pathwise derivative of $\theta_{1,1}$ at $0$ on the submodel
of $\mathcal{P}$ is again given by \eqref{eq:pathwisederivacik}. Thus, it will suffice
to verify the conditions
\begin{flalign}
 & \mathbb{E}_{P}\left[Gs_{1}\left(Y\mid S,X,G=1\right)\ell^{\prime}\left(Y\mid X,S,G=1\right)\right]\label{eq: ss cond 1}\\
 & \quad\quad=\mathbb{E}_{P}\left[\mathbb{E}_{P}\left[\mathbb{E}_{P}\left[Y\ell^{\prime}\left(Y\mid S,X,G=1\right)\mid S,X,G=1\right]\mid X,W=1,G=0\right]\mid G=1\right]\nonumber \\
 & \mathbb{E}_{P}\left[W\left(1-G\right)s_{3}\left(S\mid X,W=1,G=1\right)\ell^{\prime}\left(y\mid x,s,G=1\right)\right]\label{eq: ss cond 2}\\
 & \quad\quad=\mathbb{E}_{P}\left[\mathbb{E}_{P}\left[\nu\left(S,X\right)\ell^{\prime}\left(S\mid X,W=1,G=0\right)\mid X,W=1,G=0\right]\mid G=1\right]\nonumber \\
 & \mathbb{E}_{P}\left[\left(G\ell^{\prime}\left(X\right)+\gamma^{\prime}\left(X\right)\right)\left(\bar{\nu}\left(x\right)-\theta_{1,1}\right)\right]\label{eq: ss cond 3}\\
 & \quad\quad=\mathbb{E}_{P}\left[\bar{\nu}\left(X\right)\left(\gamma^{\prime}\left(X\right)+\ell^{\prime}\left(X\right)\gamma\left(X\right)\right)\right]-\mathbb{E}_{P}\left[\gamma^{\prime}\left(X\right)+\ell^{\prime}\left(X\right)\gamma\left(X\right)\right]\theta_{1,1}\nonumber 
\end{flalign}
These are identical to the conditions \eqref{eq:acik1}, \eqref{eq:acik2}, and \eqref{eq:acik3} and
so are verified in the Proof of \cref{thm: EIF tau_1}, Part 2.  \hfill\qed

\subsubsection{\label{sec: ss check ii}Proof of \cref{thm: ss check}, Part (ii)}

Let $\mathcal{P}$ be a regular parametric submodel of $\mathcal{M}_{\lambda}$,
indexed by $\varepsilon\in\mathbb{R}$ and such that \cref{as: uex,as: ltoc,as: ceev,as: es} hold for each $P_{\varepsilon}\in\mathcal{P}$.
Additionally, suppose that 
\[
p_{\varepsilon}\left(s\mid x,G=1\right)=f\left(s\mid x\right)
\]
for all $\varepsilon\in\mathbb{R}$ and some fixed function $f\left(s\mid x\right)$. 

In this case, we have the score 
\begin{flalign*}
\ell^{\prime}\left(b_{1}\right) & =g\cdot\ell^{\prime}\left(y\mid x,s,G=1\right)+w\left(1-g\right)\cdot\ell^{\prime}\left(s\mid W=1,x,G=0\right)\\
 & +\ell^{\prime}\left(x\right)+\frac{g-\gamma\left(x\right)}{\gamma\left(x\right)\left(1-\gamma\left(x\right)\right)}\gamma^{\prime}\left(x\right)+\left(1-g\right)\frac{w-\varrho\left(x\right)}{\varrho\left(x\right)\left(1-\varrho\left(x\right)\right)}\varrho^{\prime}\left(x\right)
\end{flalign*}
and so the tangent space $\mathcal{T}$ is given by the mean-square
closure of the linear space of the functions 
\begin{flalign*}
s\left(b_{1}\right) & =g\cdot s_{1}\left(y\mid x,s,G=1\right)+w\left(1-g\right)s_{3}\left(s\mid W=1,x,G=0\right)\\
 & +s_{4}\left(x\right)+\frac{g-\gamma\left(x\right)}{\gamma\left(x\right)\left(1-\gamma\left(x\right)\right)}s_{5}\left(x\right)+\left(1-g\right)\frac{w-\varrho\left(x\right)}{\varrho\left(x\right)\left(1-\varrho\left(x\right)\right)}s_{6}\left(x\right)
\end{flalign*}
where the functions $s_{1}$ through $s_{6}$ range over the space
of mean-zero and square integrable functions that additionally satisfy
the restrictions (\ref{eq: cond mean ss 1}) and (\ref{eq: cond mean ss 3}).
Consider the choices (\ref{eq: s_1 choice ss}) through (\ref{eq: s_5 choice ss}),
noting that in this case $s_{2}\left(\cdot\right)$ does not appear,
and additionally choose $s_{6}\left(x\right)=0$. These choices are
again mean-zero and square integrable, and satisfy the conditions
(\ref{eq: cond mean ss 1}) and (\ref{eq: cond mean ss 3}). 

Observe that in the the integral representation (\ref{eq: integral rep ss}),
the quantity $p\left(s\mid x,G=1\right)$ does not appear. Thus, the
pathwise derivative of $\theta_{1,1}$ at $0$ on the submodel of
$\mathcal{P}$ is again given by \eqref{eq:pathwisederivacik} and it will again suffice
to verify the conditions (\ref{eq: ss cond 1}) through (\ref{eq: ss cond 3}).
These are identical to the conditions \eqref{eq:acik1}, \eqref{eq:acik2}, and \eqref{eq:acik3}  and
so are verified in the Proof of \cref{thm: EIF tau_1}, Part 2.   \hfill\qed

\subsubsection{\label{sec: ss check iii}Proof of \cref{thm: ss check}, Part (iii)}

Let $\mathcal{P}$ be a regular parametric submodel of $\mathcal{M}_{\lambda}$,
indexed by $\varepsilon\in\mathbb{R}$ and such that \cref{as: uex,as: ltoc,as: ceev,as: es} hold for each $P_{\varepsilon}\in\mathcal{P}$.
Additionally, suppose that 
\[
p_{\varepsilon}\left(G=1\mid X=x\right)=\gamma\left(x\right)
\]
for all $\varepsilon\in\mathbb{R}$ and some fixed function $\gamma\left(x\right)$. 

In this case, we have the score 
\begin{flalign*}
\ell^{\prime}\left(b_{1}\right) & =g\cdot\ell^{\prime}\left(y\mid x,s,G=1\right)+g\cdot\ell^{\prime}\left(s\mid x,G=1\right)+w\left(1-g\right)\cdot\ell^{\prime}\left(s\mid W=1,x,G=0\right)\\
 & +\ell^{\prime}\left(x\right)+\left(1-g\right)\frac{w-\varrho\left(x\right)}{\varrho\left(x\right)\left(1-\varrho\left(x\right)\right)}\varrho^{\prime}\left(x\right)
\end{flalign*}
and so the tangent space $\mathcal{T}$ is given by the mean-square
closure of the linear space of the functions 
\begin{flalign*}
s\left(b_{1}\right) & =g\cdot s_{1}\left(y\mid x,s,G=1\right)+g\cdot s_{2}\left(s\mid x,G=1\right)+w\left(1-g\right)s_{3}\left(s\mid W=1,x,G=0\right)\\
 & +s_{4}\left(x\right)+\left(1-g\right)\frac{w-\varrho\left(x\right)}{\varrho\left(x\right)\left(1-\varrho\left(x\right)\right)}s_{6}\left(x\right)
\end{flalign*}
where the functions $s_{1}$ through $s_{6}$ range over the space
of mean-zero and square integrable functions that additionally satisfy
the restrictions (\ref{eq: cond mean ss 1}) through (\ref{eq: cond mean ss 3}).
Consider the choices (\ref{eq: s_1 choice ss}) through (\ref{eq: s_3 choice ss})
as well as
\begin{flalign*}
s_{4}\left(x\right) & =\frac{\gamma\left(x\right)\bar{\nu}\left(x\right)}{\pi}-\theta_{1,1}\quad\text{and}\quad s_{6}\left(x\right)=0,
\end{flalign*}
which yield the conjectured influence function
\begin{flalign*}
\check{\xi}_{1,1}\left(b_{1}\right) & =\frac{g}{\pi}\frac{\varrho\left(s,x\right)}{\varrho\left(x\right)}\frac{1-\gamma\left(s,x\right)}{\gamma\left(s,x\right)}\frac{\gamma\left(x\right)}{1-\gamma\left(x\right)}\left(y-\nu\left(s,x\right)\right)\\
 & +\frac{1}{\pi}\frac{\gamma\left(x\right)}{1-\gamma\left(x\right)}\frac{w\left(\nu\left(s,x\right)-\bar{\nu}\left(x\right)\right)}{\varrho\left(x\right)}+\frac{\gamma\left(x\right)\bar{\nu}\left(x\right)}{\pi}-\theta_{1,1}.
\end{flalign*}
These choices are again mean-zero and square integrable, and satisfy
the conditions (\ref{eq: cond mean ss 1}) and (\ref{eq: cond mean ss 3}). 

Observe that 
\[
p_{\varepsilon}\left(x\mid G=1\right)=\frac{\gamma\left(x\right)p_{\varepsilon}\left(x\right)}{\pi}.
\]
Thus, by the integral representation (\ref{eq: integral rep ss})
the pathwise derivative of $\theta_{1,1}$ at $0$ on the submodel
of $\mathcal{P}$ is given by
\[
\theta_{1,1}^{\prime}=\mathbb{E}_{P}\left[\mathbb{E}_{P}\left[\mathbb{E}_{P}\left[YU\left(Y,S,X\right)\mid S,X,G=1\right]\mid X,W=1,G=0\right]\mid G=1\right],
\]
where 
\[
U\left(y,s,x\right)=\ell^{\prime}\left(y\mid s,x,G=1\right)+\ell^{\prime}\left(s\mid x,W=1,G=0\right)+\ell^{\prime}\left(x\right).
\]
Thus, we have that 
\begin{flalign*}
\theta_{1,1}^{\prime} & =\mathbb{E}_{P}\left[\mathbb{E}_{P}\left[\mathbb{E}_{P}\left[Y\ell^{\prime}\left(Y\mid S,X,G=1\right)\mid S,X,G=1\right]\mid X,W=1,G=0\right]\mid G=1\right]\\
 & +\mathbb{E}_{P}\left[\mathbb{E}_{P}\left[\nu\left(S.X\right)\ell^{\prime}\left(S\mid X,W=1,G=0\right)\mid X,W=1,G=0\right]\mid G=1\right]\\
 & +\pi^{-1}\mathbb{E}_{P}\left[\bar{\nu}\left(X\right)\gamma\left(X\right)\ell^{\prime}\left(X\right)\right].
\end{flalign*}
Note that the first two lines in this expression are unchanged relative
to \eqref{eq:pathwisederivacik}. By the choices of $s_{1}\left(\cdot\right)$ through $s_{6}\left(\cdot\right)$
given above, we have that 
\begin{flalign*}
\mathbb{E}\left[\check{\xi}_{1,1}\left(B_{1}\right)\ell^{\prime}\left(B_{1}\right)\right] & =\mathbb{E}_{P}\left[Gs_{1}\left(Y\mid S,X,G=1\right)\ell^{\prime}\left(Y\mid X,S,G=1\right)\right]\\
 & +\mathbb{E}\left[W\left(1-G\right)s_{3}\left(S\mid X,W=1,G=1\right)\ell^{\prime}\left(Y\mid X,S,G=1\right)\right]+\mathbb{E}\left[s_{4}\left(X\right)\ell^{\prime}\left(X\right)\right].
\end{flalign*}
The equality
\[
\mathbb{E}_{P}\left[\check{\xi}_{1,1}\left(B_{1}\right)\ell^{\prime}\left(B_{1}\right)\right]=\theta_{1,1}^{\prime},
\]
follows from the fact that (\ref{eq: ss cond 1}) and (\ref{eq: ss cond 2})
hold, as $s_{1}\left(\cdot\right)$ and $s_{3}\left(\cdot\right)$
are unchanged relative to the proof of \cref{thm: EIF tau_1}, Part 2, and 
\[
\mathbb{E}_{P}\left[s_{4}\left(X\right)\ell^{\prime}\left(X\right)\right]=\pi^{-1}\mathbb{E}_{P}\left[\bar{\nu}\left(X\right)\gamma\left(X\right)\ell^{\prime}\left(X\right)\right]
\]
holds by definition, completing the proof.  \hfill\qed

\subsection{Nested Condition Expectation Estimation\label{sec:nested}}

In this section, we derive explicit rates of convergence for estimators of the nested condition expectation functions $\bar{\mu}_w(x)$ and $\bar{\nu}_w(x)$, defined in \cref{eq: mu_w(x),eq: nu(x)}, based on linear sieves. To simplify exposition, we focus attention on the conditional expectation functions 
\begin{flalign*}
\mu\left(s,x\right) & =\mathbb{E}_{P}\left[Y\mid S=s,X=x,G=1\right]\quad\text{and}\quad\bar{\mu}\left(x\right)=\mathbb{E}_{P}\left[\mu\left(S,x\right)\mid X=x,G=0\right]~,
\end{flalign*}
and let $Z_i=(S_i,X_i)$ collect the short-term outcomes and pretreatment covariates. 

\subsubsection{Estimation}

We consider the following simple linear sieve estimators for $\mu(\cdot,\cdot)$ and $\bar{\mu}(\cdot)$. For functions $h$ and $\bar{h}$ in $L_{2}\left(P_{Z}\right)$ and $L_{2}\left(P_{X}\right)$, define the loss functions 
\begin{equation}
L\left(h,A\right)=G\left(Y-h\left(Z\right)\right)^{2}\quad\text{and}\quad\bar{L}\left(h,A\right)=\left(1-G\right)\left(\mu\left(X,S\right)-\bar{h}\left(X\right)\right)^{2}\label{eq: losses}
\end{equation}
as well as the population risk functions 
\begin{equation}
R_{P}\left(h\right)=\mathbb{\mathbb{E}}_{P}\left[L\left(h,A\right)\right]\quad\text{and}\quad\bar{R}_{P}\left(h\right)=\mathbb{\mathbb{E}}_{P}\left[\bar{L}\left(h,A\right)\right].\label{eq: risks}
\end{equation}
It is clear that the functions $\mu\left(\cdot,\cdot\right)$ and $\bar{\mu}\left(\cdot\right)$ minimize the risks (\ref{eq: risks}) over $L^{2}\left(P_{Z}\right)$ and $L^{2}\left(P_{X}\right)$. 

Throughout, we assume that the functions $\mu(\cdot,\cdot)$ and $\bar{\mu}(\cdot)$ are
elements of the linear subspaces $\Theta\subseteq L_{2}\left(P_{Z}\right)$
and $\bar{\Theta}\subseteq L_{2}\left(P_{X}\right)$, respectively. 
 Define the sequences of finite-dimensional parameter spaces $\Theta_{1}\subseteq\Theta_{2}\subseteq\cdots\subseteq\Theta$
and $\bar{\Theta}_{1}\subseteq\bar{\Theta}_{2}\subseteq\cdots\subseteq\bar{\Theta}$.
Let $\Pi_{m}:\Theta\to\Theta_{m}$ denote the $L_{2}\left(P_{Z}\right)$
projection of $\Theta$ onto $\Theta_{m}$ and let $\bar{\Pi}_{m}:\bar{\Theta}\to\bar{\Theta}_{m}$
denote the $L_{2}\left(P_{X}\right)$ projection of $\bar{\Theta}$
onto $\bar{\Theta}_{m}$. 

Define the empirical criterion and associated sieve extremum estimator
\begin{equation}
R_{n}\left(h\right)=\sum_{i=1}^{n}G_{i}\left(Y-h\left(Z_{i}\right)\right)^{2}\quad\text{and}\quad\hat{\mu}_{m,n}=\underset{h\in\Theta_{m}}{\arg\min\text{ }}R_{n}\left(h\right).\label{eq: mu def}
\end{equation}
Proceeding analogously, define the infeasible empirical criterion
and associated sieve extremum estimator 
\begin{equation}
\bar{R}_{n}\left(\bar{h}\right)=\sum_{i=1}^{n}G_{i}\left(\mu\left(Z_{i}\right)-\bar{h}\left(X_{i}\right)\right)^{2}\quad\text{and}\quad\hat{\bar{\mu}}_{m,n}^{*}=\underset{h\in\Theta_{m}}{\arg\min\text{ }}\bar{R}_{n}\left(\bar{h}\right)\label{eq: infeasible bar mu def}
\end{equation}
in addition to the feasible empirical criterion and associated sieve
extremum estimator 
\begin{equation}
\bar{R}_{m,n}\left(\bar{h}\right)=\sum_{i=1}^{n}\left(1-G_{i}\right)\left(\hat{\mu}_{m,n}\left(Z_{i}\right)-\bar{h}\left(X_{i}\right)\right)^{2}\quad\text{and}\quad\hat{\bar{\mu}}_{m,n}=\underset{h\in\bar{\Theta}_{m}}{\arg\min\text{ }}\bar{R}_{m,n}\left(\bar{h}\right).\label{eq: feasible bar mu def}
\end{equation}

\subsubsection{Consistency}

The following assumptions are sufficient to show that the estimators $\hat{\mu}_{m,n}$ and $\quad\hat{\bar{\mu}}_{m,n}$ are consistent in the norm $\|\cdot\|_{P,2}$ for their respective estimands.

\begin{assumption}
\label{as: support and sieves}$\text{ }$\\
\textbf{(i)} The variables $Z_{i}$ have support $\mathcal{Z}=\mathcal{S\times\mathcal{X}},$
where $\mathcal{S}\subset\mathbb{R}^{d}$ and $\mathcal{X}\subset\mathbb{R}^{k}$
are compact. 

\noindent
\textbf{(ii)} The sieve spaces $\Theta_{m}$ and $\bar{\Theta}_{m}$
satisfy
\[
\|\Pi_{m}\left(\mu\right)-\mu\|_{P_{Z},2}\to0\quad\text{and}\quad\|\bar{\Pi}_{m}\left(\bar{\mu}\right)-\bar{\mu}\|_{P_{X},2}\to0
\]
as $n\to\infty$. 
\end{assumption}
\noindent
Consistency follows from \cref{lem: approximate minimizer,lem: Minimizer Approximation}, which are stated and proved in \cref{sec: auxlems}.
\begin{theorem}
\label{thm: double exp consistency}Under \cref{as: support and sieves}, the estimators \textup{$\hat{\mu}_{m,n}$ and $\hat{\bar{\mu}}_{m,n}$ satisfy
\[
\|\hat{\mu}_{m,n}-\mu\|_{P,2}\overset{p}{\to}0\quad\text{and}\quad\|\hat{\bar{\mu}}_{m,n}-\bar{\mu}\|_{P,2}\overset{p}{\to}0,
\]
respectively, as $n,m\to\infty$.}
\end{theorem}
\begin{proof}
By strict convexity, the population risk function $R_{P}\left(h\right)$
has a unique minimizer over the convex subspace $\Theta_{m}$, denoted
by $\mu_{m}^{*}$. We have that $\|\hat{\mu}_{m,n}-\mu_{m}^{*}\|_{P,2}\overset{p}{\to}0$
as $n\to\infty$ and $\|\mu_{m}^{*}-\mu\|_{P,2}\to0$ as $m\to\infty$
by Lemmas \ref{lem: approximate minimizer} and \ref{lem: Minimizer Approximation},
respectively. The triangle inequality then gives 
\[
\|\hat{\mu}_{m,n}-\mu\|_{P,2}\leq\|\hat{\mu}_{m,n}-\mu_{m}^{*}\|_{P,2}+\|\mu_{m}^{*}-\mu\|_{P,2}\overset{p}{\to}0,
\]
as $m,n\to\infty$. An identical argument establishes the convergence
of $\hat{\bar{\mu}}_{m,n}$.\hfill
\end{proof}

\subsubsection{Rate of Convergence}

Next, derive explicit rates of convergence for the linear sieve estimators $\hat{\mu}_{m,n}$ and $\hat{\bar{\mu}}_{m,n}$. We discipline this exercise by assuming that the functions $\mu(\cdot)$ and $\bar{\mu}(\cdot)$ are elements of Hölder classes with known parameterizations. 

To this end, we fix the following notation. If $\mathcal{X}$ is a compact set in $\mathbb{R}^{d}$ and $0<\gamma\leq1$ is a constant, then a real-valued function $h$ on $\mathcal{X}$ is said to satisfy a Hölder condition with exponent $\gamma$ if there is a positive number $c$ such that $\vert h\left(x\right)-h\left(y\right)\vert\leq c\|x-y\|_{2}^{\gamma}$ for all $x,y\in\mathcal{X}$. For a $d$-dimensional multi-index $\beta$, let $D^{\beta}$ denote the associated Differential operator and let $\vert\beta\vert=\beta_{1}+\cdots+\beta_{d}$. A real-valued function $h$ on $\mathcal{X}$ is said to be $p$-smooth if it is $m$ times continuously differentiable on $\mathcal{X}$, $D^{\beta}h$ satisfies a Hölder condition with exponent $\gamma$ for all $\vert\beta\vert=m$, and $m+\gamma=p$. The Hölder class $\Lambda^{p}\left(\mathcal{X}\right)$ denotes the space of all $p$-smooth functions on $\mathcal{X}$. Let $C^{m}\left(\mathcal{X}\right)$ denote the space of all $m$-times continuously differentiable functions on $\mathcal{X}$. The Hölder ball on $\mathcal{X}$ with width $c$ and smoothness $p$ is given by 
\begin{flalign*}
\Lambda_{c}^{p}\left(\mathcal{X}\right) & =\left\{ h\in C^{m}\left(\mathcal{X}\right):\sup_{\vert\beta\vert\leq m}\sup_{x\in\mathcal{X}}\vert D^{\beta}h\left(x\right)\vert\leq c,\sup_{\vert\beta\vert=m}\sup_{x,y\in\mathcal{X},x\neq y}\frac{\vert D^{\beta}h\left(x\right)-D^{\beta}h\left(y\right)\vert}{\vert x-y\vert_{2}^{\gamma}}\vert\leq c\right\} .
\end{flalign*}
We make the following assumptions on the regularity of the data, smoothness of the target functions, the quality of their approximation by the chosen linear sieve spaces. 

\begin{assumption}
$\text{ }$\label{as: rate}
\noindent \textbf{(i)} The moment bounds 
\begin{flalign}
\sup_{z\in\mathcal{Z}}\text{ }\mathbb{E}_{P}\left[\left(Y-\mu\left(z\right)\right)^{2}\mid Z=z,G=1\right] & <C\quad\text{and}\label{eq: conditional variance bound}\\
\sup_{z\in\mathcal{Z}}\text{ }\mathbb{E}_{P}\left[\left(\mu\left(S,x\right)-\bar{\mu}\left(x\right)\right)^{2}\mid X=x,G=0\right] & <\bar{C}\nonumber 
\end{flalign}
hold for $\lambda$-almost every $z\in\mathcal{Z}$ and $x\in\mathcal{X}$,
respectively. 

\noindent \textbf{(ii)} The densities of $X_{i}$ and $Z_{i}$ are
uniformly bounded away from zero and infinity. 

\noindent \textbf{(iii) }The function classes $\Theta$ and $\bar{\Theta}$
are given by $\Lambda_{c}^{p}\left(\mathcal{Z}\right)$ and $\bar{\mu}\in\Lambda_{\bar{c}}^{\bar{p}}\left(\mathcal{X}\right)$,
respectively.

\noindent \textbf{(iv) }The sieve spaces $\Theta_{n}$ and $\bar{\Theta}_{n}$
are linear and satisfy
\[
\mathsf{dim}\left(\Theta_{n}\right)=J_{n}^{d+k}\quad\text{and}\quad\mathsf{dim}\left(\bar{\Theta}_{n}\right)=J_{n}^{k}
\]
as well as 
\[
\inf_{h\in\Theta_{n}}\|h-\mu\|=O\left(J_{n}^{-p}\right)\quad\text{and}\quad\inf_{h\in\Theta_{n}}\|h-\bar{\mu}\|=O\left(\bar{J}_{n}^{-\bar{p}}\right)
\]
for some sequences $J_{n}$ and $\bar{J}_{n}$, respectively. 

\end{assumption}

The following Theorem, stating rates of convergence for the estimators $\hat{\mu}_{m,n}$ and $\hat{\bar{\mu}}_{m,n}$, is a consequence of Theorem 3.2 of \cite{chen2007large}. The proof uses Lemma 2 of \cite{chen1998sieve}, in addition to \cref{lem: loss rewrite,lem: linear sieve}, which are stated and proved in \cref{sec: auxlems}.

\begin{theorem}
\label{thm: double exp rate}Suppose that \cref{as: support and sieves}, (i), and \cref{as: rate} hold, and that the sieve spaces $\Theta_{n}$ and $\bar{\Theta}_{n}$
are chosen such that
\begin{equation}
J_{n}\asymp\left(\frac{n}{\log n}\right)^{\frac{1}{2p+d+k}}\quad\text{and}\quad \bar{J}_{n}\asymp\left(\frac{n}{\log n}\right)^{\frac{1}{2\bar{p}+k}}.\label{eq: choice of J}
\end{equation}
\textbf{(i)} The estimator $\hat{\mu}_{n,n}$ satisfies
\begin{flalign}
\|\hat{\mu}_{n,n}-\mu\|_{P,2} & =O_{P}\left(\left(\frac{\log n}{n}\right)^{\frac{p}{2p+d+k}}\right)\label{eq: mu rate}
\end{flalign}
as $n\to\infty$.

\noindent \textbf{(ii)} If 
\[
\frac{\bar{p}}{2\bar{p}+k}\leq\frac{1}{2}\frac{p}{2p+d+k},
\]
then the estimator $\hat{\bar{\mu}}_{n,n}$ satisfies 
\begin{equation}
\|\hat{\bar{\mu}}_{n,n}-\bar{\mu}\|_{P,2}=O_{P}\left(\left(\frac{\log n}{n}\right)^{\frac{p}{2p+d+k}}\right)\label{eq: bar mu rate}
\end{equation}
as $n\to\infty$. 
\end{theorem}

\begin{remark}
Assumption \ref{as: rate}, (iv), will hold if $\Theta_{n}$
is the $d+k$ tensor product of the space of polynomials of degree
$J_{n}$ or less. Similar statements hold for different choices of
sieve spaces, e.g., tensor products of univariate splines or orthogonal
wavelets. See e.g., Section 2.3 of \cite{chen2007large} for further discussion.
\hfill{}$\blacksquare$
\end{remark}

\begin{proof}
First, we verify the conditions of Theorem 3.2 of \cite{chen2007large}, re-stated here as Lemma \ref{lem: sieve rate}, for the loss functions $L_{P}\left(\cdot,A\right)$ and $\bar{L}_{P}\left(\cdot,A\right)$, defined in (\ref{eq: losses}). These conditions will allow for the application of \cref{lem: linear sieve} in both cases. For this task, we will make frequent use of the functions $K(y,z)$ and $\bar{K}(s,x)$, defined in \cref{lem: loss rewrite}.

In both cases, Condition (i) follows immediately from the smoothness
and convexity of the population risk functions $R_{P}\left(\cdot\right)$
and $\bar{R}_{P}\left(\cdot\right)$ about their minimizers. 

To verify Condition (ii), fix $h\in\Theta_{n}$ such that $\|h-\mu\|_{P,2}\leq\eta$,
and observe that 
\begin{flalign*}
\Var_{P}\left(L\left(h,A\right)-L\left(\mu,A\right)\right) & \leq\mathbb{E}_{P}\left[\left(L\left(h,A\right)-L\left(\mu,A\right)\right)^{2}\right]\\
 & \leq\mathbb{E}_{P}\left[\left(\left[K\left(Y,Z\right)\left(h\left(Z\right)-\mu\left(Z\right)\right)\right]\right)^{2}\right]\lesssim\|h-\mu\|_{p,2}^{2}\lesssim\eta^{2},
\end{flalign*}
 as required. Here, the second and third inequalities follow from
Lemma \ref{lem: loss rewrite}. An identical argument verifies Condition
(ii) for the loss $\bar{L}_{P}\left(\cdot,A\right)$. 

To verify Condition (iii) of Lemma
\ref{lem: sieve rate}, again fix $h\in\Theta_{n}$ with $\|h-\mu\|_{P,2}\leq\eta$,
and observe that 
\begin{flalign*}
\vert L\left(h,A\right)-L\left(\mu,A\right)\vert & \leq K\left(Y,Z\right)\|h-\mu\|_{\infty},\\
 & \lesssim K\left(Y,Z\right)\|h-\mu\|_{\lambda,2}^{\frac{2p}{2p+d+k}}\lesssim K\left(Y,Z\right)\|h-\mu\|_{P,2}^{\frac{2p}{2p+d+k}}~,
\end{flalign*}
where the first inequality follows from Lemma \ref{lem: loss rewrite},
the second inequality follows from Lemma \ref{lem: h inf h 2 bound},
and the third inequality follows from \cref{as: rate},
Part (ii), which implies that $\|\cdot\|_{\lambda,2}$ and $\|\cdot\|_{P,2}$
are equivalent norms on $\mathcal{Z}$. Similarly, if $\bar{h}\in\bar{\Theta}_{n}$
with $\|\bar{h}-\bar{\mu}\|_{P,2}\leq\eta$, then 
\begin{flalign*}
\vert\bar{L}\left(\bar{h},A\right)-\bar{L}\left(\bar{\mu},A\right)\vert & \leq\bar{K}\left(S,X\right)\|\bar{h}-\bar{\mu}\|_{\infty},\\
 & \lesssim\bar{K}\left(S,X\right)\|\bar{h}-\bar{\mu}\|_{\lambda,2}^{\frac{2\bar{p}}{2\bar{p}+k}}\lesssim\bar{K}\left(S,X\right)\|\bar{h}-\bar{\mu}\|_{P,2}^{\frac{2\bar{p}}{2\bar{p}+k}}~,
\end{flalign*}
again by Lemmas \ref{lem: loss rewrite} and \ref{lem: h inf h 2 bound}
and \cref{as: rate}, Part (ii).

Thus, in the case of $L_{P}\left(\cdot,A\right)$, Condition (iii)
is verified by setting $U\left(A\right)=C\cdot K\left(Y,Z\right)$
for some constant $C$ and $s=\frac{2p}{2p+d+k}$. Similarly, in the
case of $\bar{L}_{P}\left(\cdot,A\right)$, Condition (iii) is verified
by setting $U\left(A\right)=\bar{C}\cdot\bar{K}\left(S,X\right)$
for some constant $\bar{C}$ and $s=\frac{2\bar{p}}{2\bar{p}+k}$.
In both cases, the required bound deterministic on $\mathbb{E}\left[U\left(A\right)^{2}\right]$
follows from Assumption \ref{as: rate},
Part (i), and the definitions of $K\left(Y,Z\right)$ and $\bar{K}\left(S,X\right)$. 

Now, we apply \cref{lem: linear sieve} to verify the rate (\ref{eq: mu rate}) for the estimator $\hat{\mu}_{n,n}$. Observe that \cref{as: rates}, Part (iv) implies \cref{as: support and sieves}, Part (ii). 
Consequently, we have that $\|\hat{\mu}_{n,n}-\mu\|=o_{P}\left(1\right)$ as
$n\to\infty$ by Theorem \ref{thm: double exp consistency}. Thus, as $\hat{\mu}_{n,m}$ is a sequence of estimators that satisfy 
\[
\frac{1}{n}R_{n}\left(\hat{\mu}_{n,n}\right)=\sup_{h\in\Theta_{n}}\frac{1}{n}R_{n}\left(h\right)~,
\]
the rate (\ref{eq: mu rate}) follows from \cref{as: rates}, Part (iv), and Lemma \ref{lem: linear sieve}.

It remains to verify the rate (\ref{eq: bar mu rate}) for the estimator $\hat{\bar{\mu}}_{n,n}$. This result is achieved through an inductive application of \cref{lem: linear sieve}. 
First note that, \cref{as: rates}, Part (iv) again implies that $\|\hat{\bar{\mu}}_{n,n}-\bar{\mu}\|=o_{P}\left(1\right)$ as $n\to\infty$ by Theorem \ref{thm: double exp consistency}. 

Observe that the rate (\ref{eq: mu rate}) implies that for any $\bar{h}\in\bar{\Theta}_{n}$,
\begin{flalign}
\frac{1}{n}\bar{R}_{n}\left(\bar{h}\right) & =\frac{1}{n}\sum_{i=1}^{n}\left(1-G_{i}\right)\left(\left(\mu\left(S_{i},X_{i}\right)-\hat{\mu}_{n,n}\left(S_{i},X_{i}\right)\right)+\left(\hat{\mu}_{n,n}\left(S_{i},X_{i}\right)-\bar{h}\left(X_{i}\right)\right)\right)^{2}\nonumber \\
 & =\frac{1}{n}\bar{R}_{n,n}\left(\bar{h}\right)+\frac{1}{n}\sum_{i=1}^{n}\left(1-G_{i}\right)\left(\mu\left(S_{i},X_{i}\right)-\hat{\mu}_{n,n}\left(S_{i},X_{i}\right)\right)^{2}\nonumber \\
 & +\frac{2}{n}\sum_{i=1}^{n}\left(1-G_{i}\right)\left(\mu\left(S_{i},X_{i}\right)-\hat{\mu}_{n,n}\left(S_{i},X_{i}\right)\right)\left(\hat{\mu}_{n,n}\left(S_{i},X_{i}\right)-\bar{h}\left(X_{i}\right)\right)\nonumber\\
 & =\frac{1}{n}\bar{R}_{n,n}\left(\bar{h}\right)+O_{P}\left(\left(\frac{\log n}{n}\right)^{\frac{2p}{2p+d+k}}\right)\nonumber \\
 & +\frac{2}{n}\sum_{i=1}^{n}\left(1-G_{i}\right)\left(\mu\left(S_{i},X_{i}\right)-\hat{\mu}_{n,n}\left(S_{i},X_{i}\right)\right)\left(\hat{\mu}_{n,n}\left(S_{i},X_{i}\right)-\bar{h}\left(X_{i}\right)\right).\label{eq: cross-product-term}
\end{flalign}
The term (\ref{eq: cross-product-term}) can be written $C_n + D_n(\bar{h})$, where
\begin{align}
C_n &=\frac{2}{n}\sum_{i=1}^{n}\left(1-G_{i}\right)\left(\mu\left(S_{i},X_{i}\right)-\hat{\mu}_{n,n}\left(S_{i},X_{i}\right)\right)\cdot
\left(\hat{\mu}_{n,n}\left(S_{i},X_{i}\right)-\bar{\mu}\left(X_{i}\right)\right)\quad\text{and} \nonumber\\
D_n(\bar{h}) &= \frac{2}{n}\sum_{i=1}^{n} \left(1-G_{i}\right)\left(\mu\left(S_{i},X_{i}\right)-\hat{\mu}_{n,n}\left(S_{i},X_{i}\right)\right)
\cdot \left(\bar{\mu}\left(X_{i}\right)-\bar{h}\left(X_{i}\right)\right)\nonumber~.
\end{align}
Observe that the term $C_n$ does not depend on $\bar{h}$. Consequently, if we define the auxillary loss function
\begin{align}
\tilde{R}_n(\bar{b}) = \bar{R}_{n}\left(\bar{h}\right) - nC_n~,\nonumber
\end{align}
then the minimization problem is unchanged, in the sense that
\begin{equation}
\hat{\bar{\mu}}_{m,n}^{*}=\underset{h\in\Theta_{m}}{\arg\min\text{ }}\tilde{R}_{n}\left(\bar{h}\right)~,\nonumber
\end{equation}
as in \eqref{eq: infeasible bar mu def}, and we obtain the expansion
\begin{align}
\frac{1}{n}\tilde{R}_{n}\left(\bar{h}\right) =\frac{1}{n}\bar{R}_{n,n}\left(\bar{h}\right) + D_n(\bar{h}) + O_{P}\left(\left(\frac{\log n}{n}\right)^{\frac{2p}{2p+d+k}}\right)~, \label{eq: aux loss decomp}
\end{align}
from the expression \eqref{eq: cross-product-term}. 

Observe that the relations
\begin{flalign}
\frac{1}{n}\bar{R}_{n,n}\left(\hat{\bar{\mu}}_{n,n}\right) & =\frac{1}{n}\tilde{R}_{n}\left(\hat{\bar{\mu}}_{n,n}\right)-D_n(\hat{\bar{\mu}}_{n,n})+O_{P}\left(\left(\frac{\log n}{n}\right)^{\frac{2p}{2p+d+k}}\right),\label{eq: risk difference at hat}\\
\frac{1}{n}\bar{R}_{n,n}\left(\hat{\bar{\mu}}_{n,n}\right) & \geq\frac{1}{n}\bar{R}_{n,n}\left(\hat{\bar{\mu}}_{n,n}^{*}\right),\quad\text{and}\nonumber \\
\frac{1}{n}\bar{R}_{n,n}\left(\hat{\bar{\mu}}_{n,n}^{*}\right) & =\frac{1}{n}\tilde{R}_{n}\left(\hat{\bar{\mu}}_{n,n}^{*}\right)-D_n(\hat{\bar{\mu}}_{n,n}^{*})+O_{P}\left(\left(\frac{\log n}{n}\right)^{\frac{2p}{2p+d+k}}\right),\label{eq: risk difference at hat star}
\end{flalign}
where $\hat{\bar{\mu}}_{n,n}$ is defined in (\ref{eq: infeasible bar mu def}),
imply the inequality
\begin{equation}
\frac{1}{n}\bar{R}_{n}\left(\hat{\bar{\mu}}_{n,n}\right)\geq
\frac{1}{n}\bar{R}_{n}\left(\hat{\bar{\mu}}_{n,n}^{*}\right)
+\left(D_n(\hat{\bar{\mu}}_{n,n}) - D_n(\hat{\bar{\mu}}_{n,n}^{*})\right)
+O_{P}\left(\left(\frac{\log n}{n}\right)^{\frac{2p}{2p+d+k}}\right)~.\label{eq: first inequality}
\end{equation}
Moreover, the term $D_n(\hat{\bar{\mu}}_{n,n}^{*}))$ is bounded from above by
\begin{equation}
\frac{2}{n}\sum_{i=1}^{n}\|\mu\left(S_{i},X_{i}\right)-\hat{\mu}_{n,n}\left(S_{i},X_{i}\right)\|_{2}\|\bar{\mu}\left(X_{i}\right)-\bar{h}\left(X_{i}\right)\|_{2}=O_{P}\left(\left(\frac{\log n}{n}\right)^{\frac{p}{2p+d+k}+\frac{\bar{p}}{2\bar{p}+k}}\right)~,\label{eq: approx bound}
\end{equation}
by the Cauchy-Schwarz inequality, the rate (\ref{eq: mu rate}), and the choice \eqref{eq: choice of J}. Plugging this bound into \eqref{eq: first inequality} implies 
\begin{equation}
\frac{1}{n}\bar{R}_{n}\left(\hat{\bar{\mu}}_{n,n}\right)\geq
\frac{1}{n}\bar{R}_{n}\left(\hat{\bar{\mu}}_{n,n}^{*}\right)
+D_n(\hat{\bar{\mu}}_{n,n})
+O_{P}\left(\left(\frac{\log n}{n}\right)^{\frac{p}{2p+d+k} + \min\big\{ \frac{p}{2p+d+k},\frac{\bar{p}}{2\bar{p}+k}\big\}}\right)~.\label{eq: second inequality}
\end{equation}
Thus, in order to apply \cref{lem: linear sieve}, must must control the term $D_n(\hat{\bar{\mu}}_{n,n})$.

To this end, suppose that it were known that 
\begin{align} \label{eq: rate hypothesis}
\|\hat{\bar{\mu}}_{n,n}-\bar{\mu}\|_{2,P}=O_{P}\left(\left(\frac{\log n}{n}\right)^{s}\right)
\end{align}
for some constant $0\leq s <1$. In this case, the term $D_n(\hat{\bar{\mu}}_{n,n})$ is similarly bounded from above by 
\begin{equation}
\frac{2}{n}\sum_{i=1}^{n}\|\mu\left(S_{i},X_{i}\right)-\hat{\mu}_{n,n}\left(S_{i},X_{i}\right)\|_{2}\|\bar{\mu}\left(X_{i}\right)-\bar{h}\left(X_{i}\right)\|_{2}=O_{P}\left(\left(\frac{\log n}{n}\right)^{\frac{p}{2p+d+k}+s}\right)~,\label{eq: loose bound}
\end{equation}
by the Cauchy-Schwarz inequality and the rate (\ref{eq: mu rate}), which again implies
\begin{equation}
\frac{1}{n}\bar{R}_{n}\left(\hat{\bar{\mu}}_{n,n}\right)\geq
\frac{1}{n}\bar{R}_{n}\left(\hat{\bar{\mu}}_{n,n}^{*}\right)
+O_{P}\left(\left(\frac{\log n}{n}\right)^
{\frac{p}{2p+d+k} + \min\big\{s,\frac{p}{2p+d+k},\frac{\bar{p}}{2\bar{p}+k}\big\}}\right)
~.\label{eq: third inequality}
\end{equation} 
Hence, by \cref{as: rate}, Part (iv), Lemma \ref{lem: linear sieve} implies that
\begin{align} \label{eq: rate conclusion}
\|\hat{\bar{\mu}}_{n,n}-\bar{\mu}\|_{2,P}=O_{P}\left(\left(\frac{\log n}{n}\right)^{r(s)}\right)~,
\end{align}
where
\begin{align} \label{eq: rate update rule}
r(s) = \max\bigg\{ \frac{1}{2}\left(\frac{p}{2p+d+k} + \min\bigg\{s,\frac{p}{2p+d+k},\frac{\bar{p}}{2\bar{p}+k}\bigg\}\right), 
\frac{\bar{p}}{2\bar{p}+k} \bigg\}~.
\end{align}
Observe that we have shown $\bar{\mu}_{n,n}$ converges to $\bar{\mu}$ in $\|\cdot\|_{2,P}$ at rate $s$, then it also converges at the rate $r(s)$. 

To conclude the proof, observe that \eqref{eq: rate hypothesis} must be satisfied by $s=0$, as $\hat{\bar{\mu}}_{n,n}$ and $\bar{\mu}$ are subsets of $\Lambda^{\bar{p}}_{\bar{c}(\mathcal{X})}$. Thus, if we define the sequence $(s_i)_{i=0}^\infty$ by $s_i = r(s_{i-1})$, with initial condition $s_0=1$, we find that 
\[
\lim_{i\to\infty} s_i = \frac{p}{2p+d+k}\quad\text{if}\quad\frac{\bar{p}}{2\bar{p}+k}\leq\frac{1}{2}\frac{p}{2p+d+k}~,
\]
and so the rate of covergence \eqref{eq: bar mu rate} is obtained by induction.\hfill{}
\end{proof}

\subsubsection{Auxiliary Lemmas\label{sec: auxlems}}

\begin{lemma}[{\citealp[Theorem 2.7, ][]{newey1994large}}]
\label{lem: Newey McFadden}If there is a function $Q_{0}\left(\theta\right)$
defined on a convex set $\Theta$ such that (i) $Q_{0}\left(\theta\right)$
is uniquely minimized at $\theta_{0}$, (ii) $\theta_{0}$ is an element
of the interior of $\Theta$, (ii) the function $\hat{Q}_{n}\left(\theta\right)$
defined on $\Theta$ is concave, and (iv) $\hat{Q}_{n}\left(\theta\right)\to Q_{0}\left(\theta\right)$
for all $\theta\in\Theta$, then 
\[
\hat{\theta}_{n}=\arg\max_{\theta\in\Theta}\text{ }\hat{Q}_{n}\left(\theta\right)
\]
 exists with probability one and $\hat{\theta}_{n}\overset{p}{\to}\theta_{0}$
as $n\to\infty$.
\end{lemma}
\begin{lemma}
\label{lem: approximate minimizer}Under \cref{as: support and sieves}, the estimators \textup{$\hat{\mu}_{m,n}$ and $\hat{\bar{\mu}}_{m,n}$ satisfy
\[
\|\hat{\mu}_{m,n}-\mu_{m}^{*}\|_{P,2}\overset{p}{\to}0\quad\text{and}\quad\|\hat{\bar{\mu}}_{m,n}-\bar{\mu}_{m}^{*}\|_{P,2}\overset{p}{\to}0
\]
as $n\to\infty$ and $m$ is sufficiently large, respectively. }
\end{lemma}
\begin{proof}
It will suffice to verify the conditions of Lemma \ref{lem: Newey McFadden}.
Note first that both $\Theta_{m}$ and $\bar{\Theta}_{m}$ are convex.
The population criteria $R_{P}\left(\cdot\right)$ and $\bar{R}_{P}\left(\cdot\right)$
are both uniquely minimized at $\mu_{m}^{*}$ and $\bar{\mu}_{m}^{*}$
over $\Theta_{m}$ and $\bar{\Theta}_{m}$, respectively, by convexity.
Both $\mu_{m}^{*}$ and $\bar{\mu}_{m}^{*}$ are elements of the interiors
of $\Theta_{m}$ and $\bar{\Theta}_{m}$ for sufficiently large $m$
by Lemma \ref{lem: Minimizer Approximation}. It is clear that the
functions $R_{n}\left(h\right)$ and $\bar{R}_{m,n}\left(\bar{h}\right)$
are strictly convex. 

The pointwise convergence
\[
R_{n}\left(h\right)\overset{\text{p}}{\to}R_{p}\left(h\right)
\]
as $n\to\infty$ for $m$ sufficiently large follows from the weak
law of large numbers, and so 
\begin{equation}
\|\hat{\mu}_{m,n}-\mu_{m}^{*}\|_{P,2}\overset{p}{\to}0\label{eq: mu m conv}
\end{equation}
as $n\to\infty$ for $m$ sufficiently large follows from Lemma \ref{lem: Newey McFadden}.
Consequently, the pointwise convergence 
\[
\bar{R}_{n}\left(\bar{h}\right)\overset{\text{p}}{\to}\bar{R}_{p}\left(\bar{h}\right)
\]
follows from (\ref{eq: mu m conv}) and the weak law of large numbers,
again giving that $\|\hat{\bar{\mu}}_{m,n}-\bar{\mu}_{m}^{*}\|_{P,2}\overset{p}{\to}0$
as $n\to\infty$ for $m$ sufficiently large by Lemma \ref{lem: Newey McFadden}.\hfill
\end{proof}
\begin{lemma}
\label{lem: Minimizer Approximation}Under \cref{as: support and sieves}, the unique minimizers $\mu_{m}^{*}$ and $\bar{\mu}_{n}^{*}$ of $R_{P}\left(h\right)$ and $\bar{R}_{P}\left(h\right)$
over $\Theta_{m}$ and $\bar{\Theta}_{m}$ satisfy
\[
\|\mu_{m}^{*}-\mu\|_{P,2}\to0\quad\text{and}\quad\|\bar{\mu}_{n}^{*}-\bar{\mu}\|_{P,2}\to0
\]
as $m\to\infty$, respectively.
\end{lemma}
\begin{proof}
We prove only $\|\mu_{m}^{*}-\mu\|_{P,2}\to0$ as $m\to\infty$, as
the argument for $\|\bar{\mu}_{n}^{*}-\bar{\mu}\|_{P,2}\to0$ as $m\to\infty$
is identical. Assume, by the way of contradiction, that there exists
some $\delta>0$ such that $\|\mu_{m}^{*}-\mu\|_{P,2}>\delta$ for
all $m$ greater than some integer $m_{1}$. As $R_{P}\left(\cdot\right)$
is strictly convex on $\Theta$, we have that 
\begin{equation}
R_{P}\left(\mu_{m}^{*}\right)>R_{P}\left(\mu\right)+\varepsilon\label{eq: proj truth pop diff}
\end{equation}
for all $m\geq\max\left\{ m_{0},m_{1}\right\} $ for some $\varepsilon>0$.
By the continuity of $R_{P}\left(\cdot\right)$, there exists some
$\delta^{\prime}>0$ such that $\|\mu^{\prime}-\mu\|_{P,2}<\delta^{\prime}$
implies 
\begin{equation}
R_{P}\left(\mu^{\prime}\right)<R_{P}\left(\mu\right)+\varepsilon.\label{eq: continuity}
\end{equation}
By \cref{as: support and sieves}, (ii), there exists some integer $m_{1}$
such that $\|\Pi_{m}\left(\mu\right)-\mu\|_{P,2}<\delta^{\prime}$
for all $m\geq m_{1}$. Thus, by (\ref{eq: continuity}) and (\ref{eq: proj truth pop diff}),
we have that 
\[
R_{P}\left(\Pi_{m}\left(\mu\right)\right)<R_{P}\left(\mu\right)+\varepsilon<R_{P}\left(\mu_{m}^{*}\right),
\]
for all $m\geq\max\left\{ m_{0},m_{1}\right\} $, giving a contradiction,
as $\mu_{m}^{*}$ minimizes $R_{P}\left(\cdot\right)$ over $\Theta_{m}$.\hfill
\end{proof}

\begin{lemma}[{\citealp[Theorem 3.2,][]{chen2007large}}]
\label{lem: sieve rate}Let $\left\{ X_{i}\right\} _{i=1}^{n}$ be
an independent sequence of random variables with an identical distribution
$P$. Fix a sieve $\Theta_{1}\subseteq\Theta_{2}\subseteq\cdots\subseteq\Theta$.
For some loss function $l\left(x,\theta\right)$, let $\theta^{*}\in\Theta$
be the population risk minimizer
\[
\theta^{*}=\underset{\theta\in\Theta}{\arg\min}\text{ }\mathbb{E}_{P}\left[l\left(\theta,X_{i}\right)\right]
\]
and let $\|\cdot\|$ be some norm on $\Theta$. Define the function
class
\[
\mathcal{F}_{n}=\left\{ l\left(\theta,\cdot\right)-l\left(\theta^{*},\cdot\right):\|\theta-\theta^{*}\|\leq\delta,\theta\in\Theta_{n}\right\} .
\]
For some constant $b>0$, let 
\[
\delta_{n}=\underset{\delta\in\left(0,1\right)}{\inf}\frac{1}{\sqrt{n}\delta^{2}}\int_{b\delta^{2}}^{\delta}\sqrt{H_{[]}\left(w,\mathcal{F}_{n},\|\cdot\|\right)}\text{d}w\leq1,
\]
where $H_{[]}\left(w,\mathcal{F}_{n},\|\cdot\|\right)$ is the $L^{r}\left(P\right)$
metric entropy with bracketing of the class $\mathcal{F}_{n}$. Assume
that the following conditions hold:

\noindent \textbf{(i)} There exists constants $c_{1},c_{2}>0$ such
that 
\[
c_{1}\mathbb{E}\left[l\left(\theta,X_{i}\right)-l\left(\theta^{*},X_{i}\right)\right]\leq\|\theta-\theta^{*}\|^{2}\leq c_{2}\mathbb{E}\left[l\left(\theta,X_{i}\right)-l\left(\theta^{*},X_{i}\right)\right]
\]
 for any $\theta\in\Theta$.

\noindent \textbf{(ii)} There is a constant $C_{1}>0$ such that for
sufficiently small $\eta>0$, 
\[
\sup_{\theta\in\Theta_{n}:\|\theta-\theta^{*}\|\leq\eta}\Var\left(l\left(\theta,X_{i}\right)-l\left(\theta^{*},X_{i}\right)\right)\leq C_{1}\eta^{2}.
\]

\noindent \textbf{(iii)} For any $\eta>0$, there exists a constant
$s\in\left(0,2\right)$ such that 
\[
\sup_{\theta\in\Theta_{n}:\|\theta-\theta^{*}\|\leq\eta}\vert l\left(\theta,X_{i}\right)-l\left(\theta^{*},X_{i}\right)\vert\leq\eta^{s}U\left(X_{i}\right)
\]
with $\mathbb{E}\left[U\left(Z_{i}\right)^{\gamma}\right]\leq C_{2}$
for some $\gamma\geq2$. 

\noindent If $\hat{\theta}_{n}$ is a sequence of estimators satisfying
\begin{flalign}
\|\hat{\theta}_{n}-\theta^{*}\| & =o_{P}\left(1\right),\label{eq: consistence}\\
\frac{1}{n}\sum_{i=1}^{n}l\left(\hat{\theta}_{n},X_{i}\right) & \geq\sup_{\theta\in\Theta_{n}}\frac{1}{n}\sum l\left(\theta,X_{i}\right)+O_{P}\left(\varepsilon_{n}^{2}\right),\quad and\label{eq: optimization error}\\
\max\left\{ \delta_{n},\inf_{\theta\in\Theta_{n}}\|\theta^{*}-\theta\|\right\}  & =O\left(\varepsilon_{n}\right),\label{eq: bias variance}
\end{flalign}
 as $n\to\infty$ for some sequence $\varepsilon_{n}$, then $\|\hat{\theta}_{n}-\theta^{*}\|=O_{P}\left(\varepsilon_{n}\right)$
as $n\to\infty$. 
\end{lemma}

\begin{lemma}
\label{lem: loss rewrite}Suppose that \cref{as: rate} holds.

\noindent \textbf{(i)} There exists a function $K\left(Y,Z\right)$
and a constant $C$ such that such that 
\[
\vert L\left(h,a\right)-L\left(\mu,a\right)\vert\leq K\left(y,z\right)\vert h\left(z\right)-\mu\left(z\right)\vert,
\]
for each $a=\left(y,z,g\right)\in\mathbb{R}\times\mathcal{S}\times\left\{ 0,1\right\} $,
where 
\[
\sup_{z\in\mathcal{Z}}\text{ }\mathbb{E}_{P}\left[K\left(Y,Z\right)^{2}\mid Z=z\right]\leq C.
\]
\noindent \textbf{(ii)} There exists a function $\bar{K}\left(S,Z\right)$
and a constant $\bar{C}$ such that such that 
\[
\vert\bar{L}\left(\bar{h},a\right)-\bar{L}\left(\bar{\mu},a\right)\vert\leq\bar{K}\left(s,z\right)\vert\bar{h}\left(x\right)-\bar{\mu}\left(x\right)\vert,
\]
for each $a=\left(y,z,g\right)\in\mathbb{R}\times\mathcal{S}\times\left\{ 0,1\right\} $,
where 
\[
\sup_{x\in\mathcal{X}}\text{ }\mathbb{E}_{P}\left[\bar{K}\left(S,X\right)^{2}\mid X=x\right]\leq\bar{C}.
\]
\end{lemma}
\begin{proof}
We provide the details of the verification of the first claim, as
verification of the second claim follows by an identical argument.
For any $h\in\Lambda_{c}^{p}\left(\mathcal{Z}\right)$ and $z\in\mathcal{Z}$,
we have that 
\begin{flalign*}
\vert L\left(h,a\right)-L\left(\mu,a\right)\vert & =g\vert\left(y-h\left(z\right)\right)^{2}-\left(y-\mu\left(z\right)\right)^{2}\vert\\
 & =g\cdot\vert2\left(y-\lambda\left(z\right)h\left(z\right)+\left(1-\lambda\left(z\right)\right)\mu\left(z\right)\right)\vert\cdot\vert h\left(z\right)-\mu\left(z\right)\vert
\end{flalign*}
for some function $\lambda\left(z\right)\in\left[0,1\right]$, by
the Mean Value Theorem. Thus, we can set
\[
K\left(y,z\right)=g\cdot\vert2\left(y-\lambda\left(z\right)h\left(z\right)+\left(1-\lambda\left(z\right)\right)\mu\left(z\right)\right)\vert.
\]
We can verify that, for some constant $C$, we have that 
\begin{flalign*}
\mathbb{E}_{P}\left[K\left(Y,Z\right)^{2}\mid Z=z\right] & =4\mathbb{E}_{P}\left[G\left(Y-\lambda\left(Z\right)h\left(Z\right)+\left(1-\lambda\left(Z\right)\right)\mu\left(Z\right)\right)^{2}\mid Z=z\right]\\
 & =\mathbb{E}_{P}\left[G\left(\left(Y-\mu\left(Z\right)\right)-\lambda\left(Z\right)\left(h\left(Z\right)-\mu\left(Z\right)\right)\right)^{2}\mid Z=z\right]\leq C
\end{flalign*}
for all $z\in\mathcal{Z}$, by the bound (\ref{eq: conditional variance bound})
in Assumption \ref{as: rate} and the fact
that both $h\left(z\right)$ and $\mu\left(z\right)$ are in $\Lambda_{c}^{p}\left(\mathcal{Z}\right)$,
and are therefore bounded on $\mathcal{Z}$. \hfill
\end{proof}

\begin{lemma}[{\citealp[Lemma 2, ][]{chen1998sieve}}]
\label{lem: h inf h 2 bound}If $\mathcal{X}$ is a compact subset
of $\mathbb{R}^{d}$ and $h\in\Lambda_{c}^{p}\left(\mathcal{X}\right)$,
then
\[
\|h\|_{\infty}\leq2\|h\|_{2}^{\frac{2p}{2p+d}}c^{1-\frac{2p}{2p+d}},
\]
where $\|\cdot\|_{\infty}$ denotes the $L_{\infty}$ norm and $\|\cdot\|_{2}$
denotes the $L_{2}$ norm under the Lebesgue measure.
\end{lemma}
\begin{lemma}
\label{lem: linear sieve}Continue the notation of Lemma \ref{lem: sieve rate}.
Assume that the conditions (i)-(iii) of Lemma \ref{lem: sieve rate}
are satisfied. Suppose that the function class $\Theta=\Lambda_{c}^{p}\left(\mathcal{X}\right)$,
where $\mathcal{X}$ is some compact subset of $\mathbb{R}^{d}$,
that the sieve spaces $\Theta_{n}$ are linear and satisfy
\begin{equation}
\mathsf{dim}\left(\Theta_{n}\right)\asymp\left(\frac{n}{\log n}\right)^{\frac{d}{2p+d}}\quad\text{and}\quad\inf_{h\in\Theta_{n}}\|h-\mu\|\asymp\left(\frac{\log n}{n}\right)^{\frac{p}{2p+d}}.\label{eq: rates}
\end{equation}
If $\hat{\theta}_{n}$ is a sequence of estimators satisfying the
consistency condition (\ref{eq: consistence}) and optimization condition
(\ref{eq: optimization error}) for the sequence $\varepsilon_{n}$,
then 
\begin{equation}
\|\hat{\theta}_{n}-\theta^{*}\|_{P,2}=O_{P}\left(\max\left\{ \varepsilon_{n},\left(\frac{\log n}{n}\right)^{\frac{p}{2p+d}}\right\} \right)\label{eq: max eps, minimax rate}
\end{equation}
as $n\to\infty$.
\end{lemma}
\begin{proof}
As the conditions of Lemma \ref{lem: sieve rate} are satisfied, we
have that 
\begin{flalign}
\|\hat{\theta}_{n}-\theta^{*}\|_{P,2} & =O_{P}\left(\max\left\{ \varepsilon_{n},\delta_{n},\inf_{h\in\Theta_{n}}\|h-\mu\|\right\} \right)\nonumber \\
 & =O_{P}\left(\max\left\{ \varepsilon_{n},\delta_{n},\left(\frac{\log n}{n}\right)^{\frac{p}{2p+d}}\|\right\} \right),\label{eq: max of three}
\end{flalign}
where 
\[
\delta_{n}=\underset{\delta\in\left(0,1\right)}{\inf}\frac{1}{\sqrt{n}\delta^{2}}\int_{b\delta^{2}}^{\delta}\sqrt{H_{[]}\left(w,\mathcal{F}_{n},\|\cdot\|_{2}\right)}\text{d}w\leq1
\]
and $H_{[]}\left(w,\mathcal{F}_{n},\|\cdot\|_{2}\right)$ is the $L^{2}\left(P\right)$
metric entropy with bracketing of the function class
\[
\mathcal{F}_{n}=\left\{ L\left(h,\cdot\right)-L\left(\mu,\cdot\right):\|h-\mu\|_{P,2}\leq\delta,h\in\Theta_{n}\right\} .
\]
Now, Condition (iii) of \cref{lem: sieve rate} implies that 
\[
H_{[]}\left(w,\mathcal{F}_{n},\|\cdot\|_{2}\right)\leq\log N\left(w^{1+\frac{d}{2p}},\Theta_{n},\|\cdot\|_{P,2}\right),
\]
where $N\left(w,\mathcal{F}_{n},\|\cdot\|_{2}\right)$ is the covering
number of $\mathcal{F}_{n}$. Moreover, we have that 
\[
\log N\left(w^{1+\frac{d}{2p}},\mathcal{F}_{n},\|\cdot\|_{2}\right)\lesssim\mathsf{dim}\left(\Theta_{n}\right)\log\left(\frac{1}{w}\right)
\]
as $\Theta_{n}$ is an element of a finite dimensional linear sieve
(see e.g., Section 3.2 of \citealp{chen2007large}). Thus, we have
that 
\begin{flalign*}
\frac{1}{\sqrt{n}\delta^{2}}\int_{b\delta^{2}}^{\delta}\sqrt{H_{[]}\left(w,\mathcal{F}_{n},\|\cdot\|_{2}\right)}\text{d}w & \leq\frac{1}{\sqrt{n}\delta^{2}}\int_{b\delta^{2}}^{\delta}\sqrt{\log N\left(w^{1+\frac{d}{2p}},\Theta_{n},\|\cdot\|_{P,2}\right)}\text{d}w\\
 & \lesssim\frac{1}{\delta}\sqrt{\frac{\mathsf{dim}\left(\Theta_{n}\right)}{n}\log\frac{1}{\delta}}
\end{flalign*}
 and that therefore
\begin{equation}
\delta_{n}=O\left(\sqrt{\frac{\mathsf{dim}\left(\Theta_{n}\right)\log n}{n}}\right)=O\left(\left(\frac{\log n}{n}\right)^{\frac{p}{2p+d}}\right)\label{eq: delta n mu}
\end{equation}
by Assumption \ref{as: rate} and the choice (\ref{eq: rates}). The proof is then complete by plugging the rate (\ref{eq: delta n mu}) into the bound (\ref{eq: max of three}).\hfill
\end{proof}
 
\section{Simulation\label{sec:appendix simulation}} 

\subsection{Data\label{sec:data}}
The data from \cite{banerjee2015multifaceted} were acquired from \url{https://dataverse.harvard.edu/dataset.xhtml?persistentId=doi:10.7910/DVN/NHIXNT} on September 10, 2021. We restrict attention to data from the program evaluation in Pakistan. Prior to cleaning, there are $1299$ households in this sample. We omit data from $399$ households who attrited from the sample prior to either of the post-treatment outcome measurements. We also omit data from 46 households who were missing measurements of all consumption variables for either of the post-treatment outcome measurements. 
Consequently, we obtain a cleaned sample of 854 households. Of these households, 446 were randomly assigned to participate in the program. There was perfect compliance with treatment assignment. 

Some of the individuals are missing measurements for some of the variables. There are at most $14$ individuals missing a measurement of a given variable. For most continuously valued variables, we impute missing values with the mean value of that outcome. Five of the continuously valued variables are heavily censored, meaning that more than 50\% of the observations are equal to zero.\footnote{These variables are total informal loans outstanding, total formal loans outstanding, savings deposit amount, income from agriculture, income from business, and revenue from animals} We impute missing values with zero for these variables. For categorical or binary valued variables, we impute missing values with the modal outcome of that variable. 

We include variables belonging in five categories: consumption, food security, assets, finance, and income and revenue. The consumption variables are the total monthly consumption and the total monthly consumptions on food, non-food, and durable commodities. The food-security variables are an index for overall food-security and five binary variables indicating different aspects of food security (e.g., did a child in the household skip a meal). The assets variables are value of assets, the value of productive assets, and value of the household's assets, each aggregated in two ways. The financial variables are the total amount of formal and informal loans outstanding and the total value of the households savings. The income and revenue variables are the income from agriculture, the income from business, the income from paid labor, the revenue from animals, and the self-assessed perception of economic status. We do not include the total income from business as a baseline covariate, as it is equal to zero for all individuals in the sample. We refer the reader to the appendix of \cite{banerjee2015multifaceted} for further information on the construction of these variables.  

\subsection{Simulation Calibration\label{sec:calibration}}

We calibrate our simulation to the \cite{banerjee2015multifaceted} data with GANs. A GAN is composed of a pair of competing neural networks \citep{goodfellow2014generative}. The objective of one neural network is to generate data similar to the training data. The objective of the other neural network is to distinguish between the true and generated data. These networks iteratively compete until the true and generated data are difficult to distinguish.

There are three GANs that contribute to generating a data point $A_i$. The first GAN is fit to the marginal distribution of the pre-treatment covariates (i.e., $X_i$). The second GAN is fit to the distribution of the short-term outcomes conditioned on pre-treatment covariates and treatment (i.e., $S_i \vert X_i, W_i$).  The third GAN takes two forms, depending on whether we are generating data designed to satisfy the assumptions in the Latent Unconfounded Treatment or Statistical Surrogacy Models. For the Latent Unconfounded Treatment Model, the third GAN is fit to the distribution of the long-term outcome conditioned on the short-term outcomes, the pre-treatment covariates, and treatment (i.e., $Y_i \vert S_i, X_i, W_i$). For the Statistical Surrogacy Model, the analogous distribution does not condition on treatment. In the Statistical Surrogacy Model treatment is assumed to be independent of $Y_i$, conditional on $S_i$ and $X_i$, by \cref{as: es}. By contrast, in the Latent Unconfounded Treatment Model, treatment is only independent of the potential outcomes $Y_i(w)$, conditional on $S_i$ and $X_i$, by \cref{as: luot}.

We estimate the GANs needed to calibrate our simulation with the ``wgan'' python package associated with \cite{athey2021using}. This package is available at \url{https://github.com/gsbDBI/ds-wgan}. Table \ref{tab:tuning} displays the choices of tuning parameters used for this calibration. All other tuning parameter choices are set to the wgan defaults. 

\begin{table}
\small{
\begin{centering}
\begin{tabular}{ccccccc}
\toprule 
Tuning Parameter &  & Distribution: & $X_{i}$ & $S_{i}\mid X_{i},W_{i}$ & $Y_{i}\mid S_{i},X_{i}$ & $Y_{i}\mid S_{i},X_{i},W_{i}$\tabularnewline
\midrule
\midrule 
Batch Size &  &  & 256 & 256 & 256 & 256\tabularnewline
\midrule 
Epochs &  &  & 30000 & 30000 & 5000 & 5000\tabularnewline
\midrule 
\multirow{2}{*}{Learning Rate} & Critic &  & $10^{-4}$ & $10^{-4}$ & $10^{-4}$ & $10^{-4}$\tabularnewline
\cmidrule{2-7} \cmidrule{3-7} \cmidrule{4-7} \cmidrule{5-7} \cmidrule{6-7} \cmidrule{7-7} 
 & Generator &  & $10^{-4}$ & $10^{-4}$ & $10^{-4}$ & $10^{-4}$\tabularnewline
\midrule 
\multirow{2}{*}{Critic Dropout} & Critic &  & 0 & 0.1 & 0.1 & 0.1\tabularnewline
\cmidrule{2-7} \cmidrule{3-7} \cmidrule{4-7} \cmidrule{5-7} \cmidrule{6-7} \cmidrule{7-7} 
 & Generator &  & 0.1 & 0.1 & 0.1 & 0.1\tabularnewline
\midrule
Critic Gradient Penalty &  &  & 20 & 20 & 20 & 20\tabularnewline
\bottomrule
\end{tabular}
\par\end{centering}
\medskip{}
\medskip{}
\caption{Tuning Parameter Choices in GAN Estimation}
\label{tab:tuning}
\justifying
{\footnotesize{}Notes: Table \ref{tab:tuning} displays tuning parameter choices for the GAN estimation using the wgan package associated with \cite{athey2021using}. Any tuning parameter choices not listed are set to the wgan defaults.}{\footnotesize\par}}
\end{table}

\cref{fig:results assets} compares the moments of the data from \cite{banerjee2015multifaceted} to the GAN generated simulation data. The first two rows compare the means and variances of these data sets, and are displayed in log-scale. The second two rows compare correlations within and across variable types. For example, the third row compares the correlations of the pre-treatment covariates in the true and generated data, and the fourth row compares the correlations of the pre-treatment covariates and short-term outcomes. Note that in this comparison the long-term outcomes $Y_i$ includes both total assets, chosen as the long-term outcome in the main text, and total consumption, chosen as the long-term outcome in \cref{subsec:additional}. 

\begin{figure}
\caption{Validation}
\label{fig:validation}
\medskip{}
\begin{centering}
\begin{tabular}{c}
\includegraphics[scale=0.25]{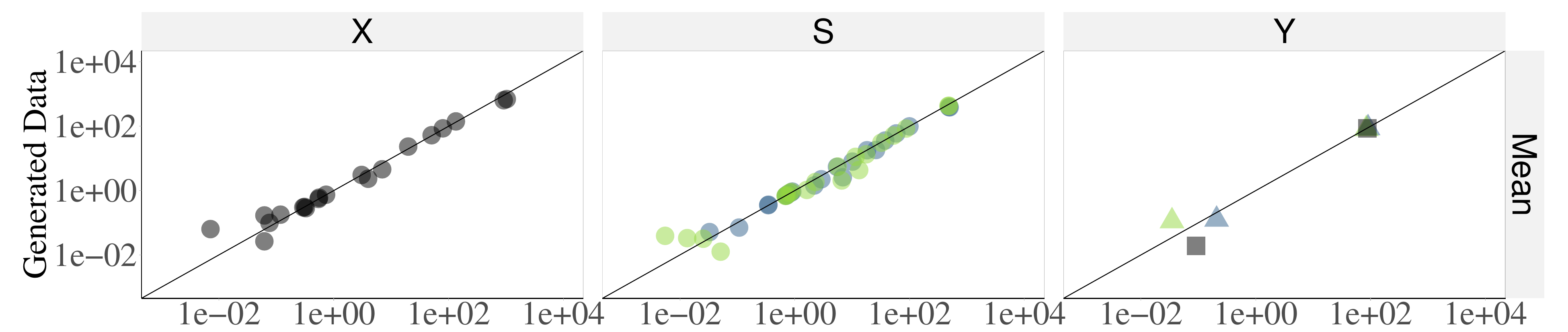}\tabularnewline
\includegraphics[scale=0.25]{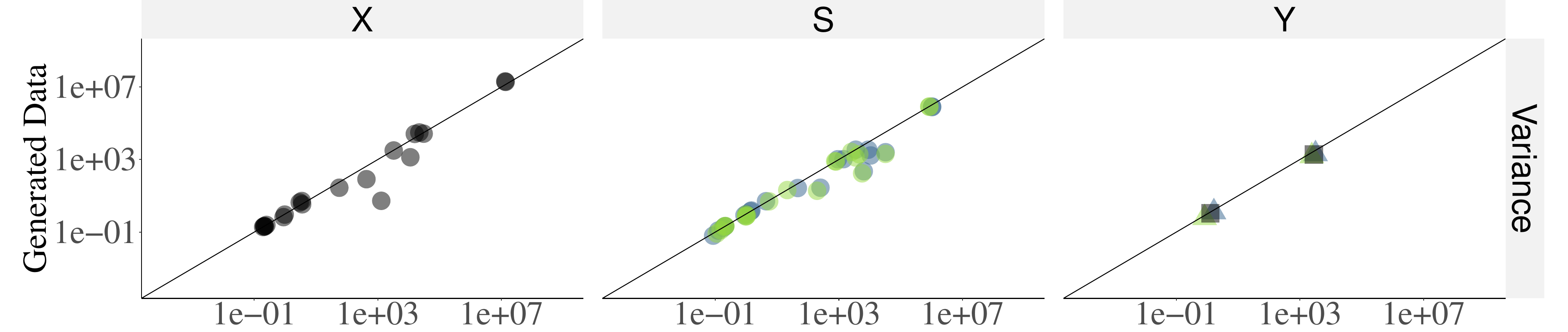}\tabularnewline
\includegraphics[scale=0.25]{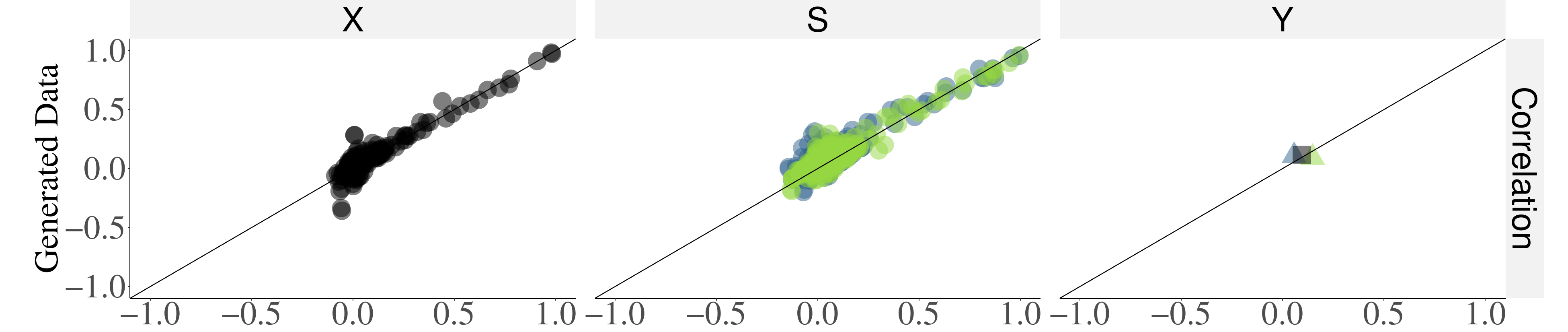}\tabularnewline
\includegraphics[scale=0.25]{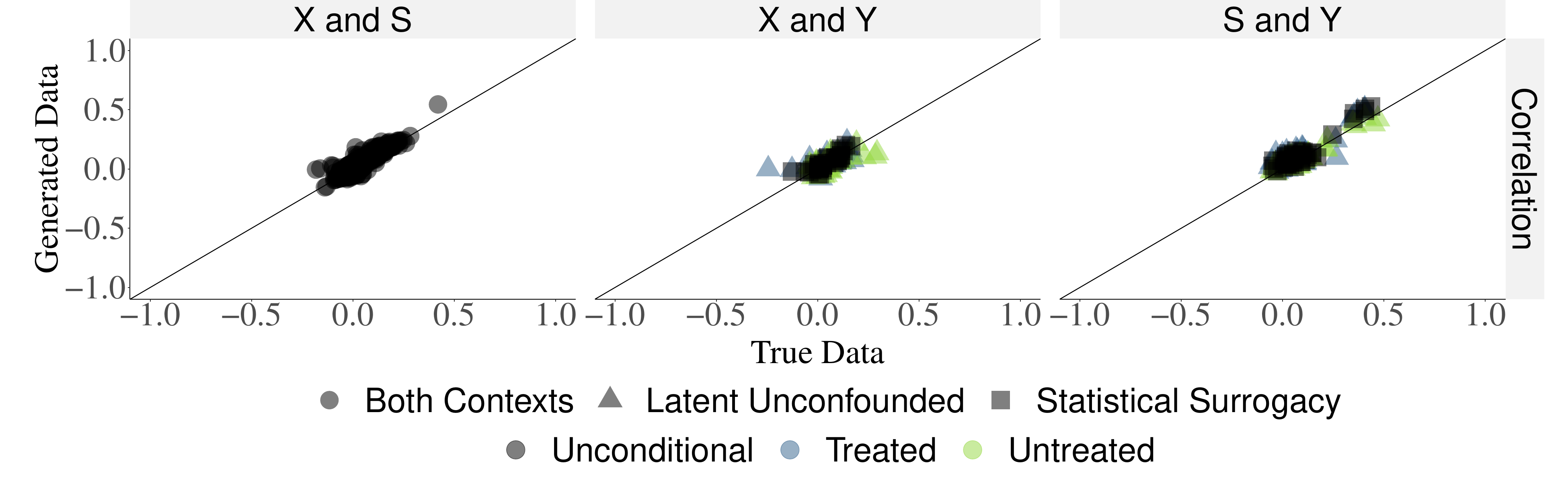}\tabularnewline
\end{tabular}
\par\end{centering}
\medskip{}
\justifying
{\footnotesize{}Notes: \cref{fig:validation} displays scatterplots comparing the moments of the data from \cite{banerjee2015multifaceted} to the GAN generated simulation data. Columns differentiate between different types of variables. In the first three rows, the columns correspond to the pre-treatment covariates, the short-term outcomes, and the long-term outcomes, respectively. The columns of the fourth row correspond to the three pairs of these three variable types. Rows differentiate between different types of moments. The first two rows compare the means and variances of the the true and generated data. The second two rows compare correlations of the true and generated data, within and across variable types, respectively. The x-axis of each sub-panel measures the moments of the true data. The y-axis of each sub-panel measures the moments of the generated data. The x and y axes in the first two rows are displayed in log-scale. A forty-five degree line is displayed in all sub-panels. Triangular and square dots correspond to data generated the Latent Unconfounded Treatment and Statistical Surrogacy Models, respectively. Circular dots correspond to data shared by both models. Blue and green dots denote moments conditioned on treatment being set to one and zero, respectively. Black dots denote unconditioned moments.}{\footnotesize\par}
\end{figure}

The joint distributions of the true and generated data match remarkably closely. This match holds both for data generated for the contexts in which treatment is and is not observed in the observational sample, and continues to hold if moments are conditioned on the treatment assignment being set to one and zero, respectively.

\cref{fig:histogram} displays histograms comparing the true and generated distributions for the long-term outcomes, conditioned on treatment assignment. These distributions again align remarkably closely. In each case, the distributions have similarly shaped supports and right tails. 

\begin{figure}
\caption{Comparison of True and Generated Long-Term Outcome Distributions}
\label{fig:histogram}
\medskip{}
\begin{centering}
\label{fig:histogram}
\begin{tabular}{c}
\textit{Panel A: Total Consumption}\tabularnewline
\includegraphics[scale=0.25]{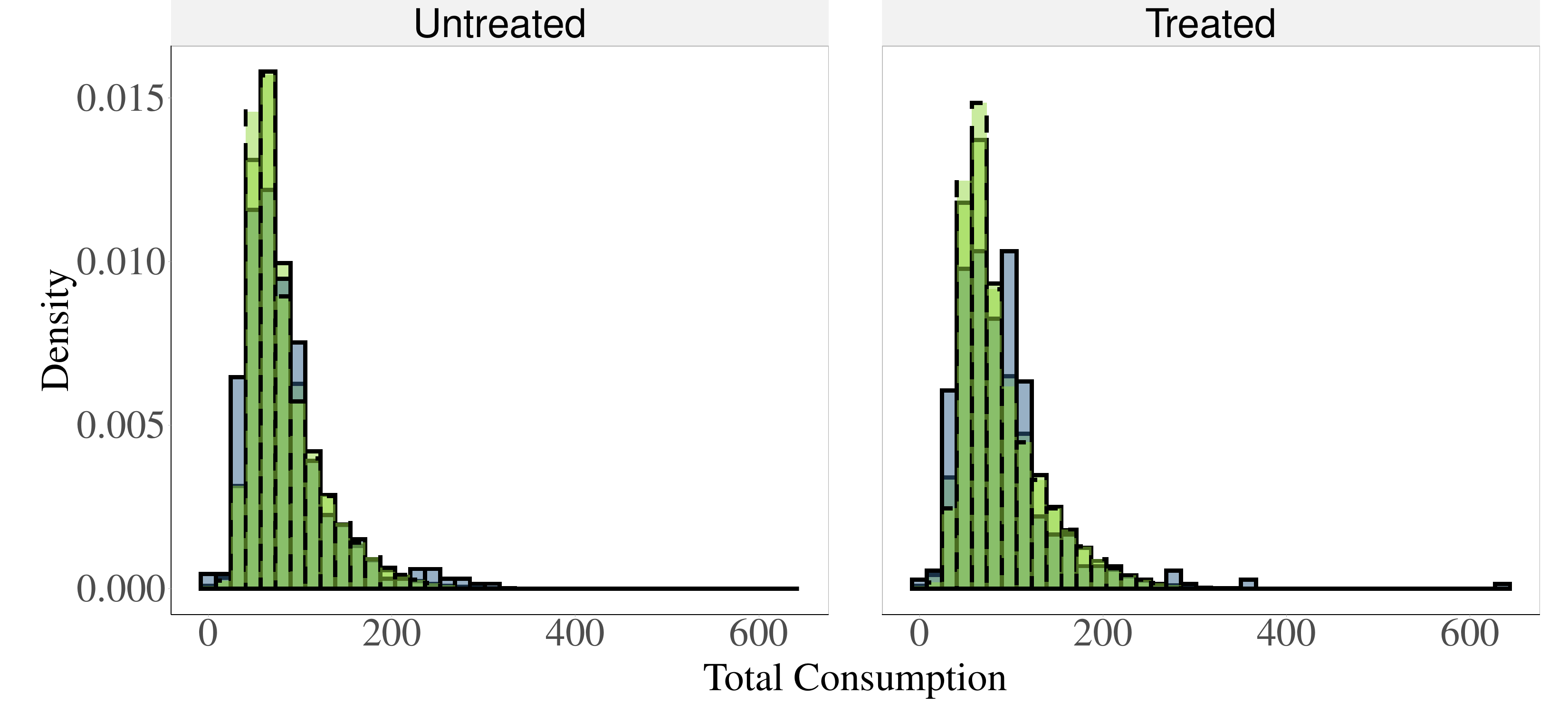}\tabularnewline
\textit{Panel B: Total Assets}\tabularnewline
\includegraphics[scale=0.25]{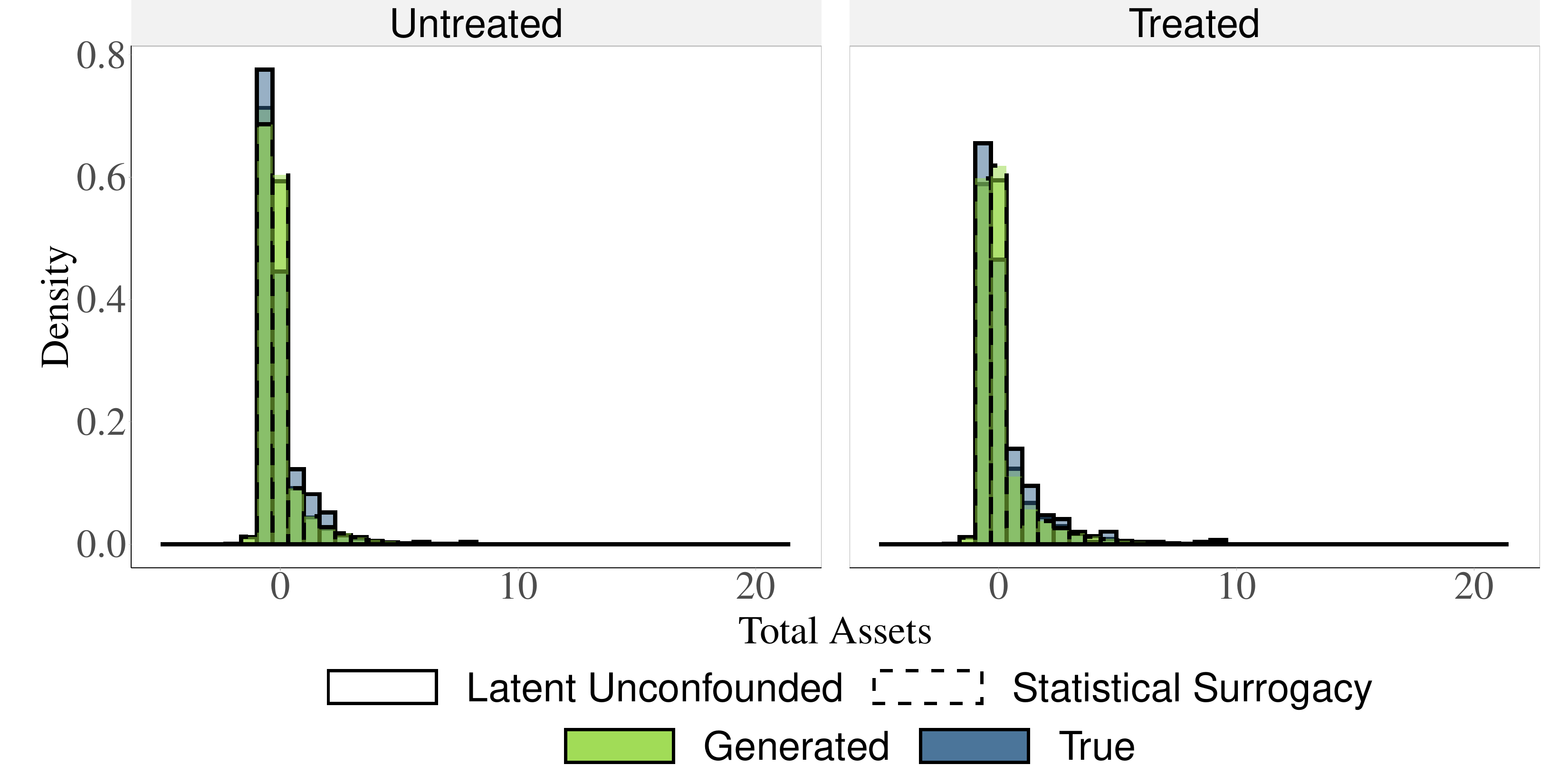}\tabularnewline
\end{tabular}
\par\end{centering}
\medskip{}
\justifying
{\footnotesize{}Notes: \cref{fig:histogram} displays histograms comparing the true and generated long term outcome distributions, conditioned by treatment assignment. Panel A compares the distributions for total consumption. Panel B compares distributions for total assets. Histograms for the generated and true data are displayed in green and blue, respectively. Solid and dotted lines border the histograms for the data generated for Latent Unconfounded Treatment
and Statistical Surrogacy Models, respectively.}{\footnotesize\par}
\end{figure}

\subsection{Design\label{sec:design}}

To generate an observation $A_i$, we first generate a value of the pre-treatment covariates $X_i$ from the first GAN. This generated value of $X_i$ is then input into the second GAN to generate values of the treated and untreated short-term potential outcomes, i.e., generating $S_i(1)$ and $S_i(0)$ by taking draws from the GAN approximating the distributions of $S_i \vert X_i, W_i=1$ and $S_i \vert X_i, W_i=0$. We repeat this step, using the generated values $(X_i,S(1))$ and $(X_i,S_i(0))$ as inputs into the third GAN to generate the long-term treated and untreated potential outcomes $Y_i(1)$ and $Y_i(0)$, respectively. Each observation is assigned to being either ``experimental'' or ``observational'' uniformly at random, i.e., we draw each $G_i$ from a Bernoulli distribution with parameter $1/2$. Treatment is always assigned uniformly at random for experimental observations. 

We consider two procedures for assigning treatment for the observational units. First, as a baseline, we assign treatment uniformly at random, leaving treatment assignment unconfounded. Second, we induce confounding. Specifically, we assign treatment with probability 
\[
\frac{1}{1+\exp(-(\tilde{\tau}_{c,i}+\tilde{\tau}_{a,i})\cdot\phi)}
\]
where $\tilde{\tau}_{c,i}$ and $\tilde{\tau}_{a,i}$ denote the true treatment effects on short-term total consumption and total assets for unit $i$, normalized across observational units, and $\phi$ is a parameter indexing the degree of confounding. Written differently, the probability of being assigned treatment increases logistically with the true treatment effects on short-term outcomes. Large values of $\phi$ induce more confounding.

By repeating this process, we generate ten million simulation draws. We take the empirical distribution of this sample as the population distribution. In this sample, we observe both treated and untreated potential outcomes, and, thereby, the true treatment effects for each individual. It is worth noting the significant computational scale of the experiment reported in this section. The simulation requires parallelization on a large-scale computing cluster and uses roughly 60 years of cpu time.

Panel A of \cref{fig:phi} plots the probability that an observational unit is assigned to treatment as a function of $\tilde{\tau}_{c,i}$ and $\tilde{\tau}_{a,i}$, the true treatment effects on short-term total consumption and total assets for unit $i$, normalized across observational units. This curve is displayed for $\phi$ equal to $0.50$ and $0.66$, which are the parameter values that we consider in out simulation analysis. Panels B and C display histograms of the treatment probabilities for all of the simulation draws for $\phi$ equal to $0.50$ and $0.66$. Observe that increasing $\phi$ increases the confounding in the data, in the sense that there are more simulation draws with propensity scores very different from one half.

\begin{figure}
\caption{Visualization of Confounding Parameterization}
\label{fig:phi}
\medskip{}
\begin{centering}
\begin{tabular}{cc}
\multicolumn{2}{c}{\textit{Panel A: Parameterization of Confounding}}\tabularnewline
\multicolumn{2}{c}{\includegraphics[scale=0.3]{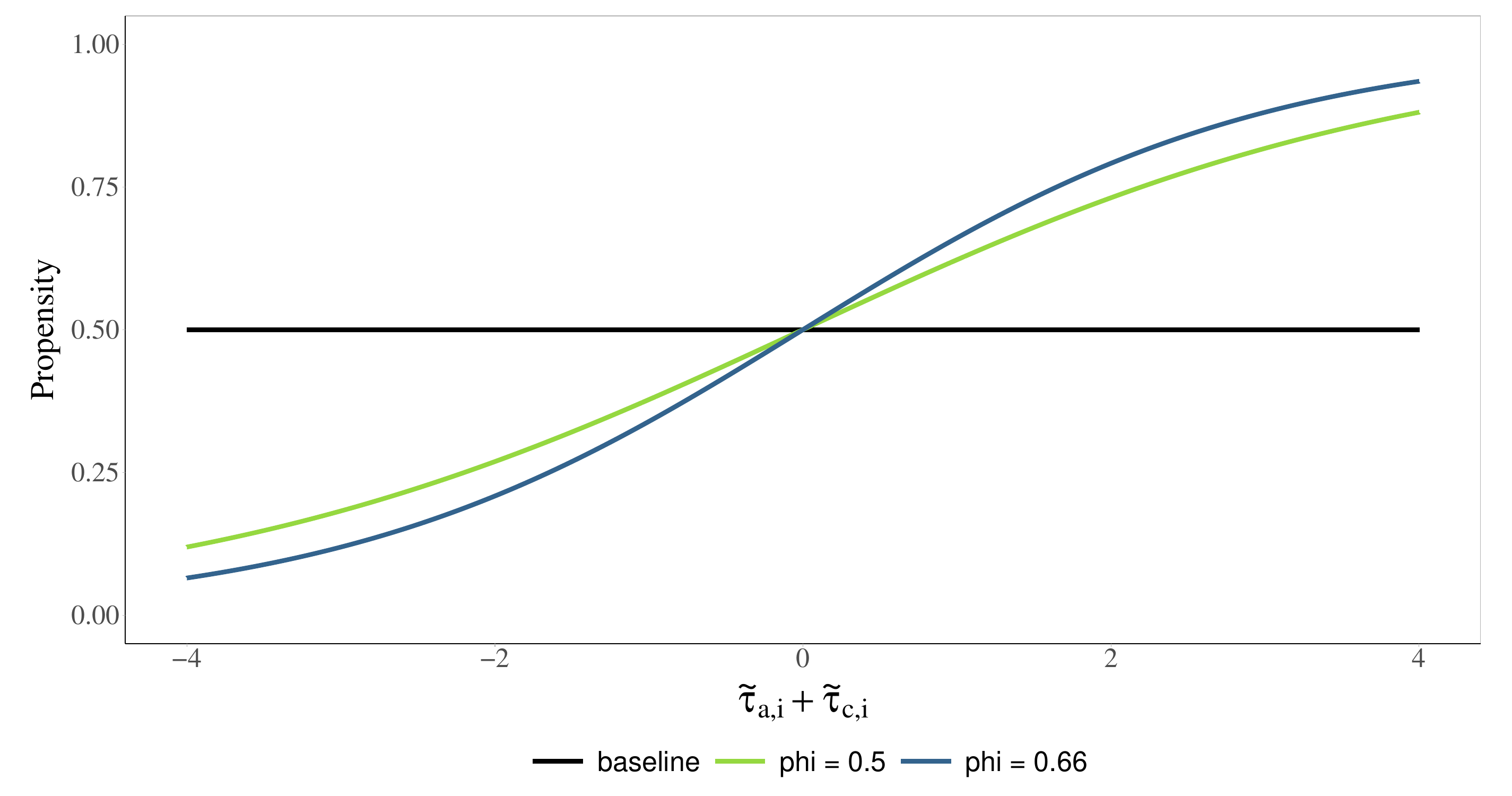}}\\
\textit{Panel B: Histogram of Propensities, $\phi=0.5$} & \textit{Panel C: Histogram of Propensities, $\phi=0.66$}\\
\includegraphics[scale=0.32]{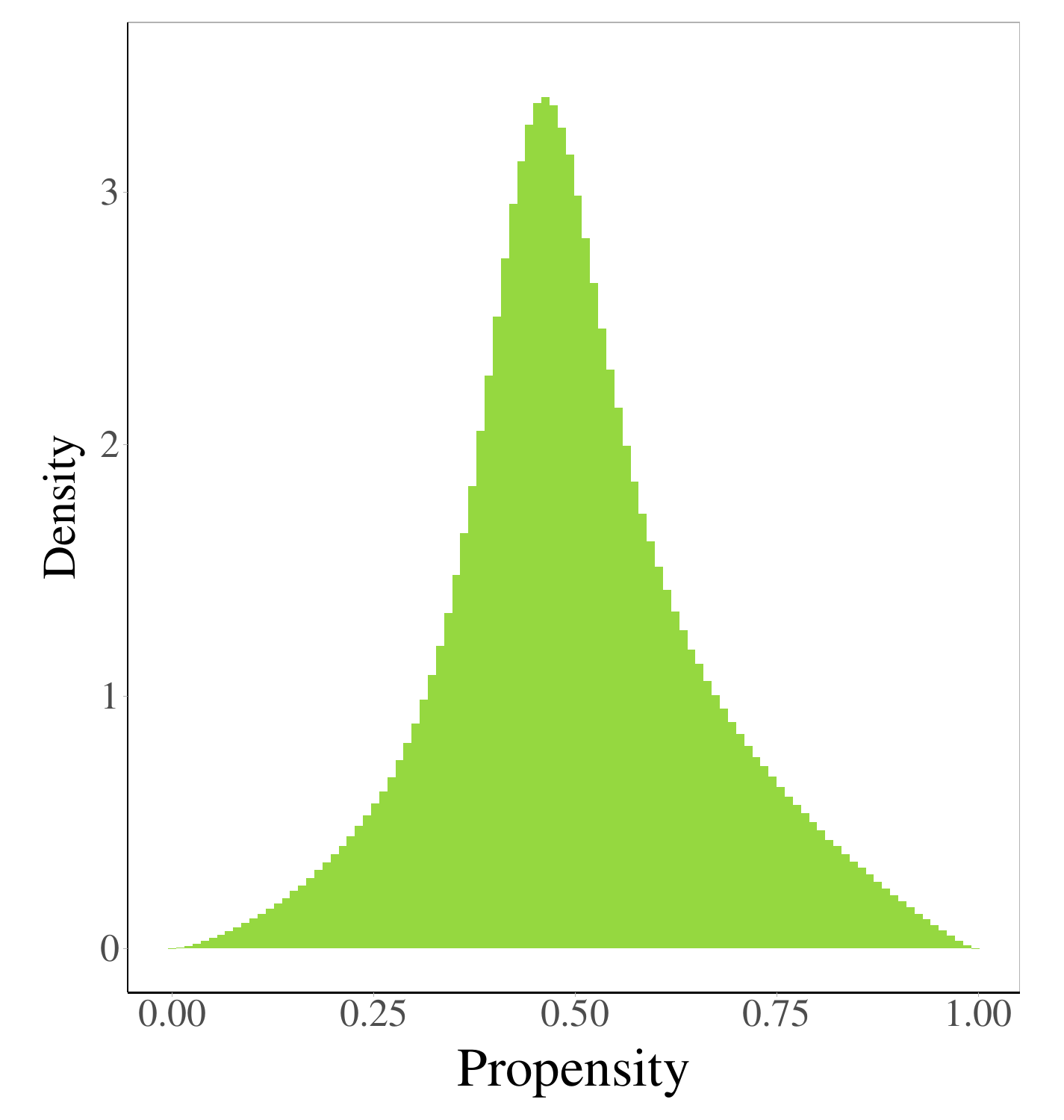} & \includegraphics[scale=0.32]{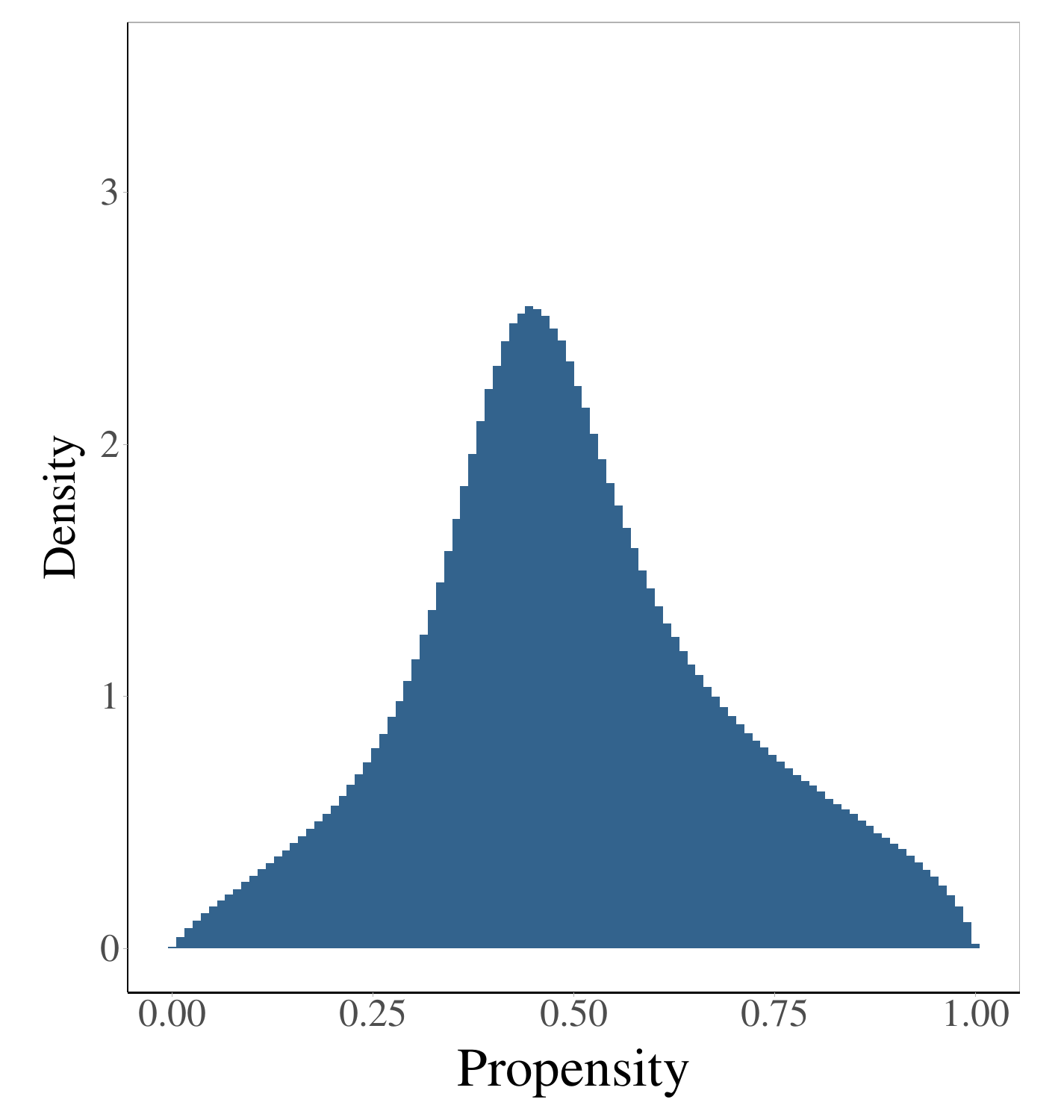}\\
\end{tabular}
\par\end{centering}
\medskip{}
\justifying
{\footnotesize{}Notes: Panel A of \cref{fig:phi} displays curves giving the
probability of a simulation draw being assigned to treatment as a function
of $\tilde{\tau}_{c,i}$ and $\tilde{\tau}_{a,i}$, the true treatment effects on short-term total consumption and total assets for unit $i$, normalized across observational units. Panels B and C display histograms of the treatment probabilities for all simulation draws for $\phi$ equal to $0.50$ and $0.66$, respectively.}{\footnotesize\par}
\medskip{}

\end{figure}

\subsection{Estimation Details\label{subsec:estimation}}

\paragraph{\textbf{Cross-fitting:}} 
Following the choice made in \cite{chernozhukov2018double}, we set $k$ equal to five for the $k$-fold cross-fitting construction. We expect that alternative choices for this hyper-parameter will not impact the results of the simulation substantially, and so do not report results for different choices of $k$.
\medskip{}

\paragraph{\textbf{Propensity Score Estimation:}} 
For the sake of numerical stability across simulations, we threshold all estimates of propensity scores to lie in $[0.05,0.95]$, which amounts to enforcing a value of $\varepsilon$ chosen for the strict overlap \cref{as: overlap} in the estimation of propensity scores.\footnote{When estimating the propensity score $\rho_w(s, x)$, we trim estimates of each of the components in \eqref{eq: q identified} in addition to the final estimate $\hat{\rho}_w(s, x)$ constructed by combining these estimates according to \eqref{eq: q identified}.} The choice of this hyper-parameter may substantively impact results in cases, e.g., in which the covariate distributions of the experimental and observational distributions have limited overlap. We do not report results for alternative choices of these hyperparameters, but view the adaptation of the methods developed in \cref{sec:estimation} to setting in which there is limited overlap along one or more direction as a useful direction for further research. 
\medskip{}

\paragraph{\textbf{Linear Regression:}} We use ordinary linear regression and a generalized linear model with a binomial response type when estimating outcome mean and propensity score type nuisance functions, respectively. 
\medskip{}

\paragraph{\textbf{Lasso:}} We implement cross-validated Lasso \citep{tibshirani1996regression} with the R package ``glmnet' \citep{hastie2014glmnet}. In each cross-fitting fold, we use $5$-fold cross-validation to choose the regularization parameter that gives the minimum mean cross-validation error. We use the binomial response type when estimating propensity scores. 
\medskip{}

\paragraph{\textbf{GRF:}} We implement generalized random forests \citep{athey2019generalized} with the R package ``grf''. When $h\in\{1,3\}$, we set the ``honesty'' hyper-parameter equal to false. In each cross-fitting fold, we choose the hyper-parameters ``sample.fraction'', ``mtry'', and ``min.node.size'' by setting the grf hyper-parameter ``tune.parameters'' equal to ``all''. This option chooses $1000$ hyper-parameter values randomly, then chooses the hyper-parameter values with the minimum 5-fold cross-validation error. When $h=10$, we set ``sample.fraction'' equal to 0.1 in order to optimize memory usage, but tuned the other parameters in the same manner. We use regression forests to predict both long-term outcome means and propensity scores. 
\medskip{}

\paragraph{\textbf{XGBoost:}} We implement XGboost \citep{chen2016xgboost} with the R package ``xgboost''  \citep{chen2015xgboost}. We set the learning rate $\eta$ and max depth equal to $0.1$ and $2$, respectively. We choose the number of rounds with 5-fold cross validation, with the maximum number of rounds set to $100$. We use the root mean squared error and log-loss as the evaluation metrics for the estimation of long-term outcome means and propensity scores, respectively. 

\subsection{Additional Results\label{subsec:additional}}

\cref{fig:baseline consumption,fig:results consumption} displays results analogous to \cref{fig:baseline assets,fig:results assets} in the main text, when total household consumption is chosen as the long-term outcome of interest. The results are very similar to the those reported in the main text.

\begin{figure}
\begin{centering}
\caption{Comparison of Estimators}
\label{fig:baseline consumption}
\medskip{}
\begin{tabular}{c}
\textit{Panel A: Latent Unconfounded Treatment}\tabularnewline
\includegraphics[scale=0.32]{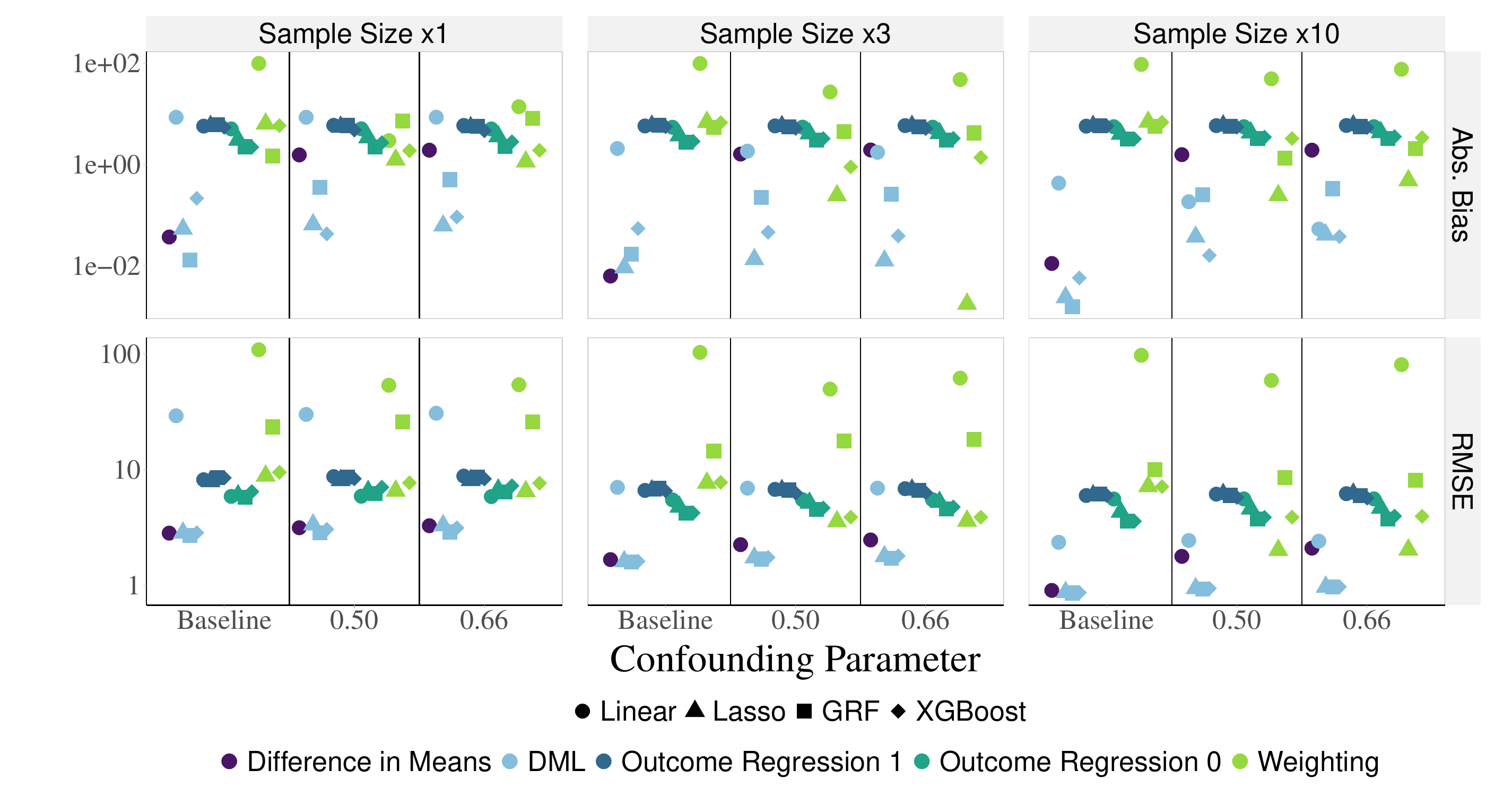}\tabularnewline
\textit{Panel B: Statistical Surrogacy}\tabularnewline
\includegraphics[scale=0.32]{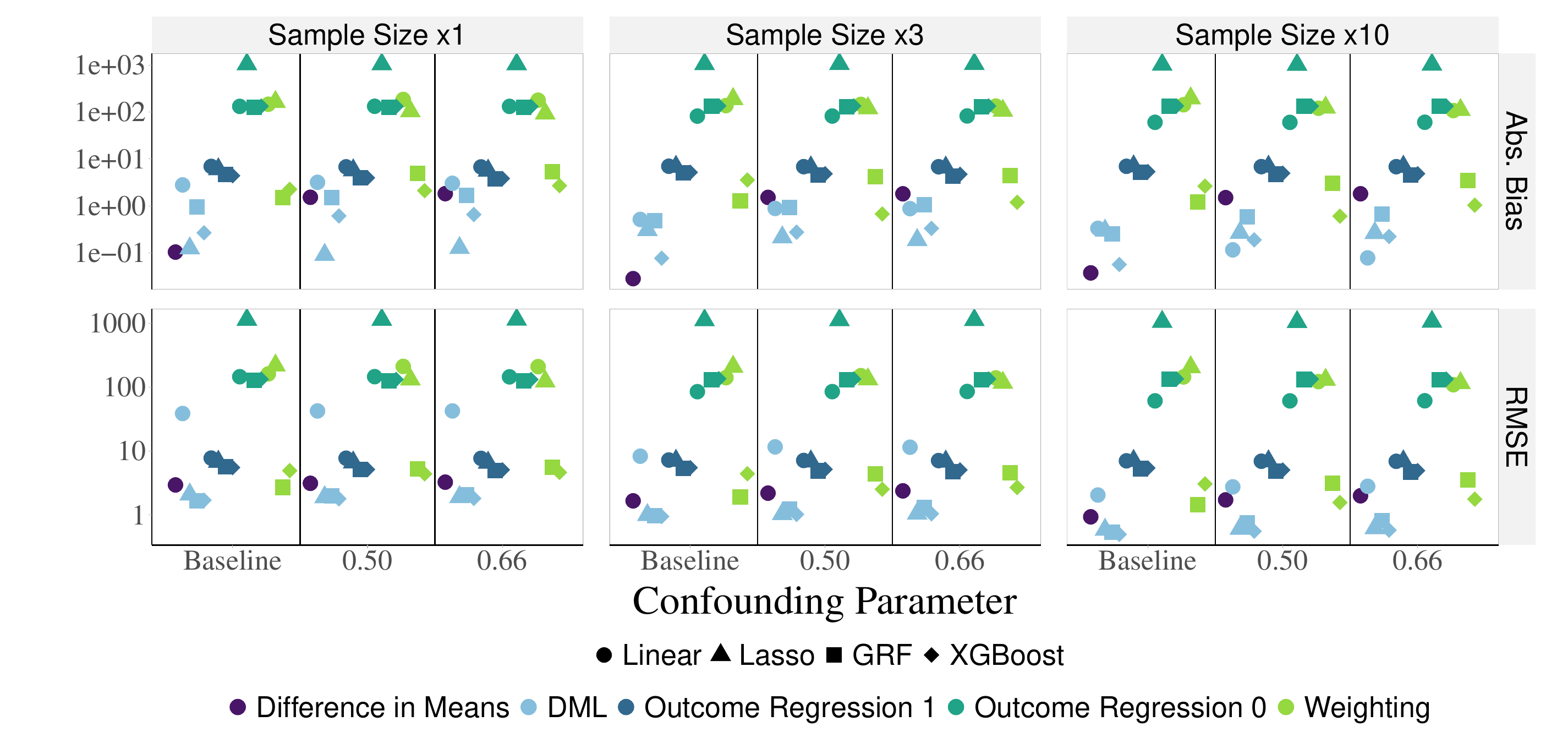}\tabularnewline
\end{tabular}
\par
\end{centering}
\medskip{}
\justifying
{\footnotesize{}Notes: \cref{fig:baseline consumption} compares measurements of the absolute
bias and root mean squared error for the estimators formulated in \cref{sec:estimation}, in addition to the 
difference in means estimator defined in \eqref{eq:dm}. The y-axes are displayed in logs, base 10. The long-term outcome is total household
consumption. Panels A and B display results for the Latent Unconfounded Treatment and
Statistical Surrogacy Models defined in \cref{def: lut} and \cref{def: sur},
respectively. The columns of each panel vary the sample size multiplier $h$. Each
sub-panel displays results for the baseline, unconfounded, case, as well as for the
cases that the confounding parameter $\phi$ has been set to $1/2$ and $2/3$. Results for each 
estimator are displayed with dots of different colors. Results for different nuisance parameter estimators
are displayed with dots of different shapes.}{\footnotesize\par}
\end{figure}

\begin{figure}
\begin{centering}
\caption{Finite-Sample Performance with Different Nuisance Parameter Estimators}
\label{fig:results consumption}
\medskip{}
\begin{tabular}{c}
\textit{Panel A: Latent Unconfounded Treatment}\tabularnewline
\includegraphics[scale=0.32]{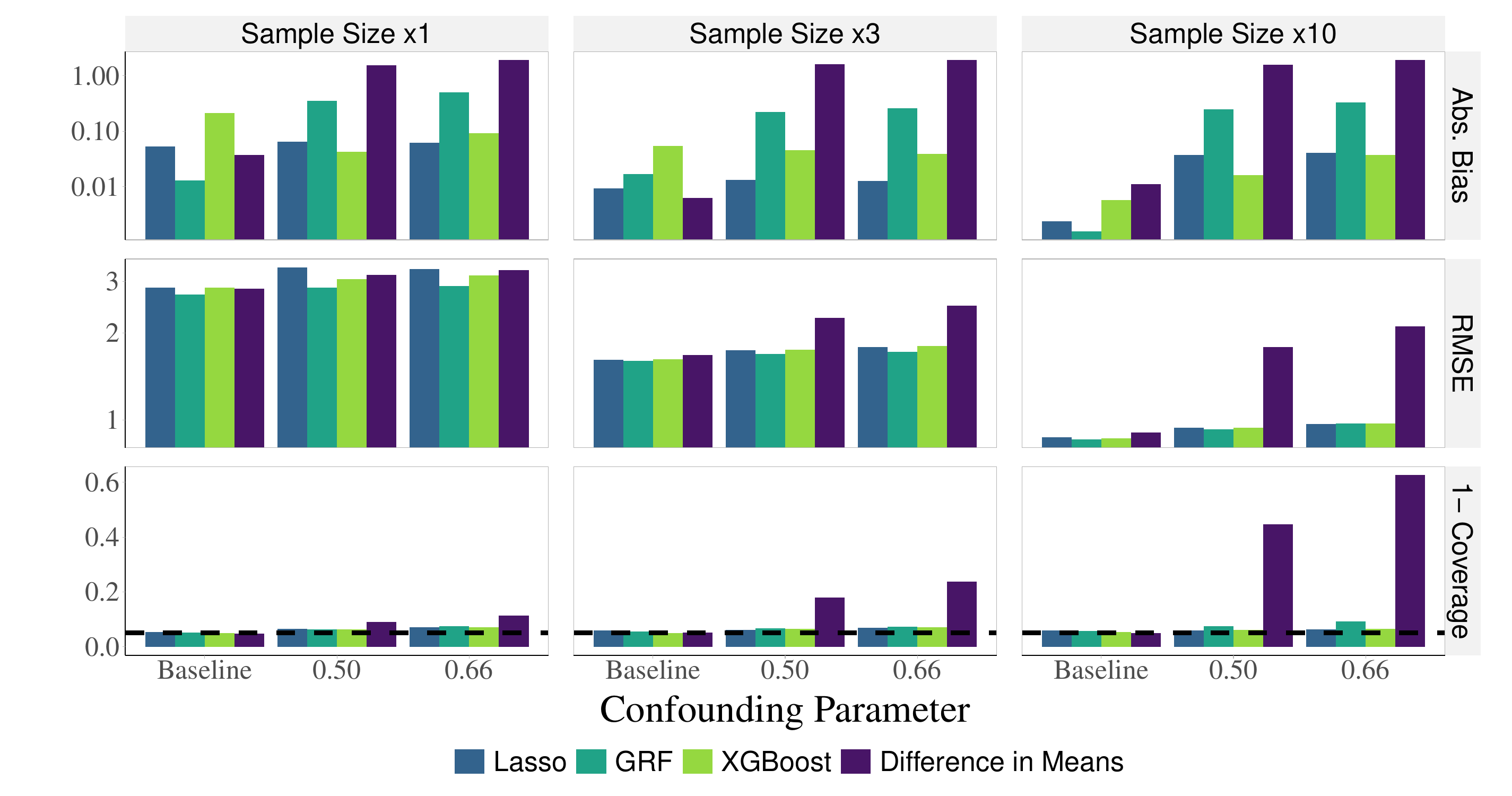}\tabularnewline
\textit{Panel B: Statistical Surrogacy}\tabularnewline
\includegraphics[scale=0.32]{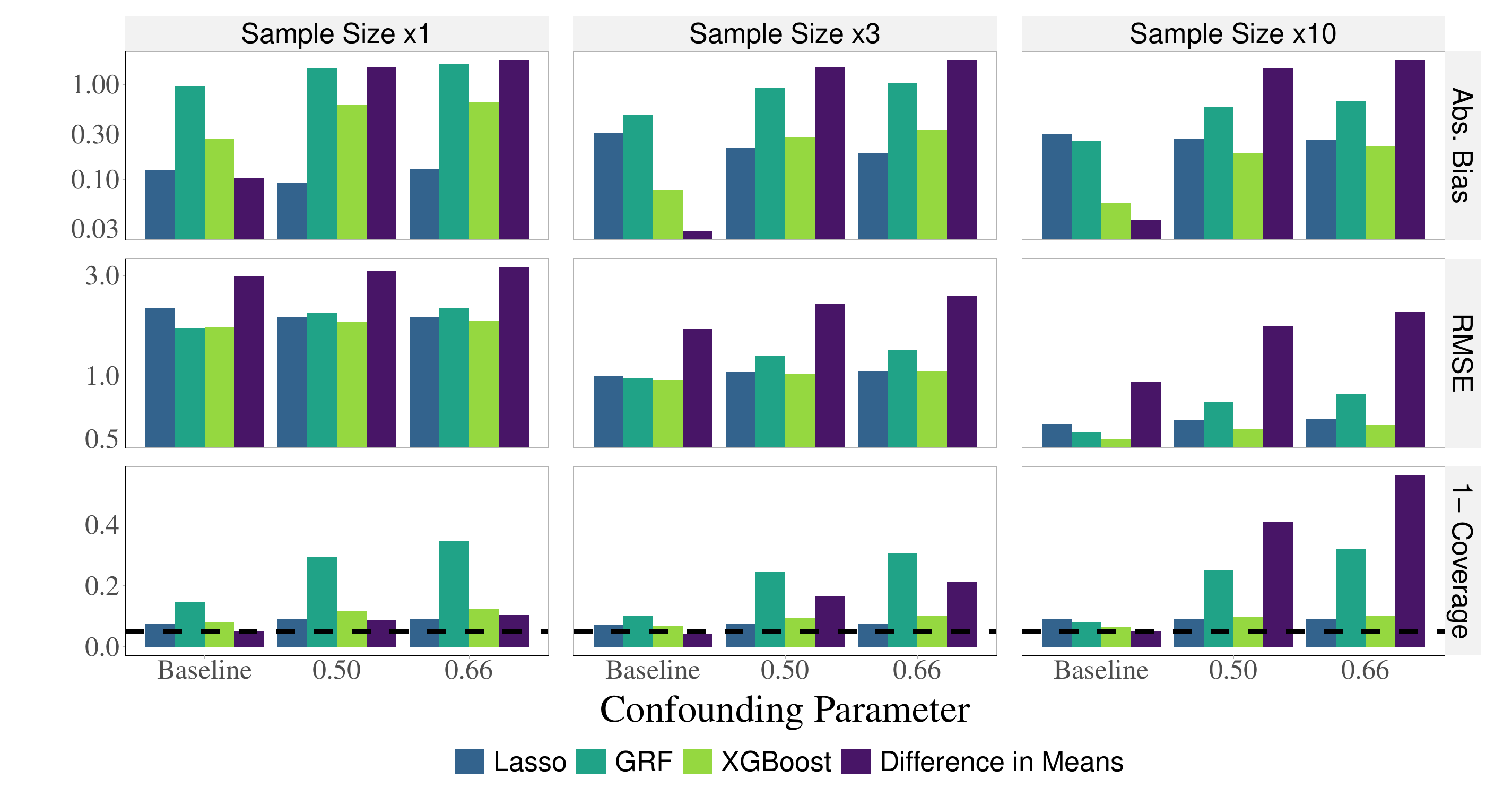}\tabularnewline
\end{tabular}
\par
\end{centering}
\medskip{}
\justifying
{\footnotesize{}Notes: \cref{fig:results consumption}  displays measurements of the
 quality of the estimators formulated in \cref{sec:orthogonal} implemented with several
 alternative choices of nuisance parameter estimators. The long-term outcome is total
 household consumption. Panels A and B display results for the estimators defined in \cref
 {def: dml} for the Latent Unconfounded Treatment and Statistical Surrogacy Models
 defined in \cref{def: lut} and \cref{def: sur}, respectively. The columns of each
 panel vary the sample size multiplier $h$. The rows of each panel display the absolute
 value of the bias, one minus the coverage probability, and the root mean squared
 error of each estimator, from top to bottom, respectively. A dotted line denoting one
 minus the nominal coverage probability, 0.05, is displayed in each sub-panel in the
third row. Each sub-panel displays a bar graph comparing measurements of the
 performance of the estimator defined in \cref{def: dml}, constructed with three types
 of nuisance parameter estimators, with the difference in means estimator \eqref
 {eq:dm}.  Each sub-panel displays results for the baseline, unconfounded, case, as
 well as for the cases that the confounding parameter $\phi$ has been set to $1/2$ and
 $2/3$.}{\footnotesize\par}
\end{figure}

\begin{figure}
\begin{centering}
\caption{Comparison of Estimators}
\label{fig:r3 assets}
\medskip{}
\begin{tabular}{c}
\textit{Panel A: Latent Unconfounded Treatment}\tabularnewline
\includegraphics[scale=0.32]{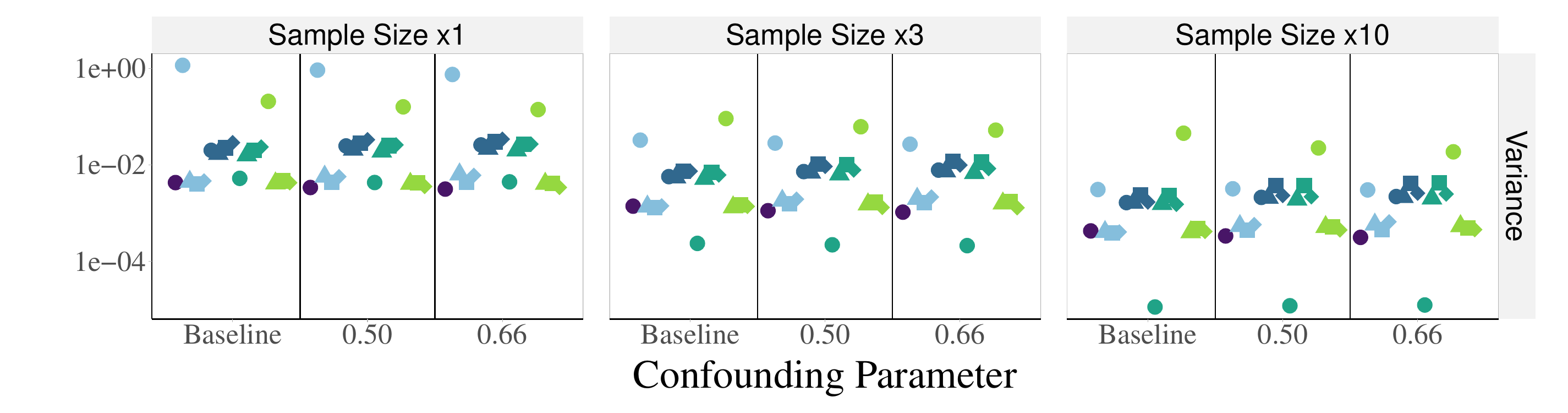}\tabularnewline
\textit{Panel B: Statistical Surrogacy}\tabularnewline
\includegraphics[scale=0.32]{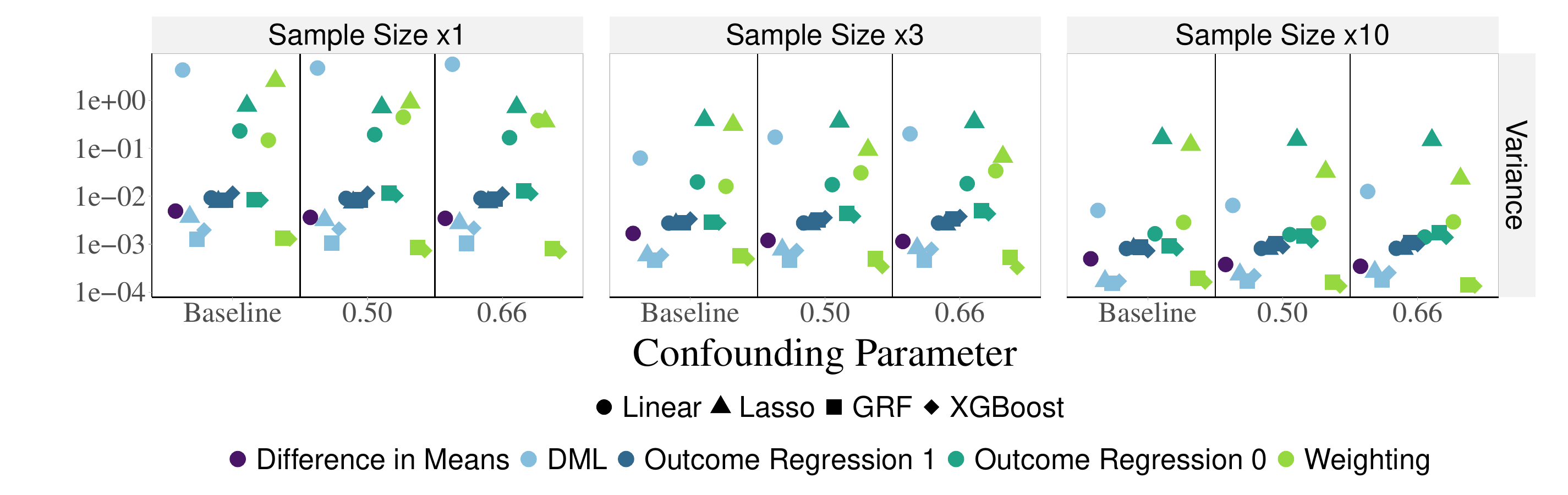}\tabularnewline
\end{tabular}
\par
\end{centering}
\medskip{}
\justifying
{\footnotesize{}Notes: \cref{fig:r3 assets} compares measurements of variance for the estimators formulated in \cref{sec:estimation}, in addition to the 
difference in means estimator defined in \eqref{eq:dm}. The y-axes are displayed in logs, base 10. The long-term outcome is total household
assets. All other formatting is analogous to \cref{fig:baseline assets} in the main text.}{\footnotesize\par}
\end{figure}

\begin{figure}
\begin{centering}
\caption{Comparison of Estimators}
\label{fig:r3 consumption}
\medskip{}
\begin{tabular}{c}
\textit{Panel A: Latent Unconfounded Treatment}\tabularnewline
\includegraphics[scale=0.32]{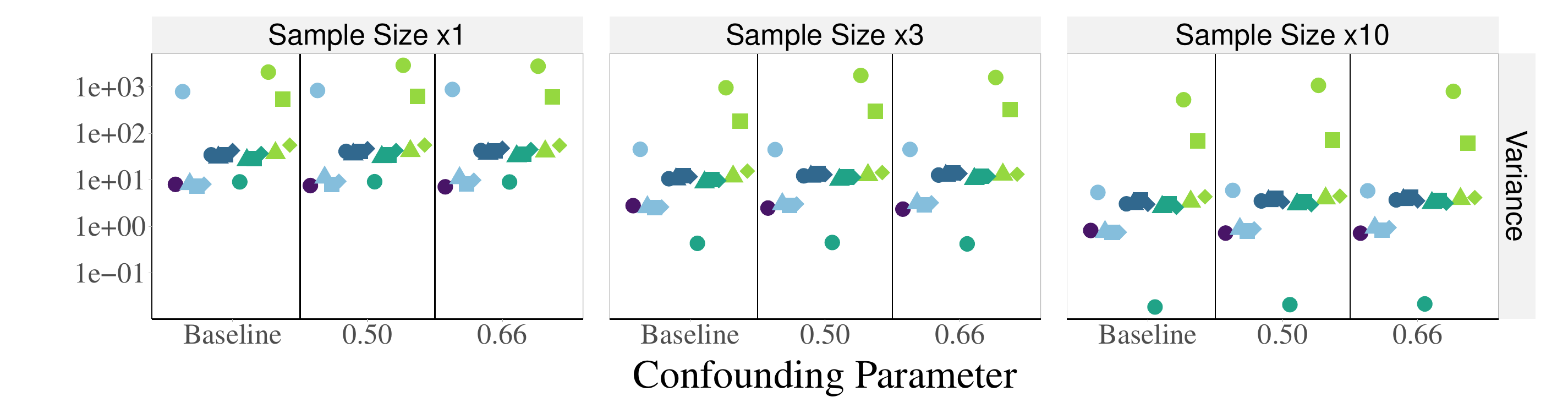}\tabularnewline
\textit{Panel B: Statistical Surrogacy}\tabularnewline
\includegraphics[scale=0.32]{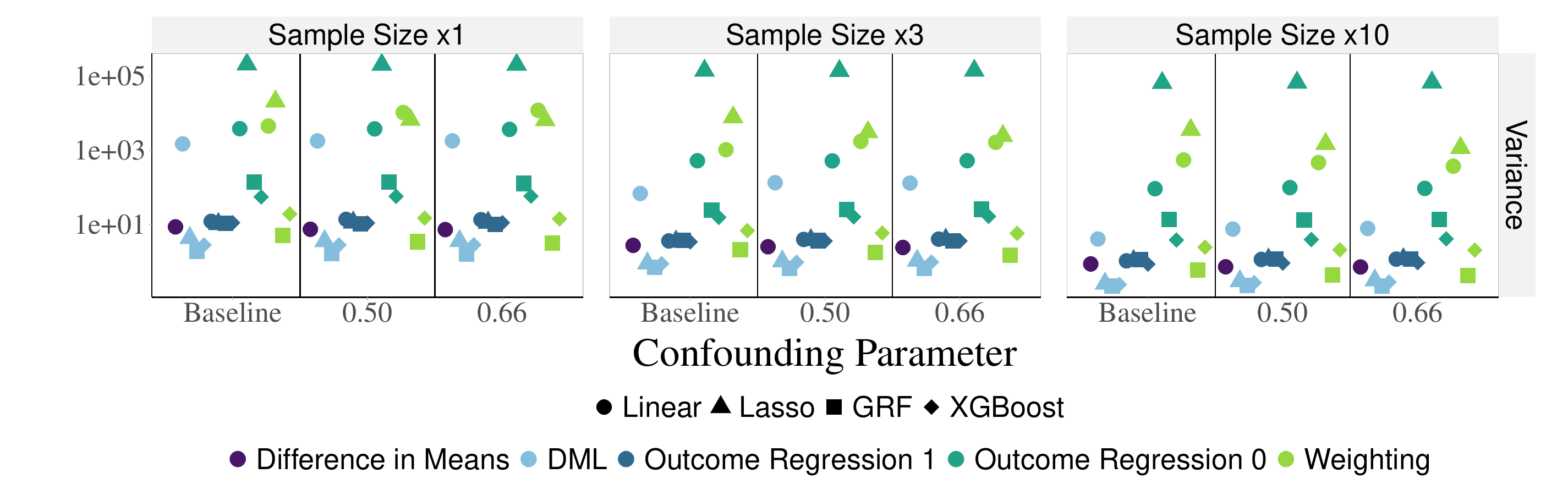}\tabularnewline
\end{tabular}
\par
\end{centering}
\medskip{}
\justifying
{\footnotesize{}Notes: \cref{fig:r3 consumption} compares measurements of variance for the estimators formulated in \cref{sec:estimation}, in addition to the 
difference in means estimator defined in \eqref{eq:dm}. The y-axes are displayed in logs, base 10. The long-term outcome is total household
consumption. All other formatting is analogous to \cref{fig:baseline consumption}}{\footnotesize\par}
\end{figure}

Tables \ref{tab:coverage assets} and \ref{tab:coverage consumption} present estimates of coverage probabilities of confidence intervals constructed around the estimators formulated in \cref{def: dml}. Wald intervals constructed around the difference in means estimator severely undercover when there is confounding in the observational data set. The confidence intervals \eqref{eq: ci} tend to undercover, although the under-coverage is not as severe. Coverage is worse under the Statistical Surrogacy Model and is very poor when generalized random forests are used to estimate nuisance parameters. This is a result of the bias exhibited in \cref{fig:results assets} and \cref{fig:results consumption}. By contrast, in the same situation, the coverage probabilities are much closer to the nominal level when nuisance parameters are estimated with the Lasso or XGBoost. As an exception, these intervals over cover if nuisance parameters are estimated with linear regression.

\newpage

\begin{table}
\small{\begin{centering}
\begin{tabular}{ccccccccc}
\toprule 
\multirow{3}{*}{} &  & \multicolumn{3}{c}{Treatment Observed} &  & \multicolumn{3}{c}{Treatment Not Observed}\tabularnewline
\cmidrule{3-5} \cmidrule{4-5} \cmidrule{5-5} \cmidrule{7-9} \cmidrule{8-9} \cmidrule{9-9} 
 & Confounding: & Baseline & 0.5 & 0.66 &  & Baseline & 0.5 & 0.66\tabularnewline
\cmidrule{2-5} \cmidrule{3-5} \cmidrule{4-5} \cmidrule{5-5} \cmidrule{7-9} \cmidrule{8-9} \cmidrule{9-9} 
 & Sample Size: & (1) & (2) & (3) &  & (4) & (5) & (6)\tabularnewline
\midrule
\midrule 
\multirow{3}{*}{Difference in Means} & $\times1$ & 0.953 & 0.803 & 0.712 &  & 0.953 & 0.811 & 0.735\tabularnewline
\cmidrule{2-9} \cmidrule{3-9} \cmidrule{4-9} \cmidrule{5-9} \cmidrule{6-9} \cmidrule{7-9} \cmidrule{8-9} \cmidrule{9-9} 
 & $\times3$ & 0.956 & 0.545 & 0.353 &  & 0.955 & 0.564 & 0.396\tabularnewline
\cmidrule{2-9} \cmidrule{3-9} \cmidrule{4-9} \cmidrule{5-9} \cmidrule{6-9} \cmidrule{7-9} \cmidrule{8-9} \cmidrule{9-9} 
 & $\times10$ & 0.953 & 0.095 & 0.011 &  & 0.954 & 0.101 & 0.021\tabularnewline
\midrule 
\multirow{3}{*}{Linear} & $\times1$ & 0.992 & 0.995 & 0.994 &  & 0.976 & 0.979 & 0.977\tabularnewline
\cmidrule{2-9} \cmidrule{3-9} \cmidrule{4-9} \cmidrule{5-9} \cmidrule{6-9} \cmidrule{7-9} \cmidrule{8-9} \cmidrule{9-9} 
 & $\times3$ & 0.999 & 0.999 & 1.000 &  & 0.974 & 0.981 & 0.983\tabularnewline
\cmidrule{2-9} \cmidrule{3-9} \cmidrule{4-9} \cmidrule{5-9} \cmidrule{6-9} \cmidrule{7-9} \cmidrule{8-9} \cmidrule{9-9} 
 & $\times10$ & 1.000 & 1.000 & 1.000 &  & 0.946 & 0.983 & 0.985\tabularnewline
\midrule 
\multirow{3}{*}{Lasso} & $\times1$ & 0.946 & 0.908 & 0.901 &  & 0.904 & 0.862 & 0.843\tabularnewline
\cmidrule{2-9} \cmidrule{3-9} \cmidrule{4-9} \cmidrule{5-9} \cmidrule{6-9} \cmidrule{7-9} \cmidrule{8-9} \cmidrule{9-9} 
 & $\times3$ & 0.945 & 0.922 & 0.911 &  & 0.823 & 0.877 & 0.884\tabularnewline
\cmidrule{2-9} \cmidrule{3-9} \cmidrule{4-9} \cmidrule{5-9} \cmidrule{6-9} \cmidrule{7-9} \cmidrule{8-9} \cmidrule{9-9} 
 & $\times10$ & 0.945 & 0.930 & 0.929 &  & 0.522 & 0.742 & 0.784\tabularnewline
\midrule 
\multirow{3}{*}{GRF} & $\times1$ & 0.947 & 0.901 & 0.868 &  & 0.791 & 0.438 & 0.334\tabularnewline
\cmidrule{2-9} \cmidrule{3-9} \cmidrule{4-9} \cmidrule{5-9} \cmidrule{6-9} \cmidrule{7-9} \cmidrule{8-9} \cmidrule{9-9} 
 & $\times3$ & 0.947 & 0.900 & 0.877 &  & 0.882 & 0.520 & 0.405\tabularnewline
\cmidrule{2-9} \cmidrule{3-9} \cmidrule{4-9} \cmidrule{5-9} \cmidrule{6-9} \cmidrule{7-9} \cmidrule{8-9} \cmidrule{9-9} 
 & $\times10$ & 0.946 & 0.842 & 0.743 &  & 0.905 & 0.583 & 0.467\tabularnewline
\midrule
\multirow{3}{*}{XGBoost} & $\times1$ & 0.947 & 0.908 & 0.891 &  & 0.868 & 0.784 & 0.746\tabularnewline
\cmidrule{2-9} \cmidrule{3-9} \cmidrule{4-9} \cmidrule{5-9} \cmidrule{6-9} \cmidrule{7-9} \cmidrule{8-9} \cmidrule{9-9} 
 & $\times3$ & 0.952 & 0.909 & 0.900 &  & 0.916 & 0.834 & 0.807\tabularnewline
\cmidrule{2-9} \cmidrule{3-9} \cmidrule{4-9} \cmidrule{5-9} \cmidrule{6-9} \cmidrule{7-9} \cmidrule{8-9} \cmidrule{9-9} 
 & $\times10$ & 0.949 & 0.925 & 0.916 &  & 0.917 & 0.849 & 0.839\tabularnewline
\bottomrule
\vspace{0.1cm}
\end{tabular}
\par\end{centering}}
\medskip{}
\caption{Coverage Probability Estimates: Total Assets}
\label{tab:coverage assets}
\justifying
{\footnotesize{}Notes: Table \ref{tab:coverage assets} compares estimates of the coverage probabilities of confidence intervals constructed around each estimator. The long-term outcome of interest is taken to be total assets. The nominal coverage probability is set to $0.95$ in each case. Standard Wald intervals are constructed around the difference in means estimator \eqref{eq:dm}. Confidence intervals are constructed with \eqref{eq: ci} for the estimators formulated in \cref{def: dml}.}{\footnotesize\par}
\end{table}

\begin{table}
\small{
\begin{centering}
\begin{tabular}{ccccccccc}
\toprule 
\multirow{3}{*}{} &  & \multicolumn{3}{c}{Treatment Observed} &  & \multicolumn{3}{c}{Treatment Not Observed}\tabularnewline
\cmidrule{3-5} \cmidrule{4-5} \cmidrule{5-5} \cmidrule{7-9} \cmidrule{8-9} \cmidrule{9-9} 
 & Confounding: & Baseline & 0.5 & 0.66 &  & Baseline & 0.5 & 0.66\tabularnewline
\cmidrule{2-5} \cmidrule{3-5} \cmidrule{4-5} \cmidrule{5-5} \cmidrule{7-9} \cmidrule{8-9} \cmidrule{9-9} 
 & Sample Size: & (1) & (2) & (3) &  & (4) & (5) & (6)\tabularnewline
\midrule
\midrule 
\multirow{3}{*}{Difference in Means} & $\times1$ & 0.953 & 0.910 & 0.887 &  & 0.948 & 0.914 & 0.895\tabularnewline
\cmidrule{2-9} \cmidrule{3-9} \cmidrule{4-9} \cmidrule{5-9} \cmidrule{6-9} \cmidrule{7-9} \cmidrule{8-9} \cmidrule{9-9} 
 & $\times3$ & 0.948 & 0.820 & 0.763 &  & 0.957 & 0.834 & 0.788\tabularnewline
\cmidrule{2-9} \cmidrule{3-9} \cmidrule{4-9} \cmidrule{5-9} \cmidrule{6-9} \cmidrule{7-9} \cmidrule{8-9} \cmidrule{9-9} 
 & $\times10$ & 0.950 & 0.554 & 0.374 &  & 0.949 & 0.590 & 0.436\tabularnewline
\midrule 
\multirow{3}{*}{Linear} & $\times1$ & 0.994 & 0.997 & 0.996 &  & 0.985 & 0.983 & 0.985\tabularnewline
\cmidrule{2-9} \cmidrule{3-9} \cmidrule{4-9} \cmidrule{5-9} \cmidrule{6-9} \cmidrule{7-9} \cmidrule{8-9} \cmidrule{9-9} 
 & $\times3$ & 0.999 & 1.000 & 0.999 &  & 0.985 & 0.984 & 0.985\tabularnewline
\cmidrule{2-9} \cmidrule{3-9} \cmidrule{4-9} \cmidrule{5-9} \cmidrule{6-9} \cmidrule{7-9} \cmidrule{8-9} \cmidrule{9-9} 
 & $\times10$ & 1.000 & 1.000 & 1.000 &  & 0.987 & 0.986 & 0.983\tabularnewline
\midrule 
\multirow{3}{*}{Lasso} & $\times1$ & 0.945 & 0.935 & 0.929 &  & 0.926 & 0.909 & 0.910\tabularnewline
\cmidrule{2-9} \cmidrule{3-9} \cmidrule{4-9} \cmidrule{5-9} \cmidrule{6-9} \cmidrule{7-9} \cmidrule{8-9} \cmidrule{9-9} 
 & $\times3$ & 0.941 & 0.938 & 0.931 &  & 0.930 & 0.924 & 0.926\tabularnewline
\cmidrule{2-9} \cmidrule{3-9} \cmidrule{4-9} \cmidrule{5-9} \cmidrule{6-9} \cmidrule{7-9} \cmidrule{8-9} \cmidrule{9-9} 
 & $\times10$ & 0.941 & 0.940 & 0.937 &  & 0.910 & 0.910 & 0.909\tabularnewline
\midrule 
\multirow{3}{*}{GRF} & $\times1$ & 0.948 & 0.937 & 0.926 &  & 0.852 & 0.705 & 0.654\tabularnewline
\cmidrule{2-9} \cmidrule{3-9} \cmidrule{4-9} \cmidrule{5-9} \cmidrule{6-9} \cmidrule{7-9} \cmidrule{8-9} \cmidrule{9-9} 
 & $\times3$ & 0.945 & 0.932 & 0.928 &  & 0.898 & 0.753 & 0.692\tabularnewline
\cmidrule{2-9} \cmidrule{3-9} \cmidrule{4-9} \cmidrule{5-9} \cmidrule{6-9} \cmidrule{7-9} \cmidrule{8-9} \cmidrule{9-9} 
 & $\times10$ & 0.943 & 0.925 & 0.908 &  & 0.918 & 0.747 & 0.680\tabularnewline
\midrule
\multirow{3}{*}{XGBoost} & $\times1$ & 0.950 & 0.936 & 0.929 &  & 0.918 & 0.884 & 0.877\tabularnewline
\cmidrule{2-9} \cmidrule{3-9} \cmidrule{4-9} \cmidrule{5-9} \cmidrule{6-9} \cmidrule{7-9} \cmidrule{8-9} \cmidrule{9-9} 
 & $\times3$ & 0.951 & 0.934 & 0.929 &  & 0.931 & 0.905 & 0.899\tabularnewline
\cmidrule{2-9} \cmidrule{3-9} \cmidrule{4-9} \cmidrule{5-9} \cmidrule{6-9} \cmidrule{7-9} \cmidrule{8-9} \cmidrule{9-9} 
 & $\times10$ & 0.946 & 0.938 & 0.935 &  & 0.937 & 0.903 & 0.898\tabularnewline
\bottomrule
\vspace{0.1cm}
\end{tabular}
\par\end{centering}}
\medskip{}
\caption{Coverage Probability Estimates: Total Consumption}
\label{tab:coverage consumption}
\justifying
{\footnotesize{}Notes: Table \ref{tab:coverage consumption} compares estimates of the coverage probabilities of confidence intervals constructed around each estimator. The long-term outcome of interest is taken to be total consumption. The nominal coverage probability is set to $0.95$ in each case. Standard Wald intervals are constructed around the difference in means estimator \eqref{eq:dm}. Confidence intervals are constructed with \eqref{eq: ci} for the estimators formulated in \cref{def: dml}.}{\footnotesize\par}
\end{table}

\end{spacing}
\end{appendix}

\end{document}